\documentclass[12pt]{article} \usepackage{amsthm}
\usepackage{latexsym} \usepackage{eucal} %%%%%%%%\usepackage{srcltx}
\usepackage{amssymb,amsmath,amsfonts,amssymb}
\usepackage{latexsym} \usepackage{eucal} \usepackage{srcltx}
\usepackage{amssymb,amsmath,amsfonts,amssymb}
\usepackage{amssymb,amsmath,amsfonts,amssymb}
\usepackage{amsmath,amsthm,amscd,amssymb}
\usepackage{latexsym}
\usepackage {verbatim}

\usepackage{graphics}
\usepackage{amsmath}
\usepackage{amsfonts,amssymb,amsthm,eucal}
\usepackage{pstricks,multido,graphicx,psfrag}
-\marginparwidth \oddsidemargin -4mm \evensidemargin
\oddsidemargin
\advance\hoffset by 5mm

\def\@abssec#1{\vspace{.05in}\footnotesize \parindent .2in
{\bf #1. }\ignorespaces}
\newtheorem{theorem}{Theorem}[section]
\newtheorem{thm}[theorem]{Theorem}
\newtheorem{lemma}[theorem]{Lemma}

\newtheorem{corollary}[theorem]{Corollary}
\newtheorem{definition}[theorem]{Definition}
\newtheorem{remark}[theorem]{Remark}

\def\1{{\bf 1}}

\chardef\set=35

\def\C{\mathbb C}

\def\x{{\vec x}}
\def\y{{\vec y}}

\def\u{{\bf u}}
\def\p{{\vec p}}

\def\j{{\bf j}}
\def\wt{\widetilde}

\def\n{{\bf n}}

\def\k{{\vec \varkappa}}

\def\b{{\bf b}}

\def\q{{\bf q}}
\def\l{{\bf l}}
\def\a{{\bf a}}
\def\h{{\vec h}}
\def\s{{\bf s}}
\def\v{{\bf v}}
\def\m{{\bf m}}

\def\R{{\mathbb R}}
\def\Z{{\mathbb Z}}
\def\N{{\mathbb N}}

\def\D{{\bf D}}

\def\DD{{\cal D}}

\def\E{{\cal E}}
\def\RR{{\cal R}}

\def\embb{\supset \kern-11.6pt\lower.335ex\hbox{$\scriptscriptstyle <$} \;}

\def\M{{\cal M}}

\def\W{{\cal W}}
\def\SS{{\cal S}}

\def\B{{\cal B}}

\def\OO{{\cal O}}
\def\MM{{\cal M}}

\def\QQ{{\cal Q}}

\def\g{{\bf g}}

\title{\ Multiscale Analysis in Momentum Space for Quasi-periodic Potential in Dimension Two}
\author{Yu. Karpeshina, R. Shterenberg}

\begin{document}

\maketitle

\begin{abstract} We consider a polyharmonic operator $H=(-\Delta)^l+V(\x)$ in
dimension two with $l\geq 2$, $l$ being an integer, and a
quasi-periodic potential $V(\x)$. We prove that the absolutely
continuous spectrum of $H$ contains a semiaxis and there is a family
of generalized eigenfunctions at every point of this semiaxis with
the following properties. First, the eigenfunctions are close to
plane waves $e^{i\langle \k,\x\rangle }$ at the high energy region.
Second, the isoenergetic curves in the space of momenta $\k$
corresponding to these eigenfunctions have a form of slightly
distorted circles with holes (Cantor type structure). A new method
of multiscale analysis in the momentum space is developed to prove
these results.\end{abstract} \tableofcontents

\newpage
\section{Introduction} \setcounter{equation}{0} We study  an operator
    \begin{equation}
    H=(-\Delta)^l+V(\x) \label{main0}
    \end{equation}
    in two dimensions, where $l$ is an integer, $l\geq 2$, $V(\x)$ is a quasi-periodic
potential being a trigonometric polynomial:
\begin{equation}\label{V}
V=\sum\limits_{\s_1,\s_2\in\Z^2,\,0<|\s_1|+|\s_2|\leq
Q}V_{\s_1,\s_2}e^{2\pi i\langle \s_1+\alpha \s_2,\x\rangle},\ \
1\leq Q<\infty . \end{equation} We assume  that the irrationality
measure $\mu$ of $\alpha$ is finite: $\mu<\infty$, or in other
words, that $\alpha$ is not a Liouville number \footnote{ Note,  that $\mu\geq 2$ for any irrational number $\alpha $.}.

%%%%% Without loss of generality we assume that
%%%%%$\int _{Q_1}V_r(x)dx=0$, $Q_1=[0,1]^2$ being the elementary
%%%%%cell of periods corresponding to $V_r(x)$.

 The one-dimensional situation $d=1$, $l=1$ is thoroughly
   investigated in discrete and continuum settings,
   see e.g. \cite{DiSi}--\cite{FK} and references there.
   It is known that a one-dimensional  quasi-periodic Schr\"{o}dinger operator demonstrates spectral and transport
   properties which
   are \underline{not} close to those of a periodic operator.
    The spectrum of the quasi-periodic
   operator is, as a rule, a Cantor set, while in the periodic case, it has a band structure.
   In the periodic case the spectrum is absolutely continuous, while in the
   quasi-periodic case, it
   can have any nature: absolutely continuous, singular continuous and  pure point.
   The transition between different types  of spectrum can happen even with a small change of
   a coefficient in a quasi-periodic operator \cite{J1}.
    The mechanism of the difference in
   spectral behavior between periodic and quasi-periodic cases can be explained by a
   phenomenon
    which is known as resonance tunneling in quantum mechanics. It is associated with small
    denominators appearing in formal series of perturbation theory. Since the spectrum of
    the one-dimensional Laplacian is thin  (multiplicity 2),  resonance
    tunneling can produce an effect strong enough to destroy
    the spectrum. If a potential is periodic, then  resonance tunneling  produces
    gaps in the spectrum near the points $\lambda _n=(\pi n/a)^2$, $n\in \Z$, $a$ being the
    period of the potential. If the potential is quasi-periodic, then it
    can be thought as a sort of combination of  infinite number of periodic potentials,
     each
    of them producing gaps near its own $\lambda _n$-s. Since the set of all
    $\lambda _n$-s can be
    dense, the number of points surrounded by gaps can be dense too. Thus, the spectrum gets a
    Cantor like structure. The properties of the operator in the high energy region for the continuum case $d=1$ are studied in
   \cite{DiSi}-\cite{MP}, \cite{E}. The KAM method is used to prove  absolute continuity of the
      spectrum and existence of quasiperiodic solutions at high
      energies.
       %%%%%%%%%%%for Schr\"{o}dinger
      %%%%%%%%%%equation with a quasi-periodic potential

    There are  important results on the density of states, spectrum,
     localization concerning the  quasi-periodic operators in $\Z^d$ and, partially, in
     $\R^d$, $d>1$, e.g. \cite{Sh1}-\cite{68a}.
   %%%%%    e.g. [6,9,10,17,21,23,27,35,\\47,48,50,59].
       %%%%%%%%%\cite{PF},\cite{RSS},\cite{S1}, \cite{S2}, \cite{St},  \cite{Sp},  \cite{ChD}.
        However, it is still much less known  about (\ref{main0}) then about its
    one-dimensional analog. The properties of the spectrum in the high energy region, existence of extended states and
    quantum transport are  still the wide open problems
    in the multidimensional case.

 Here we study properties of the spectrum and
eigenfunctions of (\ref{main0}) in the high energy region. We prove
the following results for the case $d=2$, $l\geq 2$.
    \begin{enumerate}
    \item The spectrum of the operator (\ref{main0})
    contains a semiaxis.

    This is a generalization of  a renown Bethe-Sommerfeld conjecture, which states that in the case of a periodic
    potential, $l=1$ and $d\geq 2$, the spectrum of \eqref{main0} contains a semiaxis.
    There is a variety of proofs for the periodic case, the earliest one is \cite{14r}. For a limit-periodic periodic potential, being periodic in one direction,  the conjecture is proved in \cite{SS}. For a general case
    of limit-periodic potential the conjecture is proven in \cite{KL1}--\cite{KL3}. Here we present the first proof of (a generalized) Bethe-Sommerfeld conjecture for a quasi-periodic potential.

    \item There are generalized eigenfunctions $\Psi_{\infty }(\k, \x )$,
    corresponding to the semi-axis, which are close to plane waves:
    for every $\k $ in an extensive subset $\cal{G} _{\infty }$ of
$\R^2$, there is
    a solution $\Psi_{\infty }(\k, \x)$ of the  equation
    $H\Psi _{\infty }=\lambda _{\infty }\Psi _{\infty }$ which can be
described by
    the formula:
    \begin{equation}
    \Psi_{\infty }(\k, \x)
    =e^{i\langle \k, \x \rangle}\left(1+u_{\infty}(\k,
    \x)\right), \label{qplane}
    \end{equation}
    \begin{equation}
    \|u_{\infty}\|_{L_{\infty }(\R^2)}\underset{|\k| \rightarrow
     \infty}{=}O\left(|\k|^{-\gamma _1}\right),\ \ \ \gamma _1>0,
    \label{qplane1}
    \end{equation}
    where $u_{\infty}(\k, \x)$ is a quasi-periodic
    function, namely a point-wise convergent series of exponentials $e^{2\pi i\langle\n+\alpha \m,\x\rangle}$, $\n ,\m \in \Z^2$.
   %%%%%%%%%%% \begin{equation}
   %%%%%%%%%%% u_{\infty}(\k, \x)=\sum_{r=1}^{\infty}  u_r(\k, \vec
   %%%%%%%%%%% x),\label{aplane2}
   %%%%%%%%%%% \end{equation}
   %%%%%%%%%%% $u_r(\k, \x)$ being a linear combination of thee xponent
    The  eigenvalue $\lambda _{\infty }(\k)$, corresponding to
    $\Psi_{\infty }(\k, \x)$, is close to $|\k|^{2l}$:
    \begin{equation}
    \lambda _{\infty }(\k)\underset{|\k| \rightarrow
     \infty}{=}|\k|^{2l}+
    O\left(|\k|^{-\gamma _2}\right),\ \ \ \gamma _2>0. \label{16a}
    \end{equation}
     The ``non-resonant" set $\cal{G} _{\infty }$ of
       vectors $\k$, for which (\ref{qplane}) -- (\ref{16a}) hold, is
       an extensive
       Cantor type set: ${\cal G} _{\infty }=\cap _{n=1}^{\infty }{\cal G}
_n$,
       where $\{{\cal G} _n\}_{n=1}^{\infty}$ is a decreasing sequence of
sets in $\R^2$. Each ${\cal G} _n$ has a finite number of holes  in each
bounded
       region. More and more holes appear when $n$ increases,
       however
       holes added at each step are of smaller and smaller size.
       The set $\cal{G} _{\infty }$ satisfies the estimate:
       \begin{equation}\left|\cal{G} _{\infty }\cap
        \bf B_R\right|\underset{R \rightarrow
     \infty}{=}|{\bf B_R}| \bigl(1+O(R^{-\gamma _3})\bigr),\ \ \ \gamma
_3>0,\label{full}
       \end{equation}
       where $\bf B_R$ is the disk of radius $R$ centered at the
       origin, $|\cdot |$ is the Lebesgue measure in $\R^2$.

       \item The set $\cal{D}_{\infty}(\lambda)$,
defined as a level (isoenergetic) set for $\lambda _{\infty }(
\k)$,
$$ {\cal D} _{\infty}(\lambda)=\left\{ \k \in \cal{G} _{\infty }
:\lambda _{\infty }(\k)=\lambda \right\},$$ is proven
to be a slightly distorted circle with infinite number of holes. It
can be described by  the formula:
 \begin{equation} {\cal
D}_{\infty}(\lambda)=\left\{\k:\k=\varkappa_{\infty}(\lambda, \vec
\nu){\vec \nu},
    \ {\vec \nu} \in {\cal B}_{\infty}(\lambda)\right
    \}, \label{Dinfty}
    \end{equation}
where ${\cal B}_{\infty }(\lambda )$ is a subset of the unit circle
$S_1$. The set ${\cal B}_{\infty }(\lambda )$ can be interpreted as
the set of possible  directions of propagation for  almost plane
waves (\ref{qplane}). The set ${\cal B}_{\infty }(\lambda )$ has a
Cantor type structure and an asymptotically full measure on $S_1$ as
$\lambda \to \infty $:
\begin{equation}
L\bigl({\cal B}_{\infty }(\lambda )\bigr)\underset{\lambda
\rightarrow
     \infty}{=}2\pi +O\left(\lambda^{-\gamma _4/2l}\right),\ \ \
     \gamma_4>0,
\label{B}
\end{equation}
here and below $L(\cdot)$ is a length of a curve. The value
$\varkappa_{\infty }(\lambda ,{\vec \nu} )$ in (\ref{Dinfty}) is the
``radius" of ${\cal D}_{\infty}(\lambda)$ in a direction ${\vec \nu}
$. The function $\varkappa_{\infty }(\lambda ,{\vec \nu}
)-\lambda^{1/2l}$ describes the deviation of ${\cal
D}_{\infty}(\lambda)$ from the perfect circle of the radius
$\lambda^{1/2l}$. It is proven that the deviation is asymptotically
small:
\begin{equation}
\varkappa_{\infty }(\lambda ,{\vec \nu} )\underset{\lambda
\rightarrow
     \infty}{=}\lambda^{1/2l}+O\left(\lambda^{-\gamma _5 }\right),
\ \ \ \gamma _5>0. \label{h}
\end{equation}

\item The branch of the spectrum of the operator (\ref{main0}) corresponding to the generalized eigenfunctions $\Psi_{\infty }(\k, \x
)$ is absolutely continuous.

\end{enumerate}

%We will prove absolute continuity of the branch of the spectrum  corresponding to $\Psi_{\infty }(\k,
%\x)$ (the semiaxis) in a forthcoming paper.

To prove the results listed above we  suggest a method which can be
described as {\em multiscale analysis in the space of momenta}. This
is a development of the method, which is used in \cite{KL1}--\cite{KL3} for the
case of limit-periodic potentials. The essential difference is that
in \cite{KL1}--\cite{KL3} we constructed a modification of KAM method, where the
space variable  $\x$ still plays some role (e.g. in the uniform  in
$\x$ approximation of a limit-periodic potential by periodic ones),
while in the present situation all considerations are happening in
the space of the dual variable $\k$. The KAM method in \cite{KL1}--\cite{KL3}
was motivated by \cite{2}--\cite{3}, where the method is used for
periodic problems. Multiscale analisys which we apply here is deeply
analogous to the original multiscale method developed in \cite{FrSp}
(see also \cite{BG}, \cite{74}) for the proof of localization. The
essential difference is that
 in \cite{FrSp}, \cite{BG}, \cite{74} the multiscale procedure is constructed with respect to space variable $\x$ to prove localization, while we construct a multiscale procedure in the space of momenta $\k$ to prove delocalization.
%%%%%%%% Note that the KAM method
 %%%%%%%%     was applied in 1975
%%%%%%%%       by E.I.Dinaburg
%%%%%%%%       and Ya.Sinai [20]
       %%%%%%%%%\cite{DiSi}
%%%%%%%%      to
%%%%%%%%      prove  existence of quasiperiodic solutions for one-dimensional Schr\"{o}dinger
%%%%%%%%      equation with a quasi-periodic potential at the high energy region. However, their method has not been generalized
%%%%%%%%       for the multidimensional case.

        Here is a brief description of the iteration procedure which leads to the results described above.
        Indeed,  let $\k \in \R^2$. We consider a set of finite linear combinations of plane waves
        $e^{i\langle\k+2\pi(\n+\alpha \m),\x\rangle}$,
        $ \n,\m \in \Z^2$.
         %%%%%%by ${\cal S}$.
         The set is invariant under action of the differential expression (\ref{main0}). Let $H(\k)$ be a
         matrix describing action of  (\ref{main0}) in the linear set of the exponents.
          Obviously,
          $$H(\k)=H_0(\k)+V,\ \  H_0(\k)_{(\n,\m), (\n',\m ')}=|\k+2\pi(\n+\alpha \m)|_{\R^2}^2\delta _{(\n,\n')}
          \delta_{ (\m,\m')},$$ $$V_{(\n,\m), (\n',\m')}=V_{\n-\n', \m-\m'}.$$
          Next, we consider  an expanding sequence of finite sets $\Omega_n$ in the space $\Z^2\times \Z^2$ of indices
          $(\n,\m)$: $\Omega_n\subset \Omega_{n+1}$, $\lim _{n\to \infty } \Omega_n= \Z^2\times \Z^2$. Let $P_n$ be
          the characteristic projection of  set $\Omega_n$ in the space $\ell^{2}(\Z^2\times \Z^2)$. We consider a sequence
          of finite matrices $H^{(n)}(\k)=P_nH(\k)P_n$. Each matrix corresponds to a finite dimensional operator in
          $\ell^{2}(\Z^2\times \Z^2)$, given that the operator acts as zero on $(I-P_n){\ell}^2$.
          For each $n$ we construct a ``non-resonant" set ${\cal G}_n$ in the space $\R^2$ of momenta $\k$, such that:
          if $\k \in {\cal G}_n$, then $H ^{(n)}(\k)=P_nH(\k)P_n$ has an eigenvalue $\lambda _n(\k)$ and
          its spectral projector $\E _n(\k)$ which can be described by perturbation formulas with respect to the previous
          operator $H ^{(n-1)}(\k)$. If $\k \in \cap _{n=1}^{\infty} {\cal G} _n$ then $\lambda _n(\k)$ and $\E _n(\k)$
          have  limits. The linear combinations of the exponentials, corresponding to the projectors $\E _n(\k)$, have a
          point-wise limit in $\x$, the limit being a generalized eigenfunction of \eqref{main0}.
         The generalized eigenfunction is close to the plane wave $e^{i\langle\k,\x\rangle}$ in the high energy region.

          Each matrix $H ^{(n)}$ is considered as a perturbation of a matrix $\hat H ^{(n)}$, the latter has a block
          structure, i.e., consists of a variety of blocks $H^{(s)}(\k +2\pi(\n +\alpha \m))$, $s=1,...,n-1$, and, naturally,
          some diagonal terms. Blocks with different indices $(s)$ have    sizes of different orders of magnitude
          (the size increasing
          with $s$). Thus we have a multiscale structure  in the definition of $\hat H ^{(n)}$.
          We use $\hat H ^{(n)}(\k)$ as a starting operator to construct perturbation series for $H^{(n)}(\k)$.
          At a step $n$ we apply our knowledge of spectral properties of $H^{(s)}(\k +2\pi(\n'+\alpha \m'))$, $s=1,...,n-1$,
          $\n', \m' \in  \Z^2$, obtained in the previous steps, to
        describe spectral properties of  $H^{(n)}(\k +2\pi(\n +\alpha \m))$, $\n,\m\in \Z^2$ and to construct ${\cal G} _n$.

At step one we use a regular perturbation theory and elementary geometric considerations to prove the following results. There is a set ${\cal G}_1\subset \R^2$ such that: if $\k\in {\cal G}_1$, then the operator $H^{(1)}(\k)$ has a single eigenvalue close to the unperturbed one:
\begin{equation} \lambda ^{(1)}(\k)\underset{|\k| \rightarrow
     \infty}{=}|\k|^{2l}+
    O\left(|\k|^{-\gamma _2}\right),\ \ \ \gamma _2>0.\label{lambda-1} \end{equation}
     A normalized eigenvector $\u^{(1)}$ is also close to the unperturbed one: $\u^{(1)}=\u^{(0)}+\tilde \u^{(1)}$, where
     $(\u^{(0)})_{(\n,\m)}=\delta _{\n,{\bf 0}}\delta_{ \m,{\bf 0}}$ and the $l^{1}$-norm of $\tilde \u^{(1)}$ is small:
     $\|\tilde \u^{(1)}\|_{l^{1}}<|\k|^{-\gamma _1}$, $\gamma _1>0$. It follows that:
 \begin{equation}
    \Psi_1 (\k, \x)
    =e^{i\langle \k, \x \rangle}+\tilde u_{1}(\k,
    \x),\ \ \
\|\tilde u_{1}\|_{L_{\infty }(\R^2)}\underset{|\k| \rightarrow
     \infty}{=}O(|\k|^{-\gamma _1}),\ \
\ \gamma _1>0, \label{na}
\end{equation}
    where  $\Psi_1 (\k, \x)$, $\tilde u_{1}(\k, \x)$ are the linear combinations of the exponentials corresponding to
    vectors $\u^{(1)}$ and $\tilde \u^{(1)}$, respectively. It is shown that function $\Psi_1 (\k, \x)$
    satisfies the equation for eigenfunctions with a good accuracy:
    \begin{equation} \label{eqforeigenfunctions-1}-\Delta \Psi_1+V\Psi_1=|\k|^{2l}\Psi_1+f_1,\ \ \|f_1\|_{L_{\infty }(\R^2)}\underset{|\k| \rightarrow
     \infty}{=}O(|\k|^{-\gamma _6}),\ \
\ \gamma _6>0.
\end{equation}
Relation \eqref{lambda-1} is differentiable:
\begin{equation} \nabla \lambda ^{(1)}(\k)\underset{|\k| \rightarrow
     \infty}{=}2l|\k|^{2l-2}\k+
    O\left( |\k|^{-\gamma _7}\right),\ \ \ \gamma _7>0.\label{dlambda-1} \end{equation}
     Next, we construct a sequence ${\cal G}_n$, $n\geq 2$, such  for any $\k \in {\cal G}_n$ the operator $H^{(n)}(\k)$ has a single eigenvalue $\lambda ^{(n)}(\k)$ in a super exponentially small neighborhood of $\lambda ^{(n-1)}(\k)$:
\begin{equation} \lambda ^{(n)}(\k)\underset{|\k| \rightarrow
     \infty}{=}\lambda ^{(n-1)}(\k)
   + O\left(|\k|^{-|\k|^{\gamma _8n}}\right),\ \ \ \gamma _8>0.\label{lambda-n} \end{equation}
   Similar estimates hold for the eigenvectors and the  corresponding functions $\Psi_n(\k,\x)$:
   \begin{equation}
   \Psi _n(\k,\x)=\Psi _{n-1}(\k,\x)+\tilde u _n(\k,\x),\ \ \ \|\tilde u_{n}\|_{L_{\infty }(\R^2)}\underset{k \rightarrow
     \infty}{=} O\left(|\k|^{-|\k|^{\gamma _9n}}\right),\ \ \ \gamma _9>0.\ \ \label{na-n}
\end{equation}
\begin{equation} \label{eqforeigenfunctions-1*}-\Delta \Psi_n+V\Psi_n=\lambda ^{(n)}(\k)\Psi_n+f_n,\ \ \|f_n\|_{L_{\infty }(\R^2)}\underset{|\k| \rightarrow
     \infty}{=}O\left(|\k|^{-|\k|^{\gamma _{10}n}}\right),\ \ \
\ \gamma _{10}>0.
\end{equation}
Formula \eqref{lambda-n} is differentiable with respect to $\k$:
\begin{equation} \nabla \lambda ^{(n)}(\k)\underset{|\k| \rightarrow
     \infty}{=}\nabla \lambda ^{(n-1)}(\k)+O\left(|\k|^{-|\k|^{\gamma _8n}}\right),\ \ \
\ \gamma _8>0.\label{dlambda-n} \end{equation}
In fact, for large $n$ estimates \eqref{lambda-n} -- \eqref{dlambda-n} are even stronger.
\begin{figure}
\begin{minipage}[t]{8cm}
\centering
\includegraphics
[totalheight=.25\textheight]{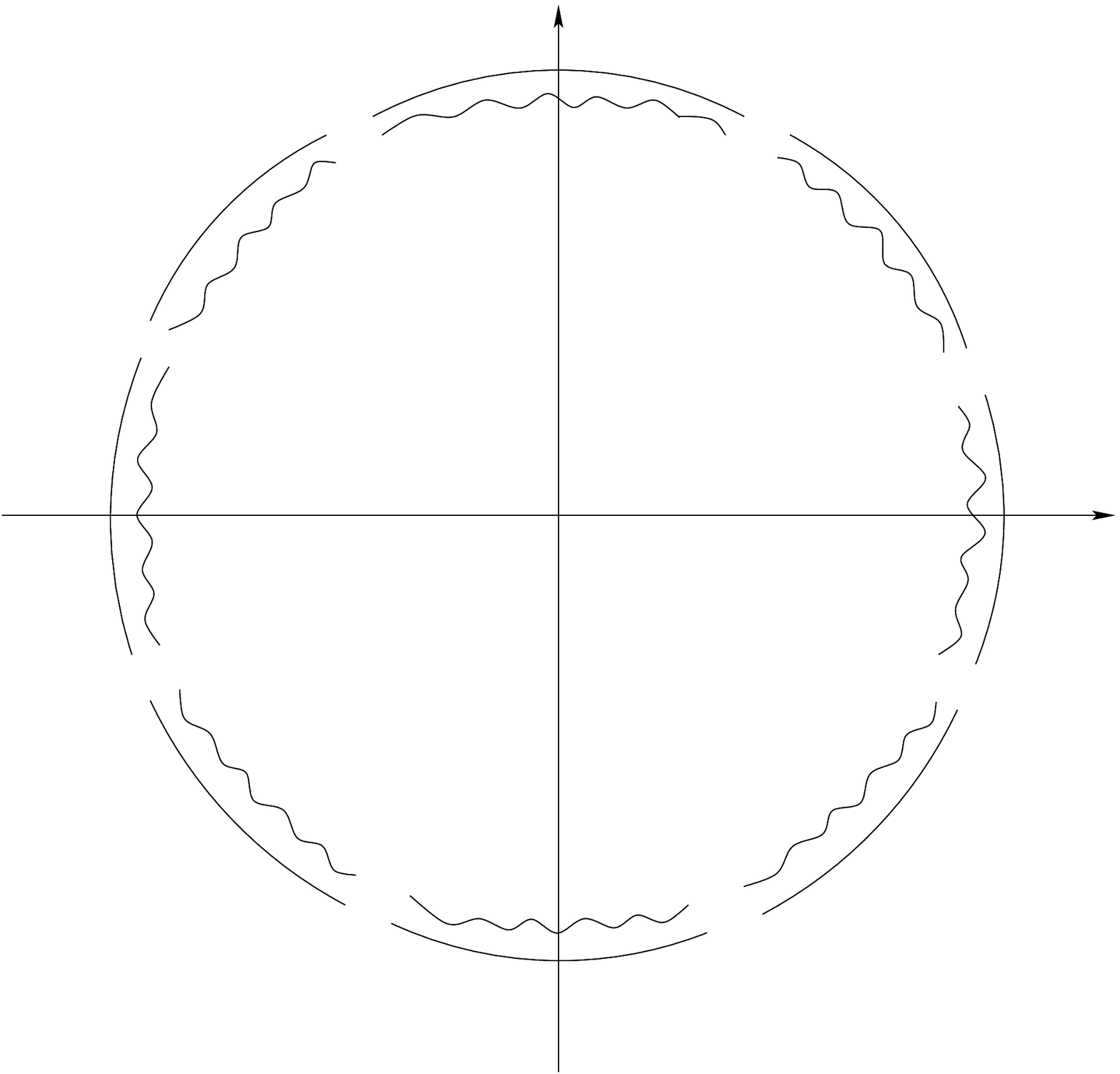} \caption{Isoenergetic curve
${\cal D}_1(\lambda)$}\label{F:1}
\end{minipage}
\hfill
\begin{minipage}[t]{8cm}
\centering
\includegraphics[totalheight=.25\textheight]{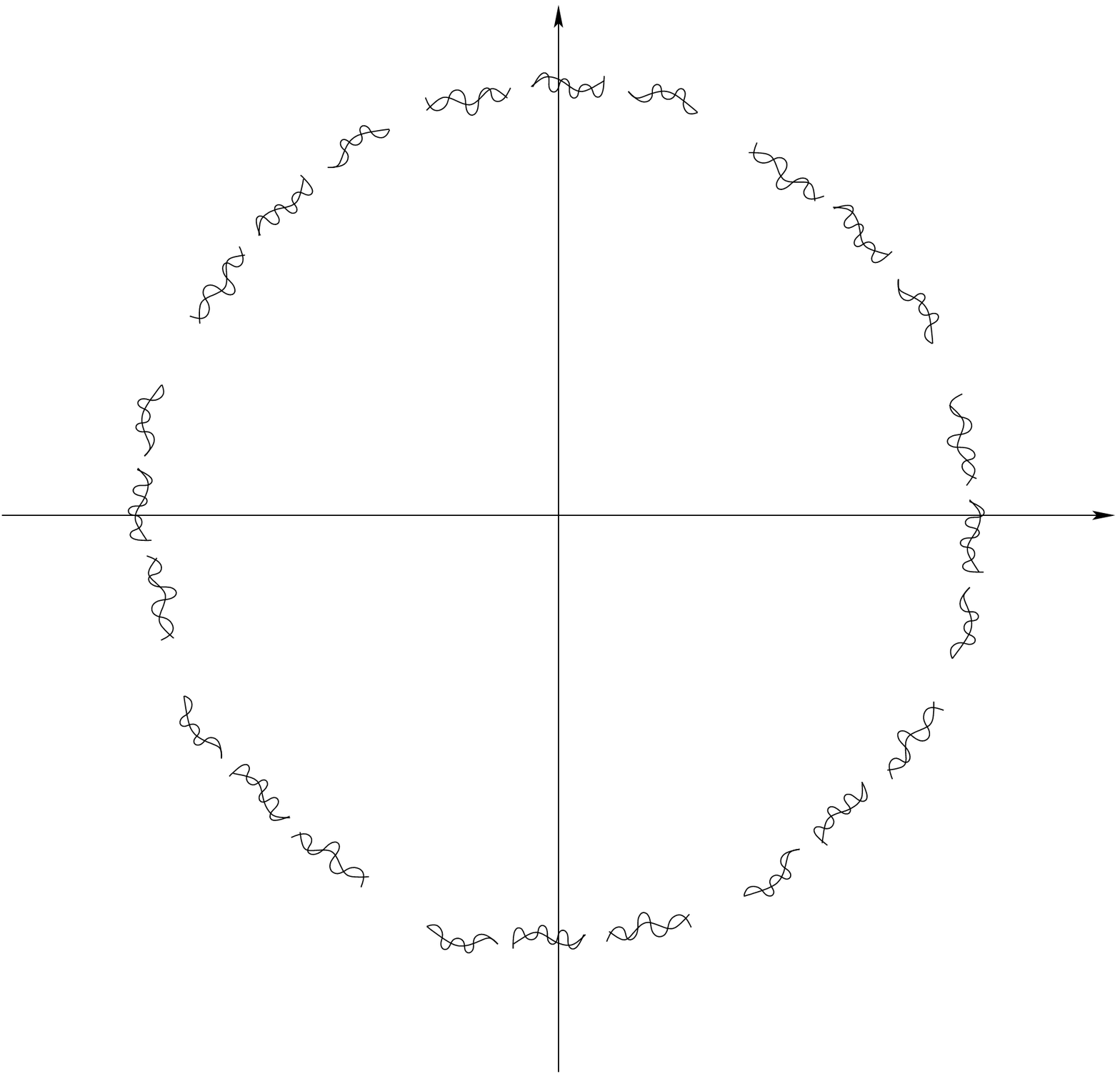}
\caption{Isoenergetic curve ${\cal D}_2(\lambda)$}\label{F:2}
\end{minipage}\hfill
\end{figure}
The non-resonant set ${\cal G} _{n}$
        is proven to be
       extensive in $\R^2$:
       \begin{equation}
       \left|{\cal G} _{n}\cap
       \bf B_R\right|\underset{R \rightarrow
     \infty}{=}|{\bf B_R}|\bigl(1+O(R^{-\gamma _3})\bigr). \label{16b}
       \end{equation}
       Estimates (\ref{lambda-n}) -- (\ref{16b}) are uniform in $n$.

The set ${\cal D}_{n}(\lambda)$ is defined as the level
(isoenergetic) set for the non-resonant eigenvalue $\lambda ^{(n)}(\k)$:
$$ {\cal D} _{n}(\lambda)=\left\{ \k \in {\cal G} _n:\lambda ^{(n)}(\k)=\lambda \right\}.$$
This set is proven to be a slightly distorted circle with a
finite number of holes (see Fig. \ref{F:1}, \ref{F:2}). The set
${\cal D} _{n}(\lambda)$ can be described by the formula:
\begin{equation}
{\cal D}_{n}(\lambda)=\left\{\k:\k=
    \varkappa^{(n)}(\lambda, {\vec \nu})\vec \nu ,
    \ \vec \nu  \in {\cal B}_{n}(\lambda)\right\}, \label{Dn}
    \end{equation}
where ${\cal B}_{n}(\lambda )$ is a subset  of the unit circle
$S_1$. The set ${\cal B}_{n}(\lambda )$ can be interpreted as the
set of possible directions of propagation for  almost plane waves $\Psi _n(\k,\x)$, see
(\ref{na}), (\ref{na-n}). It has an asymptotically full measure on $S_1$ as
$\lambda \to \infty $:
\begin{equation}
L\bigl({\cal B}_{n}(\lambda )\bigr)\underset{\lambda \to \infty
}{=}2\pi +O\left(\lambda^{-\gamma _4/2l}\right). \label{Bn}
\end{equation}
Each set ${\cal B}_{n}(\lambda)$ has only a finite number of holes,
however their number is growing with $n$. More and more holes of a
smaller and smaller size are added at each step. The value
$\varkappa^{(n)}(\lambda ,{\vec \nu} )-\lambda^{1/2l}$ gives the
deviation of ${\cal D}_{n}(\lambda)$ from the perfect circle of the
radius $\lambda^{1/2l}$  in the direction ${\vec \nu} $. It is
proven that the deviation is asymptotically small:
\begin{equation}
\varkappa^{(n)}(\lambda ,{\vec \nu})
=\lambda^{1/2l}+O\left(\lambda^{- \gamma _5}\right),\ \ \ \
\frac{\partial \varkappa^{(n)}(\lambda ,{\vec \nu})}{\partial
\varphi }=O\left(\lambda^{- \gamma _{11} }\right),\ \
\gamma_5,\gamma _{11}>0,\label{hn}
\end{equation}
$\varphi $ being an angle variable, ${\vec \nu} =(\cos \varphi ,\sin
\varphi )$.  Estimates (\ref{Bn}), (\ref{hn}) are uniform in $n$.

On each step more and more points are excluded from the
non-resonant sets ${\cal G} _n$, thus $\{ {\cal G} _n \}_{n=1}^{\infty }$ is a
decreasing sequence of sets. The set ${\cal G} _\infty $ is defined as the
limit set: ${\cal G} _\infty=\cap _{n=1}^{\infty }{\cal G} _n $. It
has an infinite number of holes, but nevertheless satisfies the
relation (\ref{full}). For every $
\k \in {\cal G} _\infty $ and every $n$, there is a generalized
eigenfunction of $H^{(n)}$ of the type  (\ref{na}), (\ref{na-n}). It is
proven that  the sequence of
$\Psi _n(\k, \x)$ has a limit in $L_{\infty }(\R^2)$ when $
\k \in {\cal G} _\infty $.
The function $\Psi _{\infty }(\k, \x)
=\lim _{n\to \infty }\Psi _n(\k, \x)$ is a generalized
eigenfunction of $H$. It can be written in the form
(\ref{qplane}) -- (\ref{qplane1}).
Naturally, the corresponding eigenvalue $\lambda _{\infty }(\k) $ is
the limit of $\lambda ^{(n)}(\k )$ as $n \to \infty $.

It is shown that $\{{\cal B}_n(\lambda)\}_{n=1}^{\infty }$  is a
decreasing sequence of sets,  on each step more and more directions
being excluded. We consider the limit ${\cal B}_{\infty}(\lambda)$
of ${\cal B}_n(\lambda)$:
    \begin{equation}{\cal B}_{\infty}(\lambda)=\bigcap_{n=1}^{\infty} {\cal
    B}_n(\lambda).\label{Dec8a}
    \end{equation}
    This set has a Cantor type structure on the unit circle.
    It is proven that ${\cal B}_{\infty}(\lambda)$ has an asymptotically
    full measure on the unit circle (see (\ref{B})).
    We prove
    that the sequence $\varkappa^{(n)}(\lambda ,{\vec \nu} )$, $n=1,2,... $,
describing the
     isoenergetic curves ${\cal D}_n(\lambda)$, quickly converges as $n\to
\infty$. We show that ${\cal D}_{\infty}(\lambda)$ can be described
as the limit of  ${\cal D}_n(\lambda)$ in the sense (\ref{Dinfty}),
where $\varkappa_{\infty}(\lambda, \vec \nu )=\lim _{n \to \infty}
\varkappa^{(n)}(\lambda, \vec \nu )$ for every $\vec \nu  \in {\cal
B}_{\infty}(\lambda)$. It is shown that the derivatives of the
functions $\varkappa^{(n)}(\lambda, \vec \nu )$ (with respect to the
angle variable on the unit circle) have a limit as $n\to \infty $
for every $\vec \nu  \in {\cal B}_{\infty}(\lambda)$. We denote this
limit by $\frac{\partial \varkappa_{\infty}(\lambda ,\vec
\nu)}{\partial \varphi }$. Using (\ref{hn}), we  prove that
    \begin{equation}\frac{\partial \varkappa_{\infty}(\lambda ,\vec \nu)}{\partial
\varphi }=O\left(\lambda^{- \gamma _{11} }\right).\label{Dec9a}
\end{equation} Thus, the limit curve ${\cal D}_{\infty}(\lambda)$ has a
tangent vector in spite of its Cantor type structure, the tangent
vector being the limit of corresponding tangent vectors for ${\cal
D}_n(\lambda)$ as $n\to \infty $.  The curve  ${\cal
D}_{\infty}(\lambda)$ looks as
  a slightly distorted circle with
infinite number of holes for every sufficiently large $\lambda $, $\lambda >\lambda _*(V)$. It immediately follows that $[\lambda _*, \infty )$ is in the spectrum of $H$ (Bethe-Sommerfeld conjecture).

The main technical difficulty to overcome is the construction of
   non-resonant sets
${\cal B} _n(\lambda)$ for every fixed sufficiently large $\lambda $,
$\lambda >
\lambda_0 (V)$,
where $\lambda_0(V)$ is the same for all $n$. The set
${\cal B} _n(\lambda)$ is obtained  by deleting a ``resonant" part
from
${\cal B}_{n-1}(\lambda)$.
Definition of  ${\cal B} _{n-1}(\lambda)\setminus {\cal B}_{n}(\lambda)$
includes
eigenvalues of $H^{(n-1)}(\k)$. To describe  ${\cal B} _{n-1}(\lambda)\setminus
{\cal B}_{n}(\lambda)$ one
has to consider
  not only non-resonant
eigenvalues of the type (\ref{lambda-1}),  (\ref{lambda-n}), but also
resonant
eigenvalues, for which no suitable  formulas are known. Absence of
formulas causes difficulties in estimating the size of ${\cal B}
_{n-1}(\lambda)\setminus {\cal B}_n(\lambda)$.
       To treat this problem we start with introducing
        an angle variable $\varphi \in [0,2\pi
       )$,  $\vec \nu  = (\cos \varphi ,
       \sin \varphi )\in S_1$ and consider sets ${\cal B}_n(\lambda)$ in
terms of this
       variable.
        Next, we show that
the resonant set  ${\cal B} _{n-1}(\lambda)\setminus
{\cal B}_{n}(\lambda)$ can be described as the set of zeros of  functions of the type
        $$ \det \Bigl(H^{(s)}\bigl(\k _{n-1}(\varphi )+2\pi(\n+\alpha
        \m)
       \bigr)-\lambda-\varepsilon \Bigr), \ \ \ s=1,...,n-1,\ \ \ (\n,\m)\in \Omega_n\setminus{(\bf 0,\bf 0)},$$
       where $\k _{n-1}(\varphi )$ is a vector-function
describing ${\cal D}_{n-1}
       (\lambda)$: $\k _{n-1}(\varphi )=\varkappa_{n-1}(\lambda
,{\vec \nu} ){\vec \nu} $. To obtain ${\cal B} _{n-1}(\lambda)\setminus
       {\cal B}_{n}(\lambda)$ we take all values of $\varepsilon $ in a small
interval
       and $(\n,\m)$ in some subset of $\Omega_n$.
        Further, we extend our
       considerations to
        a complex neighborhood $\varPhi _0$ of $[0,2\pi
       )$. We show that the    determinants are analytic functions of
        $\varphi $ and, by this,
         reduce the
        problem of estimating the size of the resonant set
         to
        a problem in complex analysis. We use theorems for analytic
functions to count
          zeros of the determinants and to investigate how far  the zeros move
when
         $\varepsilon $ changes. It enables us to estimate the size
         of the zero set of the determinants, and, hence, the size of
         the non-resonant set $\varPhi _n\subset \varPhi _0$,  which is
defined as a
         non-zero set for the determinants.
          Proving that the non-resonant set $\varPhi _n$
         is sufficiently large, we
          obtain estimates (\ref{16b}) for ${\cal G} _n$ and (\ref{Bn}) for
         ${\cal B}_n$, the set  ${\cal B}_n$ corresponding to the real part of
         $\varPhi _n$.
         \begin{figure}
\begin{minipage}[t]{8cm}
\centering \psfrag{Phi_2}{$\Phi_2$}
\includegraphics
[totalheight=.25\textheight]{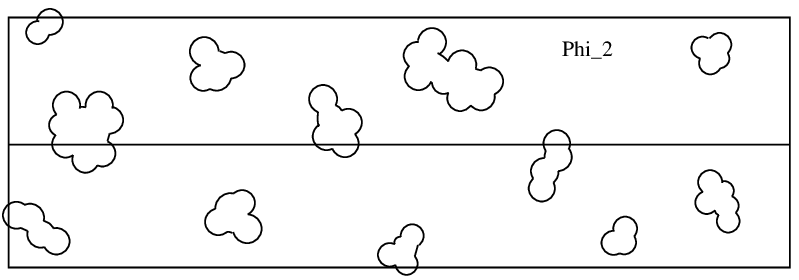} \caption{Set $\varPhi
_2$}\label{F:3}
\end{minipage}
\hfill
\end{figure}
\begin{figure}
\begin{minipage}[t]{8cm}
\centering \psfrag{Phi_3}{$\Phi_3$}
\includegraphics[totalheight=.25\textheight]{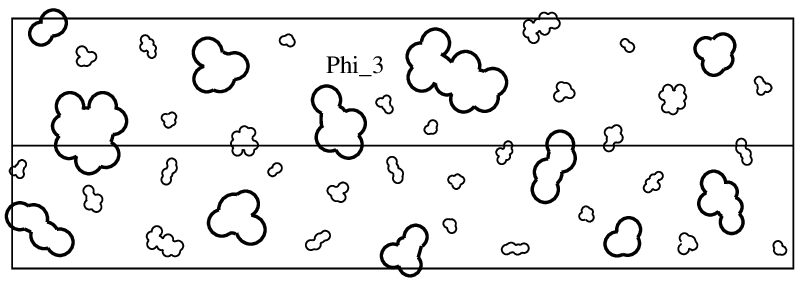}
\caption{Set $\varPhi _3$}\label{F:4}
\end{minipage}\hfill
\end{figure}

          To obtain $\varPhi _n$ we delete  from $\varPhi _0$ more and more
discs (holes) of smaller and
         smaller radii at each step. Thus, the non-resonant set $\varPhi
         _n\subset \varPhi _0$ has a structure of Swiss
         Cheese (Fig. \ref{F:3}, \ref{F:4}). Deleting  a resonance set from $\varPhi _0$ at each
         step of the recurrent procedure  we call a ``Swiss
         Cheese Method".  The essential
         difference of our method from constructions of non-resonant sets  in similar
         situations before (see e.g. \cite{2}--\cite{3}, \cite{B3}) is that
we construct a
         non-resonant set not only in the whole space of a parameter
         ($\k\in \R^2$ here), but also on isoenergetic curves
         ${\cal D}_n(\lambda )$ in
         the space of the parameter, when $\lambda $ is sufficiently large.
Estimates for the
         size of non-resonant sets on a curve require more subtle
         technical considerations than those sufficient for
         description of a non-resonant set in the whole space of
         the parameter. But as a reward, such estimates enable us to
         show that every isoenergetic set for $\lambda>\lambda_0$ is not empty and
         thus, to prove Bethe-Sommerfeld conjecture.

  Note that generalization of  the results from the case $l>1$, $l$ being an integer, to the case of rational $l$ satisfying the
         same inequality is relatively simple; it requires just slightly more careful technical considerations. The restriction $l>1
         $  is also technical, though it is more difficult to lift. The condition $l>1$ is needed only for
         the second step of the recurrent procedure. The authors plan to consider
         the case $l=1$ in a forthcoming paper.  The
         requirement $\mu < \infty $ is  essential, since we use it to estimate the minimal values of $|\n+\alpha \m|$ when
         $(\n, \m)\in \Omega_n\setminus{(\bf 0,\bf 0)}$. Such estimates are necessary for controlling small denominators in the perturbation series at each step.

         The plan of the paper is the following. Preliminary considerations are in Section 2. Sections 3 -- 7
         describe steps of the recurrent procedure. Note, that Steps I,II are designed to start the procedure. Step III is already typical,
         however uses some ``non-typical" estimates from
         Steps I,II. Step IV is completely typical: all other steps of the recurrent procedure differ from Step IV only by the change of indices. The proofs of
         convergence of the iteration procedure and of the results 1 -- 3, listed at the beginning of the introduction, are in Section
         8. The result 4 about absolutely continuous spectrum is
         proven in Section 9. Section 10 (Appendices) contains technical lemmas.

\vspace{5mm} \noindent {\bf Acknowledgement} The authors are very
grateful to Prof. Parnovski for useful
discussions and to Prof. Young-Ran Lee for allowing us to use figures 1-4 from [KL].

\section{Preliminary Remarks}

We consider two-dimensional quasi-periodic polyharmonic operator
\begin{equation}\label{main} (-\Delta)^l +V(\x),\ \
l>1 \end{equation} which is perturbation of the free operator
$H_0:=(-\Delta )^l$. Here $V$ is the potential of the form
\eqref{V}. Without  loss of generality we assume $0<\alpha <1$. We assume also that irrationality measure $\mu$ of
$\alpha$ is finite: $\mu<\infty$, or in other words, that $\alpha$
is not a Liouville number. Note, that for irrational number $\alpha$
we always have $\mu\geq 2$. It follows from the definition of the
irrationality measure that 1) {\it For any $\epsilon>0$ there exists
a constant $C_{\varepsilon} $ such that for any irreducible rational
number $\frac{\tilde{M}}{\tilde {N}}$ we have}
\begin{equation}\label{geq} \left|\alpha-\frac{\tilde{M}}{\tilde
{N}}\right|\geq\frac{C_{\varepsilon}}{\tilde {N}^{\mu+\epsilon}}.
\end{equation}

2) {\it For any $\epsilon>0$ there exists a sequence $\frac{M}{N}$
of irreducible rational numbers such that}
\begin{equation}\label{leq}
\left|\alpha-\frac{M}{N}\right|\leq\frac{1}{N^{\mu-\epsilon}}.
\end{equation}

For every pair of integer vectors $\s_1,\,\s_2\in\Z^2$ we consider
$\p_\s:=2\pi(\s_1+\alpha\s_2)$. We introduce the norm
$$ |\|\p_{\s}\||:=|\s_1|+|\s_2|.
$$
We will also use the notation $p_{\s}:=|\p_{\s}|$ and
$\p_{\s}=p_{\s}(\cos\varphi_{\s},\sin\varphi_{\s})$.
\begin{lemma}\label{psnorms}
For every $\p_\s\not=0$ we have
\begin{equation}\label{above}
p_\s\leq2\pi|\|\p_{\s}\||,
\end{equation}
\begin{equation}\label{below} p_{\s}\geq
2\pi C_\varepsilon\,\||\p_\s\||^{-(\mu-1+\epsilon)}. \end{equation}
\end{lemma}
\begin{proof} The estimate \eqref{above} is obvious. To prove \eqref{below} we notice that
if $\s_2=0$ then $p_{\s}=2\pi|\s_1|\geq2\pi$. Let
$\s_1=(s_{11},s_{12})$, $\s_2=(s_{21},s_{22})$. If, for example,
$s_{21}\not=0$ then from \eqref{geq} and definition of
$|\|\p_{\s}\||$ we obtain
\begin{equation} \label{12}
\begin{split}&
p_{\s}\geq2\pi\left|s_{11}+\alpha s_{21}\right|=
2\pi|s_{21}||\alpha+\frac{s_{11}}{s_{21}}|\geq\cr & 2\pi
C_\varepsilon|s_{21}|^{-\mu-\epsilon+1}\geq 2\pi
C_\varepsilon\,\||\p_\s\||^{-(\mu-1+\epsilon)}.
\end{split}
\end{equation}
\end{proof}

We introduce vector $\k(\varphi
):=(\varkappa_1,\varkappa_2)=\varkappa\vec
\nu:=\varkappa(\cos\varphi,\sin\varphi)$. Similar agreement will be
used for other vectors. Let $H(\k)=H(\varkappa,\varphi)$ be the
"fiber" operator acting in $L_2(\Z^4)$ with its matrix elements
given by $$
(H(\k))_{\s,\s+\q}=|\k+\p_{\s}|^{2l}\delta_{\s,\s+\q}+V_{\p_\q}. $$
Here $V_{\p_\q}:=V_{\q_1,\q_2}$.  (see \eqref{V})
\begin{equation}\label{V_q=0} V_{\p_\q}=0,\ \ \mbox{when  }
\||\p_\q\||>Q,\ \ \ (Q<\infty) .\end{equation} To simplify the
notation in what follows we will write $V_\q$ instead of $V_{\p_\q}$
when it does not lead to confusion.

\section{Step I}

\subsection{Operator $H^{(1)}$}
Let $\delta$ be some small parameter, $0<\delta <(100\mu )^{-1}$. We
put
$$ \Omega(\delta):=\{\m\in\Z^2:\ \ |\|\p_{\m}|\|\leq k^\delta\},\ \ \ \tilde\Omega(\delta):=\{\m\in\Z^2:\ \ |\|\p_{\m}|\|\leq 4k^\delta\}. $$

By $P(\delta)$ we denote orthogonal (diagonal) projection in $l^2(\Z^2)$ on the set
of elements supported in $\Omega(\delta)$. We call it the characteristic projector of $ \Omega(\delta)$. The dimension of the projector is equal to the number of elements in $ \Omega(\delta)$ and, obviously, does not exceed
$(8k^\delta )^4$. We have $$
\big(P(\delta)H_0(\k)P(\delta)\big)_{\m,\n}=|\k+\p_{\m}|^{2l}\delta_{\m,\n}\,\chi_{\Omega(\delta)}(\m),
$$ where as usual $\chi_{\Omega(\delta)}(\m)$ is the characteristic
function of the set $\Omega(\delta)$. We are going to consider
 $H^{(1)}(\k)=P(\delta)H(\k)P(\delta)$ as a perturbation of the operator
$P(\delta)H_0(\k)P(\delta)$.

\subsection{Perturbation Formulas}

Now we construct a ``non-resonant" set of $\varphi $, for which the operator $H^{(1)}(\k (\varphi ))$ can be constructively considered as a perturbation of $H^{(1)}_0(\k (\varphi ))$ corresponding to $V=0$.
In what follows $\tau $ is an auxiliary
parameter $\frac{1}{32}\leq \tau \leq 32$.

\begin{lemma}[Geometric] \label{L:G1} For every $k>800$ there is a subset $\omega^{(1)} (k,\delta
,\tau )$ of the interval $[0,2\pi )$ such that: \begin{enumerate}
\item For every $\varphi \in \omega^{(1)} (k,\delta ,\tau )$ and $\m\in
\tilde\Omega (\delta )\setminus \{0\}$, the following inequality
holds:
\begin{equation} \left||\vec k(\varphi )+\p_{\m}|^{2l}-k^{2l}\right|>l\tau k^{2l-1-40\mu
\delta },\ \ \ \vec k:=k(\cos \varphi, \sin \varphi). \label{G1-1}\end{equation} \item For every $\varphi $ in
the real $\frac{\tau }{16}k^{-(40\mu +1)\delta }$-neighborhood of
$\omega^{(1)} (k,\delta ,\tau )$ and $\varkappa\in \R:
|\varkappa-k|<\frac{\tau }{16}k^{-40\mu \delta }$, a slightly weaker
inequality holds for $\k(\varphi )=\varkappa(\cos \varphi ,\sin
\varphi )$ and $\m \in \tilde\Omega (\delta )\setminus \{0\}$:
\begin{equation}
\left||\k(\varphi )+\p_{\m}|^{2l}-{k}^{2l}\right|>\frac{\tau l
}{2}k^{2l-1-40\mu \delta }. \label{G1-2}\end{equation}
\item The set $\omega^{(1)} (k,\delta ,\tau )$ has an asymptotically full
measure in $[0,2\pi )$ as $k\to \infty $. Namely,
\begin{equation}
|\omega^{(1)} (k,\delta ,\tau )|=2\pi +O(k^{-37\mu \delta}),\ \
{k\to \infty }. \label{G1-3}
\end{equation}\end{enumerate}
\end{lemma}
\begin{corollary} \label{C:L:G1} If $\varphi $ is in
the real $\frac{\tau }{16}k^{-(40\mu +1)\delta }$-neighborhood of $\omega^{(1)}
(k,\delta ,\tau )$  and $z$ is on the circle
\begin{equation} \label{G-4}
C_1=\{z:|z-k^{2l}|=\frac{\tau l}{4}k^{2l-1-40\mu \delta }\},
\end{equation}
then  the following inequality holds for all $\m \in \tilde\Omega
(\delta )$:
\begin{equation}
\left||\vec k(\varphi) +\p_{\m}|^{2l}-z\right|\geq \frac{\tau l
}{4}k^{2l-1-40\mu \delta }, \ \ z\in C_1. \label{G1-5}\end{equation}
\end{corollary}

The lemma is proved in Section \ref{GC} (Corollaries \ref{Parts 1,2}
and \ref{Part 3}.) The corollary from the lemma is proven at the end
of Section \ref{GC}. Note that in Section \ref{GC} we construct
non-resonance set of $\varphi $ in the set of complex numbers. Such
complex non-resonance set we need for construction of further steps
of approximation.

Let $r=1,2...$ and
\begin{equation}\label{g} g^{(1)}_r({\k}):=\frac{(-1)^r}{2\pi
ir}\hbox{Tr}\oint_{C_1}\left((P(\delta)(H_0({\k})-zI)P(\delta))^{-1}VP(\delta)\right)^rdz,
\end{equation} \begin{equation}\label{G}
G^{(1)}_r({\k}):=\frac{(-1)^{r+1}}{2\pi
i}\oint_{C_1}\left((P(\delta)(H_0({\k})-zI)P(\delta))^{-1}VP(\delta)\right)^r(P(\delta)(H_0({\k})-zI)P(\delta))^{-1}dz.
\end{equation}
Note that $g^{(1)}_1(\k)=0$ since $V=0$. Coefficient
$g^{(1)}_2(\k )$ admits representation:
\begin{equation}\begin{aligned}
g^{(1)}_2(\k)
    &=\sum _{\q\in \Omega (\delta )\setminus \{0\}}| V_\q| ^2(|\k|^{2l}-
        |{\k}+\p_\q|^{2l})^{-1} \\
    &=-\frac12\sum _{\q\in \Omega (\delta )\setminus \{0\}}\frac{| V_\q| ^2
        \left(|{\k}+\p_\q|^{2l}+
        |{\k}-\p_\q|^{2l}-2|{\k}|^{2l}\right)}{(|{\k}|^{2l}-
        |{\k}+\p_\q|^{2l})(|{\k}|^{2l}-
        |{\k}-\p_\q|^{2l})},\label{2.14}
 \end{aligned}\end{equation}
 From now on  $\|A \|_1$ means the norm of an operator $A$ in the trace class.

\begin{theorem} \label{Thm1} Suppose $\varphi $ is in
the real  $\frac{\tau }{16}k^{-(40\mu +1)\delta }$-neighborhood of
$\omega^{(1)} (k,\delta,\tau )$ and   $\varkappa\in\R$,
$|\varkappa-k|\leq \frac{\tau}{16}k^{-40\mu \delta }$,
$\k=\varkappa(\cos \varphi ,\sin \varphi )$. Then, for sufficiently
large $k>k_0(V,\delta ,\tau )$ there exists a single eigenvalue of
$H^{(1)}({\k})$ in the interval $\varepsilon _1( k,\delta,\tau )=(
k^{2l}-\frac{\tau l}{2}k^{2l-1-40\mu \delta }, k^{2l}+\frac{\tau l
}{2}k^{2l-1-40\mu \delta })$. It is given by the absolutely
converging series:
\begin{equation}\label{eigenvalue}\lambda^{(1)}({\k})=\varkappa^{2l}+\sum\limits_{r=2}^\infty
g^{(1)}_r({\k}).\end{equation} For coefficients $g^{(1)}_r({\k})$
the following estimates hold:
\begin{equation}\label{estg} |g^{(1)}_r({\k})|\leq
%%%%%%%%%(Ck)^{-r(2l-1-40\mu\delta)+(2l-1-40\mu\delta)+4\delta}\leq
(Ck)^{-(r-1)(2l-1-40\mu\delta)+4\delta}. \end{equation}
 Moreover,
 \begin{equation}\label{estg_2}
|g^{(1)}_2({\k})|\leq Ck^{-2l+(80\mu+6)\delta}. \end{equation} The
corresponding spectral projection is given by the series:
\begin{equation}\label{sprojector}
\E^{(1)}({\k})=\E_0({\k})+\sum\limits_{r=1}^\infty G^{(1)}_r({\k}),
\end{equation} $\E_0({\k})$ being the unperturbed spectral
projection. The operators $G^{(1)}_r({\k})$ satisfy the estimates:
\begin{equation}
\label{jan27}
\left\|G^{(1)}_r({\k})\right\|_1<(Ck)^{-r(2l-1-44\mu\delta)}.
\end{equation}
Matrix elements of $G^{(1)}_r({\k})$ satisfy the following
relations:
\begin{equation}
G^{(1)}_r({\k})_{\s\s'}=0,\ \ \mbox{if}\ \ \
rQ<\||\p_\s\||+\||\p_{\s'}\||. \label{zeros} \end{equation}
\end{theorem}
\begin{corollary} \label{corthm1} For the perturbed eigenvalue and its spectral
projection the following estimates hold:
 \begin{equation}\label{perturbation}
\lambda^{(1)}({\k})=\varkappa^{2l}+O\left(k^{-2l+(80\mu+6)\delta}\right),
\end{equation}
\begin{equation}\label{perturbation*}
\left\|\E^{(1)}({\k})-\E_0({\k})\right\|_1<ck^{-2l+1+44\mu\delta}.
\end{equation}
Matrix elements of spectral projection $\E^{(1)}(\k)$ also satisfy
the estimate:
\begin{equation}
\label{matrix elements}
\left|\E^{(1)}(\k)_{\s\s'}\right|<(Ck)^{-d^{(1)}(\s,\s')},\ \ \
d^{(1)}(\s,\s')=Q^{-1}\left(\||\p_\s\||+\||\p_{\s'}\||\right)(2l-1-44\mu
\delta ).
\end{equation}
\end{corollary}
The last estimate easily follows from the formula \eqref{zeros} and
estimate \eqref{jan27}.
\begin{proof}The proof is based on expansion of the resolvent in
perturbation series on the circle $C_1$. Indeed, let us consider
the series
\begin{equation}
\left(H^{(1)}-z \right)^{-1}=\sum ^{\infty }
_{r=0}\left(H_0^{(1)}-z\right)^{-1}\left( -P(\delta)VP(\delta
)\left(H_0^{(1)}-z\right)^{-1}\right)^{r} \label{seriesforresolvent}
\end{equation}
where $H_0^{(1)}=P(\delta )H_0$ and $z\in C_1$. It easily follows
from \eqref{G1-5} that
\begin{equation}\left\|\left(H_0^{(1)}(\k)-z\right)^{-1}\right\|<\frac{8}{\tau l
}k^{-2l+1+40\mu \delta }.   \label{ocenka}
\end{equation}
Hence,
\begin{equation}\left\|\left(H^{(1)}(\k)-z\right)^{-1}\right\|<\frac{16}{\tau l
}k^{-2l+1+40\mu \delta }   \label{ocenka*}
\end{equation}
for sufficiently large $k$.
Substituting the series into the formula $\E^{(1)}
(\k)=-\frac{1}{2\pi i}\oint _{C_1}(H^{(1)}(\k)-z)^{-1}dz$ and
integrating term-wise, we arrive at
 \eqref{sprojector}. Estimates \eqref{jan27} easily follow from \eqref{ocenka} and the obvious inequality $\|P(\delta)\|_1\leq (2k^{\delta })^4$.
 It follows $\E^{(1)}= \E_0+O(k^{-2l+1+44\mu \delta })$. This means that there is
 a single eigenvalue of $H^{(1)}(\k)$ inside $C_1$.   In a similar way (using \eqref{g}, \eqref{2.14} and $V=0$) we obtain
 the formula for the eigenvalue and \eqref{estg}, \eqref{estg_2}, for details see
 [K]. To prove \eqref{zeros} we consider the operator $A=VP(\delta
)\left(H_0^{(1)}-z\right)^{-1}$ and represent it as $A=A_0+A_1+A_2$,
where $A_0=\left(P(\delta )-\E_0({\k})\right)A \left(P(\delta
)-\E_0({\k})\right)$, $A_1=\left(P(\delta )-\E_0({\k})\right)A
\E_0({\k})$, $A_2= \E_0({\k})A \left(P(\delta )-\E_0({\k})\right)$.
It is easy to see that $\E_0({\k})A \E_0({\k})=0$ because of $V=0$.
Note that
$$\oint _{C_1}\left(H_0^{(1)}-z\right)^{-1}A_0^r dz=0,$$ since the
integrand is a holomorphic function inside $C_1$. Therefore,
$$G^{(1)}_r({\k})=\frac{(-1)^{r+1}}{2\pi i}\sum _{j_1,...j_r=0,1,2,\
j_1^2+...+j_r^2\neq 0}\oint
_{C_1}\left(H_0^{(1)}-z\right)^{-1}A_{j_1}.....A_{j_r} dz. $$ At
least one of indices in each term is equal to $1,2$. We take into
account that $(A_2)_{\s\s'}=(A_1)_{\s'\s}=0$ if $\s\neq 0$ and
$A_{\s\s'}=0$ if $\||\p_{\s-\s'}\||>Q$. It follows that
$G^{(1)}_r({\k})_{\s\s'}$ can differ from zero only if $rQ\geq
\||\p_\s\||+\||\p_{\s'}\||$.

\end{proof}

%%%%%%%%%%%%It is easy to see that estimates are stable in the complex set
%%%%%%%%%%%%$\tilde{\k}:\ |\tilde{\k}-\k|<\frac{\tau }{8}k^{-(40\mu+1)\delta }$.

It will be shown (Corollary \ref{Corollary 3.8}) that coefficients
$g^{(1)}_r({\k})$ and operators $G^{(1)}_r({\k})$ can be
analytically extended into the complex $\frac{\tau }{16}k^{-(40\mu
+1)\delta }$-neighborhood of $\omega^{(1)}(k,\delta, \tau )$ as
functions of $\varphi $
 and to the complex
$\frac{\tau }{8}k^{-(40\mu+1)\delta }-$ neighborhood of $k$ as
functions of $\varkappa$, estimates \eqref{estg}, \eqref{estg_2},
\eqref{jan27} being preserved. Now, we use formulae \eqref{g},
\eqref{eigenvalue} to extend
$\lambda^{(1)}({\k})=\lambda^{(1)}(\varkappa,\varphi)$ as an
analytic function. Obviously, series \eqref{eigenvalue} is
differentiable. Using Cauchy integral we get the following lemma.
\begin{lemma} \label{L:derivatives-1}Under
conditions of Theorem \ref{Thm1} the following estimates hold when
$\varphi $ is in $\omega^{(1)}(k,\delta, \tau )$ or its complex
$\frac{\tau }{32}k^{-(40\mu +1)\delta }$-neighborhood and
$\varkappa$ is in the complex $\frac{\tau }{16}k^{-40\mu\delta
}$-neighborhood of $\varkappa=k$ :
 \begin{equation}\label{perturbation-C}
\lambda^{(1)}({\k})=\varkappa^{2l}+O\left(k^{-2l+(80\mu+6)\delta}\right),
\end{equation}
\begin{equation}\label{estgder1}
\frac{\partial\lambda^{(1)}}{\partial\varkappa}=2l\varkappa^{2l-1} +
O\left(k^{-2l +(120\mu+6)\delta}\right), \ \
\frac{\partial\lambda^{(1)}}{\partial \varphi }=
 O\left(k^{-2l +(120\mu+7)\delta}\right),\end{equation}
\begin{equation}\label{estgder2} \begin{split} &
\frac{\partial^2\lambda^{(1)}}{\partial\varkappa^2}=2l(2l-1)\varkappa^{2l-2}+O\left(k^{-2l+(160\mu+6)\delta}\right),\cr
& \frac{\partial^2\lambda^{(1)}}{\partial\varkappa\partial \varphi
}= O\left(k^{-2l+(160\mu+7)\delta}\right) ,\ \
\frac{\partial^2\lambda^{(1)}}{\partial\varphi ^2}=
O\left(k^{-2l+(160\mu+8)\delta}\right).
\end{split}\end{equation}
%%%%%%%%\begin{equation}\label{estgder1}
%%%%%%%%\nabla_{\tilde{\k}}\lambda^{(1)}\left(\tilde{k},\varphi\right)=2l|\tilde{\k}|^{2l-2}\tilde{\k}
%%%%%%%%+ O\left(k^{-2l +(120\mu+7)\delta}\right);\end{equation}
%%%%%%%%\begin{equation}\label{estgder2}
%%%%%%%%\frac{\partial^2\lambda^{(1)}}{\partial\tilde{k}_j\partial\tilde{k}_i}=2l|\tilde{\k}|^{2l-2}\delta_{ij}+
%%%%%%%%2l(2l-2)|\tilde{\k}|^{2l-4}\tilde{k}_j\tilde{k}_i +
%%%%%%%%O\left(k^{-2l+(160\mu+8)\delta}\right),
%%%%%%%%\end{equation}
%%%%%%%%where $ \tilde k_1=\tilde k \sin \varphi ,  \tilde k_2=\tilde k \cos  \varphi $, $\tilde \k = (\tilde k_1, \tilde k_2)$.
\end{lemma}

%%%%%%%%%For any $\varphi\in\W\cap[0,2\pi]$, and $|\tilde{k}-k|\leq
%%%%%%%%%\frac{\tau}{8}k^{-40\mu \delta },\ \tilde{k}\in\R$, there exists a
%%%%%%%%%unique eigenvalue $\lambda(\tilde{\k})$ of the operator
%%%%%%%%%$P(\delta)H(\tilde{\k})P(\delta)$ such that
%%%%%%%%%$|\lambda(\tilde{\k})-k^{2l}|\leq
%%%%%%%%%\frac{\tau}{2}k^{2l-1-40\mu\delta}$. We also have

%%%%%%%%%It follows from \eqref{g}, \eqref{perturbest} that

 %%%%%%%%%It follows:

%%%%%%%%%%%%%%%%%%%%%%%%%%%%%%%%%%%%%%%%%%%%%%%%%%%%%%%%%%%%%%%%%%%%%%%%%%%%%%%%%%%%%%%%%%%%%%%%%%%%%%
%%%%%%%%%%%The derivatives satisfy the estimates: \begin{equation}\label{estg}
%%%%%%%%%%%|T(m)g^{(1)}_r(\tilde{\k})|\leq
%%%%%%%%%%%(Ck)^{-(r-1)(2l-1-40\mu\delta)+|m|40\mu \delta }, \end{equation}
%%%%%%%%%%%\begin{equation}\label{estg_2} |g^{(1)}_2(\tilde{\k})|\leq
%%%%%%%%%%%Ck^{-2l+(8\mu+6)\delta+|m|40\mu \delta }, \end{equation} where
%%%%%%%%%%%$T(m)=\frac{\partial ^{|m|}}{\partial ^{m_1}x_1\partial ^{m_2}x_2}$,
%%%%%%%%%%%$|m|=m_1+m_2$.
%%%%%%%%%%%%%%%%%%%%%%%%%%%%%%%%%%%%%%%%%%%%%%%%%%%%%%%%%%%%%%%%%%%%%%%%%%%%%%%%%%%%

%%%%%%%%%%%??????/// If $\tilde k: |\tilde k-k|<\frac{\tau }{16}k^{-40\mu
%%%%%%%%%%%\delta }$ is fixed, then $\lambda\left(\tilde{\k}\right)$ can be
%%%%%%%%%%%extended as an analytic function of $\varphi $ into the $\frac{\tau
%%%%%%%%%%%}{16}k^{-(40\mu +1)\delta } $-neighborhood of every $\varphi \in
%%%%%%%%%%%{\cal W}$, since the estimates \eqref{perturbest} are stable with
%%%%%%%%%%%respect to such perturbation of $\varphi $.??????

\subsection{\label{GC}Geometric Considerations}
In this section we prove Lemma \ref{L:G1} and its corollary.
However, we will prove a version of this lemma  for a complex set of
$\varphi $. We need this complex version for further steps. Lemma
\ref{L:G1} is a simple corollary of the result proven in this
section. We will use the notation $|\a|^2_\R:=(\a,\a)_\R$ where
$(\a,\b)_\R:=a_1b_1+a_2b_2$ when $\a, \b \in \C^2$. It is easy to
see that $|\k (\varphi )+\p_{\m}|^2_{\R}$ is an analytic extension
in $\varkappa$ and $\varphi$ of $$
|\k+\p_{\m}|^2=\varkappa^2+p_{\m}^2+2\varkappa
p_{\m}\cos(\varphi-\varphi_{\m})
$$ defined for real $\varkappa,\varphi $. Note that $|\cdot |$ is the
canonical norm in $\C$ or $\R^2$. For every fixed $k\geq1$ and
$\frac{1}{32}\leq \tau\leq 32$, we describe the resonance set
$\OO^{(1)}=\OO^{(1)}(k,\tau )$ of $\varphi\in \C$. We put
\begin{equation} \label{52a} \OO^{(1)}(k,\tau
):=\cup_{\m\in\tilde\Omega(\delta)\setminus\{0\}}\OO_{\m}(k,\tau
),\end{equation} where
\begin{equation}\label{resonance} \begin{split}& \OO_{\m}(k,\tau
):=\{\varphi\in\C:\ \ \left||\vec k+\p_{\m}|^2_{\R}-k^2\right|\leq
\tau k^{1-40\mu\delta}\}=\cr & \{\varphi\in\C:\ \
\left|p_{\m}^2+2kp_{\m}\cos(\varphi-\varphi_{\m})\right|\leq \tau
k^{1-40\mu\delta}\}. \end{split} \end{equation} In most cases
parameter $\tau$ will be equal to $1$. But sometimes we will use
different choice of $\tau$. It easily follows from the definition
\eqref{resonance} and the estimate \eqref{above} that for any
$\varkappa\in\C$ such that $|\varkappa-k|\leq1$ and any $\varphi\in
\OO_{\m}(k,\tau )$ we have
\begin{equation}\label{complex}
\left||p_{\m}^2+2\varkappa
p_{\m}\cos(\varphi-\varphi_{\m})|-|p_{\m}^2+2kp_{\m}\cos(\varphi-\varphi_{\m})|\right|\leq
\frac{\tau}{4} k^{1-40\mu\delta}, \end{equation} provided
$2(1+40\mu)\delta\leq1$ and $k\geq 800$ which will be assumed in
what follows.

Let  $\W_0:=\{\varphi \in \C: |\Im \varphi |<1\}.$ We introduce a complex non-resonant set:
\begin{equation} \label{W1} \W^{(1)}(k,\tau ):=\W_0 \setminus \OO^{(1)}(k,\tau
). \end{equation} Clearly, it  is open. We also note that the set $\OO^{(1)}\cap[0,2\pi]$ is
symmetric, i.e.
$\OO^{(1)}\cap[0,2\pi]+\pi\,(\hbox{mod}\,2\pi)\,=\OO^{(1)}\cap[0,2\pi]$,
since $\varphi _{-\m}=\varphi _\m+\pi $.
%%%%%%%%%%%%%Let $40\mu\delta\leq1/2$.
We define $\omega^{(1)} (k,\delta, \tau )$ as a real part of
$\W^{(1)} (k,\delta, \tau )$:
\begin{equation}\omega^{(1)} (k,\delta, \tau )=\W^{(1)}(k,\tau )\cap [0,2\pi ).
\label{omega} \end{equation}
\begin{lemma} Let $\varphi $ be in $\W^{(1)}(k,\tau )$, then
\begin{equation}\label{jan28a}
\left||\vec k (\varphi )+\p_{\m}|^2_{\R}-k^2\right|\geq \tau
k^{1-40\mu\delta}\mbox{   for all  }\m\in \tilde\Omega (\delta
)\setminus\{0\}.\end{equation} If $\varphi $ is in the complex
$k^{-(40\mu +1)\delta }$-neighborhood of $\W^{(1)}(k,\tau )$ and
$\varkappa\in \C: |\varkappa-k|<\frac{\tau }{8}k^{-40\mu \delta }$.
Then, for $\k =\varkappa (\cos \varphi ,\sin \varphi )$ the
following estimate holds:
\begin{equation}\label{jan28b}
\left||\k (\varphi )+\p_{\m}|^2_{\R}-k^2\right|\geq \frac{\tau}{2}
k^{1-40\mu\delta}\mbox{   for all  }\m\in \tilde\Omega (\delta
)\setminus\{0\}.\end{equation}\end{lemma} The lemma easily follows
from \eqref{resonance} and \eqref{complex}.
\begin{corollary} \label{Parts 1,2} Parts 1 and 2 of Lemma \ref{L:G1} hold.
\end{corollary}
\begin{corollary} \label{Corollary 3.8} Coefficients $g^{(1)}_r({\k})$ and
operators $G^{(1)}_r({\k})$ can be analytically extended into the
complex $\frac{\tau }{16}k^{-(40\mu +1)\delta }$-neighborhood of
$\omega^{(1)}(k,\delta, \tau )$ as functions of $\varphi $
 and to the complex
$\frac{\tau }{16}k^{-(40\mu+1)\delta }-$ neighborhood of $k$ as
functions of $\varkappa$, estimates \eqref{estg}, \eqref{estg_2},
\eqref{jan27} being preserved. \end{corollary}

\begin{lemma} The  measure of the resonance set
$\OO^{(1)}\cap[0,2\pi]$ satisfies the estimate:
\begin{equation}\label{meas1} meas(\OO^{(1)}\cap[0,2\pi])\leq  Ck^{-37\delta\mu}.
\end{equation}
\end{lemma}
\begin{corollary} \label{Part 3} Part 3 of Lemma \ref{L:G1} holds. \end{corollary}
\begin{proof}

Let $\m \neq 0$ and $\varphi^{\pm}_{\m}$ be two (mod $2\pi$)
solutions of the equation $$
p_{\m}^2+2kp_{\m}\cos(\varphi-\varphi_{\m})=0. $$ Obviously,
$\varphi^{\pm}_{\m}-\varphi_{\m}=\pm \frac{\pi }{2}+O(k^{-1+\delta
})$.  Put $$ \Phi^{\pm}_{\m}:=\{\varphi\in\C:\ \
|\varphi-\varphi^{\pm}_{\m}|\leq \tau k^{-39\delta\mu}\}.
$$ Then, taking into account \eqref{below}, it is not difficult to see that $\OO_{\m}\subset
\cup_{\pm,j\in\Z}(\Phi^{\pm}_{\m}+2\pi j)$. Thus,
\begin{equation}\label{meas} meas(\OO^{(1)}\cap[0,2\pi])\leq4\tau
k^{-39\delta\mu}(8k^\delta)^4\leq C k^{-37\delta\mu}.
\end{equation} \end{proof}

{\bf Proof of Corollary \ref{C:L:G1}.} Let $C_1:=\{z\in\C:\ \
|z-k^{2l}|=\frac{\tau }{4}k^{2l-1-40\mu\delta}\}$ be the contour
around eigenvalue $k^{2l}$ of the unperturbed operator $H_0(\vec
k)$. Then it follows from \eqref{jan28a} that for any
$\varphi\in\W^{(1)}(k,\tau)$,
$\m\in\tilde\Omega(\delta)\setminus\{0\}$,
%$\tilde{k}\in \R,
%|\tilde{k}-k|\leq \frac{\tau}{16}k^{-40\mu \delta }$
and $z:\
|z-k^{2l}|\leq \frac{\tau }{4}k^{2l-1-40\mu\delta}$ we have
%(here we
%denote $\tilde{\k}:=\tilde{k}(\cos\varphi,\sin\varphi)$)
\begin{equation}\label{perturbest} \begin{split}&
||\vec k+\p_{\m}|^{2l}_{\R}-z|\geq ||\vec
k+\p_{\m}|^{2l}_{\R}-k^{2l}|-\frac{\tau
l}{4}k^{2l-1-40\mu\delta}\geq \cr & {\tau l}(1-O(k^{\delta-1}))
k^{2l-1-40\mu\delta}-\frac{\tau l}{4}k^{2l-1-40\mu\delta}\geq
\frac{\tau l}{4} k^{2l-1-40\mu\delta},
\end{split}
\end{equation} for sufficiently large $k$. For $\m=0$ the estimate
follows from the definition of $C_1$.

\subsection{\label{IS1}Isoenergetic Surface for  Operator $H^{(1)}$}

\begin{lemma}\label{ldk} \begin{enumerate}
\item For every sufficiently large $\lambda $, $\lambda :=k^{2l}$, and $\varphi $ in the real $\frac{\tau }{32} k^{-(40\mu +1)\delta }$-neighborhood
of $\omega^{(1)}(k,\delta, \tau )$ , there is a unique
$\varkappa^{(1)}(\lambda, \varphi )$ in the interval
$I_1:=[k-\frac{\tau }{32}k^{-40 \mu \delta },k+\frac{\tau
}{32}k^{-40 \mu \delta }]$, such that
    \begin{equation}\label{2.70}
    \lambda^{(1)} \left(\k
^{(1)}(\lambda ,\varphi )\right)=\lambda ,\ \ \k ^{(1)}(\lambda
,\varphi ):=\varkappa^{(1)}(\lambda ,\varphi )\vec \nu(\varphi).
    \end{equation}
\item  Furthemore, there exists an analytic in $ \varphi $ continuation  of
$\varkappa^{(1)}(\lambda ,\varphi )$ to the complex  $\frac{\tau
}{32} k^{-(40\mu +1)\delta }$-neighborhood of
$\omega^{(1)}(k,\delta, \tau )$ such that $\lambda^{(1)} (\k
^{(1)}(\lambda, \varphi ))=\lambda $. Function
$\varkappa^{(1)}(\lambda, \varphi )$ can be represented as
$\varkappa^{(1)}(\lambda, \varphi )=k+h^{(1)}(\lambda, \varphi )$,
where
\begin{equation}\label{dk0} |h^{(1)}|=O(k^{-4l+1+(80\mu +6)
\delta }), \end{equation}

\begin{equation}\label{dk}
\frac{\partial{h}^{(1)}}{\partial\varphi}=O\left(k^{-4l+1+(120\mu+7)\delta}\right),\
\ \ \ \
\frac{\partial^2{h}^{(1)}}{\partial\varphi^2}=O\left(k^{-4l+1+(160\mu+8)\delta}\right),
\end{equation}
\begin{equation}\label{dk1} \frac{\partial \varkappa^{(1)}}{\partial \lambda }=\frac{1}{2lk^{2l-1}}\left(1+O(k^{-4l+1+(120\mu +6)\delta })\right). \end{equation}\end{enumerate}

 \end{lemma} \begin{proof} \begin{enumerate} \item Let us prove
 existence of  $\varkappa^{(1)}(\lambda, \varphi ) $. By Theorem \ref{Thm1}, there
exists an eigenvalue $\lambda ^{(1)} (\k)$, given by
(\ref{eigenvalue}), for all $\varkappa$ in the interval $I_1$.  Let
    ${\cal L}^{(1)}(\varphi ):=
    \{\lambda^{(1)}(\k) : \varkappa \in
I_1\}.$ Using the definition of $I_1$, \eqref{perturbation}, and
continuity of $\lambda^{(1)}(\k)$ is continuous in $\varkappa$, we
easily obtain
    $
    {\cal L}^{(1)}(\varphi ) \supset
    [k^{2l}-t,k^{2l}+t]$,
    $t=c_1k^{2l-1-40\mu \delta}$, $
    0<c_1 \neq c_1(k).
    $
     Hence,  there exists a $\varkappa^{(1)}$ such that
$\lambda^{(1)}(\k^{(1)})=k^{2l}$, $\varkappa^{(1)} \in I_1$.

 Now we show that there is only one $\varkappa^{(1)}$  in the interval $I_1$ satisfying
(\ref{2.70}). Indeed, by \eqref{estgder1},
    $
    \dfrac{\partial \lambda^{(1)}(\k)}{\partial \varkappa } \geq 2lk^{2l-1}\bigl( 1+o(1)
\bigr)$. This implies that $\lambda^{(1)}(\k)$ is monotone with
respect to $\varkappa$ in $I_1$. Thus, there is only one $\varkappa
\in I_1$ satisfying ~(\ref{2.70}).

\item  We  consider $\lambda^{(1)}\left(\k
(\varphi )\right)$ as a function of complex variable $\varkappa$ in
the disc $|\varkappa-k|<\frac{\tau }{32}k^{-40\mu \delta }$. Taking
into account \eqref{perturbation-C} and applying Rouch\'{e}'s
theorem, we obtain that for any $\varphi$ in $\frac{\tau
}{32}k^{-(40\mu +1)\delta }$-neighborhood of $\omega^{(1)}(k,\delta,
\tau )$ there exists unique value of $\varkappa^{(1)}(\varphi )$
such that $|\varkappa^{(1)}(\varphi )-k|<\frac{\tau }{32}k^{-40\mu
\delta }$ and $\lambda^{(1)}\left(\k^{(1)}(\varphi
)\right)=\lambda:=k^{2l}$. Actually,
\begin{equation} |\varkappa^{(1)}(\varphi )-k|<k^{-4l+1+(80\mu +6) \delta
}.\label{kappa1} \end{equation} Then it follows from
\eqref{estgder1} and  implicit function theorem that
$\varkappa^{(1)}(\varphi)$ is locally analytic. Combined with
uniqueness this implies global analyticity.

 The
estimate \eqref{dk0} follows from \eqref{kappa1}.  Applying standard
arguments with the Cauchy formula we obtain \eqref{dk}. Using
\eqref{estgder1} we get \eqref{dk1}.
\end{enumerate}
\end{proof}

Let us consider the set of points in $\R^2$ given by the formula:
$\k=\k^{(1)} (\varphi), \ \ \varphi \in \omega^{(1)} (k,\delta, \tau
)$. By Lemma \ref{ldk} this set of points is a slightly disturbed
circle with holes, see Fig. 1. All the points of this curve satisfy
the equation $\lambda^{(1)} (\k ^{(1)}(\lambda, \varphi ))=k^{2l}$.
We call it isoenergetic surface of the operator $H^{(1)}$ and
denote by ${\cal D}_{1}(\lambda)$, see figure \ref{F:1}. The ``radius"
$\varkappa^{(1)}(\lambda, \varphi )$ of ${\cal D}_{1}(\lambda)$
monotonously increases with $\lambda $, see \eqref{dk1}.

%%%%%%%% Let
%%%%%%%%$\tau:=\vec \nu'_\varphi=(-\sin(\varphi),\cos(\varphi))$. We have $$
%%%%%%%%0=\frac{\partial\lambda}{\partial\varphi}=\left(\nabla_{{\k^{(1)}}}\lambda,\frac{\partial{k}^{(1)}}{\partial\varphi}\vec \nu\right)_\R+
%%%%%%%%\left(\nabla_{{\k^{(1)}}}\lambda,{k}^{(1)}\tau\right)_\R. $$ Now,
%%%%%%%%using \eqref{estgder1} and orthogonality of $\vec \nu$ and $\tau$ we
%%%%%%%%obtain the first estimate in \eqref{dk}. The second bound can be
%%%%%%%%obtained in the same way.

\subsection{Preparation for Step II. Construction
of the Second Nonresonant Set}
\subsubsection{Model Operator for Step II \label{MOforStep2}}
Here we will describe an operator $PHP$, see \eqref{PHP}, which will
be used  for constructing perturbation series in the second step.
The operator $PHP$ has a block structure, the size of blocks being
of order $k^{\delta}$.

Let $r_1$ be some fixed number $2<r_1$. An upper bound on $r_1$ we
will introduce in Step II. We defined $\OO _\m$ by formula
(\ref{resonance}) for all $\m $: $0<\||\p_\m\||\leq 4k^{\delta }$.
Now we define $\OO _\m$ by the formula
\begin{equation}\label{resonance1} \begin{split}& \OO_{\m}(k,\tau
):=\{\varphi\in\C:\ \ \left||\vec k+\p_{\m}|^2_{\R}-k^2\right|\leq
\tau k^{-40\mu\delta}\}=\cr & \{\varphi\in\C:\ \
\left|p_{\m}^2+2kp_{\m}\cos(\varphi-\varphi_{\m})\right|\leq \tau
k^{-40\mu\delta}\}. \end{split} \end{equation} for $\m $:
$4k^{\delta}<\||\p_\m\||\leq k^{r_1}$. Note that  the right-hand
part in the inequality here is smaller than the corresponding one in
\eqref{resonance}. Obviously, $\OO _\m$ contains the whole interval
$[0,2\pi )$ for sufficiently small $p_\m$. As in Step I let
$\varphi^{\pm}_{\m}$ be two (mod $2\pi$) solutions of the equation
\begin{equation} p_{\m}^2+2kp_{\m}\cos(\varphi-\varphi_{\m})=0.
\label{Jan23a} \end{equation}

\begin{lemma} \label{L:3.1} The set $\OO _\m (k, \tau )$ has the following properties:
\begin{enumerate}
\item If $p_\m>4k$, then $\W_0\cap \OO _\m (k, \tau )=\emptyset $.
\item If $k^{-1-39\mu \delta }\leq p_\m\leq 4k$ and $|4k^2-p_\m^2|>4\tau k^{-40\mu \delta }$, then
$\OO_{\m}\subset \cup_{\pm,j\in\Z}(\Phi^{\pm}_{\m}+2\pi
j)$, where
$$ \Phi^{\pm}_{\m}:=\left\{\varphi\in\C:\ \
|\varphi-\varphi^{\pm}_{\m}|\leq \frac{\tau k^{-1-40\mu\delta
}}{p_\m\sqrt{1-p_\m^2(2k)^{-2}}}\right\},
$$
and $\Phi^{+}_{\m}\cap \Phi^{-}_{\m}=\emptyset $.
\item
If  $|4k^2-p_\m^2|\leq 4\tau k^{-40\mu \delta }$, then
$\OO_{\m}\subset \cup_{\pm,j\in\Z}(\Phi^{\pm}_{\m}+2\pi j)$, where
$$ \Phi^{\pm}_{\m}:=\left\{\varphi\in\C:\ \
|\varphi-\varphi^{\pm}_{\m}|\leq 32\tau k^{-1-20\mu\delta }\right\}.
$$
\end{enumerate}
\end{lemma}
In the proof we use the Taylor series with respect to $\varphi $ for
$|\vec k(\varphi
)+\p_\m|_\R^2-k^{2}$ near its zeros, see Appendix 1.\\

Let $\varphi _0\in[0,2\pi)\setminus\OO^{(1)}(k,8)$, where $\OO^{(1)}(k,8)$ is given by \eqref{52a}. We define $\MM(\varphi
_0)\subset \Z^2$ as follows:  \begin{equation} \label{M} \MM(\varphi _0):=\{\m:\ \ \
0<\,|\|\p_{\m}|\|\leq k^{r_1}\ \hbox{and}\ \varphi
_0\in\OO_{\m}(k,1)\}. \end{equation} We will also need a larger set $$
\MM'(\varphi _0):=\{\m:\ \ \ 0<\,|\|\p_{\m}|\|\leq 2k^{r_1}\
\hbox{and}\ \varphi _0\in\OO_{\m}(k,1)\}. $$ In fact, $\MM(\varphi
_0)$, $\MM'(\varphi _0)$ do not include $\m:|\|\p_{\m}|\|<4k^{\delta
}$, since $\varphi _0\in[0,2\pi)\setminus\OO^{(1)}(k,8)$.

We split $\MM (\varphi _0)$ into two
 components  $\MM:=\MM_1\cup\MM_2$.
 %%%%%To define the components we put $$
%%%%%\MM_\m(\varphi _0):=\tilde{\MM}_{\m}(\varphi _0)\cap \MM(\varphi
%%%%%_0). $$ Obviously, $\MM_\m(\varphi _0)$ contains $\m$, however it
%%%%%may contain more then one point of $\MM(\varphi _0)$.
By definition, $\m\in \MM_1$ if $$\min _{\m'\in \MM' (\varphi _0),
\m'\neq \m}\||\p _{\m-\m'}\||>k^{\delta }.$$   Let $\MM_2=\MM
\setminus \MM_1$. Next, let $\tilde{\MM}_{\m}$ be $(k^{\delta
}/3)$-neighborhood of $\m$ in $\||\cdot \||$ norm:
$$\tilde{\MM}_{\m}:=\{\n:\ \ \ |\|\p_{\n-\m}|\|<
k^{\delta}/3\ \hbox{for a given}\ \m\in\MM(\varphi _0)\},$$
 Obviously,
$$\tilde{\MM}_{\m}(\varphi _0)\cap \tilde{\MM}_{\m'}(\varphi
_0)=\emptyset,\ \ \ \mbox{for any }\m \in \MM_1\mbox{ and }\m'\in
\MM',\ \ \m'\neq \m.$$
Let  $\tilde{\MM}_1(\varphi _0)$ be $(k^{\delta }/3)$-neighborhood of $\MM_1$ in $\||\cdot \||$ norm:$$\tilde{\MM}_1(\varphi _0):=\cup _{\m \in \MM_1(\varphi
_0)}\tilde{\MM}_{\m}(\varphi_0) =\{\n:\ \ \ |\|\p_{\n-\m}|\|<
k^{\delta}/3\ \hbox{for some}\ \m\in\MM_1(\varphi _0)\}.$$

\bigskip

Let us introduce an equivalence relation in $\M'$. We say $\m_0 \sim
\m_0'$ if there is a sequence $\m_j\in \MM'$, $j=1,...J,$ such that
$\min_{k<j}\||\p _{\m_j-\m_{k}}\||\leq k^{\delta }$ for all
$j=1,..,J$ and $\m_J=\m_0'$.  We denote the equivalence class
containing $\m \in \MM_2$ by  $\MM_2^{(\m)}$. By definition of
$\MM_2$ such equivalence class contains at least one more element.
In the next lemma we prove that an equivalence class contains no more than 4 elements. Namely in this lemma the restriction $l>1$ plays a crucial role.

\begin{lemma}\label{2.12} Let $\m_0\in \MM_2$ and $\m_j \in \M'$, $j=1,...,J$, are such that all $\m_{j}$, $j=0,...J$, are different  and
$\min_{k<j}\||\p _{\m_j-\m_{k}}\||\leq k^{\delta }$ for all
$j=1,..,J.$ Then, $1\leq J\leq 3$. \label{L:4}\end{lemma} The proof
is in Appendix 2.

Obviously, for any pair
$\m,\m'\in \MM_2$ either $ \MM_2^{(\m)}= \MM_2^{(\m')}$ or $
\MM_2^{(\m)}\cap \MM_2^{(\m')}=\emptyset .$ We can enumerate
different equivalence classes $ \MM_2^{(\m)}$ by an index $j$ and
denote them by $ \MM_2^{j}$, $j=1,...,J_0$. By construction, $\MM
_2\subset \cup _{j=1}^{J_0} \MM_2^{j}\subset \MM '$.

Let $\tilde \MM_2^{j}$ be $(k^{\delta }/3)$-neighborhood of $\MM_2^{j}$ in $\||\cdot \||$ norm:
$$\tilde{\MM}_2^{j}(\varphi _0):=\{\n:\ \ \ |\|\p_{\n-\m}|\|<
k^{\delta}/3\ \hbox{for an }\ \m\in\MM_2^{j}(\varphi _0)\}.$$
Obviously,
$$\tilde \MM _2^j =\cup _{\m\in \MM
_2^j}\tilde \MM _{\m},$$  $$\tilde \MM _2^j\cap \tilde \MM
_2^{j'}=\emptyset,\mbox{ when }j\neq j',$$ $$\tilde \MM _2^j\cap
\tilde \MM _\m =\emptyset,\mbox{ when }\m\in \MM_1.$$ Let $$\tilde
\MM _2=\cup _{j=1}^{J_0}\tilde \MM _2^j,$$
$$\tilde \MM =\tilde \MM _1\cup \tilde \MM _2.$$
Moreover, $\||\cdot \||$ distance between these sets is greater than $\frac{1}{3}k^{\delta }$.
It is easy to see that $\tilde \MM \subset \MM'$. Hence, the number
of elements in $\tilde \MM $ does not exceed $ck^{4r_1}$.

 We consider   the diagonal
projection $P$ corresponding to $\tilde{\MM}(\varphi _0)$:
$$P(\varphi _0)_{\m\m}=\left\{\begin{array}{ll}1,&\mbox{when } \m\in
\tilde{\MM}(\varphi _0),\\0,& \mbox{otherwise.}\end{array}\right.$$
We consider $PH(\k ^{(1)}(\varphi ))P:\ PL_2(\Z^2) \to PL_2(\Z^2)$
for $\varphi \in \C$, $|\varphi -\varphi _0|<k^{-2-\delta (40\mu
+1)}$. Since $\varphi _0\in[0,2\pi)\setminus\OO^{(1)}(k,8)$, perturbation series (\ref{eigenvalue}),
(\ref{sprojector}) converge in the disc.

By construction, the set $\tilde \MM (\varphi_0)$ is split into
several nonintersecting components: \begin{equation} \label{def-tildeM}\tilde \MM (\varphi
_0)=\left(\cup _{\m \in \MM_1}\tilde \MM _\m \right) \cup \left(\cup
_j\tilde \MM_2^j\right).\end{equation} Obviously, \begin{equation}\label{defP}P=\sum _{\m \in
\MM_1}P_{\m}+\sum _{j}P_{2}^j,\end{equation} where $P_{\m}$, $P_{2}^j$ are
diagonal projectors corresponding to the sets $\tilde \MM _{\m}$ and
$\MM _2^j$, the projectors being orthogonal. Considering
\eqref{V_q=0} and taking into account that $Q<k^{\delta }/3$ for
sufficiently large $k$, we readily  show:
\begin{equation}P_{\m}VP_{\m'}=P_{\m}VP_{2}^j=P_{2}^jVP_{\m}=P_{2}^jVP_{2}^{j'}=0,
\mbox{when } \m,\,\m'\in\MM_1,\  \m \neq \m',\ j\neq j'. \label{PVP}
\end{equation}
Therefore,
\begin{equation} \label{PHP} PHP=\sum _{\m \in
\MM_1}P_{\m}HP_{\m}+\sum _{j}P_{2}^jHP_{2}^j.\end{equation} Since
\eqref{G1-1} holds for any $\m\in \tilde \Omega (\delta )\setminus
\{0\}$, we have $\MM (\varphi _0)\cap \tilde \Omega (\delta
)=\emptyset $. This means that the $\||\cdot \||$-distance between
$\tilde \MM (\varphi _0)$ and $ \Omega (\delta )$ is no less than
$3k^{\delta }$. Hence,
\begin{equation}P_{\m}VP(\delta)=P(\delta )VP_{\m}=P(\delta )VP_{2}^{j}=P_{2}^jVP(\delta)=0.
 \label{PVP*}
\end{equation}
\subsubsection{Estimates for the Resolvent of the Model Operator}
In the next lemma we use the restriction $l>1$ for the first time. In fact, we need this restriction only in the second step of the procedure.
\begin{lemma}\label{L:estnonres1} Let $\varphi _0\in \omega ^{(1)}(k,8)$.
\begin{enumerate}
\item If $\m\in \M_1(\varphi _0): p_\m>4k^{\delta}, |2k-p_\m| \geq 1$, then,
the operator $$\left(P_\m\left(H\big( \k ^{(1)}(\varphi
)\big)-k^{2l}I\right)P_\m\right)^{-1}$$ has no more than one pole in
the disk $|\varphi -\varphi _0|<2k^{-2-\delta (40\mu +1) }$. The
following estimate holds: \begin{equation} \label{Mon3-1}
\left\|\left(P_\m\left(H\big(\k ^{(1)}(\varphi
)\big)-k^{2l}I\right)P_\m\right)^{-1}\right\|<ck^{-2l+1}\varepsilon
_0 ^{-1}, \  \ \varepsilon _0=\min\{\varepsilon , k^{-2-(40\mu +1)
\delta}\}, \end{equation} when $\varphi $ is in the smaller disk
$|\varphi -\varphi _0|<k^{-2-\delta (40\mu +1)}$, $\varepsilon $
being the distance from $\varphi $ to the nearest pole of the
operator.
\item If $\m\in \M: |2k-p_\m| <1$, then, in fact $\m\in \M_1$ and the operator $$\left(P_\m\left(H\big( \k ^{(1)}(\varphi
)\big)-k^{2l}I\right)P_\m\right)^{-1}$$ has no more than two poles
in the disk $|\varphi -\varphi _0|<2k^{-2-\delta (40\mu +1)}$. The
following estimate holds: \begin{equation} \label{Mon3-3}
\left\|\left(P_\m\left(H\big(\k ^{(1)}(\varphi
)\big)-k^{2l}I\right)P_\m\right)^{-1}\right\|<ck^{-2l}\varepsilon
_0^{-2},\ \ \varepsilon _0=\min\{\varepsilon, k^{-2-\delta(40\mu
+1)}\},
\end{equation}  when $\varphi $ is in the smaller disk $|\varphi -\varphi _0|<k^{-2-\delta (40\mu
+1)}$, $\varepsilon $ being the distance from $\varphi $ to the
nearest pole of the operator.
\item If $\m\in \M: p_\m<4k^{\delta}$, then, in fact $\m\in \M_1$ and the operator
$$\left(P_\m\left(H\big( \k ^{(1)}(\varphi
)\big)-k^{2l}I\right)P_\m\right)^{-1}$$ has no more than one pole in
the disk $|\varphi -\varphi _0|<2k^{-2-\delta (40\mu +1)}$. The
following estimate holds: \begin{equation} \label{Mon3-2}
\left\|\left(P_\m\left(H\big(\k ^{(1)}(\varphi
)\big)-k^{2l}I\right)P_\m\right)^{-1}\right\|<
8k^{-2l+1}p_\m^{-1}\varepsilon _0 ^{-1},\  \  \varepsilon _0=\min
\{\varepsilon , k^{-2l+1 +\delta }\},
\end{equation} when $\varphi $ is in the smaller disk $|\varphi -\varphi _0|<k^{-2-\delta (40\mu
+1)}$, $\varepsilon $ being the distance from $\varphi $ to the
nearest pole of the operator.
\item The operator
$\left(P_2^{j}\left(H\big( \k ^{(1)}(\varphi
)\big)-k^{2l}I\right)P_2^{j}\right)^{-1}$ has no more than four
poles in the disk $|\varphi -\varphi _0|<2k^{-2-\delta (40\mu +1)}$.
The following estimate holds: \begin{equation} \label{Mon3-4}
\left\|\left(P_2^j\left(H\big(\k ^{(1)}(\varphi
)\big)-k^{2l}I\right)P_2^j\right)^{-1}\right\|<ck^{-2l-2-120\mu
\delta }\varepsilon _0^{-4}, \ \ \varepsilon _0=\min \{\varepsilon ,
k^{-1-40\mu \delta }\},\end{equation}  when $\varphi $ is in the smaller disk  $|\varphi -\varphi
_0|<k^{-2-\delta (40\mu +1)}$, $\varepsilon $ being the distance
from $\varphi $ to the nearest pole of the operator.
\end{enumerate}
\end{lemma}
\begin{corollary}\label{estnonres} Let $\varphi _0\in {\omega ^{(1)}}(k,8)$. Then, the operator
$\left(P\left(H\big( \k ^{(1)}(\varphi
)\big)-k^{2l}I\right)P\right)^{-1}$ has no more than $64k^{4r_1}$
poles in the disk $|\varphi -\varphi _0|<2k^{-2-\delta (40\mu +1)}$.
 The
following estimate holds: \begin{equation} \label{Mon3}
\left\|\left(P\left(H\big(\k ^{(1)}(\varphi
)\big)-k^{2l}I\right)P\right)^{-1}\right\|<ck^{\mu
r_1}\varepsilon _0^{-1}+ck^{-2l}\varepsilon  _0^{-2}+ck^{-2l-2}\varepsilon _0
^{-4},  \  \varepsilon _0=\min
\{\varepsilon , k^{-2l+1 +\delta }\},
\end{equation}  when $\varphi $ is in the smaller disk  $|\varphi -\varphi _0|<k^{-2-\delta
(40\mu +1)}$, $\varepsilon $ being the distance from $\varphi $ to
the nearest pole of the operator.\end{corollary}

Indeed, the number of blocks in $PHP$ (see \eqref{PHP})  does not
exceed $16k^{4r_1}$ (the number of elements in $\Omega (r_1)$). The
resolvent of each block has no more than  four poles. Therefore, the
resolvent of $PHP$ has no more than $64k^{4r_1}$ poles. Using
\eqref{Mon3-1}-\eqref{Mon3-4} and, using that $p_\m>k^{-\mu r_1}$ in
\eqref{Mon3-2}, we obtain the corollary.
\begin{corollary}\label{estnonres1}
If $\varepsilon =k^{-r_1'}$, $r_1'\geq\mu r_1$, then
\begin{equation} \label{Mon3a} \left\|\left(P\left(H\big(\k
^{(1)}(\varphi
)\big)-k^{2l}I\right)P\right)^{-1}\right\|<ck^{4r_1'},\end{equation}
\begin{equation} \label{Mon3a'}
\left\|\left(P\left(H\big(\k^{(1)}(\varphi
)\big)-k^{2l}I\right)P\right)^{-1}\right\|_1<ck^{4r_1'+4r_1}.\end{equation}  \end{corollary}
 The first formula follows from \eqref{Mon3}. The second formula follows from the fact that the
dimension of $P$ does not exceed $k^{4r_1}$.\\

\begin{proof}\begin{enumerate} \item Let $|2k-p_\m|\geq 1$,
$p_\m>4k^{\delta }$.
 Clearly only
the case $p_\m<4k$ is significant, since otherwise $\OO _{\m}(k,1)$
cannot intersect the disc $|\varphi -\varphi _0|<k^{-2-40(\mu
+1)\delta }$ by Lemma \ref{L:3.1}. It is easy to see that the set
$\OO _{\m}(k,1)$ consists of two separate discs $\OO _{\m}^{\pm}(k,1)$, the
distance between them being greater than $ck^{-1/2}$. Let us assume for definiteness $\varphi _0\in \OO _{\m}^{+}(k,1)$. This means  the disc  $|\varphi -\varphi
_0|<k^{-2-\delta (40\mu +1)}$ does not intersect  $\OO _{\m}^-(k,1)$. Let us first show that the operator
\begin{equation}\left(P_\m\left(H _0\big(\k ^{(1)}(\varphi
)\big)-k^{2l}I\right)P_\m\right)^{-1} \label{freeres} \end{equation}
has exactly one pole inside $\OO _{\m}^+(k,1)$, which is, in fact,
inside $\OO _{\m}^+(k,1/4)$.  Note that $\k ^{(1)}(\varphi )$ is
defined in $\OO_\m^+(k,1)$, since the size of $\OO_\m^+(k,1)$ is
much less than that of any circle in $\OO^{(1)} $. It satisfies the
estimate $\k ^{(1)}(\varphi )= \vec k (\varphi )+o(k^{-2})$ in $\OO
_\m ^+$. If $\varphi _0\in \OO_\m ^+(k,1)\setminus \OO_\m^+(k,1/4)$,
then the estimates $\left||\k ^{(1)}(\varphi _0
)+\p_{\m+\q}|^2-k^2\right|>\frac{1}{4} k^{-40\mu \delta }$, hold for
$0\leq \||\p_\q\||<k^{\delta }$ (see definition of $\MM_1(\varphi
_0)$) and can be extended to the $(k^{-2-(40\mu +1) \delta
})$-neighborhood of $\varphi _0$ ($\frac14$ becomes $\frac18$).
Thus,
\begin{equation} \label{Mon3a*} \left\|\left(P_\m\left(H_0\big(\k
^{(1)}(\varphi
)\big)-k^{2l}I\right)P_\m\right)^{-1}\right\|<ck^{-2l+2+40\mu
\delta}
\end{equation}
$$\mbox{when}\ \ |\varphi -\varphi _0|<k^{-2-(40\mu +1)\delta },\ \
\varphi _0\in \OO_\m ^+(k,1)\setminus \OO_\m^+(k,1/4).$$
Clearly the resolvent \eqref{Mon3a*} does not have poles in the set
$|\varphi -\varphi _0|<k^{-2-(40\mu +1)\delta }$. The estimate
(\ref{Mon3-1})  with $\varepsilon _0=k^{-2-(40\mu +1)\delta}$ follows
from (\ref{Mon3a*}) and Hilbert identity.

 Now, suppose that $\varphi _0\in \OO _\m^+(k,\frac{1}{4})$.
The function $| \vec k (\varphi )+\p_\m|^{2l}_\R-k^{2l}$ has a
single zero inside $\OO_\m^+(k,\frac{1}{4})$. Using Rouch\'{e}'s
theorem, we obtain that $|\k ^{(1)}(\varphi )+\p_\m|^{2l}_\R-k^{2l}$
also has a single zero inside $\OO_\m^+(k,\frac{1}{4})$. Note that
the following inequality holds in $\OO_\m^+(k,\frac{1}{4})$ for
$0<\||\p_\q\||<k^{\delta }$: $$\left||\k ^{(1)}(\varphi
)+\p_{\m+\q}|^{2l}_\R-k^{2l}\right|>\frac{1}{4}k^{2l-2-40\mu\delta
}.$$ Indeed, if $\left||\k ^{(1)}(\varphi
)+\p_{\m+\q}|^{2l}_\R-k^{2l}\right|\leq
\frac{1}{4}k^{2l-2-40\mu\delta }$ for some $\q\neq (0,0)$ and
$\varphi \in \OO_\m^+(k,\frac{1}{4})$, then
\begin{equation} \left|2(\k ^{(1)}(\varphi
)+\p_{\m},\p_\q)_\R+p_\q^2\right|<\frac{1}{2}k^{-40\mu \delta
}.\label{Mon21}\end{equation}
Considering that the size of $\OO_\m^+(k,\frac{1}{4})$ is
$\frac{k^{-1-40\mu\delta }}{p_\m \sqrt{1-p_\m^2(2k)^{-2}}}(1+o(1))$
and that $p_\m>4k^{\delta}>4p_\q/2\pi$, we obtain the inequality
analogous to (\ref{Mon21}) for $\varphi _0$ with $\frac{3}{4}$
instead of $\frac{1}{2}$. This contradicts to the assumption
$\varphi _0\in \MM _1$. Thus,
 the following inequality holds for all
$\q:\||\p_\q\||<k^{\delta }$ including $\q=(0,0)$: $$\left||\k
^{(1)}(\varphi )+
\p_{\m+\q}|^{2l}_\R-k^{2l}\right|>\frac{1}{4}k^{2l-2-40\mu \delta
},$$ when $\varphi $ is on the boundary of
$\OO_\m^+(k,\frac{1}{4})$. Hence, the resolvent $$\left(P_\m\left(H
_0\big(\k ^{(1)}(\varphi )\big)-k^{2l}I\right)P_\m\right)^{-1}$$ of
the free operator $P_\m H_0$ has exactly one pole inside $\OO_\m
^+(k,\frac{1}{4})$ and \begin{equation} \left\|\left(P_\m\left(H
_0\big(\k ^{(1)}(\varphi )\big)-k^{2l}I\right)P_\m\right)^{-1}\right
\|\leq 4k^{-2l+2+40\mu \delta }, \label{Mon5} \end{equation} when
$\varphi $ is on the boundary on the disc $\OO _\m
^+(k,\frac{1}{4})$. Considering that the dimension of $P_\m$ does
not exceed $16k^{4\delta }$ we obtain:
\begin{equation} \left\|\left(P_\m\left(H _0\big(\k ^{(1)}(\varphi
)\big)-k^{2l}I\right)P_\m\right)^{-1}\right \|_1\leq
64k^{-2l+2+\delta (40\mu+4) }. \label{Mon5a} \end{equation} It
remains to prove the analogous result for the perturbed operator
$H$.  We introduce the determinant $$D(\varphi )=\det
\left(P_\m\left(H \big(\k ^{(1)}(\varphi
)\big)-k^{2l}I\right)P_\m\left(H _0\big(\k ^{(1)}(\varphi
)\big)-k^{2l}I\right)^{-1}P_\m \right).$$ Obviously, $D(\varphi
)=\det(I+A)$, where $I,A:P_\m L_2(\Z^2)\to P_\m L_2(\Z^2)$
$$A(\varphi )=P_\m V\left(H _0\big(\k ^{(1)}(\varphi
)\big)-k^{2l}I\right)^{-1}P_\m .$$ Taking into account that
$$D(\varphi )=\frac{\det \left(P_\m \left(H \big(\k ^{(1)}(\varphi
)\big)-k^{2l}I\right)P_\m \right)}{\det \left(P_\m \left(H_0\big(\k
^{(1)}(\varphi )\big)-k^{2l}I\right)P_\m \right)},$$ we see that
$D(\varphi )$ is a meromorphic function inside $\OO_\m
^+(k,\frac{1}{4})$. Next, we employ a well-known inequality for the
determinants, see \cite{RS}: \begin{equation} \left|\det (I+A)-\det
(I+B)\right|\leq \|A-B\|_1exp (\|A\|_1+\|B\|_1+1),\ \ A,B\in \bf
S_1. \label{determinants}\end{equation}  Putting $A=A(\varphi )$,
$B=0$, we obtain $$\left|\det (I+A)-1\right|\leq \|A\|_1exp
(\|A\|_1+1).$$ It is easy to see that $$\|A\|_1\leq
\|V\|\|P_\m\left(H_0 \big(\k ^{(1)}(\varphi
)\big)-k^{2l}I\right)^{-1}P_\m\|_1.$$ Considering the estimate
(\ref{Mon5a}) for the resolvent of the free operator, we obtain
$\|A_1(\varphi )\|_1<1/200$ on the boundary of
$\OO^+_\m(k,\frac{1}{4})$ for sufficiently large $k$. By
Rouch\'{e}'s theorem, $D(\varphi )$ has only one zero in $\OO_\m^+
(k,\frac{1}{4})$. Thus, $\det P_\m\left(H \big(\k ^{(1)}(\varphi
)\big)-k^{2l}I\right)P_\m$ has exactly one zero in
$\OO^+_\m(k,\frac{1}{4})$. Using this, we immediately obtain that
operator $\left(P_\m\left(H \big(\k ^{(1)}(\varphi
)\big)-k^{2l}I\right) P_\m\right)^{-1}$ has one pole inside
$\OO_\m^+(k,\frac{1}{4})$. Considering the estimate for the free
resolvent and using Hilbert
 identity, we immediately obtain,
\begin{equation}\left\|\left(P_\m\left(H (\k ^{(1)}(\varphi
))-k^{2l}I\right) P_\m \right)^{-1}\right\|\leq 8k^{-2l+2+40\mu
\delta } \label{Mon9} \end{equation} for all $\varphi $ on the
boundary of $\OO_\m^+(k,\frac{1}{4})$. Taking into account that the
size of $\OO^+_\m(k,\frac{1}{4})$ does not exceed $k^{-1-40\mu\delta
}$, we obtain:
\begin{equation} \left\|\left(P_\m\left(H \big(\k ^{(1)}(\varphi
)\big)-k^{2l}I\right)P_\m \right)^{-1}\right\|\leq 8k^{-2l+2+40\mu
\delta }(k^{-1-40\mu \delta }/\varepsilon
)\label{Jan10}\end{equation} when $\varphi \in
\OO_\m^+(k,\frac{1}{4})$  on the distance $\varepsilon $,
 from the pole. If $|\varphi -\varphi _0|<k^{-2-(40\mu +1)\delta }$, but $\varphi \not \in \OO_\m^+(k,\frac{1}{4})$,
 then $\varphi $ is on the distance less than $k^{-2-(40\mu +1)\delta }$ from the boundary of $\OO_\m^+(k,\frac{1}{4})$, since $\varphi _0$ is inside
  $\OO_\m^+(k,\frac{1}{4})$. The estimate  \eqref{Mon9} holds on the boundary and stable with respect to such a
  small perturbation of $\varphi $. Thus, estimate \eqref{Mon3-1} is proven.

%%%%%%%Considering that the dimension of $P_\m$ does not exceed $k^{4\delta
%%%%%%%}$, we obtain: $$ \left\|\left(P_\m\left(H \big(\k ^{(1)}(\varphi
%%%%%%%)\big)-k^2I\right) P_\m\right)^{-1}\right\|_1\leq 2k^{-2l+2+(40\mu
%%%%%%%+4)\delta }(k^{-1-40\mu \delta }/\varepsilon )$$ when $\varphi $ is
%%%%%%%on the distance $\varepsilon $, $0<\varepsilon <k^{-40\mu \delta }$
%%%%%%%from the pole.

\item Let $\m\in \MM$, $|2k-p_\m|<1$. Then,
$\OO_\m^+$ and $\OO_\m^-$ can overlap. The case $\varphi _0\in
\OO_\m (k,1)\setminus \OO_\m(k,1/4)$ we consider in the same way as
for $|2k-p_\m|\geq 1$. Suppose $\varphi _0\in \OO_\m(k,1/4)$.
Combining $\left|\bigl|\vec k(\varphi _0
)+\p_\m\bigr|_\R^2-k^2\right|<\frac{1}{4}k^{1-40\mu \delta }$ with
$|2k-p_\m|<1$, we obtain that the vectors $2\vec k(\varphi )$ and
$-\p_\m $ are close: $$|2\vec k(\varphi _0)+\p_\m|_\R^2<5k.$$
Therefore, $( \vec k(\varphi _0)+\p_\m , \p_\q)_\R =-( \vec
k(\varphi _0), \p_\q)_\R +O(k^{1/2+\delta })$ for all
$0<\||\p_\q|\|<k^{\delta }$. Considering that the size of $\OO_\m $
does not exceed $ck^{-1-20\mu \delta }$ (Lemma \ref{L:3.1}) and the
distance between $\OO_\m^+$ and $\OO_\m^-$ is $O(k^{-1/2})$, we
obtain the analogous estimate for all $\varphi $ in $\OO_\m$:
$$(\vec k(\varphi )+\p_\m , \p_\q )_\R=-(\vec k(\varphi ),
\p_\q)_\R+O(k^{1/2+\delta })\ \ \mbox{ for all
  }0<\||\p_\q\||<k^{\delta }.$$ It immediately follows:
$$|\vec k(\varphi )+\p_{\m+\q}|^2_\R-|\vec k(\varphi )+\p_{\m}|^2_\R=
|\vec k(\varphi )-\p_{\q}|^2_\R-|\vec k(\varphi )|^2_\R
+O(k^{1/2+\delta })$$ for all $0<\||\p_\q\||<k^{\delta }$. The size
of $\OO_\m$ is much smaller than that of $\OO_{-\q}$ and
  $\OO_\m$ is not completely in $\OO_{-\q}$. Hence,
  $$\left||\vec k(\varphi )+\p_{\m+\q}|^2_\R-|\vec k(\varphi
  )+\p_{\m}|^2_\R\right|>\frac{1}{2}k^{1-40\mu \delta }$$
  for all $\varphi \in \OO_\m$. Considering that $\left||\vec k(\varphi
  )+\p_{\m}|^2_\R-k^2\right|\leq
  \frac{1}{2}k^{-40\mu \delta }$ in $\OO_\m$, we obtain
  $$\left||\vec k(\varphi )+\p_{\m+\q}|^2_\R-k^2\right|>\frac{1}{2}k^{1-40\mu \delta }, \mbox{when}\ \varphi \in \OO_\m.$$
  In particular, $\m\in \MM_1(\varphi _0)$. Considering that
  $\left||\vec k(\varphi
  )+\p_{\m}|^2_\R-k^2\right|=
  \frac{1}{2}k^{-40\mu \delta }$ on the boundary of $\OO_\m$, we
  obtain \eqref{Mon5} and \eqref{Mon5a}.
Considering as before, we show that the resolvent (\ref{freeres})
has at most two poles inside $\OO_\m $. It follows that $$
\left\|\left(P_\m\left(H \big(\k ^{(1)}(\varphi
)\big)-k^{2l}I\right)P_\m \right)^{-1}\right\|\leq k^{-2l+2+40\mu
\delta }(ck^{-1-20\mu \delta }/\varepsilon )^2$$ when $\varphi $ is
on the distance $\varepsilon $ from the pole.

%%%%%%% Considering that the
%%%%%%%dimension of $P_\m$ does not exceed $k^{4\delta }$, we obtain:
%%%%%%%$$ \left\|\left(P_\m\left(H \big(\k ^{(1)}(\varphi
%%%%%%%)\big)-k^2I\right) P_\m\right)^{-1}\right\|_1\leq k^{-2l+2+(40\mu
%%%%%%%+4)\delta }(2k^{-1-20\mu \delta }/\varepsilon )^2$$ when $\varphi $
%%%%%%%is on the distance $\varepsilon $ from the pole.

\item Let $\m\in \MM$, $0<p_\m\leq 4k^{\delta },\ \m\not \in \Omega
(\delta)$. The case $\varphi _0 \in {\cal O_\m}(k,1)\setminus {\cal
O_\m}(k,\frac{1}{4})$ is considered the same way as in the previous
steps, see (\ref{Mon3a*}).  From now on we assume $\varphi _0 \in
{\cal O}_{\m}(k,\frac{1}{4})$. There is an eigenvalue $\lambda^{(1)}
\big(\k^{(1)}(\varphi )+\p_\m\big)$ of $P_\m H(\k^{(1)}(\varphi
))P_\m$ given by the perturbation series. Indeed,
$$\left|\bigl|\vec k(\varphi
_0)+\p_\m\bigr|_\R^2-k^2\right|<\frac{1}{4}k^{-40\mu \delta },$$
since $\varphi _0 \in {\cal O}_{\m}(k,\frac{1}{4})$. Considering
that $\varphi _0 \not \in {\cal O}^{(1)}(k,8)$, we easily obtain
that $\left|( \vec k (\varphi _0 ),\p_{\q})_\R\right|\gtrsim
k^{1-40\mu \delta }$ for all $\q\in \Omega (\delta )\setminus
\{0\}$. Taking into account that and $p_\m\leq 4k^{\delta }$ we
arrive at the estimate:
$$\left|\bigl|\vec k (\varphi _0 )+\p_{\m+\q}\bigr|_\R^2-k^{2}\right|\gtrsim
k^{1-40\mu \delta }$$ for all $\q\in \Omega (\delta )\setminus
\{0\}$ and any $\varphi _0\in \omega^{(1)}(k,8)\cap {\cal
O_\m}(k,\frac{1}{4})$. It follows  $\m\in \MM_1$. By Lemma \ref{ldk},
$\k^{(1)}(\varphi )$ is defined in $
\frac{1}{4}k^{-(40\mu+1)\delta}$-neighborhood of $\omega^{(1)}(k,8)$ which
we denote by $\tilde{\cal W}^{(1)}(k,\frac14)$.
%%%%%$1\leq\tilde\tau\leq8$ (here we consider smaller neighborhood than
%%%%%before just not to introduce additional notation later in
%%%%%subsection~\ref{GSII}). Whenever the estimates are uniform with
%%%%%respect to $\tilde\tau$ we will omit it in the notation.
It is easy
to show that the estimates similar to the last two hold for
$\k^{(1)}(\varphi )$, $\varphi \in \tilde{\cal W}^{(1)}(k,\frac14))\cap
{\cal O_\m}(k,\frac{1}{2})$ . Therefore,
\begin{equation}\label{metka1}\left|\bigl|\k^{(1)}(\varphi )+\p_\m\bigr|_\R^{2l}-k^{2l}\right|\lesssim
\frac{1}{2}k^{2l-2-40\mu \delta },\end{equation} \begin{equation} \label{metka2}\left|\bigl|\k^{(1)}(\varphi
)+\p_{\m+\q}\bigr|^{2l}-k^{2l}\right|\gtrsim k^{2l-1-40\mu \delta
}\end{equation} for all $\q\in \Omega (\delta )\setminus \{0\}$. It follows from
the last two estimates that the perturbation series for
$\lambda^{(1)} \big(\k^{(1)}(\varphi )+\p_\m\big)$ and
$\lambda^{(1)} \big(\k^{(1)}(\varphi )\big)$ converge. Both are holomorphic functions of $\varphi $ in $\tilde{\cal W}^{(1)}(k,\frac14)\cap
{\cal O_\m}(k,\frac{1}{2})$. Using
Rouch\'{e}'s theorem, it is not difficult to show (for details see
Appendix 3, Lemma~\ref{L: Appendix 1}) that the equation
\begin{equation}\lambda^{(1)} \big(\k^{(1)}(\varphi
)+\p_\m\big)=k^{2l}+\varepsilon _0,\ \ \ |\varepsilon _0|\leq p_\m
k^{\delta },\label{25} \end{equation} has no more than two solutions
$\varphi^\pm(\varepsilon_0) $ in the $\tilde{\cal W}^{(1)}(k,\frac18)\cap
{\cal O}_{\m}(k,\frac{1}{2})$. They satisfy the estimates:
 \begin{equation} \label{87a} \big|\varphi^{\pm }(\varepsilon_0)-\varphi _\m ^{\pm
}\big|<4k^{-2l+1+2\delta }.\end{equation} Considering that $\varphi _\m ^{\pm
}=\varphi_\m\pm \pi/2+ O(k^{-1+\delta })$, we see that  the distance
between two solutions is approximately equal to $\pi $. For any
$\varphi \in \tilde{\cal W}^{(1)}(k,\frac14)\cap {\cal
O}_{\m}(k,\frac{1}{2})$ satisfying the estimate $\big|\varphi
-\varphi _\m ^{\pm }\big|<k^{-\delta },$
\begin{equation}\frac{\partial }{\partial \varphi }\lambda^{(1)}
\big(\k^{(1)}(\varphi )+\p_\m\big)=\pm 2lp_\m k^{2l-1}(1+o(1)),
\label{26}
\end{equation} for details see Appendix 3, Lemma~\ref{L:Appendix 2}. Therefore (for details
see Appendix 3, Lemma~\ref{L:3.7.1}),
\begin{equation}\big|\lambda^{(1)} \big(\k^{(1)}(\varphi
)+\p_\m\big)-k^{2l}\big|\geq k^{2l-1} p_\m\varepsilon
\label{Mon10}\end{equation} if $\varphi \in \tilde{\cal
W}^{(1)}(k,\frac18)\cap {\cal O}_{\m}(k,\frac{1}{2})$ is outside $\tilde
{\cal O}_{\m,\varepsilon}^+ \cup \tilde {\cal
O}_{\m,\varepsilon}^-$, here and below $\tilde {\cal
O}_{\m,\varepsilon }^{\pm }$ are the open discs of the radius
$\varepsilon $, $0<\varepsilon <k^{-2l+1+\delta}$ , centered at
$\varphi^\pm (0)$. It is shown in Appendix 3, Lemma~\ref{L:July5}
that
\begin{equation} \left \|(\lambda ^{(1)}(\y(\varphi
))-k^{2l})\left(P_\m(H(\k^{(1)}(\varphi))-k^{2l})P_\m\right)^{-1}\right\|\leq
8, \ \ \ \y(\varphi ):= \k^{(1)}(\varphi )+\p _\m, \label{July3a'''}
\end{equation}
for any $\varphi $ in $\tilde{\cal
W}^{(1)}(k,\frac18)\cap {\cal O}_{\m}(k,\frac{1}{2})$.

If $|\varphi -\varphi _0|<2k^{-2-\delta (40\mu +1)}$ and $\varphi _0\in \omega (k,\delta ,8)\cap {\cal O}_{\m}(k,\frac{1}{4})$, then $\varphi \in \tilde{\cal
W}^{(1)}(k,\frac18)\cap {\cal O}_{\m}(k,\frac{1}{2})$ and, hence, \eqref{Mon10}, \eqref{July3a'''} hold. Now \eqref{Mon3-2} easily follows from \eqref{Mon10} and \eqref{July3a'''}.
%%%%%Thus,  \begin{equation}
%%%\label{Mon11}\left\|\left(P_\m(H(\k^{(1)}(\varphi )
%%%)-k^{2l}I)P_\m\right)^{-1}\right\|\leq c(p_\m k^{2l-1}
%%%%%\varepsilon)^{-1},\end{equation} $\varepsilon $ being the distance
%%%%from $\varphi $ to one of two poles.

%If $\varphi $ is real, then (\ref{Mon11}) easily follows from  (\ref{Mon10}).

%%% Now, the desired estimate follows from simple geometry.

\item
Let, now, $\m,\,\m'$ be two elements from the same set $\M_2^j$. It
means that there exist elements $\p_{\q_i}\in\Omega(\delta)$,
$i=1,\dots, I$, such that $\m'=\m+\sum_{i=1}^I\q_i$ and
$\m+\sum_{i=1}^s\q_i\in\M_2$ for any $1\leq s\leq I$. We have proved
in Lemma \ref{L:4} that $I\leq 3$. Next, we consider $\tilde{\cal
O}^j=\cup _{\m\in \M_2^j}\cal O_{\m}$. Each connected component of
$\tilde{\cal O}^j$ contains no more than four discs. We have proven
above that all $\m: p_\m<4k^{\delta } \mbox{ or } |2k-p_\m|<1 $
belong to $\M_1$. Using this fact and applying Lemma~\ref{L:3.1} we
see that the size of each component does not exceed $o(k^{-1-40\mu
\delta })$. Considering as above, we show that
\begin{equation}\left\|\left(P_\m\left(H (\k ^{(1)}(\varphi
))-k^{2l}I\right) P_\m \right)^{-1}\right\|\leq 2k^{-2l+2+40\mu
\delta} \label{Mon9'}
\end{equation} for all $\varphi $ on the boundary of $\tilde{\cal
O}^j$. Each component contains no more than four poles of the
resolvent. It follows that $$ \left\|\left(P_\m\left(H \big(\k
^{(1)}(\varphi )\big)-k^{2l}I\right)P_\m \right)^{-1}\right\|\leq
ck^{-2l+2+40\mu \delta}(k^{-1-40\mu \delta }/\varepsilon _0)^4$$
with $\varepsilon _0=\min\{\varepsilon, k^{-1-40\mu \delta }\},$
when $\varphi $ is on the distance $\varepsilon $ from the poles.
\end{enumerate}
\end{proof}

\subsubsection{Resonant and Nonresonant Sets for Step II \label{GSII}}

We divide $[0,2\pi )$ into $[2\pi k^{2+\delta(40\mu+1)}]+1$
intervals $\Delta_j^{(1)}$ with the length not bigger than
$k^{-2-\delta(40\mu+1)}$. If a particular interval belongs to
$\OO^{(1)}(k,8)$ we ignore it; otherwise, let
$\varphi_0(j)\not\in\OO^{(1)}(k,8)$ be a point inside the $\Delta_j^{(1)}$. Let
$$\W_j^{(1)}=\{\varphi \in \W^{(1)}:\ |\varphi -\varphi
_j^{(0)}|<2k^{-2-\delta(40\mu+1)}\}.$$ Clearly, neighboring sets
$\W_j^{(1)}$ overlap (because of the multiplier 2 in the
inequality), they cover the $2k^{-2-\delta(40\mu+1)}$-neighborhood of $\omega (k,8)$. We denote this neighborhood by $\hat \W^{(1)}(k,2)$. For each $
\varphi $ in the neighborhood there is a $j$ such that $|\varphi
-\varphi _{0}^j|<k^{-2-\delta(40\mu+1)}$. We consider the poles of
the operator $\left(P(\varphi _j^0)\left(H(\k^{(1)}(\varphi))-k^{2l}\right)P(\varphi
_j^0)\right)^{-1}$ in a $\W_j^{(1)}$ and denote them by $\varphi _{jm}$,
$m=1,...,M_j$. By Corollary \ref{estnonres}, $M_j<64k^{4r_1}$. Next,
let $\OO^{(2)}_{jm}$ be the disc of the radius $k^{-r_1'}$ around
$\varphi _{jm}$, $r_1'>\mu r_1$.
\begin{definition} The set
\begin{equation}\OO^{(2)}=\cup _{jm}\OO^{(2)}_{jm} \label{O2}
\end{equation}
we call the second resonant set. The set
\begin{equation}\W^{(2)}= \hat \W^{(1)}(k,2)\setminus \OO^{(2)}\label{W2}
\end{equation}
is called the second nonresonant set. The set
\begin{equation}\omega^{(2)}= \W^{(2)}\cap [0,2\pi) \label{w2}
\end{equation}
is called the second real nonresonant set. \end{definition}
\begin{lemma}\label{L:geometric2}Let  $r_1'>\mu r_1$,
$\varphi \in \W^{(2)}$ and $\varkappa \in \C:
|\varkappa-\varkappa^{(1)}(\varphi )|<k^{-4r'_1-2l+1-\delta }$.
Then,
\begin{equation} \label{Mon3a**}
\left\|\left(P\left(H\big(\k(\varphi
)\big)-k^{2l}I\right)P\right)^{-1}\right\|<ck^{4r_1'},\end{equation}
\begin{equation} \label{Mon3a***}
\left\|\left(P\left(H\big(\k(\varphi
)\big)-k^{2l}I\right)P\right)^{-1}\right\|_1<ck^{4r_1'+4r_1},\end{equation}
where $P$ is the projection \eqref{defP} corresponding to the
interval $\Delta _j^{(1)}$ containing $\Re \varphi $.
\end{lemma}
\begin{proof} For $\k=\k^{(1)}(\varphi )$ the lemma follows
immediately from the definition of $\W^{(2)}$ and Corollary
\ref{estnonres1}. Considering the Hilbert identity, it is easy to
see that estimates \eqref{Mon3a} and \eqref{Mon3a'} are  stable with
respect to perturbation of $\varkappa^{(1)}$ of order
$k^{-4r_1'-2l+1-\delta}$. This stability ensure \eqref{Mon3a**} and
\eqref{Mon3a***}.

 \end{proof} By total size of
the set $\OO^{(2)}$ we mean the sum of the sizes of its connected
components.
\begin{lemma} \label{L:O2size} Let $r_1'\geq (\mu +4)r_1$. Then, the size of each connected component of $\OO^{(2)}$ is less
than $128k^{4r_1-r_1'}$. The total size of $\OO^{(2)}$ is less than
$128k^{2+\delta(40\mu+1)+4r_1-r_1'}$, where $2+\delta(40\mu+1)+4r_1-r_1'<0$.
\end{lemma}
\begin{corollary} \label{C:O2size} If  a connected component of $\OO^{(2)}$ intersects $[0,2\pi)$ or its
$\frac{1}{2}k^{-2-\delta(40\mu+1)}$-neighborhood, then it is strictly inside
$\tilde \W^{(1)}$. \end{corollary}
\begin{proof} Indeed, each set $\W_j^{(1)}$ contains no more than
$64k^{4r_1}$ discs $\OO_{jm}$. Therefore,  the size of $\OO^{(2)}\cap
\W_j^{(1)}$ is less than $128k^{-r_1'+4r_1}$. Considering that
$128k^{-r_1'+4r_1}$ is much smaller than the length of $\Delta _j^{(1)}$, we
obtain that there is no connected components which go across the
whole set $\W_j^{(1)}$ and the size of each connected component of
$\OO^{(2)}$ is less than $128k^{4r_1-r_1'}$. Considering that
$j<k^{2+\delta(40\mu+1)}$, we obtain the required estimate for the
total size of $\OO^{(2)}$. \end{proof}

%%%%%%%{\it We say that $\varphi\in I_j$ is a resonant point of the second
%%%%%%%kind if its distance from the nearest pole of the operator
%%%%%%%$(P(H(\k_j)-k^2)P)^{-1}$, $\k_j=k(\cos \varphi_0(j), \sin
%%%%%%%\varphi_0(j))$  is less than $k^{-r_1'}$.} The set of all resonant
%%%%%%%points of the second kind we denote by $\OO^{(2)}$. Notice that
%%%%%%%$$meas\,(\OO^{(2)}\cap[0,2\pi))\leq(2\pi k^{2+\delta(40\mu+1)}+1)ck^{4r_1}k^{-r_1'}\leq k^{-r_1}.$$
%%%%%%%{\it Here and in what follows we assume that $r_1'>5r_1+2$.} We also
%%%%%%%put $\W^{(2)}:=\W ^{(1)}\setminus\OO^{(2)}$,
%%%%%%%$\W^{(2)}_j:=\W^{(1)}\cap\{\varphi:\
%%%%%%%|\varphi-\varphi_0(j)|<k^{-2-\delta(40\mu+1)}\}$.

We will also need the estimates for the resolvent in the
neighborhood of $\m=0$. From the definition of $\k^{(1)}(\varphi)$
we obtain the following lemma.
\begin{lemma}\label{estnonres0} Let
$\varphi\in\W^{(1)}_j$ and $C_2$ be the circle $|z-k^{2l}|=(2
ck^{4r_1'})^{-1}$ (where we use the same constant as in
Lemma~\ref{L:geometric2}). Then
$$
\left\|\left(P(\delta)(H(\k^{(1)}(\varphi))-z)P(\delta)\right)^{-1}\right\|\leq
8ck^{4r_1'}. $$ \end{lemma}
\begin{proof} This estimate is sufficiently obvious and can be
obtained in many different ways. Here though we will use the
construction which we often will keep in mind when stating similar
estimates in what follows. We apply for the $z$ variable the
"squeezing" arguments which we used in the proof of
Lemma~\ref{L:estnonres1} for the variable $\varphi$. Namely, by
 \eqref{ocenka*} (or rather its
identical analogues for complex $\k^{(1)}(\varphi)$),
$$
\left\|\left(P(\delta)(H(\k^{(1)}(\varphi))-z)P(\delta)\right)^{-1}\right\|\leq
\frac{16}{\tau l}k^{-2l+1+40\mu\delta}, $$ when
$|z-k^{2l}|=\frac{\tau l}{4}k^{2l-1-40\mu\delta}$. Let us show that
analytic function
$\hbox{det}\,(P(\delta)(H(\k^{(1)}(\varphi))-z)P(\delta))$ has the
single simple zero $z=k^{2l}$ inside the circle $C_1$ . Indeed,
consider $$D(\varphi )=\det \left(P(\delta )\left(H \big(\k
^{(1)}(\varphi )\big)-zI\right)\left(P(\delta )\left(H _0\big(\k
^{(1)}(\varphi )\big)-zI\right)P(\delta )\right)^{-1} \right).$$
Obviously, $D(\varphi )=\det(I+A)$, where $I,A:P(\delta
)L_2(\Z^2)\to P(\delta ) L_2(\Z^2)$,
$$A(\varphi )=P(\delta )V\left(H _0\big(\k ^{(1)}(\varphi
)\big)-zI\right)^{-1}P(\delta ) .$$ Obviously, $D(\varphi )$ is a
meromorphic function inside $C_1$.  Next, we employ
\eqref{determinants} putting $A=A(\varphi )$, $B=0$. We obtain
$\left|\det (I+A)-1\right|\leq \|A\|_1exp (\|A\|_1+1).$  By
\eqref{ocenka}, $\|A(\varphi )\|_1<ck^{-2l+1+\delta (40\mu +4)}$
when $z\in C_1$. By Rouch\'{e}'s theorem, $D(\varphi )$ has only one
zero in $C_1$. Thus, $\det P(\delta )\left(H \big(\k ^{(1)}(\varphi
)\big)-zI\right)P(\delta )$ has one zero in $C_1$. Using this, we
immediately obtain that operator $\left(P(\delta )\left(H \big(\k
^{(1)}(\varphi )\big)-zI\right) P(\delta )\right)^{-1}$ has exactly
one pole inside $C_1$,  the pole being at the point $z=k^{2l}$.
Using the maximum principle, we obtain the required estimate on the
circle $C_2$.
\end{proof}
We also notice that the statement of the Lemma~\ref{L:geometric2}
still holds (with $2c$ instead of $c$) if we use $z\in C_2$ instead
of $k^{2l}$. Thus, if we put \begin{equation}P_j:=P(\varphi_j(0)),\
\  \ \wt{P}_j:=P_j+P(\delta),  \label{defP_j} \end{equation} then
(notice that $P_jVP(\delta)=0$ by \eqref{PVP*})
\begin{equation}\label{estfull}
\left\|\left(\wt{P}_j(H(\k^{(1)}(\varphi))-z)\wt{P}_j\right)^{-1}\right\|\leq
8ck^{4r_1'},\ \ \ |z-k^{2l}|=(2ck^{4r_1'})^{-1},\
\varphi\in\W^{(1)}_j. \end{equation} Here we also used
Lemma~\ref{estnonres0}. At last, considering from the beginning the
discs $\OO^{(2)}_{jm}$ with radius $\frac12 k^{-r_1'}$ instead of
$k^{-r_1'}$ one can easily see that similar estimates (with probably
larger constants) hold in $k^{-r_1'-\delta}$-neighborhood of ${\cal
W}^{(2)}$.

\section{Step II}
\subsection{Operator $H^{(2)}$. Perturbation Formulas}
Let $P(r_1)$ be an orthogonal projector onto $\Omega(r_1):=\{\m:\
|\|\p_\m\||\leq k^{r_1}\}$ and $H^{(2)}=P(r_1)HP(r_1) $.  From now
on we assume \begin{equation}r_1'=40\mu r_1+2l,\ \ \ 2<r_1<k^{\delta /8}. \label{Aug13-1} \end{equation} We
consider $H^{(2)}(\k^{(1)}(\varphi ))$ as a perturbation of
\begin{equation}\label{gulf1} \tilde H^{(1)}=\tilde P_jH(\k^{(1)}(\varphi ))\tilde
P_j+\left(P(r_1)-\tilde P_j\right)H_0(\k^{(1)}(\varphi
))\left(P(r_1)-\tilde P_j\right),\end{equation}where $\tilde P_j=P_j+P(\delta )$ and
$P_j$ is the projection $P$ corresponding to the interval $\Delta _j^{(1)}$
containing $\varphi $. By \eqref{PVP}, \eqref{PVP*}, the first term on the right-hand side of \eqref{gulf1} has a block structure:
$$ \tilde P_jH(\k^{(1)}(\varphi ))\tilde
P_j=  P(\delta)HP(\delta )+PHP=P(\delta)HP(\delta )+\sum _{\m \in
\MM_1}P_{\m}HP_{\m}+\sum _{j}P_{2}^jHP_{2}^j.$$ The second term in \eqref{gulf1} is, obviously, diagonal. Thus, $\tilde H^{(1)}$ has a block-diagonal structure.
Let $W$ be the perturbation of $\tilde H^{(1)}$, i.e, $W=H^{(2)}-\tilde H^{(1)}$. It is easy to see that:
\begin{equation}W=P(r_1)VP(r_1)-\tilde P_j V\tilde P_j. \label{W}\end{equation}
By analogy with \eqref{g}, \eqref{G},
\begin{equation}\label{g2} g^{(2)}_r({\k}):=\frac{(-1)^r}{2\pi
ir}\hbox{Tr}\oint_{C_2}\left(W(\tilde
H^{(1)}({\k})-zI)^{-1}\right)^rdz,
\end{equation} \begin{equation}\label{G2}
G^{(2)}_r({\k}):=\frac{(-1)^{r+1}}{2\pi i}\oint_{C_2}(\tilde
H^{(1)}({\k})-zI)^{-1}\left(W(\tilde
H^{(1)}({\k})-zI)^{-1}\right)^rdz.
\end{equation}

\begin{theorem} \label{Thm2} Suppose $\varphi $ is in
the real  $k^{-r_1'-\delta }$-neighborhood of $\omega
^{(2)}(k,\delta,\tau )$ and $\varkappa\in\R$,
$|\varkappa-\varkappa^{(1)}(\varphi )|\leq k^{-4r'_1-2l+1-\delta}$,
$\k=\varkappa(\cos \varphi ,\sin \varphi )$. Then, for sufficiently
large $k>k_1(V,\delta ,\tau )$ there exists a single eigenvalue of
$H^{(2)}({\k})$ in the interval\\ $\varepsilon _2( k,\delta,\tau
)=\left( k^{2l}-(2ck^{4r_1'})^{-1},
k^{2l}+(2ck^{4r_1'})^{-1}\right)$. It is given by the absolutely
converging series:
\begin{equation}\label{eigenvalue-2}\lambda^{(2)}({\k})=\lambda^{(1)}({\k})+
\sum\limits_{r=2}^\infty g^{(2)}_r({\k}).\end{equation} For
coefficients $g^{(2)}_r({\k})$ the following estimates hold:
\begin{equation}\label{estg2} |g^{(2)}_r({\k})|<ck^{-2k^{\delta
}Q^{-1}}k^{-(2l-2-40\mu \delta )(r-1)+2+8\delta }.
\end{equation}
The corresponding spectral projection is given by the series:
\begin{equation}\label{sprojector-2}
\E ^{(2)}({\k})=\E^{(1)}({\k})+\sum\limits_{r=1}^\infty
G^{(2)}_r({\k}), \end{equation} $\E^{(1)}({\k})$ being the spectral
projection of $H^{(1)}(\k)$. The operators $G^{(2)}_r({\k})$ satisfy
the estimates:
\begin{equation}
\label{Feb1a} \left\|G^{(2)}_r({\k})\right\|_1<ck^{-k^{\delta
}Q^{-1}}k^{-(2l-2-40\mu \delta )r+1+4\delta }.
\end{equation}
\begin{equation}G^{(2)}_r({\k})_{\s\s'}=0,\ \ \mbox{if}\ \ 10rk^{\delta}<\||\p_\s\||+\||\p_{\s'}\|| \label{Feb6a}
\end{equation}
\end{theorem}
\begin{corollary} \label{corthm2} For the perturbed eigenvalue and its spectral
projection the following estimates hold:
 \begin{equation}\label{perturbation-2}
\lambda^{(2)}({\k})=\lambda^{(1)}({\k})+ O\left(k^{-2k^{\delta
}Q^{-1}}k^{-2l+4+48\mu \delta }\right),
\end{equation}
\begin{equation}\label{perturbation*-2}
\left\|\E^{(2)}({\k})-\E^{(1)}({\k})\right\|_1<ck^{-k^{\delta
}Q^{-1}}k^{-(2l-3-40\mu \delta -4\delta )}.
\end{equation}
\begin{equation}
\left|\E^{(2)}({\k})_{\s\s'}\right|<k^{-d^{(2)}(\s,\s')},\ \
\mbox{when}\ \||\p_\s\||>k^{\delta } \mbox{\ or }
\||\p_{\s'}\||>k^{\delta },\label{Feb6b}
\end{equation}
$$d^{(2)}(\s,\s')=\frac{1}{10}(\||\p_\s\||+\||\p_{\s'}\||)k^{-\delta }(2l-2-40\mu\delta )+k^{\delta
}Q^{-1}-1-4\delta.$$
\end{corollary}
Formulas \eqref{perturbation-2} and \eqref{perturbation*-2} easily
follow from \eqref{eigenvalue-2}, \eqref{sprojector-2} and
\eqref{estg2} and \eqref{Feb1a}. The estimate \eqref{Feb6b} follows
from \eqref{sprojector-2}, \eqref{Feb1a} and \eqref{Feb6a}. Indeed,
using these estimates, we obtain
$\left|\left(\E^{(2)}({\k})-\E^{(1)}({\k})\right)_{\s\s'}\right|
<k^{-d^{(2)}(\s,\s')}$. Considering that $\E^{(1)}({\k})_{\s\s'}=0$
when $\||\p_\s\||>k^{\delta }$ or $\||\p_{\s'}\||>k^{\delta }$, we
arrive at  \eqref{Feb6b}.

\begin{proof}Let
$\varphi\in\W^{(1)}_j$ for some $j$. Put $P'_j:=P(r_1)-\wt{P}_j$ (we
will omit the index $j$ in what follows). By \eqref{gulf1}, \eqref{W},
$$\tilde{H}^{(1)}\big(\k ^{(1)}(\varphi)\big):=\wt{P}H\big(\k
^{(1)}(\varphi)\big)\wt{P}+P'H_0\big(\k ^{(1)}(\varphi)\big)P',\ \ \
W:=P'VP'+P'V\wt{P}+\wt{P}VP'.$$
%%%%%Then $$H^{(2)}\big(\k
%%%%^{(1)}(\varphi)\big):=P(r_1)H\big(\k
%%%%^{(1)}(\varphi)\big)P(r_1)=\tilde{H}^{(1)}\big(\k
%%^{(1)}(\varphi)\big)+W.$$
We will often omit $\k ^{(1)}(\varphi)$ in the arguments when it
cannot lead to confusion. By \eqref{estfull}, we have
\begin{equation}\label{step2raz}
\left\|(\tilde{H}^{(1)}-zI)^{-1}\right\| <8ck^{4r'_1}.
\end{equation} Let us consider the perturbation series
\begin{equation}\label{step2dva}
(H^{(2)}-z)^{-1}=\sum_{r=0}^\infty(\tilde H^{(1)}-z)^{-1}A^r,
\end{equation}
$$A=-W(\tilde H^{(1)}-z)^{-1}.$$
To check the convergence it is enough to show that
\begin{equation}\label{||A||}\|A\|<ck^{-2l+2+40\mu \delta }.
\end{equation}
Estimates \eqref{step2raz}, \eqref{||A||} yield:
\begin{equation}\label{step2raz*}
\left\|(H^{(2)}(\k ^{(1)}(\varphi))-z)^{-1}\right\| <16ck^{4r'_1}.
\end{equation}
To prove \eqref{||A||} it suffice to establish the following estimates:
\begin{equation}\label{step2tri}
\begin{split}& \|P'VP'(\tilde H^{(1)}-z)^{-1}\|<ck^{-2l+2+40\mu \delta },\cr &
\|P'V\wt{P}(\tilde H^{(1)}-z)^{-1}\|<ck^{-2l+2+40\mu \delta },\cr &
\|\wt{P}VP'(\tilde H^{(1)}-z)^{-1}\|<ck^{-2l+2+40\mu \delta }.
\end{split} \end{equation} The first and the third inequalities in
\eqref{step2tri} are simple.  They follow from the definition of
$P'$ and $\wt{P}$ and identities $$ (\tilde
H^{(1)}-z)^{-1}P'=P'(\tilde
H^{(1)}-z)^{-1}=P'(H_0\big(\k^{(1)}(\varphi)\big)-z)^{-1}P'.$$ Indeed, by
definition, $(P')_{\m\m}=1$ if and only if $\varphi
_j(0)\not \in \OO_\m(k,8)$ and $\||\p_\m\||>k^{\delta }$. Therefore,
$\left||\k^{(1)}(\varphi
_0)+\p_\m|_\R^{2l}-k^{2l}\right|>k^{2l-2-40\mu \delta }$ when
$(P')_{\m\m}=1$ and the estimate is stable with respect to
perturbation of order $k^{-1-40\mu \delta }$.  The first and third inequalities in \eqref{step2tri} easily follow. Let us prove the
second estimate. We represent $\wt{P}(\tilde H^{(1)}-z)^{-1}$ as follows:
\begin{align}\nonumber
&\wt{P}(\tilde H^{(1)}-z)^{-1}=\cr
&\sum_{r=0}^{r_0}\left((H_0-z)^{-1}\wt{P}V\wt{P}\right)^r(H_0-z)^{-1}\wt{P}+
\left((H_0-z)^{-1}\wt{P}V\wt{P}\right)^{r_0+1}(\tilde
H^{(1)}-z)^{-1}\wt{P}, \end{align}
 where $r_0$ to be fixed later.
Then,
\begin{equation}\label{step25}
\begin{split}& \|P'V\wt{P}(\tilde H^{(1)}-z)^{-1}\|\leq
\sum_{r=0}^{r_0}\left\|B_r\right\|+
\left\|P'V\left((H_0-z)^{-1}\wt{P}V\wt{P}\right)^{r_0+1}\right\|\|(\tilde
H^{(1)}-z)^{-1}\wt{P}\|, \cr
&
B_r:=P'V\left((H_0-z)^{-1}\wt{P}V\wt{P}\right)^r(H_0-z)^{-1}\wt{P}.
\end{split} \end{equation}
Note that $B_r=P'B_r\wt{P}$ and matrix elements $(B_r)_{\j\s}$ are
equal to zero if $|\|\p_\j-\p_\s\||>Q(r+1)$ (see \eqref{V_q=0}).
Thus, the only non-trivial elements $(B_r)_{\j\s}$ are such that $$
\j\in\Omega(r_1)\setminus\left(\tilde{\MM}(\varphi _0)\cup \Omega (\delta )\right),\ \ \
\s\in\tilde{\MM}(\varphi _0)\cup \Omega (\delta ),\ \ \ |\|\p_\j-\p_\s\||\leq Q(r+1). $$
Let $r:Q(r+1)\leq k^\delta/6$. It follows that  $(B_r)_{\j\s}=0$ if $\s\in \MM$ or $\s=0$, since  such $\s$ have a distance greater than $\frac{1}{3}k^{\delta }$ from $\j$. If
$\s \not \in \MM$ or $\s \neq 0$, then $\left||\k^{(1)}(\varphi) +\p_{\s}|_{\R}^{2l}-z\right|>k^{2l-2-40\mu \delta }$. Therefore, for $r:Q(r+1)\leq k^\delta/6$ we have:
$$ \|B_r\|\leq
(Ck^{-2l+2+40\mu\delta})^{r+1},\ \ \
\left\|P'V\left((H_0-z)^{-1}\wt{P}V\wt{P}\right)^{r+1}\right\|\leq
(Ck^{-2l+2+40\mu\delta})^{r+1}, $$ with some absolute constant
$C$. Now, we fix $r_0:=k^\delta /(6Q)-1$. Then the condition $Q(r+1)\leq
k^\delta/6$ is satisfied for all $r\leq r_0$ and $$ \|P'V\wt{P}(\tilde H^{(1)}-z)^{-1}\|\leq
\sum_{r=0}^{r_0}(Ck^{-2l+2+40\mu\delta})^{r+1}+(Ck^{-2l+2+40\mu\delta})^{r_0+1}4ck^{4r'_1}.
$$  Assuming that $k$ is
large enough (in particular, $\frac{(2l-2-40\mu\delta)
k^\delta}{6Q}>5r'_1$) we obtain the second inequality in \eqref{step2tri}.
%%%%$$\|P'V\wt{P}(\tilde
%%%%%H^{(1)}-z)^{-1}\|<ck^{-2l+2+40\mu \delta }.$$

To prove \eqref{Feb1a} we consider the operator
$AP(\delta)=WP(\delta )\left(\tilde H^{(1)}-z\right)^{-1}$ and
represent it as $AP(\delta)=A_0+A_1+A_2$, where
$A_0=\left(P(r_1)-\E^{(1)}({\k})\right)A \left(P(r_1)
-\E^{(1)}({\k})\right)$, $A_1=\left(P(r_1 )-\E^{(1)}({\k})\right)A
\E^{(1)}({\k})$, $A_2= \E^{(1)}({\k})A
\left(P(r_1)-\E^{(1)}({\k})\right)$. Note that we have
$\E^{(1)}({\k})W\E^{(1)}({\k})=0$, because of \eqref{W}. We see that
$$\oint _{C_2}\left(\tilde H^{(1)}-z\right)^{-1}A_0^r dz=0,$$ since
the integrand is a holomorphic function inside $C_2$. Therefore,
\begin{equation} \label{Feb1} G^{(2)}_r({\k})=\frac{(-1)^{r}}{2\pi i}\sum
_{j_1,...j_r=0,1,2,\ j_1^2+...+j_r^2\neq 0}I_{j_1...j_r},\ \ \ \
I_{j_1...j_r}:=\oint _{C_2}\left(\tilde
H^{(1)}-z\right)^{-1}A_{j_1}.....A_{j_r} dz.
\end{equation} At least one of indices in each term is equal to 1 or 2.
Let us show that
\begin{equation} \label{A_2}
\|A_2\|_1<ck^{-k^{\delta }Q^{-1}}k^{-(2l-2-40\mu \delta )+1+4\delta}.
\end{equation} First, we notice that
$\E^{(1)}W(P(r_1)-\E^{(1)})=\E^{(1)}WP'$ by \eqref{W} and \eqref{PVP*}. It suffices
to show that
\begin{equation}\label{Feb6}\|\E^{(1)}WP'\|_1<k^{-k^{\delta
}Q^{-1}+1+4\delta},\end{equation} since $\|P'\left(\tilde
H^{(1)}-z\right)^{-1}\|=\|P'\left(
H_0-z\right)^{-1}\|<k^{-(2l-2-40\mu \delta)}$ for $z\in C_2 $.
Indeed,
$$\left(\E^{(1)}WP'\right)_{\s\s'}=\sum _{\s'':\ \||\p_{\s''}\||\leq k^{\delta},\ \||\p_{\s''-\s'}\||\leq
Q}\E^{(1)}_{\s\s''}W_{\s''-\s'}$$ when $\||\p_{\s'}\||>k^{\delta }$
and it is equal to zero otherwise. Hence,
$$\left|\left(\E^{(1)}WP'\right)_{\s\s'}\right|\leq \|W\|
\sum _{\s'':\ k^{\delta }-Q\leq \||\p_{\s''}\||\leq k^{\delta
}}\E^{(1)}_{\s\s''}$$ if $\||\p_{\s'}\||<k^{\delta }+Q$ and zero
otherwise. Using \eqref{matrix elements}, we obtain
\begin{equation} \label{E1-1}\left|\left(\E^{(1)}WP'\right)_{\s\s'}\right|<
ck^{4\delta }\max _{\||\p_{\s''}\||>k^{\delta
}-Q}k^{-d^{(1)}(\s,\s'')}.\end{equation} It easily follows:
$$\left|\left(\E^{(1)}WP'\right)_{\s\s'}\right|<ck^{4\delta }k^{-(2l-1-40\mu\delta)(k^{\delta }Q^{-1}-1+\||\p_\s\||Q^{-1})}$$
when  $\||\p_{\s'}\||< k^{\delta }+Q$, and zero otherwise. It
follows $\left\|\E^{(1)}WP'\right\|<ck^{-k^{\delta
}Q^{-1}+1+4\delta}$. Considering that $\E^{(1)}$ is a
one-dimensional projection, we obtain the same estimate for $\bf S_1
$-norm, namely, \eqref{Feb6}. Thus, we have proved \eqref{A_2}.
%%%$$\left\|\E^{(1)}WP'\right\|_1<ck^{-k^{\delta
%%%}Q^{-1}+1+4\delta}.$$ Considering that $\|P'\left(\tilde
%%H^{(1)}-z\right)^{-1}\|<ck^{-2l+2+40\mu \delta}$ and
%%$Q<k^{\delta/3}$, we obtain
Let us estimate $I_{j_1...j_r}$. Suppose one of the indices is equal to 2.
Substituting \eqref{A_2} into \eqref{Feb1} and taking into account
$\|\E^{(1)}\left(\tilde H^{(1)}-z\right)^{-1}\|<b_2^{-1}$, where
$b_2$ is the radius of $C_2$, we obtain:
$$\left\|I_{j_1...j_r}\right\|<ck^{-k^{\delta
}Q^{-1}}k^{-(2l-2-40\mu \delta )r+1+4\delta }.$$
Note that the operator
$A_1$ is always followed by $A_2$ unless $A_1$ occupies the very
last position in the product. Thus, it remains to consider the case
$A_{j_1}.....A_{j_r}= A_0^{r-1}A_1$. It is easy to see that
$$\left(\tilde
H^{(1)}-z\right)^{-1}A_{0}^{r-1}A_{1}=\left(\left(\tilde
H^{(1)}-\bar z\right)^{-1}A_2(\bar z)A_{0}^{r-1}(\bar z)\right)^*.$$
This implies the estimate for this case too.
Therefore,
$$\left\|G^{(2)}_r({\k})\right\|<ck^{-k^{\delta
}Q^{-1}}k^{-(2l-2-40\mu \delta )r+1+4\delta }.$$ The same estimate
can be written for the $\bf S_1$ norm of this operator, since
$\E^{(1)}$ is one-dimensional.

Let us obtain the estimate for $g_r({\k})$.
Obviously,\begin{equation} \label{Feb1'}
g^{(2)}_r({\k})=\frac{(-1)^r}{2\pi ir}\sum _{j_1,...j_r=0,1,2,\
j_1^2+...+j_r^2\neq 0}Tr\oint _{C_2}A_{j_1}.....A_{j_r} dz.
\end{equation}
Note that each term contains both $A_1$ and $A_2$, since we compute
the trace of the integral. Using \eqref{Feb6}, we obtain:
$\|A_1\|_1<cb^{-1}_2k^{-k^{\delta }Q^{-1}+1+4\delta}$. Combining
this estimate with \eqref{A_2} and \eqref{||A||}, we obtain
\eqref{estg2} for $r\geq 2$. Finally, applying \eqref{g2} in the
case $r=1$, we see that $g^{(2)}_1({\k})=0$, since
$\E^{(1)}W\E^{(1)}=0$.

To prove \eqref{Feb6a} it's enough to notice that the biggest block
of $\tilde H^{(1)}$ has the size not greater than $2k^\delta$.

\end{proof}

%%%%%%%%%%%It is easy to see that estimates are stable in the complex set
%%%%%%%%%%%$\tilde{\k}:\ |\tilde{\k}-\k^{(1)}(\varphi )|<
%%%%%%%%%%%k^{-4r_1-2l+1-\delta}$.
It is easy to see that coefficients $g^{(2)}_r({\k})$ and operators
$G^{(2)}_r({\k})$ can be analytically extended into the complex
$k^{-r_1'-\delta}$-neighborhood of $\omega ^{(2)}$ (in fact, into
$k^{-r_1'-\delta}$-neighborhood of $\W^{(2)}$) as functions of
$\varphi $ and to the complex $(k^{-4r'_1-2l+1-\delta})-$
neighborhood of $\varkappa =\varkappa^{(1)}(\varphi )$ as functions
of $\varkappa$, estimates \eqref{estg2}, \eqref{perturbation-2}
being preserved. Now, we use formulae \eqref{g2},
\eqref{eigenvalue-2} to extend
$\lambda^{(2)}({\k})=\lambda^{(2)}(\varkappa,\varphi)$ as an
analytic function. Obviously, series \eqref{eigenvalue-2} is
differentiable. Using Cauchy integral and Lemma
\ref{L:derivatives-1} we get the following lemma.
\begin{lemma} \label{L:derivatives-2}Under conditions of Theorem \ref{Thm2} the following
estimates hold when $\varphi \in \omega ^{(2)}(k,\delta )$ or its
complex $k^{-r_1'-\delta}$-neighborhood and $\varkappa\in \C:$
$|\varkappa-\varkappa^{(1)}(\varphi )|<k^{-4r'_1-2l+1-\delta}:$
\begin{equation}\label{perturbation-2c}
\lambda^{(2)}({\k})=\lambda^{(1)}({\k})+ O\left(k^{-2k^{\delta
}Q^{-1}}k^{-2l+4+48\mu \delta }\right),
\end{equation}
\begin{equation}\label{estgder1-2k}
\frac{\partial\lambda^{(2)}}{\partial\varkappa}=\frac{\partial\lambda^{(1)}}{\partial\varkappa}
+ O\left(k^{-2k^{\delta }Q^{-1}}k^{4r_1'+3+(48\mu +1)\delta
}\right), \end{equation}
\begin{equation}\label{estgder1-2phi}\frac{\partial\lambda^{(2)}}{\partial \varphi }=\frac{\partial\lambda^{(1)}}{\partial \varphi }+O\left(k^{-2k^{\delta }Q^{-1}}k^{r_1'-2l+4+(48\mu +1)\delta }\right),
 \end{equation}
\begin{equation}\label{estgder2-2} \frac{\partial^2\lambda^{(2)}}
{\partial\varkappa^2}= \frac{\partial^2\lambda^{(1)}}
{\partial\varkappa^2}+O\left(k^{-2k^{\delta
}Q^{-1}}k^{8r_1'+2l+2+(48\mu +2)\delta }\right), \end{equation}
\begin{equation} \label{gulf2} \frac{\partial^2\lambda^{(2)}}
{\partial\varkappa\partial \varphi
}=\frac{\partial^2\lambda^{(1)}}{\partial\varkappa\partial \varphi
}+ O\left(k^{-2k^{\delta }Q^{-1}}k^{5r_1'+3+(48\mu +2)\delta
}\right),
\end{equation}
\begin{equation} \label{gulf3}
\frac{\partial^2\lambda^{(2)}}{\partial\varphi ^2}=\frac{\partial^2\lambda^{(1)}}{\partial\varphi ^2}+O\left(k^{-2k^{\delta }Q^{-1}}k^{2r_1'-2l+4+(48\mu +2)\delta }\right).
\end{equation}\end{lemma}
%%%%\begin{equation}\label{estgder1-2}
%%%%%\frac{\partial\lambda^{(2)}}{\partial\tilde{k}}=2l\tilde k^{2l-1}
%%%%%+ O\left(k^{-2l +(120\mu+6)\delta}\right); \ \ \frac{\partial\lambda^{(2)}}{\partial \varphi }=
%%%%% O\left(k^{-2l +(120\mu+7)\delta}\right)\end{equation}
%%%%%\begin{equation}\label{estgder2-2} \frac{\partial^2\lambda^{(2)}}{\partial\tilde k^2}=2l(2l-1)\tilde k^{2l-2}+O\left(k^{-2l+(160\mu+6)\delta}\right),\ \ \frac{\partial^2\lambda^{(2)}}{\partial\tilde k\partial \varphi }
%%%%%,\
%%%%%%\frac{\partial^2\lambda^{(2)}}{\partial\varphi ^2}=
%%%%O\left(k^{-2l+(160\mu+8)\delta}\right).
%%%%%\end{equation}
%%%\begin{equation}\label{estgder1-2}
%%%%%\nabla_{\tilde{\k}}\lambda
%%%%%%^{(2)}\left(\tilde{k},\varphi\right)=2l|\tilde{\k}|^{2l-2}\tilde{\k}
%%%%%%+ O\left(k^{-2l +(120\mu+7)\delta}\right);\end{equation}
%%%%%\begin{equation}\label{estgder2-2}
%%%%%%\frac{\partial^2\lambda
%%%%%^{(2)}}{\partial\tilde{k}_j\partial\tilde{k}_i}=2l|\tilde{\k}|^{2l-2}\delta_{ij}+
%%%%%2l(2l-2)|\tilde{\k}|^{2l-4}\tilde{k}_j\tilde{k}_i +
%%%%%%O\left(k^{-2l+(160\mu+8)\delta}\right). \end{equation}

\subsection{\label{IS2}Isoenergetic Surface for Operator $H^{(2)}$}

\begin{lemma}\label{ldk-2} \begin{enumerate}
\item For every sufficiently large $\lambda $, $\lambda :=k^{2l}$, and $\varphi $ in the real  $\frac{1}{2} k^{-r_1'-\delta }$-neighborhood
of $\omega^{(2)}(k,\delta, \tau )$ , there is a unique
$\varkappa^{(2)}(\lambda, \varphi )$ in the interval
$I_1:=[\varkappa^{(1)}(\lambda, \varphi )-\frac{1
}{2}k^{-4r_1'-2l+1-\delta },\varkappa^{(1)}(\lambda, \varphi
)+\frac{1 }{2}k^{-4r_1'-2l+1\delta },]$, such that
    \begin{equation}\label{2.70-2}
    \lambda^{(2)} \left(\k
^{(2)}(\lambda ,\varphi )\right)=\lambda ,\ \ \k ^{(2)}(\lambda
,\varphi ):=\varkappa^{(2)}(\lambda ,\varphi )\vec \nu(\varphi).
    \end{equation}
\item  Furthermore, there exists an analytic in $ \varphi $ continuation  of
$\varkappa^{(2)}(\lambda ,\varphi )$ to the complex  $\frac{1}{2}
k^{-r_1'-\delta }$-neighborhood of $\omega^{(2)}(k,\delta, \tau )$
such that $\lambda^{(2)} (\k ^{(2)}(\lambda, \varphi ))=\lambda $.
Function $\varkappa^{(2)}(\lambda, \varphi )$ can be represented as
$\varkappa^{(2)}(\lambda, \varphi )=\varkappa^{(1)}(\lambda, \varphi
)+h^{(2)}(\lambda, \varphi )$, where
\begin{equation}\label{dk0-2} |h^{(2)}(\varphi )|=
O\left(k^{-2k^{\delta }Q^{-1}}k^{-4l+5+48\mu \delta }\right),
\end{equation}
\begin{equation}\label{dk-2}
\frac{\partial{h}^{(2)}}{\partial\varphi}= O\left(k^{-2k^{\delta
}Q^{-1}}k^{r_1'-4l+5+49\mu \delta }\right),\ \ \ \ \
\frac{\partial^2{h}^{(2)}}{\partial\varphi^2}=O\left(k^{-2k^{\delta
}Q^{-1}}k^{2r_1'-4l+5+50\mu \delta }\right).
\end{equation} \end{enumerate}\end{lemma}
\begin{proof}  The proof is completely analogous to that of Lemma \ref{ldk}, estimates \eqref{perturbation-2c} --\eqref{gulf3} being used. \end{proof}

%%%%%It follows from
%%%%%\eqref{perturbation-2c} and Rouch\'{e}'s theorem that for any
%%%%%It$\varphi$ in $k^{-r_1'-\delta}$-neighborhood of $\omega ^{(2)}$
%%%%%Itthere exists unique value of ${k}^{(2)}$ such that
%%%%%$|{k}^{(2)}(\varphi )-k^{(1)}(\varphi)|<k^{-4r'_1-2l+1-\delta}$ and
%%%%%$\lambda^{(2)}\left({k}^{(2)},\varphi\right)=\lambda_0:=k^{2l}$.
%%%%%Actually,
%%%%%\begin{equation} |{k}^{(2)}(\varphi )-k^{(1)}(\varphi)|<k^{-2k^{\delta }Q^{-1}}k^{-4l+5+48\mu \delta
%%%%%}.\label{kappa2} \end{equation} Then it follows from
%%%%%\eqref{estgder1-2} and  implicit function theorem that
%%%%%${k}^{(2)}(\varphi)$ is locally analytic. Combined with uniqueness
%%%%%this implies global analyticity. Considering (\ref{kappa2}), it is
%%%%%easy to show that perturbation series  for $\lambda^{(2)}
%%%%%\big(\k^{(1)} (\varphi) \big)$, $\E^{(2)} \big(\k^{(1)} (\varphi
%%%%%)\big)$ converge when $\varphi$ is in
%%%%%$k^{-r_1'-\delta}$-neighborhood of $\omega ^{(2)}$.

%%%%%The estimate \eqref{dk0-2} follows from (\ref{perturbation-2}) and
%%%%%(\ref{estgder1-2}). Applying standard arguments with the Cauchy
%%%%%integral we easily obtain the estimates for the derivatives in
%%%%%$\frac12 k^{-r_1'-\delta}$-neighborhood of $\omega ^{(2)}$.
%%%%%\end{proof}

Let us consider the set of points in $\R^2$ given by the formula:
$\k=\k^{(2)} (\varphi), \ \ \varphi \in \omega ^{(2)}(k,\delta, \tau
)$. By Lemma \ref{ldk-2} this set of points is a slight distortion
of ${\cal D}_{1}$, see Fig. 2. All the points of this curve satisfy
the equation $\lambda^{(2)} (\k ^{(2)}(\varphi ))=k^{2l}$. We call
it isoenergetic surface of the operator $H^{(2)}$ and denote by
${\cal D}_{2}$.

%%%%%%%% Let
%%%%%%%%$\tau:=\vec \nu'_\varphi=(-\sin(\varphi),\cos(\varphi))$. We have $$
%%%%%%%%0=\frac{\partial\lambda}{\partial\varphi}=\left(\nabla_{{\k^{(1)}}}\lambda,\frac{\partial{k}^{(1)}}{\partial\varphi}\vec \nu\right)_\R+
%%%%%%%%\left(\nabla_{{\k^{(1)}}}\lambda,{k}^{(1)}\tau\right)_\R. $$ Now,
%%%%%%%%using \eqref{estgder1} and orthogonality of $\vec \nu$ and $\tau$ we
%%%%%%%%obtain the first estimate in \eqref{dk}. The second bound can be
%%%%%%%%obtained in the same way.

\subsection{Preparation for Step III - Geometric Part. Properties of the Quasiperiodic Lattice}\label{geomIII}
Let
\begin{equation} \label{Aug25-1}
\SS (k, \varepsilon _0)=\left\{ \k \in \R^2:\left\|\left(H^{(1)}(\k)-k^{2l}\right)^{-1}\right\|>\varepsilon _0^{-1}\right\}.\end{equation}
In this section we prove that the number of the lattice points $ \k _0+\p_\m$,
$\||\p_\m\||<k^{r_1}$  in $\SS (k, \varepsilon _0)$  does not exceed $Ck^{\frac{2r_1}{3}+1}$ when $\varepsilon _0$ is  sufficiently small and  $\k  _0$ is fixed. For this we split
$\SS $ into two subsets: `` non-resonant" and ``resonant", the non-resonant set being just a vicinity of $\DD _1(k^{2l})$.
An estimate for the number of lattice points in the non-resonant set is proven in Lemma \ref{L:number of points-1}. An estimate for the number of lattice points in the resonant set is proven in Lemma \ref{L:Resnumberofpoints}.
 These estimates play an important
role in the further construction.
\subsubsection{General Lemmas} We consider
$\p_\m=2\pi(\s_1+\alpha\s_2)$ with integer vectors $\s_j$ such that
$|\s_j|\leq 4k^{r_1}$.

It is easy to see that there exists a pair $(q,p)\in\Z^2$ such that
$0<q\leq 4k^{r_1}$ and
\begin{equation}\label{q}
|\alpha q+p|\leq 16k^{-r_1}.
\end{equation}
We choose a pair $(p,q)$ which gives the best approximation. In
particular, $p$ and $q$ are mutually simple. Put
$\epsilon_q:=\alpha+\frac{p}{q}$.
%%%%%Without loss of generality and for
%%%%%the sake of convenience of notation we will assume that $\epsilon_q$
%%%%%is positive.
We have
\begin{equation}k^{-2r_1\mu}\leq |\epsilon_q| \leq 16q^{-1}k^{-r_1}.\label{epsilon_q}\end{equation}
 We  write any $\s_2$ in the form\begin{equation}\s_2=q\s_2'+\s_2'' \label{s1}
\end{equation}
with  integer vectors $\s_2'$ and $\s_2''$, $0\leq (\s_2'')_j< q$
for $j=1,2$. Hence, $|(\s_2')_j|\leq 4k^{r_1}/q+1$. It follows
$$
(2\pi)^{-1}\p_\m=(\s_1-p\s_2')+(-\frac{p}{q}\s_2''+\epsilon_q\s_2'')+\epsilon_q
q\s_2'.
$$
Denote $\s:=\s_1-p\s_2'$. Then $|\s|\leq 8k^{r_1}$. The number of
different vectors $\tilde{\s}:=-\frac{p}{q}\s_2''+\epsilon_q\s_2''$
is not greater than $(2q)^2$. For each fixed pair $\tilde \s,\ \s$
we obtain a lattice parameterized by $\s_2'$. We call this lattice a
cluster corresponding to given $\tilde \s,\ \s$. Each cluster,
obviously, is a square lattice with the step $\epsilon _qq$. It
contains no more than $\left(9k^{r_1}q^{-1}\right)^2$ elements,
since $|(\s_2')_j|\leq 4k^{r_1}q^{-1}+1$, $j=1,2$. The size of each
cluster is less than $5|\epsilon _q|k^{r_1}$. If $\epsilon _q$
satisfies slightly stronger inequality, than \eqref{epsilon_q} than
clusters don't overlap, see the following lemma.

     \begin{lemma}\label{Lattice-1}Suppose that $\epsilon _q$ satisfies the
     inequality
     \begin{equation}|\epsilon_q|\leq \frac{1}{64}q^{-1}k^{-r_1}.\label{epsilon_q'}\end{equation}
     Then, the size of each cluster is less that $\frac{1}{8q}$. The distance between clusters is greater than
     $\frac{1}{2q}$. \end{lemma}\begin{proof}Let us estimate the
     distance between  points $\s_2'=\bf 0$ of two different
     clusters. Indeed, $\s-\frac{p}{q}\s_2''\neq \bf 0$, since $p\s_2''$,
see \eqref{s1},
     is not a multiple of $q$. Therefore,  $\left|(\s-\frac{p}{q}\s_2'')_j\right|\geq \frac{1}{q} $, $j=1,2$.
     Considering that
     $0\leq (\s_2'')_j<q$, $j=1,2$, we obtain that the distance between two points
     where $\s_2'=\bf 0$ is greater than $\frac{1}{q}-|\epsilon _q|q$, that is greater
     than $\frac{15}{16q}$. The size of each cluster is obviously
     less than $|\epsilon _q|q(4k^{r_1}q^{-1}+1)\leq\frac{1}{8q}$. Thus, two clusters cannot overlap,
     the distance between them being greater than
     $\frac{1}{2q}$.\end{proof}
     We need two more properties of the lattice $\p_\m$,
     $\||\p_{\m}\||<2k^{r_1}$.
     \begin{lemma}\label{Lattice-2} The number of vectors $\p_\m$,
      satisfying the inequalities $\||\p_{\m}\||<2k^{r_1}$,
     $p_{\m}<|\epsilon _q|qk^{r_1/3}$, does not exceed
     $k^{2r_1/3}$.\end{lemma}
     \begin{proof} Suppose vectors $\p_{\m}$ and $\p_{\m'}$ satisfy
     the conditions of the lemma. Then, $\||\p_{\m}-\p_{\m'}\||<4k^{r_1}$.
     By definition of $\epsilon _q$, $(2\pi )^{-1}|\p_{\m}-\p_{\m'}|\geq |\epsilon _q|q$. Thus, the distance
     between the points $\p_{\m},\p_{\m'}$ is greater than $2\pi |\epsilon _q|q$ and each point can be
      surrounded by the disc  of the radius $\pi|\epsilon _q|q$,
      the discs being disjoint. Dividing the area of the disc of the radius $2 |\epsilon _q|qk^{r_1/3}$
      (we increased radius to take into account points $\p_\m$ near the boundary of the disc $p_{\m}<|\epsilon _q|qk^{r_1/3}$) by the area of a disc of the radius $\pi |\epsilon _q|q$, we
      obtain that the number of vectors satisfying the inequality
     $p_{\m}<|\epsilon _q|qk^{r_1/3}$ does not exceed
     $k^{2r_1/3}$.\end{proof}
     \begin{lemma} \label{Lattice-3} Suppose $q$ in the inequality
     \eqref{q} satisfies the estimate $q>k^{2r_1/3}$. Then, the
     number of vectors  $\p_\m$,
     $\||\p_{\m}\||<2k^{r_1}$, satisfying the inequality
     $p_{\m}<k^{-2r_1/3}$ does not exceed
     $2^{12}\cdot k^{2r_1/3}$.\end{lemma}
     \begin{proof} First assume
     $|\epsilon_q|>\frac{1}{64}q^{-1}k^{-r_1}$. Then, dividing
     the area of the disc of the radius $2 k^{-2r_1/3}$
      by the area of a disc of the radius $\pi |\epsilon _q|q>\frac{1}{32}k^{-r_1}$, we
      obtain that the number of vectors satisfying the inequality
     $p_{\m}<k^{-2r_1/3}$ does not exceed
     $2^{12}k^{2r_1/3}$.

     Second, we consider the case $|\epsilon_q|\leq
     \frac{1}{64}q^{-1}k^{-r_1}$. According to Lemma
     \ref{Lattice-1}, the clusters do not overlap. The distance
     between clusters is greater than $\frac{1}{2q}$. Therefore, dividing the area of a disc with radius $\frac32 k^{-2r_1/3}$ by the area of a disc with radius $\frac{1}{4q}$, the last number being smaller than $\frac14 k^{-2r_1/3}$ by the conjecture of the lemma, we obtain that the
     number of clusters intersecting the disc of the radius $k^{-2r_1/3}$ is
     less than $\left(6k^{-2r_1/3}q\right)^2$. Each cluster contains
     less than $\left(9k^{r_1}q^{-1}\right)^2$ points. Therefore, the
     total number of of vectors  $\p_\m$,
     $\||\p_{\m}\||<2k^{r_1}$, satisfying the inequality
     $p_{\m}<k^{-2r_1/3}$ does not exceed
     $\left(6k^{-2r_1/3}q\right)^2\cdot \left(9k^{r_1}q^{-1}\right)^2
     <2^{12}\cdot k^{2r_1/3}$. \end{proof}
     \subsubsection{Lattice Points in the Nonresonant Set\label{Lattice Points in a Nonresonant Set}}
     \begin{lemma} \label{L:number of points-1}
     Let $N(k, r_1, \k _0,\varepsilon _0)$  be the number of points $\k  _0+\p_{\n}$,
$\||\p_{\n}\||<k^{r_1}$ in the $\varepsilon _0$-neighborhood of
     ${\cal D}_1(k^{2l})$, where $\varepsilon _0=k^{-5\mu r_1}$ and
     $\k  _0\in\R^2$ being fixed. Then,
     $$ N(k,r_1,\k _0,\varepsilon _0)<1000  \cdot k^{\frac{2r_1}{3}+1}.$$
     \end{lemma}
     \begin{proof}Let us consider the segment $\p_{\n-\n'}$ between
     two points $\k  _0+\p_{\n}$ and $\k  _0+\p_{\n'}$ in the neighborhood.
     Obviously, $\||\p_{\n-\n'}\||<2k^{r_1}$ and $p_{\n-\n'}>k^{-\mu
     r_1}>>\varepsilon _0$. This means that the direction of the
     segment cannot be orthogonal to the curve (in fact they are almost parallel to the curve) and each end can be
     assigned its own angle coordinate $\varphi _{\n}, \varphi _{\n'}$, $\varphi _{\n}\neq \varphi _{\n'}$. We
     enumerate the points $\k  _0+\p_{\n}$ in the order of increasing $\varphi
     _{\n}$ and connect neighboring points by segments. First we
     consider the segments with the length greater or equal to
     $\frac{1}{64}k^{-\frac{2r_1}{3}}$. Since the length of ${\cal D}_1(k^{2l})$ does
     not exceed $3\pi k$, the number of such segments does not
     exceed $650 k^{\frac{2r_1}{3}+1}$.

      It remains to estimate the
     number of segments with the length less than
     $\frac{1}{64}k^{-\frac{2r_1}{3}}$.
 First, we prove that no two
     segments $\p_{\n_1-\n'_1}$, $\p_{\n_2-\n'_2}$ can be equal to
     each other.
We use concavity of the curve ${\cal D}_1(k^{2l})$ and
      a small size $\varepsilon _0$ of its neighborhood. We show that for every
     $\p_{\n_1-\n'_1}$ with both ends in the neighborhood, there is a point on the
     curve where the tangent vector is parallel to $\p_{\n_1-\n'_1}$. Since the tangent vector
      changes monotonously with $\varphi $, no two vectors
      $\p_{\n_1-\n'_1}$ can have the same direction. Indeed, let us consider a segment $\p_{\n_1-\n'_1}$.
      Let $(x,y)$ be local coordinates associated with  $\p_{\n_1-\n'_1}$,  the beginning
      of the segment being at the origin and the end having
      the coordinates $(\tau, 0)$, $\tau=p_{\n_1-\n'_1}$. The curve is described by the equation $y=y(x)$. It easily follows
      from Lemma \ref{ldk} that $y'(x)=o(1)$ and the curvature $\kappa $ of the curve ${\cal D}_1(k^{2l})$ is
      $\frac{1}{k}(1+o(1))$ at all  points of the curve. Using the formula
      $\kappa (x)=|y''(x)|\left(1+y'(x)^2\right)^{-3/2}$, we easily obtain $y''(x)=-\frac{1}{k}(1+o(1))$.
      Using a Taylor formula, we get $y(\tau )=y(0)+y'(0)\tau
      -\frac{1}{2k}(1+o(1))\tau^2$. Note that $|y(0)|, |y(\tau
      )|<2\varepsilon _0$, since both ends of the segment are in the $\varepsilon
      _0$-neighborhood of the curve. Considering also that $\tau
      >k^{-r_1\mu }$ and the estimate on $\varepsilon _0$, we
      conclude: $\frac{\tau }{k}=2y'(0)(1+o(1))+O(k^{-4r_1\mu })$.
      Substituting this into the Taylor formula \begin{equation}
      y'(\tau
      )=y'(0)-\frac{\tau}{k}(1+o(1)),\label{y'}
      \end{equation}
       we obtain: $y'(\tau
      )=-y'(0)(1+o(1))+O(k^{-4r_1\mu })$.
      If $y'(\tau
      )$ and $y'(0)$ have the same sign or one of them is zero, the last relation yields $|y'(\tau
      )|+|y'(0)|=O(k^{-4r_1\mu })$. This contradicts to \eqref{y'},
      since $\tau >k^{-r_1\mu }$.
      Therefore, $y'(\tau
      )$ and $y'(0)$ have different signs. Considering that $y'(x)$
      is continuous, we obtain that there is a point $x_0$ in $(0,\tau
      )$ such that $y'(x_0)=0$. This means that the isoenergetic
      curve at this point is parallel to $\p_{\n_1-\n'_1}$.

      To finish the proof of the lemma  we consider two cases. Suppose $q$ in the inequality
     \eqref{q} satisfies the estimate $q>k^{2r_1/3}$. Then, by Lemma \ref{Lattice-3}, the
     number of vectors  $\p_\n$,
     $\||\p_{\n}\||<2k^{r_1}$, satisfying the inequality
     $p_{\n}<\frac{1}{64}k^{-2r_1/3}$ does not exceed
     $2^{12}\cdot k^{2r_1/3}$. Since each of them can be used only once, the
     total number of short segments  does not exceed $2^{12}\cdot k^{2r_1/3}$.

     Let  $q\leq k^{2r_1/3}$.
     If $|\epsilon _q|>\frac{1}{64}q^{-1}k^{-r_1}$. Then, obviously,
     $\frac{1}{64}k^{-2r_1/3}<|\epsilon _q|qk^{r_1/3}$. Applying Lemma
     \ref{Lattice-2}, we obtain that the number of segments with the
     length less than $\frac{1}{64}k^{-2r_1/3}$ is less than
     $k^{2r_1/3}$. Since each of them can be used only once, the
     total number of short segments  does not exceed $k^{2r_1/3}$.
     It remains to consider the case $q\leq
     k^{2r_1/3}$, $|\epsilon _q|\leq \frac{1}{64}q^{-1}k^{-r_1}$. By
     Lemma \ref{Lattice-1}, clusters are well separated. Considering
     that the distance between clusters is greater than
     $\frac{1}{2q}$ and the size of each cluster is less than
     $\frac{1}{8q}$, we obtain that no more than $8\pi qk$ clusters
     can intersect $\varepsilon _0$-neighborhood of ${\cal D}_1(k^{2l})$.
     The part of the curve inside the clusters has the length $L_{in}$ which is less
     than the double size of a cluster $10|\epsilon _q|k^{r_1}$ (the curve is concave) multiplied by
     the number of clusters $8\pi qk$, i.e., $L_{in}<80\pi |\epsilon
     _q|qk^{r_1+1}$.  Next, the segments with the length less than
     $\frac{1}{2}k^{-2r_1/3}$ cannot connect different clusters, since the
     distance between clusters is greater than $\frac{1}{2q}\geq
     \frac{1}{2}k^{-2r_1/3}$. Therefore, any segment of the length
     less than $\frac{1}{2}k^{-2r_1/3}$ is inside one
     cluster. If we consider the segments with the length greater
     than $|\epsilon _q|qk^{r_1/3}$, then the number of such segments
     is less than $L_{in}/|\epsilon _q|qk^{r_1/3}$, i.e., it is less
     than $80 \pi k^{2r_1/3+1}$. By Lemma \ref{Lattice-2}, the total
     number of segments of the length less than $|\epsilon
     _q|qk^{r_1/3}$ is less than $k^{2r_1/3}$. Each of them can be used only once. Thus, the total
     number of segments is less than $ 300 k^{2r_1/3+1}$.
\end{proof}
\subsubsection{Lattice Points in the Resonant Set\label{Lattice
Points in a Resonant Set}}
 Let $\QQ_0=\{\bf 0\}\cup \QQ $ where $\QQ =\{\q_1,...\q_J\}$, $J\geq 1$ and
$\min_{\q,\q'\in\QQ_0}\||\p_{\q-\q'}\||<k^{\delta}$. We assume that
all elements of $\QQ_0$ are different.
  We say that $\k \in
\RR_{\QQ _0}\subset \R^2$ if all the following inequalities hold:
$$\left||\k+\p_{\q}|_{\R}^2-k^2\right|<k^{-40\mu \delta }\
\ \mbox{when
  }\q\in \QQ _0,$$
  \begin{equation}\label{Sunday}\left||\k+\p_{\q'}|_{\R}^2-k^2\right|\geq k^{-40\mu \delta }\ \ \mbox{when
  }\q' \not \in \QQ _0\ \mbox{and}\ \min _{\q \in \QQ _0}\||\p_{\q'-\q}\||<k^{\delta }.\end{equation}
 By Lemma \ref{2.12} such $\k$ may exist only if  $J\leq 3$. Let
$P_{\QQ_0}$ be a diagonal projection: $(P_{\QQ_0})_{\n\n}=1$ if and
only if $\min _{\q\in \QQ _0}\||\p_{\n-\q}\||<k^{\delta }$. The
dimension of $P_{\QQ_0}$ clearly does not exceed
$\left(2(J+1)k^{\delta }\right)^4$. Suppose $\k \in \RR_{\QQ _0}$.
We consider operator $P_{\QQ _0}\left(H(\k)-(k^{2l}+\varepsilon
_0)I\right)P_{\QQ _0}$, $|\varepsilon _0|<1$, and its determinant
$D(\k, k^{2l}+\varepsilon _0)$. Let $\SS_{\QQ _0}(k,\varepsilon
_0)\subset \RR_{\QQ _0}$ be the set:
\begin{equation}\label{Feb21-1*} \SS_{\QQ _0}(k,\varepsilon _0)=
\left\{\k \in \RR_{\QQ _0}: D(\k, k^{2l}+\varepsilon _0')=0\mbox{
for some } |\varepsilon _0'|<\varepsilon _0\right\}.\end{equation}
Obviously, $ \SS_{\QQ _0}(k,\varepsilon _0)=\SS (k, \varepsilon _0) \cap \RR_{\QQ _0}$, see \eqref{Aug25-1}.
Let $(\tau _1,\tau _2)$ be new orthogonal coordinates with the
origin at the point $\frac{1}{2}\p_{\q_1}$, $ \tau _1$-axis being in
the direction of $\p_{\q_1}$. It is easy to see that $\RR_{\QQ
_0}\subset \{(\tau _1,\tau _2):\ |\tau _1 |<k^{-39 \mu \delta },\
\frac12k<|\tau _2|<2k\}$.

\begin{lemma}\label{L:curves-2} The set $D(\k, k^{2l})=0$ has the following properties
 in $\RR _{\QQ _0}$:
 %%%%the region $\tau _1\in \R, |\tau
%%%_1|<k^{-39\mu\delta }$:
\begin{enumerate} \item  The equation
$D(\k, k^{2l})=0$ describes at most 8 curves. They are described by
the equations $\tau _2=f_i(\tau _1)$, where $|f_i'(\tau
_1)|<ck^{-1+\delta }$, $0\leq i\leq \tilde J$, $\tilde J\leq 8$.
\item The set $\SS_{\QQ_0}(k,\varepsilon _0)$ belongs to  $\cup
_i^{\tilde J}\SS_i(k,\varepsilon _0)$, $\SS_i(k,\varepsilon
_0)=\{\k: |\tau _2-f_i(\tau _1)|<2\varepsilon _0, |\tau
_1|<k^{-39\mu\delta } \}$.
\item The curves $\tau _2=f_i(\tau _1)$, $0\leq i\leq \tilde J$, all
together have no more than $2^{31}l^2\cdot k^{8\delta }$ inflection
points.
\item Let $\l$ is a segment of a straight line,
\begin{equation}\l=\{\k=(\tau_1, \beta _1 \tau _1+\beta _2),\  \tau _{1,0}<\tau _1<\tau _{1,0}+\eta
\},\ \ |\tau_{1,0}|<k^{-39\mu\delta},\label{segment-1}\end{equation}
such that both its ends belong to $\SS_{i}(k,\varepsilon _0)$,
    $2\varepsilon _0<\eta^8k^{-8l}$, $0<\eta<1$. Then, there is an
inner part $\l'$ of the segment
 which is not in $\SS_{i}(k,\varepsilon _0)$. Moreover, there is a point $(\tau _{1*},\tau _{2*})$ in $\l'$ such that
 $f'_i(\tau _{1*})=\beta _1$, i.e., the curve and the segment have
 the same direction when $\tau _{1}=\tau _{1*}$.
 \end{enumerate}
 \end{lemma}
 \begin{proof}    \begin{enumerate} \item Let us
consider eigenvalues $\hat \lambda _i(\k)$ of $P_{\QQ _0}H(\k)P_{\QQ
_0}$. Obviously, there are no more than $J+1$ eigenvalues satisfying
the inequality $|\hat \lambda
_i(\k)-k^{2l}|<\frac{1}{2}k^{2l-2-40\mu \delta }$
%%%%%and there are no less
%%%%%than $J+1$ eigenvalues satisfying the inequality $|\hat \lambda
%%%%%_i(\hat \k)-k^{2l}|<\frac{3}{2}k^{2l-2-40\mu \delta }$
 when
$\k \in \RR_{\QQ _0}$. Let $E_0$ be the diagonal projection:
$(E_0)_{\m'\m'}=1$ if and only if $\m'\in \QQ_0$. Let
$$T=P_{\QQ _0}\frac{\partial H(\k)}{\partial \tau _2}=
P_{\QQ _0}\frac{\partial H _0(\k)}{\partial \tau _2}.$$ Considering
that other $\p_{\q_i}$, if any, have directions close to
$\p_{\q_1}$,  and $T$ is a diagonal operator, we obtain:
\begin{equation}\label{T}
TE_0=2l\tau _2k^{2l-2}E_0(I+o(1)), \ \ |\tau _2|\approx
k,\end{equation} where $o(\cdot)$ stands in the sense of the norm of
bounded operators. Let $\hat e_i$ be a normalized eigenvector
corresponding to $\hat \lambda _i$:  $|\hat \lambda
_i(\k)-k^{2l}|<\frac{1}{2}k^{2l-2-40\mu \delta }$. By \eqref{Sunday}
and regular perturbation formulas, $\hat e_i=E_0\hat e_i+o(1)$.
Hence,
\begin{equation}
 \frac{\partial \hat \lambda
_i}{\partial \tau _2}=(T\hat e_i,\hat e_i)=2l\tau
_2k^{2l-2}(1+o(1)). \label{4m}
\end{equation}
By simple perturbation arguments, $\frac{\partial \hat \lambda
_i}{\partial \tau _1}=O(k^{2l-2+\delta })$. Hence, the number of
curves satisfying the equations $\hat \lambda _i(\k)=k^{2l}$ in $\RR
_{\QQ _0}$ is at most $2(J+1)$, each corresponding to a particular
$i$, $1\leq i\leq J+1$, and a sign of $\tau _2$. They can be
described as $\tau _2 =f_i(\tau _1)$,  $f_i(\tau _1)$ being
piecewise differentiable and $|f'_i(\tau _1)|<ck^{-1+\delta }$.
\item  It easily
follows from \eqref{Feb21-1*} and \eqref{4m} that
$\SS_{\QQ_0}\subset \cup _{i=1}^{\tilde J}\SS_{i}$, $1\leq \tilde
J\leq 8$.
\item Inflection points of the curves $\tau _2 =f_i(\tau _1)$ are
       described by the system of equations:
       \begin{equation} \label{D}D=0,\end{equation}
       \begin{equation}
       D_{\tau _2\tau _2}(D_{\tau _1})^2+2D_{\tau _1\tau _2}D_{\tau _1}D_{\tau
       _2}+D_{\tau _1\tau _1}(D_{\tau _2})^2=0
       \label{D''}\end{equation}
       where
       $D=D(\k, k^{2l})$. The left hand sides of \eqref{D} and
\eqref{D''} are polynomials of the degree $K$ and $3K-4$, $K\leq
2l\left(2(J+1)k^{\delta }\right)^4\leq 2^{13}lk^{4\delta }$, with
respect to $\tau _1$, $\tau _2$. If they are mutually irreducible
than, by Bezout theorem, the number of inflection points does not
exceed $K(3K-4)<2^{28}l^2\cdot k^{8\delta}$. Suppose the left hand
sides of \eqref{D} and \eqref{D''}  are mutually reducible. Then,
there is a solution $\tau _2= f_i(\tau _1)$ with the zero curvature
everywhere, i.e. a straight line $\tau_2=a\tau _1+b$. Considering
that $D(\tau _1, a\tau _1+b)=\left(\tau _1^{2l}+a^{2l}\tau
_1^{2l}\right)^{K/(2l)}(1+o(1))$ as $\tau _1 \to \infty $, we
conclude that a straight line cannot satisfy the equation \eqref{D}.
Thus, the total number of inflection points for the curves $\tau _2=
f_i(\tau _1)$ all together does not exceed $2^{31}l^2\cdot
k^{8\delta}$.
\item

Let us consider  a  segment \eqref{segment-1} of a straight line,
 such that both its ends are in
$\SS_{i}(k,\varepsilon _0) $ and $2\varepsilon _0<\eta^8k^{-8l}$. It
follows that $\beta _1=O(k^{-1+\delta })$, $|\beta _2|\approx k$,
since $f_i'(\tau _1)=O(k^{-1+\delta })$. Next, we show that the
there is a part $\l'$ of $\l$ which is outside of
$\SS_{i}(k,\varepsilon _0) $. Note that $D(\k, k^{2l}+\varepsilon
_0')=0$ if and only if $\hat D(\k, k^{2l}+\varepsilon _0')=0$ where
$\hat D(\k, \lambda )$ is the determinant of the matrix
\begin{equation} \hat H^{-1}P_{\QQ_0}(H-\lambda I)P_{\QQ _0}
,\ \ \ \hat H= (H_0-\lambda I)(I-E_0) +E_0.\label{d2}
\end{equation}
 Note that diagonal terms of the matrix $\hat H^{-1}P_{\QQ_0}(H-\lambda I)P_{\QQ _0}$ are
 equal to 1 unless they correspond to $E_0$ and
 \begin{equation}\label{Hhat}\|\hat H^{-1}(I-E_0)\|<ck^{-(2l-2-40\mu \delta
 )} \ \ \mbox{when}\ \ \k \in \RR_{\QQ_0}.\end{equation}
%%%%%%%%%%%It is easy to see that (\ref{d1}) holds for all points in the disk.
%%%%%5Using (\ref{d3}), we obtain \begin{equation}
%%%%%\hat D=f(\tau _1)\prod _{i=1}^4(\tau _1-\tau
%%%%%_1^{(i)})-|v_q|^2+O(k^{-(2l-2-\mu \delta +2\delta }),\ \ |f(\tau
%%%%%_1)|\approx k^{4l-4}.\label{d3}
%%%%%\end{equation}
Let us extend $\hat D(\tau _1, \beta _1 \tau _1+\beta _2)$ as an
analytic function of $\tau _1$ into the complex disc $\D=\{\tau
_1\in \C: |\tau _1|<k^{-39\mu \delta }\}$.
%%%%%%$$\D=\left\{\tilde \k=(\tau _1, \beta _1 \tau _1+\beta _2):\left||\tilde
%%%%%%\k|_{\R}^{2l}-k^{2l}\right|<\frac{1}{8}k^{2l-2-40\mu \delta }, |\Im
%%%%%%\tau _1|<1\right\}$$
%%%%%%$\tilde \k=(\tau _1, \beta _1 \tau _1+\beta _2)$.
%%%%%%%%$\D_1=\{\tau _1:  |\tau
%%%%%%%%_1+p_{\q_1}/2|<p_{\q_1}k^{-\delta }\}$.
We also  consider  the following regions in $\C$:
 %%%%%$$\D=\left\{\tau _1\in \C:\left||\hat
%%%%%\k|_{\R}^{2l}-k^{2l}\right|<k^{\delta }, \hat \k=(\tau
%%%%_1, \beta _1 \tau _1+\beta _2)\right\},$$
$$\D_{\q}=\left\{\tau _1\in \C:\left||\k+\p_{\q}|_{\R}^{2l}-k^{2l}\right|<k^{\delta
},\ \k=(\tau _1, \beta_1 \tau _1+\beta_2 )\right\}, \  \ \q\in \QQ
_0.$$ By $\hat \D_{\q}$ we will denote the set of vectors $\k $
corresponding to $ \D_{\q}$. We are interested only in the connected
component(s) of $\cup _{\q\in \QQ _0}\hat\D_{\q}$ having the
non-empty intersection with $\RR _{\QQ_0}$. It is easy to show that
$\l\cap \SS_{\QQ _0} \subset \cup _{\q\in \QQ _0}\hat \D_{\q}$.
Therefore, we assume that $\l \subset \cup _{\q\in \QQ _0}\hat
\D_{\q}$ (otherwise the lemma is proved). Note that the estimate
\eqref{Hhat} is preserved for such regions, since each $\D_{\q}$ can
be included  in the balls of the radius $ck^{-l+1+\delta /2}$
centered at the points $|\k+\p_{\q}|_{\R}^{2l}-k^{2l}=0$ and for any
pair $\q,\q'$ $\frac{d }{d \tau
_1}\left(|\k+\p_{\q}|_{\R}^{2}-|\k+\p_{\q'}|_{\R}^{2}\right)=O(|\p_{\q}-\p_{\q'}|)$.
%%%%%% It is
%%%%%%easy to see that
%%%%%%$$\left||\tilde k|_{\R}^{2l}-k^{2l}\right|\geq k^{\delta /2}, \left||\tilde \k+\vec
%%%%%%p_{\q_i}|_{\R}^{2l}-k^{2l}\right|\geq k^{\delta /2}, \tau _2=\alpha
%%%%%%\tau _1+\beta $$ outside these small disks. There are no more than
%%%%%%$2(J+1)$ roots of
Let $$\hat D_0(\tau _1, \beta _1 \tau _1+\beta _2)=\prod_{\q \in \QQ
_0}\left||\k+\p_{\q}|_{\R}^{2l}-k^{2l}\right|,\ \ \k=(\tau _1, \beta
_1 \tau _1+\beta _2).$$ Obviously $\hat D_0$ has at most $2(J+1)$
roots inside $\cup _{\q\in \QQ _0}\D_{\q}$. We denote the number of
roots by $\hat J$, $\hat J\leq 2(J+1)\leq 8$. It is easy to see that
$|\hat D_0|>k^{\delta (J+1)}$ on the boundary of $\cup _{\q\in \QQ
_0}\D_{\q}$. Using \eqref{Hhat} and \eqref{determinants}, it is easy
to show that $
 \hat D=\hat D_0+O(k^{\delta J})$ on the boundary.
Applying Rouchet's theorem, we see that $\hat D$ has $\hat J$ roots
inside the union of the disks and satisfies the estimate $|\hat
D|>\frac{1}{2}k^{(J+1)\delta }$ on the boundary $\cup _{\q\in
\QQ_0}\D_{\q}$. Therefore, $\hat D$ can be represented in the form:
\begin{equation}
\hat D(\tau _1, \beta _1 \tau _1+\beta _2)=\tilde f(\tau _1)\prod
_{i=1}^{\hat J}(\tau _1-\tau _1^{(i)}),\ \ \tau _1\in \cup _{\q\in
\QQ_0}\D_{\q}, \ \ \ 0\leq \hat J\leq 8.\label{d3'}
\end{equation}
Note that each $\D_{\q}$ can be included  in the balls of the radius
$ck^{-l+1+\delta /2}$ centered at the points
$|\k+\p_{\q}|_{\R}^{2l}-k^{2l}=0$. From the minimum principle
($\tilde f \neq 0$) it follows:
\begin{equation}
|\tilde f|\geq ck^\gamma, \   \  \gamma =\delta (J+1) +(l-1-\delta
/2)\hat J >0 ,\label{f}
\end{equation}
when $\tau _1 \in \cup _{\q\in \QQ_0}\D_{\q} $. It easily follows
that \eqref{d3'}, \eqref{f} hold in $\RR_{\QQ _0}$.
 Let us consider a segment of the straight line
$\k=[\tau _1,\beta_1 \tau _1+\beta_2],$ $\tau _1\in (\tau _{1,0},
\tau _{1,0}+\eta)$, $0<\eta <1$. By \eqref{d3'}, there is a point
$(\tau _1', \beta_1 \tau _1'+\beta_2 )$ in this segment, where
$|\hat D|>ck^\gamma \eta^{\hat J}$. Considering that  $\hat D(\tau
_1', f_i(\tau _1'))=0$ by the definition of the curve and the obvious
inequality $| \hat D_{\tau _2} (\tau _1,\tau _2)|<k^{2l(J+1)-1}$, we
obtain $\left|f_i(\tau _1')-(\beta_1 \tau
_1'+\beta_2)\right|>ck^\gamma \eta ^{\hat J}/
ck^{2l(J+1)-1}>k^{-2l(J+1)}\eta ^{\hat J}$. If $\eta^8k^{-8l}
>2\varepsilon _0$, then there are points in the segment which are outside
$\SS _i(k,\varepsilon _0)$. At one of these points the function
$\left|f_i(\tau _1')-(\beta_1 \tau _1'+\beta_2)\right|$ attains its
maximum value. At this point the curve and the line are parallel.

{\em Remark.} Note that the perturbation series for
$\left(P_{\QQ_0}(H-k^{2l}I)P_{\QQ _0}\right)^{-1}$ converges (with
respect to  $\left(P_{\QQ_0}(H_0-k^{2l}I)P_{\QQ _0}\right)^{-1}$)
when $\tau _1$ is on the boundary of $\cup _{\q\in \QQ_0}\D_{\q} $
and
\begin{equation}
\left\| \left(P_{\QQ_0}(H-k^{2l}I)P_{\QQ
_0}\right)^{-1}\right\|<ck^{\delta }. \label{Aug10a}
\end{equation}
By \eqref{d3'}, the resolvent has no more than $\hat J$ poles inside
$\cup _{\q\in \QQ_0}\D_{\q} $. Considering that each $\D_{\q}$ can
be included  in a ball of the radius $ck^{-l+1+\delta /2}$  and
applying the maximum principle for the norm of a holomorphic
operator, we obtain the following estimate inside $\cup _{\q\in
\QQ_0}\D_{\q} $:
\begin{equation}
\left\| \left(P_{\QQ_0}(H-k^{2l}I)P_{\QQ
_0}\right)^{-1}\right\|<ck^{\delta }\left(\frac{ck^{-l+1+\delta
/2}}{d}\right)^{\hat J},\label{Aug10b}
\end{equation}
where $d$ is the distance from a point $\tau _1 \in \cup _{\q\in
\QQ_0}\D_{\q} $ to a nearest pole of the  resolvent. If $d\geq \eta
/16 $, then
\begin{equation}
\left\| \left(P_{\QQ_0}(H-k^{2l}I)P_{\QQ
_0}\right)^{-1}\right\|<ck^{\delta }\left(\frac{ck^{-l+1+\delta
/2}}{\eta }\right)^{\hat J}<\varepsilon _0^{-1}. \label{Aug10c}
\end{equation}
It also shows that there is a point in $\bf l$ which is not in
$\SS(k,\varepsilon _0)$.
\end{enumerate}
\end{proof}

Let $N_{\QQ_0}(k ,r_1,\k _0,\varepsilon_0)$ be the number of points $\k _0
+\p_{\n}$, $\||\p_{\n}\||<k^{r_1}$ in $S_{\QQ_0}(k,\varepsilon _0)$, $\k _0$ being fixed.
%\marginpar{$\tilde \k$ or $\hat \k$?}

\begin{lemma} \label{L:Resnumberofpoints} Let $\delta <r_1<\infty $. If $\varepsilon _0<k^{-16\mu r_1}$ then the number
of points $N_{\QQ_0}(k ,r_1,\k_0,\varepsilon_0)$ admits the estimate
\begin{equation}\label{Resnumberofpoints}N_{\QQ_0}(k ,r_1,\k_0,\varepsilon_0)\leq
2^{44}l^2\cdot k^{2r_1/3+8\delta }.\end{equation}\end{lemma}
%\marginpar{$\tilde \k$ or $\hat \k$?}
\begin{proof} The proof of the lemma is analogous to that of Lemma
\ref{L:number of points-1}, when we replace properties of a
distorted circle ${\cal D}_1(\lambda )$ by analogous properties of
the set $D(\k)=0$ proven in the previous lemma. Indeed, let us
consider the segment $\p_{\n-\n'}$ between
     two points $\k _0+\p_{\n}$ and $\k _0+\p_{\n'}$ in the
     $\varepsilon _0$-neighborhood of a concave component of a curve $\tau _2=f_i(\tau _1)$.
     Obviously, $\||\p_{\n-\n'}\||<2k^{r_1}$ and $p_{\n-\n'}>k^{-\mu
     r_1}>>\varepsilon _0$. This means that the direction of the
     segment cannot be orthogonal to the curve and each end can be
     assigned its own  coordinate $\tau _{1\n}, \tau_{1\n'}$, $\tau _{1\n}\neq \tau _{1\n'}$. We
     enumerate the points $\k _0+\p_{\n}$ in the order of increasing
     $\tau
     _{1\n}$ and connect neighboring points by segments. Consider the segments with the length greater or equal to
     $\frac{1}{64}k^{-\frac{2r_1}{3}}$. Obviously, the length of the curve does
     not exceed $1$. Hence, the number of such segments does not
     exceed $128k^{\frac{2r_1}{3}}$. It remains to estimate the
     number of segments with the length less than
     $\frac{1}{64}k^{-\frac{2r_1}{3}}$.

      First, we prove that no two
     segments $\p_{\n_1-\n'_1}$, $\p_{\n_2-\n'_2}$ can be equal to
     each other in the same concave component of a curve $\tau _2=f_i(\tau _1)$. Indeed, both end of $\p_{\n_1-\n'_1}$ are in
     $\SS_{\QQ_0}(k,\varepsilon _0)$. By Lemma \ref{L:curves-2}, part 4, there is a point on the curve between two ends of the segment,
     where the curve is parallel to a segment (we notice that now we use the lemma for $\eta>k^{-\mu r_1}$). The same is true for
     $\p_{\n_2-\n'_2}$. Since we consider a concave component of a
     curve, it cannot be true.

     To finish the proof of the lemma  we consider two cases. Suppose $q$ in the inequality
     \eqref{q} satisfies the estimate $q>k^{2r_1/3}$. Then, by Lemma \ref{Lattice-3}, the
     number of vectors  $\p_\n$,
     $\||\p_{\n}\||<2k^{r_1}$, satisfying the inequality
     $p_{\n}<\frac{1}{64}k^{-2r_1/3}$ does not exceed
     $2^{12}\cdot k^{2r_1/3}$. Since each of them can be used only once, the
     total number of short segments  does not exceed $2^{12}\cdot k^{2r_1/3}$.

     Let  $q\leq k^{2r_1/3}$.
     If $|\epsilon _q|>\frac{1}{64}q^{-1}k^{-r_1}$. Then, obviously,
     $\frac{1}{64}k^{-2r_1/3}<|\epsilon _q|qk^{r_1/3}$. Applying Lemma
     \ref{Lattice-2}, we obtain that the number of segments with the
     length less than $\frac{1}{64}k^{-2r_1/3}$ is less than
     $k^{2r_1/3}$. Since each of them can be used only once, the
     total number of short segments  does not exceed $k^{2r_1/3}$.
     It remains to consider the case $q\leq
     k^{2r_1/3}$, $|\epsilon _q|\leq \frac{1}{64}q^{-1}k^{-r_1}$. By
     Lemma \ref{Lattice-1}, clusters are well separated. Considering
     that the distance between clusters is greater than
     $\frac{1}{2q}$ and the size of each cluster is less than
     $\frac{1}{8q}$, we obtain that no more than $4q$ clusters
     can intersect $\varepsilon _0$-neighborhood of a concave component of a curve
     $\tau _2=f_i(\tau _1)$.
     The part of the curve inside the clusters has the length $L_{in}$ which is less
     than the double size of a cluster $10|\epsilon _q|k^{r_1}$ multiplied by
     the number of clusters $4q$, i.e., $L_{in}<40|\epsilon
     _q|qk^{r_1}$.  Next, the segments with the length less than
     $\frac{1}{2}k^{-2r_1/3}$ cannot connect different clusters, since the
     distance between clusters is greater than $\frac{1}{2q}\geq
     \frac{1}{2}k^{-2r_1/3}$. Therefore, any segment of the length
     less than $\frac{1}{2}k^{-2r_1/3}$ is inside one
     cluster. If we consider the segments with the length greater
     than $|\epsilon _q|qk^{r_1/3}$, then the number of such segments
     is less than $L_{in}/|\epsilon _q|qk^{r_1/3}$, i.e., it is less
     than $40k^{2r_1/3}$. By Lemma \ref{Lattice-2}, the total
     number of segments of the length less than $|\epsilon
     _q|qk^{r_1/3}$ is less than $k^{2r_1/3}$. Each of them can be used only once. Thus, the total
     number of segments is less than $41k^{2r_1/3}$.

     We proved that the number of segments in $\varepsilon _0$-neighborhood of each concave component
     of a curve $\tau _2=f_i(\tau _1)$ does not exceed
     $2^{13}k^{2r_1/3}$. By Lemma \ref{L:curves-2}, part 3, the number of
     such components does not exceed $2^{31}l^2\cdot k^{8\delta}$. The estimate
     \eqref{Resnumberofpoints} easily follows.
\end{proof}

\subsection{Preparation for  Step III - Analytic Part \label{S:3}}
\subsubsection{Model Operator for Step III\label{MOforStep3}}
Let $r_2>r_1>10^8$. Further we use the notation: \begin{equation}
\label{box} \Omega (r_2)=\{\m:\||\p_\m\||<k^{r_2}\}.
\end{equation}
%%%%%%%%% $c(r_2)$ being a technical constant in the interval
%%%%%%%%%%%%$[1,2]$ to be defined later. We don't put upper bound on $r_2$. It
%%%%%%%%%%5will be done in Step III.
 We repeat for $r_2$ the construction of
the section \ref{MOforStep2} which was done for an arbitrary
$r_1>2$. It is easy to see that the whole construction is monotonous
with respect to $r_1$. Namely,
$$\M(\varphi _0,r_1) \subseteq\M(\varphi _0,r_2), \ \M'(\varphi _0,r_1)\subseteq \M'(\varphi
_0,r_2),\ \M_1(\varphi _0,r_1)\subseteq \M_1(\varphi _0,r_2),$$ $$
\M_2(\varphi _0,r_1)\subseteq \M_2(\varphi _0,r_2),\ \tilde
\M_1(\varphi _0,r_1)\subseteq \tilde \M_1(\varphi _0,r_2),\ \tilde
\M_2(\varphi _0,r_1)\subseteq \tilde \M_2(\varphi _0,r_2).$$
%%%%%%%%In
%%%%%%%5particular, the collection of sets $\M_j^{(2)}(r_1)$,
%%%%%55555$j=1,...,J_0(r_1)$ is a part of the collection $\M_j^{(2)}(r_2)$,
%%%%%%%%%%$j=1,...,J_0(r_1),...,J_0(r_2),$  $J_0(r_1)\leq J_0(r_2)$.
%%%%%%%%Using
%%%%%%%%\eqref{PVP} we readily show:
%%%%%%%%\begin{equation}
%%%%%%%%P(r_2)HP(r_2)=P(r_1)HP(r_1)+\left(P(r_2)-P(r_1)\right)H\left(P(r_2)-P(r_1)\right),
%%%%%%%%\label{Feb8}
%%%%%%%%\end{equation}
%%%%%%%%\begin{equation}P(r_2)-P(r_1)=\sum _{\m \in \M_1(r_2)\setminus
%%%%%%%%\M_1(r_1)}P_{\m}+\sum _{J_0(r_1)+1}^{J_0(r_2)}P_2^j .\label{r1r2}
%%%%%%%%\end{equation}
%%%%%%%%Further we use the notation $P(\m)$ for  $P_{\m}$ if $\m\in
%%%%%%%%\M_1(r_2)$ and for $P_2^j$ if $\m\in \M_1(r_2)$, $j$ being the index
%%%%%%%%of equivalence class containg $\m$.

Let $\varphi_0\in \omega^{(2)} (k, \delta ,1)$. Put
\begin{equation}\label{M^2} {\MM}^{(2)}:={\MM}^{(2)}(\varphi _0, r_2)=\{\m\in \M(\varphi _0, r_2)
:\ \varphi_0\in{\cal O}_\m^{(2)}(10r_1',1)\},\end{equation} where
${\cal O}_\m^{(2)}(10r_1',\tau)$ is the union of the disks of the
radius $\tau k^{-10r_1'}$ with the centers at poles of the resolvent
of $k^\delta$-component containing $\k^{(1)}(\varphi_0)+\p_\m$. More
precisely, for each $\m\in \M(\varphi _0, r_2)$ we construct
$k^\delta$-box around it. We establish $3k^\delta$ equivalence
relation between such boxes. Such components separated by $k^\delta$
from each other we call $k^\delta$-components. Then ${\cal
O}_\m^{(2)}(10r_1',\tau)$ is the union of the disks of the radius
$\tau k^{-10r_1'}$ with the centers at poles of the operator
$(P_{\m}(H(\k^{(1)}(\varphi_0))-k^{2l}I)P_{\m})^{-1}$, where $P_{\m}$ is the
projection onto a particular $k^\delta$-component containing
$\k^{(1)}(\varphi_0)+\p_\m$. We notice (see the proof of Lemma
\ref{2.12} with $3k^\delta$ and $k^{r_2}$ instead of $k^\delta$ and
$k^{r_1}$) that each $k^\delta$-component contains not more than 4
elements $\m\in \M(\varphi _0, r_2)$. For such $\m$ corresponding
sets ${\cal O}_\m^{(2)}(10r_1',\tau)$ are identical. By construction
of the non-resonant set $\omega^{(2)} (k, \delta ,1)$, we have
${\MM}^{(2)}\cap \Omega (r_1)=\emptyset $.

Further we use the property of the set $\MM^{(2)}$ formulated in the
next lemma.

\begin{lemma}\label{L:2/3-1} Let  $\m _0\in \Omega (r_2)$, $1/20<\gamma '<20$ and $\Pi _{\m_0}$ be the $k^{\gamma 'r_1}$-neighborhood (in $\||\cdot\||$-norm) of $\m_0$.
Then, the set $\Pi
_{\m _0}$ contains less than $ck^{2\gamma'r_1/3+1}$ elements of
$\MM^{(2)}$.
\end{lemma}
\begin{proof}
If $\m\in \MM ^{(2)}$, then there is a $\varphi _*: |\varphi
_0-\varphi _*|<k^{-10r_1'}$ such that
\begin{equation}\label{gulf5}\det \Big(P_\m\big(H(\k^{(1)}(\varphi _*))-k^{2l}I\big)P_\m\Big)=0.\end{equation}
Therefore, for some $\varepsilon_0':|\varepsilon _0'|<\varepsilon
_0,\ \varepsilon _0:=ck^{2l-1-10r_1'}$, \begin{equation}
\label{Feb21-1}\det \Big(P_\m\big(H(\k^{(1)}(\varphi
_0))-(k^{2l}+\varepsilon _0')I\big)P_\m\Big)=0.
\end{equation}
Indeed, if \eqref{Feb21-1} holds for no $\varepsilon '$, then $\left\|\Big(P_\m\big(H(\k^{(1)}(\varphi
_0))-k^{2l}I\big)P_\m\Big)^{-1}\right\|<ck^{-2l+1+10r_1'}$, since $\varphi _0$ is real and, hence, $H(\k^{(1)}(\varphi _0)$ is selfadjoint. Using Hilbert identity, we obtain that
$\Big(P_\m\big(H(\k^{(1)}(\varphi
_*))-k^{2l}I\big)P_\m\Big)^{-1}$ is bounded. This contradicts to \eqref{gulf5}. Hence, \eqref{Feb21-1} holds for some $\varepsilon_0':|\varepsilon _0'|<\varepsilon
_0$.

Suppose $\m\in \MM_1(\varphi _0, r_2)$. Then, \eqref{Feb21-1} means
that $|\lambda ^{(1)}(\k^{(1)}(\varphi
_0)+\p_{\m})-k^{2l}|<\varepsilon _0$. Introducing the notation
$\k_0=\k^{(1)}(\varphi _0)+\p_{\m_0}$, we rewrite the last inequality
in the form: $|\lambda ^{(1)}(\k_0+\p_{\m-\m_0})-k^{2l}|<\varepsilon
_0$, where $\||\p_{\m-\m_0}\||<k^{\gamma ' r_1}$. It follows that
$\k_0+\p_{\m-\m_0}$ is in the real $c\varepsilon _0
k^{-2l+1}$-neighborhood of ${\cal D}_1(k^{2l})$.  Applying Lemma
\ref{L:number of points-1}, we obtain that the number of such points does not exceed
$ck^{2\gamma ' r_1/3+1}$.  Let $\m \in \MM _2(\varphi _0,r_2)$.
Namely, let $\m$ belongs to a component $\MM_2^j(\varphi _0, r_2)$.
Then, $\left||\k^{(1)}(\varphi _0)+\p_{\m}|^2_{\R}-k^2\right|<k^{-40
\mu \delta }$ and \eqref{Feb21-1} holds, $P_\m$ being the projection
on $\tilde \MM_2^j(\varphi _0, r_2)$. Using again the notation
$\k_0=\k^{(1)}(\varphi _0)+\p_{\m_0}$, and the definition of
$\MM_2^j(\varphi _0, r_2)$, we obtain:
$\left||\k_0+\p_{\m-\m_0}+\p_\q|^2-k ^2\right|<k^{-40\mu \delta }$ for
all $\q \in \QQ _0$, where $\QQ_0:=\MM _2^j(\varphi _0,r_2)-\m$.
%%%%%%% Suppose
%%%%%%%$\left||\hat \k+\p_\q|^2- k ^2\right|\geq \frac{1}{8}k^{-40\mu
%%%%%%%\delta }$ for all $\q\neq \bf 0$. Then, this case can be treated as
%%%%%%%$\m\in \MM_1 (\varphi _0,r_2)$. Assume that there is a $\q_1,...
%%%%%%%\q_J$, such that $\left||\hat\k+\p_{\m-\m_0}+\p_\q|^2-\tilde k
%%%%%%%%%%%%%%^2\right|<\frac{1}{8}k^{-40\mu \delta }$. By Lemma \ref{.}, $J\leq
%%%%%%%3$. Let $\QQ=\{\q_1,... \q_J\}$.
In terms of Section \ref{Lattice Points in a Resonant Set},
\eqref{Feb21-1} means $\k_0+\p_{\m-\m_0}\in \SS_{\QQ_0}(k, \varepsilon
_0)$, see \eqref{Feb21-1*}. Applying Lemma \ref{L:Resnumberofpoints}
 and using
\eqref{Feb21-1}, we obtain that the number of such points does not
exceed $ck^{2\gamma ' r_1/3+8 \delta}$ for a fixed $\QQ_0$.
Considering that the total number of sets $\QQ_0$ does not exceed
$ck^{12\delta }$, we obtain that the number of points
$\k_0+\p_{\m-\m_0}\in \cup _{\QQ_0}\SS_{\QQ _0}(k, \varepsilon _0)$
does not exceed $ck^{2\gamma ' r_1/3+20 \delta}$. Adding the
estimates for the total number of resonant and non-resonant sets, we
prove it is strictly
less than $C_0k^{2\gamma ' r_1/3+1}$.\\
%%%%%%%Therefore,
%%%%%%%$\MM
%%%%%%% ^{(2)j_2} \subset  \hat \MM _{\m_0}$ and contains less than
%%%%%%% $cC_0^{2/3}k^{2\gamma r_1+3}$ points. This means $\tilde \MM
%%%%%%% ^{(2)j_2}$ contains less than $cC_0^{2/3}k^{3\gamma r_1+3}$ points.
%%%%%%% If $C_0$ is sufficiently large, then this number is strictly less than $C_0k^{3\gamma
%%%%%%% r_1+3}$. Therefore, any connected component is strictly inside the
%%%%%%% block of $\||\cdot \||$-size $C_0k^{3\gamma
%%%%%%% r_1+3}$ and contains less than $C_0k^{2\gamma
%%%%%%% r_1+3}$ points in $\MM^{(2)j_2}$.
\end{proof}

Let us split $k^{r_2}$-box into $k^{\gamma r_1}$-boxes as described
below. In the whole construction below we will have $\gamma
=\frac{1}{5}$, but in some cases we will refer to the similar
estimates with other values of $\gamma$. That's why in what follows
we prefer to use implicit notation.  The procedure consists of
several steps. On each step we introduce a new scale of a box.
Further structure will acquire additional scales at each step of
approximation procedure. This is why we call the procedure
Multiscale Construction in the Space of Momenta.
\begin{enumerate} \item {\em Simple region.} \label{simple} Let $\Omega _s^{(2)}(r_2)$ be the  collection of $\m\in \Omega(r_2)$ with small values of $p_\m$, namely,
$\Omega _s^{(2)}(r_2)=\{\m\in \Omega(r_2):0<p_\m\leq k^{- 5r_1'}\}$.
It is easy to see  that $\Omega _s^{(2)}(r_2)\subset \MM (\varphi _0,r_2)$, since $p_{\m}$ is small, see \eqref{M}, \eqref{resonance1}. Next, if  $\m\in \Omega_s^{(2)}(r_2)$, then there are
 no other elements of $\MM (\varphi _0,r_2)$ in
 the $k^{\delta }$-box  around  $\m$. Indeed, let $\k=\k^{(1)}(\varphi _0)+\p_\m$. It is  is a small perturbation of
$\k^{(1)}(\varphi _0)$, hence it satisfies $\left||\k +\p_{\n}|^2-
|\k|^2\right|> \frac \tau 2 k^{-40\mu \delta }(1+o(1))$ when
$0<\||\p_\n\||<k^{\delta }$, see \eqref{jan28b}. This means $\m +\n
\not \in \MM (\varphi _0,r_2)$. Further, if $\m\in
\Omega_s^{(2)}(r_2)$, then there are no other  elements  of $\Omega
_s^{(2)}(r_2)$ in the surrounding box of the size $k^{r_1}$, see
\eqref{below}. Last, $\m$ itself can belong or do not belong to
$\MM^{(2)}$,  there are
 no other elements of
$\MM^{(2)}$ in the $k^{r_1}$-box  around such $\m$. Indeed,
$\k^{(1)}(\varphi _0)$ satisfies the conditions of Lemma
\ref{L:geometric2}.  This means that the $k^{\delta }$-cluster
around each $\q$: $0<\||\p_\q\||<k^{r_1}$ is non-resonant. Moreover,
the $k^{\delta }$-box around each $\m+\q$: $0<\||\p_\q\||<k^{r_1}$
is non-resonant too, since $p_{\m}$ is sufficiently small. This
means $\m+\q \not \in \MM^{(2)}$.

For each $\m \in \Omega _s^{(2)}(r_2)$  we consider its $k^{r_1/2}$-neighborhood. The union of
such boxes we call the simple region and denote it by $\Pi _s(r_2)$. The
corresponding projection is $P_s$. Note, that the distance from the simple region to the nearest point of $\MM^{(2)}$ is greater than $\frac12k^{r_1}$.

\item {\em Black region.} Next, we split $\Omega (r_2)\setminus \left(\Omega (r_1)\cup \Pi _s\right)$ into boxes of the size $k^{\gamma r_1}$. All elements $\m \in \MM^{(2)}$ there satisfy
$p_\m>k^{-5 r_1'}$. We call a box black, if  together with its
neighbors it contains more than $k^{\gamma r_1/2+\delta _0r_1}$
elements of $\MM ^{(2)}$, $\delta_0=\gamma /100$ (in particular
$\delta_0r_1>100$). Let us consider all "black" boxes together with
their $k^{\gamma r_1+\delta _0r_1}$-neighborhoods. We call this the
black region.  Note that that the size of the neighborhoods involved is much smaller than the size of the neighborhoods $k^{r_1/2}$ for the simple region, since $\gamma +\delta _0<\frac12$.
The estimates for the size of the black region will be proven in Lemma \ref{L:black}. We denote the black  region by $\Pi _b$. The
corresponding projector is $P_{b}$. Obviously the distance between black and simple regions is greater than $\frac12k^{r_1}$.

\item {\em Grey region.} By a white box we mean a
$k^{\gamma r_1}$-box, which together with its neighbors contains no
more than $k^{\gamma r_1/2+\delta_0r_1}$ elements of $\MM ^{(2)}$.
Every white box we split into "small" boxes of the size $k^{\gamma
r_1/2+2\delta_0r_1}$. We call a small box "grey", if together with
its neighbors it contains more than $k^{\gamma r_1/6-\delta_0r_1}$
elements of $\MM ^{(2)}$.  The grey region is the union of all grey
small boxes together with their $k^{\gamma
r_1/2+2\delta_0r_1}$-neighborhoods.  Note that that the size of the
neighborhoods involved is much smaller than the size of the
neighborhoods  the simple and black regions. The estimates for the
size of the grey region will be proven in Lemma \ref{L:grey}. The
notation for this region is $\Pi _g$. The corresponding projector is
$P_{g}$. The part of the grey region, which is outside the black
region, we denote by $\Pi _g'$ and the corresponding projection by
$P_g'$.   Obviously, the distance between grey and simple regions is
greater than $\frac12k^{r_1}$.

\item {\em White region.} By a white small box we mean a small box, which
together with its neighbors has no more than $k^{\gamma
r_1/6-\delta_0r_1}$ elements of $\MM ^{(2)}$.  In each small white
box we consider $k^{\gamma r_1/6}$-boxes around each point of
$\MM^{(2)}$. The union of such $k^{\gamma r_1/6}$-boxes we call the
white region and denote it by $\Pi _w$. The corresponding projection
is $P_w$.  Note that the size of the neighborhoods involved is much
smaller than the size of the neighborhoods  the simple, black and
grey regions. The estimates for the size of the white region will be
proven in Lemma \ref{L:white}. The part of the white region which is
outside the black and grey regions, we denote $\Pi _w'$ and the
corresponding projection by $P_w'$. Obviously, the distance between
grey and simple regions is greater than $\frac12k^{r_1}$.

\item {\em Non-resonant region.} We also consider $\frac13 k^{\delta }$-neighborhoods of all points in the set
$\MM(r_2, \varphi _0)\setminus \left(\MM(r_1, \varphi _0)\cup
\MM^{(2)}\cup \Omega _s^{(2)}(r_2)\right)$. The union of this neighborhoods we call
the non-resonance region $\Pi _{nr}$. The corresponding projection
is $P_{nr}$.  The part of the non-resonant region which is outside
$\Pi_s\cup\Pi_b\cup\Pi_g\cup\Pi_w$, we denote $\Pi _{nr}'$ and the
corresponding projection by $P_{nr}'$.

\end{enumerate}

Let
$$P_r:=P_s+P_b+P_g'+P_w',\ \ \ P^{(2)}:=P_r+P_{nr}'+P(r_1)$$
index $r$ standing for "resonant".

%%%%%The model operator $\tilde H^{(2)}$ for the third step we define by
%%%%%the formula:
%%%%%\begin{equation}
%%%%%\begin{split}&
%%%%%\tilde H^{(2)}=H^{(2)}+(P(r_2)-P(r_1)-P^{(2)})H_0+\hat{H}^{(2)}:=\cr
%%%%%& H^{(2)}+(P(r_2)-P(r_1)-P^{(2)})H_0+
%%%%%P_sHP_s+P_bHP_b+P_g'HP_g'+P_w'HP_w'+P_{nr}'HP_{nr}'.
%%%%%\label{MO3}\end{split}\end{equation} Obviously, it has a block
%%%%%structure. Further we prove some properties of the operator $\tilde
%%%%%H^{(2)}$ which we need for constructing the third non-resonance set.

First, we establish $3k^{\gamma r_1+\delta _0r_1}$-equivalence
relation between black boxes. Then the set $\Pi _b$ can be
represented as the union of components (clusters) separated by distance no less
than $k^{\gamma r_1+\delta _0r_1}$. We denote  such a component by
$\Pi _b^j$.
\begin{lemma} \label{L:black} \begin{enumerate} \item
Each $\Pi _b^j$ contains no more than $ck^{\gamma r_1/2-\delta
_0r_1+3}$ black boxes.
\item The size of $\Pi _b^j$ in $\||\cdot \||$ norm is less than
$ck^{3\gamma r_1/2+3}$.
\item Each $\Pi _b^j$ contains no more than $ck^{\gamma r_1+3}$ elements of
$\MM^{(2)}$. Moreover, any box of $\||\cdot \||$-size $ck^{3\gamma r_1/2+3}$ containing $\Pi _b^j$ has no more than $ck^{\gamma r_1+3}$ elements of
$\MM^{(2)}$ inside.
\end{enumerate}
\end{lemma}
\begin{proof} Let $n_b$ be the number of black boxes in $\Pi
_b^j$,
 $L_b$ be the size of $\Pi _b^j$ and $N_b$ the number of elements of
$\MM^{(2)}$ in $\Pi _b^j$.  Obviously, $L_b<n_b3k^{\gamma r_1+\delta
_0r_1}$ and $N_b>cn_bk^{\gamma r_1/2+\delta _0r_1}$. By Lemma
\ref{L:2/3-1}, $N_b<cL_b^{2/3}k$. Solving the last three
inequalities for $n_b$, we get $n_b<ck^{\gamma r_1/2-\delta
_0r_1+3}$. It follows $L_b<ck^{3\gamma r_1/2+3}$. Next, we consider
a box of the size $k^{3\gamma r_1/2+3}$, containing $\Pi _b^j$.
Using again  Lemma \ref{L:2/3-1}, we obtain that the number of
elements of $\MM^{(2)}$ in this box is less than $cL_b^{2/3}k$.
Therefore,  $N_b<ck^{\gamma r_1+3}$.
\end{proof}

Second, we establish $3k^{\gamma r_1/2+2\delta _0r_1}$-equivalence
relation between small grey boxes. Then the set $\Pi _g$ can be
represented as the union of components separated by distance no less
than $k^{\gamma r_1/2+2\delta _0r_1}$. We denote each such component
as $\Pi _g^j$.
\begin{lemma}\label{L:grey} \begin{enumerate} \item
Each $\Pi _g^j$ contains no more than $ck^{\gamma r_1/3+2\delta
_0r_1}$ grey boxes.
\item The size of $\Pi _g^j$ in $\||\cdot \||$ norm is less than
$ck^{5\gamma r_1/6+4\delta _0r_1}$.
\item Each $\Pi _g^j$ contains no more than $k^{\gamma r_1/2+\delta _0r_1}$ elements of
$\MM^{(2)}$.\end{enumerate}
\end{lemma}
\begin{proof} Let us consider a part of $\Pi _g^j$
belonging to one "big" white box. Let $n_g$ be the number of grey
boxes in $\Pi _g^j$.
 Let $L_g$ be the size of $\Pi _g^j$ and $N_g$ be  the number of elements of
$\MM^{(2)}$ in $\Pi _g^j$.  Obviously, $N_g>cn_gk^{\gamma r_1
/6-\delta _0r_1}$. By definition of a big white box $N_g<k^{\gamma
r_1/2+\delta _0r_1}$. Therefore, $n_g<ck^{\gamma r_1/3+2\delta
_0r_1}$. Clearly, $L_g<n_g3k^{\gamma r_1/2+2\delta
_0r_1}<ck^{5\gamma r_1/6+4\delta _0r_1}$. Since $\delta  _0<\gamma
/24$, we obtain that the size of each grey component is much less
than the size $k^{\gamma r_1}$ of a big box. The lemma is proven under condition that
$\Pi _g^j$ is inside one of white boxes. Suppose $\Pi _g^j$ intersects more
than one white box. Considering that the size of $\Pi _g^j$ in each
big white box is much less than the size of this box, we conclude that $\Pi
_g^j$
 fits into neighboring boxes and satisfies the estimates proven
above.
\end{proof}

Third, we consider points of $\MM^{(2)}$ in small white boxes. We
establish $3k^{\gamma r_1/6}$-equivalence relation between them.
Considering $k^{\gamma r_1/6}$-neighborhoods of the points in
$\MM^{(2)}$, we see that this neighborhoods form clusters $\Pi _w^j$
of $\Pi _w $ separated by the distance no less  than $k^{\gamma
r_1/6}$. The number of $\MM^{(2)}$ points in a white cluster we
denote by $N_w^j$.
\begin{lemma}\label{L:white} \begin{enumerate} \item The size of $\Pi _w^j$ in $\||\cdot \||$ norm is less than
$ck^{\gamma r_1/3-\delta _0r_1}$.
\item Each $\Pi _w^j$ contains no more
than $k^{\gamma r_1/6-\delta _0r_1}$ points of $\MM^{(2)}$.
\end{enumerate}
\end{lemma}
\begin{proof} Let us consider points of $\MM^{(2)}$ in a small white
box. By the definition of the white small box, the number of such
points does not exceed $k^{\gamma r_1/6-\delta _0r_1}$. We consider
the $k^{\gamma r_1/6}$-neighborhoods of these points. They can form
clusters. The total contribution from all points of $\MM^{(2)}$ in
the small white box and its neighbors, obviously, does not exceed
$3k^{\gamma r_1/3-\delta _0r_1}$, which is much less than the size
of a small white box. Therefore, each $\Pi _w^j$ can't spread
outside of the small white box and its neighbors. This proves both
statements of the lemma.
\end{proof}

\hskip 1cm At last, we also establish $k^\delta$-equivalence
relation between all points in $\MM(r_2, \varphi _0)\setminus
\left(\MM(r_1, \varphi _0)\cup
\MM^{(2)}\cup\Omega^{(2)}_s(r_2)\right)$.  Then the set $\Pi _{nr}$
can be represented as the union of components (clusters) separated
by distance no less than $\frac13 k^{\delta}$. We denote such a
component by $\Pi _{nr}^j$.

Next, we slightly change definitions of the black, grey and white
areas to adjust their boundary to the structure of clusters. Namely,
if $k^{\delta }$-cluster containing points of $\MM(\varphi
_0,r_2)\setminus \MM^{(2)}$  has a distance less than $k^{\delta }$ to a white, grey or black
area, then we include it into the area with the lightest color.
% (it can be one of those only, since the clusters of the same color are well separated) .
 This ``addition" does not change
formulation of Lemmas \ref{L:black}, \ref{L:grey}, \ref{L:white},
since the size of a $k^{\delta }$-cluster is much smaller that the
sizes of  $\Pi _b^j$, $\Pi _g^j$, $\Pi _w^j$. If a white cluster has a distance less than $k^{\gamma r_1 /6}$ to the grey or black area, we include it into that with the lighter color.
This ``addition" also does not change formulation of Lemmas
\ref{L:black}, \ref{L:grey},  since the size of a white cluster is
much smaller that the sizes of $\Pi _b^j$, $\Pi _g^j$.  If a grey
cluster has a distance less than $k^{\gamma r_1/2+2\delta _0r_1}$  to the  black area, we include it into one of these  areas.
This ``addition" does not change formulation of Lemma \ref{L:black},
since the size of a grey cluster is much smaller that the size of
any $\Pi _b^j$. The new structure has the following properties: if
the intersection of the $k^{\delta }$-neighborhood of  a $k^{\delta }$-cluster with white, grey or
black area is not empty, then this cluster is completely in this
area. If the intersection of the $k^{\gamma r_1/6}$-neighborhood of a white cluster with  grey or black
area is not empty, then this cluster is completely in this area. If
  the intersection of the $k^{\gamma r_1/2+2\delta _0r_1}$-neighborhood of a grey cluster with  the black area is
not empty, then this cluster is completely in this area.
%%%%%%%%%%It follows
%%%%%that projections $P_s$, $P_b$, $P_g'$, $P_w'$ and $P_{nr}'$ are
%%%%%mutually orthogonal. Note that $\gamma r_1+\delta _0r_1<r_1$, since
$\delta _0<\gamma /24$, $\gamma <1/3$. This means that each
component of the white, grey, black and non-resonance region is much
smaller in $\||\cdot \||$-size than $\Omega (r_1)$. Moreover, there
are no points of $\MM^{(2)}$ inside $\Omega (r_1)$. If the
$k^{\gamma r_1/6}$-neighborhood of a white cluster intersect $\Omega
(r_1)$, we reduce $\Omega (r_1)$ by this neighborhood. This
insignificant reduction does not change Step II. We make a similar
reduction of $\Omega (r_1)$ if it is intersected by neighborhoods of
grey or black clusters.
%%%Thus, the projection $P(r_1)$ is orthogonal to $P_s$, $P_b$, $P_g'$,
%%%%%$P_w'$ and $P_{nr}'$.
Sometimes it will be convenient to numerate the projections
$P(r_1)$, $P_b$, $P_g'$, $P_w'$, $P_{nr}'$, $P_s$ by indices
0,1,2,3,4,5 as $P_0,P_1,P_2,P_3,P_4,P_5$. The corresponding sets are
$\Pi _i$. Note that each $\Pi _i$ consists of the components
$\Pi_{ij}$, $j=1, ...,J(i)$ as described in the construction of the
sets $\Pi_i$. The distance between closest components $\Pi _{ij}$,
$i=1,2,3,4,5$ with the same first index is greater than $k^{\gamma r_1+\delta_0 r_1}$,  $k^{\gamma r_1/2+2\delta _0 r_1}$, $k^{\gamma r_1/6}$, $\frac13 k^{\delta }$, $k^{\gamma r_1}$, correspondingly. Then we can rewrite
\begin{equation}P^{(2)}=\sum _{i=0}^{5}P_i.\label{P(2)}
\end{equation}
We introduce the boundaries $\partial \Omega (r_1)$, $\partial \Pi
_b$, $\partial \Pi _g'$, $\partial \Pi _w'$, $\partial \Pi _{nr}'$,
$\partial \Pi _s$
%%%%%%%%%and the interiors $\Omega
%%%%%%%%^{(int)} (r_1)$ $ \Pi _b^{(int)}$, $ \Pi _g^{'(int)}$, $ \Pi
%%%%%%%%_w^{'(int)}$
of the sets $\Omega (r_1)$, $\Pi_b$, $\Pi_g'$, $\Pi_w'$,
$\Pi_{nr}'$, $\Pi_s$ as follows: $\partial \Omega (r_1)$, $\partial
\Pi _b$, $\partial \Pi _g'$, $\partial \Pi _w'$, $\partial \Pi
_{nr}'$, $\partial \Pi _s$ are the sets of points in $\Omega (r_1)$,
$\Pi_b$, $\Pi_g'$, $\Pi_w'$, $\Pi _{nr}'$, $\Pi_s$ whose $\||\cdot
\||$-distance to the complements of $\Omega (r_1)$, $\Pi _b$, $\Pi
_g'$, $\Pi _w'$, $\Pi _{nr}'$, $\Pi _s$, respectively, is less than
$\frac{1}{3}k^{\delta}$.
%%%%%%%% The interiors $\Omega ^{(int)}(r_1)$, $ \Pi
%%%%%%%%_b^{(int)}$, $ \Pi _g^{'(int)}$, $ \Pi _w^{'(int)}$, $\Pi
%%%%%%%%_{nr}{'(int)}$ are the complements of the boundaries in $\Omega
%%%%%%%%(r_1)$, $\Pi_b$, $\Pi_g'$, $\Pi_w'$, $\Pi _{nr}$ .
 The corresponding
projectors we denote as $P^{\partial }(r_1)$, $P_b^{\partial}$,
$P_g^{'\partial}$, $P_w^{'\partial}$, $P_{nr}^{'\partial}$,
$P_s^{\partial}$ or $P_i^{\partial}$, $i=0,1,2,3,4,5$.

\begin{lemma} \label{L:boundary} Let $i,i'=0,1,2,3,4,5$, $i\neq i'$. The following relations hold:
\begin{equation}
P_iP_{i'}=0, \label{ort}
\end{equation}
\begin{equation} \label{PVP-2} P_iVP_{i'}=0,
%%%%P_iVP_{i'}^{\partial}=P_{i}^{\partial}VP_{i'}=P_{i}^{\partial}VP_{i'}^{\partial}
\end{equation}
\begin{equation}
(I-P^{(2)})VP_i=(I-P^{(2)})VP_{i}^{\partial}. \label{boundary}
\end{equation}\end{lemma}
\begin{corollary} \label{C:PHP-2}Operators $P^{(2)}VP^{(2)}$ and $P^{(2)}HP^{(2)}$ have a block structure. Namely,  \begin{equation}
P^{(2)}VP^{(2)}=\sum _{i=0}^5 P_iVP_i, \   \  \ P^{(2)}HP^{(2)}=\sum
_{i=0}^5 P_iHP_i. \label{PHP-2}
\end{equation}
\begin{equation}
P_iVP_i=\sum _{j}P_{ij}VP_{ij},\  \  \ P_iHP_i=\sum _{j}P_{ij}HP_{ij}, \label{blocks}
\end{equation}\end{corollary}
The lemma easily follows from the construction of the projectors and the fact that $V_{jj'}=0$ when
$|j-j'|\geq \frac{1}{3}k^{\delta}$ .

{\it Remark.} Thus, we have constructed a multiscale structure
inside $P^{(2)}HP^{(2)}$,  blocks of different colors having
distinctly different size. Merging blocks of a smaller size
 (a lighter color) with neighboring blocks of a bigger size (a darker color), we made the blocks to be
 separated by the $\|||\cdot \||$ distance greater then $k^{\delta }$.

\begin{lemma}\label{L:Pnr}Let $\varphi _0\in \omega^{(2)}(k,\delta ,\tau )$,
$|\varphi-\varphi _0|<k^{-2l-40r_1'-\delta }$. Then,
\begin{equation}\label{Pnr}
\left\|\Bigl(P_{nr}\bigl(H(\k^{(2)}(\varphi
))-k^{2l}I\bigr)P_{nr}\Bigr)^{-1}\right\|<ck^{40r_1'}.
\end{equation} \end{lemma}
\begin{proof} The set $\Pi _{nr}$ can be presented as $\cup _{j} \Pi _{nr}
^j$,  each $\Pi _{nr}
^j$ being a $k^{\delta }$-cluster, and the distances between sets $\Pi _{nr} ^j$ is no less than
$\frac13 k^{\delta }$. Therefore, $P_{nr}=\sum _j
P_{nr}^{j}$, where $P_{nr}^{j}$ are projections corresponding
to $\Pi _{nr} ^j$, and $P_{nr}HP_{nr}=\sum
_jP_{nr}^{j}HP_{nr}^{j}$. Hence, it is enough to prove
\begin{equation}\label{Pnrj}
\left\|\Bigl(P_{nr}^j\bigl(H(\k^{(2)}(\varphi
))-k^{2l}I\bigr)P_{nr}^{j}\Bigr)^{-1}\right\|<ck^{40r_1'}.
\end{equation}
It follows from Lemma \ref{L:estnonres1}. Indeed,
by construction, each $\Pi _{nr} ^j$ contains $\m \in \M(\varphi , r_2)$, $\M(\varphi , r_2)=\M _1(\varphi , r_2)\cup \M _2(\varphi , r_2)$. We can apply Lemma \ref{L:estnonres1}, since it is proven for any $r_1$ (no restrictions from above). We take $\varepsilon _0=k^{-10r_1'}$ in Lemma \ref{L:estnonres1}, since the distance from $\varphi _0$ to the
nearest pole of the operator
$\Bigl(P_{nr}^{j}\bigl(H(\k^{(1)}(\varphi _0
))-k^{2l}I\bigr)P_{nr}^{j}\Bigr)^{-1}$ is greater than
$k^{-10r_1'}$.
By analogy with
Corollary \ref{estnonres1}, we obtain (recall that now
$p_\m>k^{-5r_1'}$)
\begin{equation}\label{Pnrj-1}
\left\|\Bigl(P_{nr}^{j}\bigl(H(\k^{(1)}(\varphi _0
))-k^{2l}I\bigr)P_{nr}^{j}\Bigr)^{-1}\right\|<ck^{40r_1'}.
\end{equation} Taking into account that
$\varkappa^{(2)}(\varphi _0)-\varkappa^{(1)}(\varphi
_0)=o(k^{-2l-40\mu r_1'})$ and $\k^{(2)}(\varphi )-\k^{(2)}(\varphi
_0)=o(k^{-2l+1-40r_1'})$, we arrive at \eqref{Pnrj}.
\end{proof}
%%%%%%%%Let us recall that $\MM^{(2)j_2}$, $j_2=1,...,J_2$ are $k^{\gamma
%%%%%%%%r_1}$-equivalence classes of $\MM^{(2)}$ and $\tilde \MM^{(2)j_2}$
%%%%%%%%are their $\frac{1}{3}k^{\gamma r_1}$ neighborhoods. The projections
%%%%%%%%corresponding to neighborhoods we denote by $P_r^{(2)j_2}$. Clearly,
%%%%%%%%$P^{(2)}=\sum _{j_2}P_r^{(2)j_2}$ and $P_r^{(2)}HP_r^{(2)}=\sum
%%%%%%%%_{j_2}P_r^{(2)j_2}HP_r^{(2)j_2}.$
%%%%%%%\begin{lemma}\label{L:Pr} Let $\varphi _0\in \omega _2(k,\delta ,\tau )$,
%%%%%%%%$P_{r}^{(2)j_2}=P_{r}^{(2)j_2}(\varphi _0, r_2)$, $j_2=1,...,J_2$,
%%%%%%%and $|\varphi-\varphi _0|<k^{-11r_1-\delta}$. Then,
%%%%%%%\begin{equation}\label{Pr-1}
%%%%%%%\left\|\Bigl(P_{r}^{(2)j_2}\bigl(H(\k^{(2)}(\varphi
%%%%%%%))-k^2I\bigr)P_{r}^{(2)j_2}\Bigr)^{-1}\right\|<k^{44\mu
%%%%%%%r_1}\left(\frac{k^{2\gamma r_1-11r_1}}{\varepsilon
%%%%%%%_0}\right)^{N_{j_2}}.
%%%%%%%\end{equation}
%%%%%%%\begin{equation}\label{Pr-2}
%%%%%%%\left\|\Bigl(P_{r}^{(2)j_2}\bigl(H(\k^{(2)}(\varphi
%%%%%%%))-k^2I\bigr)P_{r}^{(2)j_2}\Bigr)^{-1}\right\|_1<k^{44\mu
%%%%%%%r_1+4\gamma r_1}\left(\frac{k^{2\gamma r_1-11r_1}}{\varepsilon
%%%%%%%_0}\right)^{N_{j_2}}.
%%%%%%%\end{equation}
%%%%%%%where $N_{j_2}$ is the number of elements in the set $\MM^{(2)j_2}$,
%%%%%%%$N_{j_2}<c_1k^{2\gamma r_1}$ and  $\varepsilon _0$ is the distance
%%%%%%%to the nearest pole of the resolvent in ${\cal W}^{(2)}$.
%%%%%%%\end{lemma}
\begin{lemma}\label{L:Pr} Let $\varphi _0\in \omega^{(2)}(k,\delta ,\tau )$,
%%%%%%%%$P_{r}^{(2)j_2}=P_{r}^{(2)j_2}(\varphi _0, r_2)$, $j_2=1,...,J_2$,
and $|\varphi-\varphi _0|<k^{-44r_1'-2l-\delta}$, i=1,2,3. Then, \begin{enumerate}
\item The number of poles of the resolvent $\Bigl(P_i\bigl(H(\k^{(2)}(\varphi
))-k^{2l}I\bigr)P_i\Bigr)^{-1}$ in the disc  $|\varphi-\varphi _0|<k^{-44r_1'-2l-\delta}$ is no greater than $N_i^{(1)}$, where $N_1^{(1)}=k^{\gamma r_1+3}$, $N_2^{(1)}=k^{\gamma r_1/2+\delta _0r_1}$,
$N_3^{(1)}=k^{\gamma r_1/6-\delta _0r_1}$.
\item Let  $\varepsilon$ be the
distance to the nearest pole of the resolvent in ${\cal W}^{(2)}$
and let $\varepsilon_0:=\min\{\varepsilon,\,k^{-11r_1'}\}$. Then,
the following estimates hold:
\begin{equation}\label{Pr-1}
\left\|\Bigl(P_i\bigl(H(\k^{(2)}(\varphi
))-k^{2l}I\bigr)P_i\Bigr)^{-1}\right\|<ck^{44r_1'}\left(\frac{k^{-11r_1'}}{\varepsilon
_0}\right)^{N_{i}^{(1)}},
\end{equation}
\begin{equation}\label{Pr-2}
\left\|\Bigl(P_i\bigl(H(\k^{(2)}(\varphi
))-k^{2l}I\bigr)P_i\Bigr)^{-1}\right\|_1<ck^{44r_1'+8\gamma
r_1}\left(\frac{k^{-11r_1'}}{\varepsilon _0}\right)^{N_{i }^{(1)}}.
\end{equation}
\end{enumerate}
\end{lemma}
\begin{proof}  Let $\Pi $ be a component  $\Pi _b^j$, $\Pi _g^j$ or $\Pi _w^j$.
By Lemmas \ref{L:black}, \ref{L:grey}, \ref{L:white} the number $N$
of elements  $\MM^{(2)}\cap \Pi$ does not exceed $ck^{\gamma
r_1+3}$.  Let us recall that the set ${\MM}^{(2)}$ is defined by the
formula \eqref{M^2}, where ${\cal O}_\m^{(2)}$ is the union of the
disks of the radius $k^{-10r_1'}$ with the centers at poles of the
resolvent of $k^\delta$-components containing
$\k^{(1)}(\varphi_0)+\p_\m$. Let us consider ${\cal O}^{(2)}_{\Pi
}=\cup_{\m\in \Pi \cap {\MM}^{(2)}}{\cal O}_\m^{(2)}$ and an
analogous set consisting of smaller discs: $\tilde{\cal
O}^{(2)}_{\Pi}=\cup_{\m\in \Pi \cap {\MM}^{(2)}}\tilde{\cal
O}_\m^{(2)}$, where $\tilde{\cal O}_\m^{(2)}$ have the radius
$k^{-11r_1'}$. Since $N<ck^{\gamma r_1+3}$, the total size of
$\tilde{\cal O}^{(2)}_{\Pi}$ is less than $k^{-11r_1'+\gamma
r_1+3}=o(k^{-10r_1'})$.

First, assume  $\varphi _0\not \in \tilde{\cal O}^{(2)}_{\Pi}$.
Then, we can apply Lemma \ref{L:estnonres1} and Corollary
\ref{estnonres1}. Indeed, let us consider a $k^{\delta }$-component
in $\Pi$. We denote it by $\tilde \MM (\m)$ and the corresponding
projection $P_\m$. By the definitions of $\OO _{\m}^{(2)}$, $\tilde
\OO _{\m}^{(2)}$, the distance from $\varphi _0$ to the nearest pole
of $\left(P_\m(H(\k^{(1)})-k^{2l}I)P_\m\right)^{-1}$ is greater than
$k^{-11r_1'}$. Applying Lemma \ref{L:estnonres1} to these
resolvents, we obtain (recall that now $p_\m>k^{-10\mu r_1}$):
\begin{equation} \label{March3-1a}\left\|\left(P_\m(H(\k
^{(1)}(\varphi
_0))-k^{2l}I)P_\m\right)^{-1}\right\|<ck^{44r_1'},\end{equation}
\begin{equation}
\label{March3-1b}\left\|\left(P_\m(H(\k ^{(1)}(\varphi
_0))-k^{2l}I)P_\m\right)^{-1}\right\|_1<ck^{44r_1'+4\delta}.\end{equation}
By analogy with Corollary \ref{estnonres1},
$$\left\|\left(P(H(\k^{(1)}(\varphi _0))-k^{2l}I)P\right)^{-1}\right\|<ck^{44r_1'},$$
$$\left\|\left(P(H(\k^{(1)}(\varphi _0))-k^{2l}I)P\right)^{-1}\right\|_1<ck^{44r_1'}L^4,$$ where
$P$ is the projection for all $k^{\delta }$ components in $\Pi $,
$L$ is the size of $\Pi $.  Arguing as in the proof of Theorem
\ref{Thm2}, we show that the perturbation series for the resolvent
$\left(P_{\Pi }(H(\k^{(1)}(\varphi _0))-k^{2l}I)P_{\Pi
}\right)^{-1}$ converges when we take $PH(\k^{(1)}(\varphi
_0))P+(P_{\Pi }-P)H_0$ as the unperturbed operator. Therefore,
$$\left\|\left(P_{\Pi }(H(\k^{(1)}(\varphi
_0))-k^{2l}I)P_{\Pi }\right)^{-1}\right\|<ck^{44 r_1'},$$ no poles being inside of the disc. Taking into
account that $|\varphi -\varphi _0|<k^{-44 r_1'-2l-\delta}$ and
$|\k^{(2)}-\k^{(1)}|=o(k^{-44 r_1'-2l})$,  we obtain
$$\left\|\left(P_{\Pi }(H(\k^{(2)}(\varphi
))-k^{2l}I)P_{\Pi }\right)^{-1}\right\|<ck^{44 r_1'}.$$ By Lemmas
\ref{L:black}, \ref{L:grey}, \ref{L:white}, $L<k^{2\gamma r_1}$,
$$\left\|\left(P_{\Pi }(H(\k^{(2)}(\varphi
))-k^{2l}I)P_{\Pi }\right)^{-1}\right\|_1<ck^{8\gamma r_1+44r_1'}.$$
Thus, the resolvent has no poles inside the disk around $\varphi _0$ and the estimates \eqref{Pr-1}, \eqref{Pr-2} hold with $\varepsilon_0:=k^{-11r_1'}$.
Second, if $\varphi \not \in \tilde{\OO }_{\Pi }^{(2)}$, then $\varphi
_0 \not \in \tilde{\OO}_{\Pi }^{(2)}(11r_1',\frac{1}{2})$. Therefore
the estimate similar to the last two  hold. Now estimates
\eqref{Pr-1}, \eqref{Pr-2} easily follow.

It remains to consider the case $\varphi _0, \varphi \in \tilde{\cal
O}_{\Pi }^{(2)}$.   Obviously, $\varphi _0, \varphi $ belong to the same connected component of $\tilde{\cal
O}_{\Pi }^{(2)}$ or to different components being at the distance less than $k^{-44r_1'-2l-\delta}$ from each other.  We consider a $\varphi _*\in \partial \tilde{\cal O}_{\Pi
}^{(2)}$, where $\partial \tilde{\cal O}_{\Pi
}^{(2)}$ is the boundary of the component(s) containing  $\varphi _0, \varphi $.
Note that  $\varphi _*\not \in \OO _{\m}^{(2)}(11r_1', 1)$ for all $\m \in \Pi$. Indeed, for $\m \in  \MM^{(2)}$, it follows from the relations $\varphi _*\in \partial \tilde{\cal O}_{\Pi }^{(2)}$ and the definitions of
${\cal
O}_{\Pi }^{(2)}, \tilde {\cal
O}_{\Pi }^{(2)}$.  If $\m\in \Pi \setminus \MM^{(2)}$, then $\varphi _0$ is not in $\OO _{\m}^{(2)}(10r_1', 1)$ by the definition of $\MM^{(2)}$.  Since  $\varphi _0, \varphi \in \tilde{\cal
O}_{\Pi }^{(2)}$ and
 the length of $\tilde{\cal O}_{\Pi }^{(2)}$ is $o(k^{-10r_1'})$,
we have $\varphi _*\not \in \OO _{\m}^{(2)}(10r_1', \frac{1}{2})$.   Now, considering as in the case
$\varphi _0\not \in \tilde{\cal O}_{\Pi }^{(2)}$, we obtain that the
perturbation series for the resolvent $\left(P_{\Pi
}(H(\k^{(1)}(\varphi _*))-k^{2l}I)P_{\Pi }\right)^{-1}$ converges
when we take $PH(\k^{(1)}(\varphi _*))P+(P_{\Pi }-P)H_0$ as the
unperturbed operator. Therefore,
$$\left\|\left(P_{\Pi }(H(\k^{(2)}(\varphi
_*))-k^{2l}I)P_{\Pi }\right)^{-1}\right\|<ck^{44 r_1'}.$$ The number
of poles of the resolvent $\left(P_{\Pi }(H(\k^{(2)}(\varphi
))-k^{2l}I)P_{\Pi }\right)^{-1}$ in $\tilde \OO_{\Pi }^{(2)}$ is the same
as the number of poles of the resolvent of unperturbed operator.
Hence, it is $N$. Using the Maximum principle, we get \eqref{Pr-1}
for the case $\varepsilon \leq k^{-11r_1'}$, where $N_i=N$ and
depends on color of $\Pi $. Considering that the dimension of
$P_{\Pi }$  does not exceed $k^{8\gamma r_1}$, we obtain
\eqref{Pr-2}
\end{proof}

At last, let $\Pi_s^j$ be a particular $k^{r_1/2}$-box around $\m$
containing. Let $P_s^j$ be corresponding projection.
\begin{lemma}\label{L:Ps} Let $\varphi _0\in \omega^{(2)}(k,\delta ,\tau )$. Then, the operator
$\Bigl(P_s^j\bigl(H(\k^{(2)}(\varphi
))-k^{2l}I\bigr)P_s^j\Bigr)^{-1}$ has no more than one pole in the
disk $|\varphi-\varphi _0|<k^{-r_1'-\delta}$. Moreover,
\begin{equation}\label{Ps-1}
\left\|\Bigl(P_s^j\bigl(H(\k^{(2)}(\varphi
))-k^{2l}I\bigr)P_s^j\Bigr)^{-1}\right\|<\frac{8k^{-2l+1}}{p_{\m}\varepsilon
_0},
\end{equation}
\begin{equation}\label{Ps-2}
\left\|\Bigl(P_s^j\bigl(H(\k^{(2)}(\varphi
))-k^{2l}I\bigr)P_s^j\Bigr)^{-1}\right\|_1<\frac{8k^{-2l+1+
4r_1}}{p_{\m}\varepsilon_0},
\end{equation}
$\varepsilon _0=\min \{\varepsilon , k^{-r_1'-\delta }\}$, where
$\varepsilon $ is the distance to the pole of the operator.
\end{lemma}
\begin{proof}  The proof is similar to that of Lemma \ref{L:estnonres1} (part 3).  Indeed, when $p_\m<k^{-5r_1'}$, the series for $\lambda ^{(2)}(\k^{(2)}(\varphi )+\p_{\m})$ converges
in the complex $k^{-r_1'-\delta }$ neighborhood of $\omega^{(2)}(k,\delta ,\tau )$ and $\lambda ^{(2)}(\k^{(2)}(\varphi )+\p_{\m})=\lambda ^{(1)}(\k^{(1)}(\varphi )+\p_{\m})+
o(k^{-100r_1})$, see \eqref{perturbation-2c}. By Lemma \ref {L: Appendix
1}, the equation $\lambda^{(1)} \big(\k^{(1)}(\varphi
)+\p_\m\big)=k^{2l}+\varepsilon _0,\
|\varepsilon _0|\leq p_\m k^{\delta }$ has no more than two solutions in this neighborhood  of $\omega^{(2)}(k,\delta ,\tau )$. Using Rouche's theorem, we obtain the same
 fact for $\lambda ^{(2)}(\k^{(2)}(\varphi )+\p_{\m})=\varepsilon _0$. It is easy to show that the analogs of Lemmas \ref{L:Appendix 2}, \ref{L:3.7.1} and
  \ref{L:July5} hold for  $\lambda ^{(2)}(\k^{(2)}(\varphi )+\p_{\m})$. Thus, we obtain  \eqref{Ps-1}, \eqref{Ps-2}. \end{proof}

%%%%%%%%%\begin{lemma}\label{L:P2} Let $\varphi _0\in \omega _2(k,\delta ,\tau )$,
%%%%%%%%%and $|\varphi-\varphi _0|<k^{-11r_1-\delta}$. Then,
%%%%%%%%%\begin{equation}\label{P2-1}
%%%%%%%%%\left\|\Bigl(P^{(2)}\bigl(H(\k^{(2)}(\varphi
%%%%%%%%%))-k^2I\bigr)P^{(2)}\Bigr)^{-1}\right\|<k^{44\mu
%%%%%%%%%r_1}\left(\frac{k^{2\gamma r_1-11r_1}}{\varepsilon
%%%%%%%%%_0}\right)^{c_1k^{2\gamma r_1}}.
%%%%%%%%%\end{equation}
%%%%%%%%%\begin{equation}\label{P2-2}
%%%%%%%%%\left\|\Bigl(P^{(2)}\bigl(H(\k^{(2)}(\varphi
%%%%%%%%%))-k^2I\bigr)P^{(2)}\Bigr)^{-1}\right\|_1<k^{44\mu r_1+4\gamma
%%%%%%%%%r_2}\left(\frac{k^{2\gamma r_1-11r_1}}{\varepsilon
%%%%%%%%%_0}\right)^{c_1k^{2\gamma r_1}}.
%%%%%%%%%\end{equation}
%%%%%%%%%where $\varepsilon _0$ is the distance to the nearest pole of the
%%%%%%%%%resolvent in ${\cal W}^{(2)}$.
%%%%%%%%%\end{lemma}
%%%%%%%%%\begin{proof} Estimate \eqref{P2-1} follows from the first relation in Corollary \ref{C:blocks} and
%%%%%%%%%estimates \eqref{Pnr}, \eqref{Pr-1}. Considering that the number of
%%%%%%%%%elements in $\MM(\varphi _0,r_2)$ does not exceed $k^{4r_2}, $, we
%%%%%%%%%obtain that the dimension of $P^{(2)}$ does not exceed $k^{4r_2}$.
%%%%%%%%%Now \eqref{P2-2} follows from \eqref{P2-1}. \end{proof}

\subsubsection{Resonant and Nonresonant Sets for Step III \label{GSIII}}

We divide $[0,2\pi )$ into $[2\pi k^{44r_1'+2l+\delta}]+1$ intervals
$\Delta_m^{(2)}$ with the length not bigger than $k^{-44r_1'-2l-\delta }$.
If a particular interval belongs to $\OO^{(2)}$ we ignore it;
otherwise, let $\varphi_0(m)\not\in\OO^{(2)}$ be a point inside the
$\Delta_m^{(2)}$. Let
\begin{equation}\W_m^{(2)}=\{\varphi \in \W^{(2)}:\ | \varphi -\varphi
_0(m)|<4k^{-44r_1'-2l-\delta }\}. \label{W2m} \end{equation} Clearly, neighboring sets
$\W_m^{(2)}$  overlap (because of the multiplier 4 in the
inequality), they cover $\hat \W^{(2)}$ , which is
%%%%%$\omega^{(2)}$ and its $k^{-44r_1'-2l-\delta
%%%%%}$-neighborhood, i.e
the restriction of $\W^{(2)}$ to the $2k^{-44r_1'-2l-\delta
}$-neighborhood of $[0,2\pi )$. For each $\varphi \in \hat \W^{(2)}$
there is an $m$ such that $|\varphi -\varphi
_{0}(m)|<4k^{-44r_1'-2l-\delta }$. We consider the poles of the
resolvent  $\left(P^{(2)}
(H(\k^{(2)}(\varphi))-k^{2l})P^{(2)}\right)^{-1}$ in $\hat
\W_m^{(2)}$ and denote them by $\varphi^{(2)} _{mj}$, $j=1,...,M_m$.
By Corollary \ref{C:PHP-2} the resolvent has a block structure. The
number of blocks clearly cannot exceed the number of elements in
$\Omega (r_2)$, i.e. $k^{4r_2}$. Using the estimates for the number
of poles for each block, the estimate being provided by Lemma
\ref{L:Pr} Part 1, we can roughly estimate the number of poles of
the resolvent by $k^{4r_2+r_1}$.

%%%%We us consider all the components $\Pi_s^j$, $\Pi _b^j$, $\Pi _g^j$
%%%%and $\Pi _w^j$ for a given $\varphi_0(m)$ and numerate them just by
%%%%index $j$: $\Pi ^j(\varphi _0(m))$. Let $P(\Pi )$ be the projection
%%%%corresponding to a $\Pi ^j(\varphi _0(m))$ (we omit indices $j,m$
%%%%for shortness). Let us consider all the poles of the operator $(P
%%%%_{\Pi }(H(\k^{(2)}(\varphi))-k^{2l})P_{\Pi })^{-1}$ in $\W_m^{(2)}$
%%%%and denote them by $\varphi^{(2)} _{mj}$, $j=1,...,M_m$. Considering as in Appendix 8, we obtain
%%%%M_m<Ck^{4r_2+8r_1}$.

 Next, let $r_2'>11r_1'$ and $\OO^{(3)}_{mj}$ be the disc of the radius
$k^{-r_2'}$ around $\varphi ^{(2)}_{mj}$.
\begin{definition} The set
\begin{equation}\OO^{(3)}=\cup _{mj}\OO^{(3)}_{mj} \label{O3}
\end{equation}
we call the third resonant set. The set
\begin{equation}\W^{(3)}= \hat  \W^{(2)}\setminus \OO^{(3)}\label{W3}
\end{equation}
is called the third non-resonant set. The set
\begin{equation}\omega^{(3)}= \W^{(3)}\cap [0,2\pi) \label{w3}
\end{equation}
is called the third real non-resonant set. \end{definition}
\begin{lemma}\label{L:geometric3}Let  $r_2'>\mu r_2>44r_1'$, $\varphi \in \W^{(3)}$, $\varphi
_0(m)$ corresponds  to an interval $\Delta _m^{(2)}$ containing $\Re
\varphi $. Let $\Pi $ be one of the components $\Pi _s^j(\varphi
_0(m))$, $\Pi _b^j(\varphi _0(m))$, $\Pi _g^j(\varphi _0(m))$, $\Pi
_w^j(\varphi _0(m))$ and
 $P(\Pi )$ be the projection corresponding to $\Pi $. Let also
 $\varkappa \in \C: |\varkappa-\varkappa^{(2)}(\varphi )|<k^{-r_2'k^{2\gamma
r_1}}$. Then,
\begin{equation} \label{March3-2} \left\|\left(P(\Pi )\left(H\big(\k(\varphi
)\big)-k^{2l}I\right)P(\Pi )\right)^{-1}\right\|<ck^{2\mu
r_2+r_2'N^{(1)}},\end{equation}
\begin{equation} \label{March3-3}
\left\|\left(P(\Pi )\left(H\big(\k(\varphi
)\big)-k^{2l}I\right)P(\Pi )\right)^{-1}\right\|_1<ck^{(2\mu+1)
r_2+r_2'N^{(1)}},\end{equation} $N^{(1)}$ corresponding to the color
of $\Pi $: $N^{(1)}=1,\ k^{\gamma r_1+3 },\ k^{\gamma r_1/2+\delta
_0r_1 },\ k^{\gamma r_1/6-\delta _0r_1 }$ for simple, black, grey
and white clusters, correspondingly.
\end{lemma}
\begin{proof} For $\k=\k^{(2)}(\varphi )$ the lemma follows
immediately from the definition of $\W^{(3)}$ and Lemmas \ref{L:Pr}
and \ref{L:Ps} ($p_{\m}>k^{-2\mu r_2}$). It is easy to see that
estimates \eqref{March3-2} and \eqref{March3-3} are stable with
respect to perturbation of $\varkappa^{(2)}$ of order
$k^{-r_2'k^{2\gamma r_1}}$.
%%%%%%%%%%This stability ensure \eqref{March3-2} and \eqref{March3-3}.
\end{proof}

 By total size of the set $\OO^{(3)}$ we mean the sum of
the sizes of its connected components.
%\marginpar{check coeff-s}
\begin{lemma}\label{4.16} Let $r_2'\geq (\mu+10)r_2$,
$r_2>45r_1'+2l$. Then, the size of each connected component of
 $\OO^{(3)}$ is less
than $k^{5r_2-r_2'}$. The total size of $\OO^{(3)}$ is less than
$k^{-r_2'/2}$.
\end{lemma}
\begin{proof} Indeed, each set $\W_m^{(2)}$ contains no more than
$k^{4r_2+r_1}$ discs $\OO_{mj}^{(3)}$. Therefore, the size of
$\OO^{(3)}\cap \W_m^{(2)}$ is less than $k^{-r_2'+5r_2}$.
Considering that $k^{-r_2'+5r_2}$ is much smaller that the length
of $\Delta _m^{(2)}$, we obtain that there is no connected components
which go across the whole set $\W_m^{(2)}$ and the size of each
connected component of $\OO^{(3)}$ is less than $k^{5r_2-r_2'}$.
Considering that the number of intervals $\Delta _m^{(2)}$ is less than
$k^{45r_1'+2l+\delta }$, we obtain the required estimate for the total size
of $\OO^{(3)}$.
\end{proof}

%%%%%%%{\it We say that $\varphi\in I_j$ is a resonant point of the second
%%%%%%%kind if its distance from the nearest pole of the operator
%%%%%%%$(P(H(\k_j)-k^2)P)^{-1}$, $\k_j=k(\cos \varphi_0(j), \sin
%%%%%%%\varphi_0(j))$  is less than $k^{-r_1'}$.} The set of all resonant
%%%%%%%points of the second kind we denote by $\OO^{(2)}$. Notice that
%%%%%%%$$meas\,(\OO^{(2)}\cap[0,2\pi))\leq(2\pi k^{2+\delta(40\mu+1)}+1)ck^{4r_1}k^{-r_1'}\leq k^{-r_1}.$$
%%%%%%%{\it Here and in what follows we assume that $r_1'>5r_1+2$.} We also
%%%%%%%put $\W^{(2)}:=\W ^{(1)}\setminus\OO^{(2)}$,
%%%%%%%$\W^{(2)}_j:=\W^{(1)}\cap\{\varphi:\
%%%%%%%|\varphi-\varphi_0(j)|<k^{-2-\delta(40\mu+1)}\}$.

\begin{lemma}\label{estnonres0-1} Let $\varphi\in\W^{(2)}$ and
$C_3$ be the circle $|z-k^{2l}|=k^{-2r_2'k^{2\gamma r_1}}$. Then
$$
\left\|\left(P(r_1)(H(\k^{(2)}(\varphi))-z)P(r_1)\right)^{-1}\right\|\leq
4^2k^{2r_2'k^{2\gamma r_1}}.$$ \end{lemma}
\begin{proof} The proof is similar to the proof of Lemma~\ref{estnonres0} if we take into account
\eqref{step2raz*} and \eqref{dk0-2}. We notice also, that since in
the proof of the lemma we use the estimates from the previous step
along with some perturbation arguments: first, the series
decomposition (cf. \eqref{step2raz} and \eqref{step2raz*}), and
second, the shift from $\k^{(1)}$ to $\k^{(2)}$, we accumulate
additional factor $4$.\end{proof}

%%%%We also notice that the
%%%%%statement of Lemma \eqref{L:geometric2} still holds (with $2c$
%%%%%%%nstead of $c$) if we use $z$ instead of $k^{2l}$. Thus, if we put
%%%%%%$P_j^{(2)}:=P^{(2)}(\varphi_j(0))$,
%%%%%$\wt{P}_j^{(2)}:=P_j^{(2)}+P(r_1)$ then the following lemma holds:
%%%%%%\begin{lemma}\label{L:estfull-2}\begin{equation}\label{estfull-2}
%%%%%%\left\|\left(\wt{P}_j^{(2)}(H(\k^{(2)}(\varphi))-z)\wt{P}_j^{(2)}\right)^{-1}\right\|\leq
%%%%%%2k^{2r_2'k^{2\gamma r_1}},\ \ \ |z-k^{2l}|=(2k^{2r_2'k^{2\gamma
%%%%%r_1}})^{-1},\ \varphi\in\W^{(3)}_j.
%%%%%\end{equation}\end{lemma}

\section{Step III}
Let $k_*$ be sufficiently large to satisfy the estimates:
$$ k_*\geq k_{1}(V,\delta ,\tau), \ \ \  \ k^{\delta /8}_*>10^{8}+\|V\|+\mu +2l,$$
$k_{1}(V,\delta ,\tau)$ being introduced in the formulation of Theorem \ref{Thm2}. We also assume that $k_*$ is such that all constants $c$ in previous estimates (e.g. \eqref{March3-2}, \eqref{March3-3}) satisfy $c<k^{\delta /8}_*$. Since now on we consider $k>k_*$. This restriction on $k$ won't change in all consecutive steps.
We introduce a new notation $O_T(\cdot )$: let $f(k)=O_T(k^{-\gamma })$ mean that $|f(k)|<Tk^{-\gamma }$ when $k>k_*$.
\subsection{Operator $H^{(3)}$. Perturbation Formulas}
Let $P(r_2)$ be an orthogonal projector onto $\Omega(r_2):=\{\m:\
|\|\p_\m\||\leq k^{r_2}\}$ and $H^{(3)}=P(r_2)HP(r_2) $. From now on
we assume \begin{equation}
\label{r_2}k^\delta<r_2<k^{\gamma10^{-7}
r_1}.\end{equation} Note that $45r_1'+2l<k^\delta<k^{\gamma10^{-7}
r_1}$ for all $k>k_*$, since $10^{8}<r_1<k^{\delta /8}$. Let  $
\beta =2l-2-41\mu \delta$ and \begin{equation}\label{r_2'}5\mu r_2<r_2'<\frac{\beta
}{128}k^{\delta_0r_1-\delta -3}. \end{equation} We consider
$H^{(3)}(\k^{(2)}(\varphi ))$ as a perturbation of $\tilde
H^{(2)}(\k^{(2)}(\varphi ))$:
$$\tilde H^{(2)}:=\tilde
P_j^{(2)}H\tilde
P_j^{(2)}+\left(P(r_2)-\tilde
P_j^{(2)}\right)H_0,$$
where $H=H(\k^{(2)}(\varphi ))$, $H_0=H_0(\k^{(2)}(\varphi ))$ and $\tilde
P_j^{(2)}$
is the projection $P^{(2)}$, see \eqref{P(2)}, corresponding to $\varphi _{0}(j)$ in
the interval $\Delta _j^{(2)}$ containing $\varphi $.
Note that the operator $\tilde H^{(2)}$ has a block structure, the block $\tilde
P_j^{(2)}H\tilde
P_j^{(2)}$ being composed of smaller blocks $P_iHP_i$, $i=0,...,5$, see \eqref{PHP-2}, \eqref{blocks}.
%%%%%\begin{equation}
%%%%%\\begin{split}
%%%%%\&\tilde H^{(2)}=P(r_1)H(\k^{(2)}:=
%%%%%\))P(r_1)+P_{j,s}^{(2)}H(\k^{(2)}(\varphi
%%%%%\))P_{j,s}^{(2)}+P_{j,b}^{(2)}H(\k^{(2)}(\varphi ))P_{j,b}^{(2)}\cr &
%%%%%\+P_{j,g}^{'(2)}H(\k^{(2)}(\varphi
%%%%%\))P_{j,g}^{'(2)}+P_{j,w}^{'(2)}H(\k^{(2)}(\varphi
%%%%%\))P_{j,w}^{'(2)}+P_{j,nr}^{'(2)}H(\k^{(2)}(\varphi
%%%%%\))P_{j,nr}^{'(2)}\cr &
%%%%%\+\left(P(r_2)-\tilde
%%%%%\P_j^{(2)}\right)H_0(\k^{(2)}(\varphi ))\left(P(r_2)-\tilde
%%%%%\P_j^{(2)}\right),
%%%%%\\end{split}
%%%%%\\end{equation}
 Let
\begin{equation}W^{(2)}=H^{(3)}-\tilde H^{(2)}=P(r_2)VP(r_2)-\tilde P_j^{(2)}V\tilde P_j^{(2)}, \label{W2*}\end{equation}
\begin{equation}\label{g3} g^{(3)}_r({\k}):=\frac{(-1)^r}{2\pi
ir}\hbox{Tr}\oint_{C_3}\left(W^{(2)}(\tilde
H^{(2)}({\k})-zI)^{-1}\right)^rdz,
\end{equation} \begin{equation}\label{G3}
G^{(3)}_r({\k}):=\frac{(-1)^{r+1}}{2\pi i}\oint_{C_3}(\tilde
H^{(2)}({\k})-zI)^{-1}\left(W^{(2)}(\tilde
H^{(2)}({\k})-zI)^{-1}\right)^rdz,
\end{equation}
where $C_3$ is the circle $|z-k^{2l}|=\varepsilon _0^{(3)}$,
$\varepsilon _0^{(3)}=k^{-2r_2'k^{2\gamma r_1}}.$
\begin{theorem} \label{Thm3} Suppose  $k>k_*$, $\varphi $ is in
the real  $k^{-r_2'-\delta }$-neighborhood of $\omega
^{(3)}(k,\delta,\tau )$ and $\varkappa\in\R$,
$|\varkappa-\varkappa^{(2)}(\varphi )|\leq \varepsilon
^{(3)}_0k^{-2l+1-\delta }$, $\k=\varkappa(\cos \varphi ,\sin \varphi
)$. Then,  there exists a single eigenvalue of $H^{(3)}({\k})$ in
the interval $\varepsilon _3( k,\delta,\tau )=\left(
k^{2l}-\varepsilon _0^{(3)}, k^{2l}+\varepsilon _0^{(3)}\right)$. It
is given by the absolutely converging series:
\begin{equation}\label{eigenvalue-3}\lambda^{(3)}({\k})=\lambda^{(2)}({\k})+
\sum\limits_{r=2}^\infty g^{(3)}_r({\k}).\end{equation} For
coefficients $g^{(3)}_r({\k})$ the following estimates hold:
\begin{equation}\label{estg3} |g^{(3)}_r({\k})|<k^{-\frac{\beta}{5} k^{r_1-\delta
}-\beta (r-1)}.
\end{equation}
The corresponding spectral projection is given by the series:
\begin{equation}\label{sprojector-3}
\E ^{(3)}({\k})=\E^{(2)}({\k})+\sum\limits_{r=1}^\infty
G^{(3)}_r({\k}), \end{equation} $\E^{(2)}({\k})$ being the spectral
projection of $H^{(2)}(\k)$. The operators $G^{(3)}_r({\k})$ satisfy
the estimates:
\begin{equation}
\label{Feb1a-3}
\left\|G^{(3)}_r({\k})\right\|_1<k^{-\frac{\beta}{10} k^{r_1-\delta
} -\beta r},
\end{equation}
\begin{equation}G^{(3)}_r({\k})_{\s\s'}=0,\ \mbox{when}\ \ \
2rk^{\gamma
r_1+3}+3k^{r_1}<\||\p_\s\||+\||\p_{\s'}\||.\label{Feb6a-3}
\end{equation}
\end{theorem}
\begin{corollary}\label{corthm3} For the perturbed eigenvalue and its spectral
projection the following estimates hold:
 \begin{equation}\label{perturbation-3}
\lambda^{(3)}({\k})=\lambda^{(2)}({\k})+ O_2\left(k^{-\frac15 \beta
k^{r_1-\delta }-\beta }\right),
\end{equation}
\begin{equation}\label{perturbation*-3}
\left\|\E^{(3)}({\k})-\E^{(2)}({\k})\right\|_1<2k^{-\frac{\beta}{10}
k^{r_1-\delta }-\beta }.
\end{equation}
\begin{equation}
\left|\E^{(3)}({\k})_{\s\s'}\right|<k^{-d^{(3)}(\s,\s')},\ \
\mbox{when}\ \||\p_\s\||>4k^{r_1} \mbox{\ or }
\||\p_{\s'}\||>4k^{r_1 },\label{Feb6b-3}
\end{equation}
$$d^{(3)}(\s,\s')=\frac18(\||\p_\s\||+\||\p_{\s'}\||)k^{-\gamma r_1-3}\beta +\frac{1}{10}\beta k^{r_1-\delta
}.$$
\end{corollary}
Formulas \eqref{perturbation-3} and \eqref{perturbation*-3} easily
follow from \eqref{eigenvalue-3}, \eqref{estg3} and \eqref{sprojector-3}, \eqref{Feb1a-3}. The estimate \eqref{Feb6b-3}
follows from \eqref{sprojector-3}, \eqref{Feb1a-3}, \eqref{Feb6a-3}
and \eqref{Feb6b}.
\begin{proof}Let us consider the perturbation series
\begin{equation}\label{step3dva}
(H^{(3)}-z)^{-1}=\sum_{r=0}^\infty (\tilde
H^{(2)}-z)^{-1}\left(-W^{(2)} (\tilde H^{(2)}-z)^{-1}\right)^r,
\end{equation} here and below all the operators are computed at
$\k $. Further, we consider $\k$ and, therefore, the operators, as
analytic functions of $\varphi $ in $\W_j^{(2)}$, assuming
$\varkappa$ is fixed. By \eqref{W2*} and \eqref{PHP-2},
$W^{(2)}=V-\sum _{i=0}^5 P_iVP_i $. By assumption on $\varkappa$ and
Lemmas \ref{L:geometric3} and \ref{estnonres0-1},
\begin{equation}
\left\|(\tilde
H^{(2)}(\k)-z)^{-1}\right\|<2\cdot4^2k^{2r_2'k^{2\gamma r_1}}.
\label{tik-tak}
\end{equation}
To check the convergence it is enough to show that
%%\begin{equation}\label{step25}
%%\begin{split}& \|P'V\wt{P
%%%%Let
%%%%%%\begin{equation}P^{(3)}=P(r_1)+P_b+P_{wl}+P_{ws}+P_{nr}.\label{P3} \end{equation} We consider a model
%%%%%operator
%%%%%$$\tilde H^{(2)}=P_{b}HP_{b}+P_{ws}HP_{ws}+P_{wl}HP_{wl}+P_{nr}HP_{nr} +\left(P(r_2)-P^{(3)}\right).$$
%%%%In this notations,
%%%%%$$H^{(3)}=\tilde H^{(2)}+W^{3},$$
%%%%%%%where $W^{3}=W^{3}_1+W^{3}_b+W^{3}_{ws}+W^{3}_{wl}+W^{3}_{nr}$.
%%%%%%%$W^{(3)}_1=P_bV(P^{(3)}-P_b)+(P^{(3)}-P_b)VP_b$, etc.,
%%%%%%%$$W^{3}_2=\left(P(r_2)-P^{(3)}\right)VP^{(3)} +P^{(3)}V\left(P(r_2)-P^{(3)}\right)+
%%%%%%%\left(P(r_2)-P^{(3)}\right)V\left(P(r_2)-P^{(3)}\right).$$
 \begin{equation} \label{March5}\left\| \left(\tilde
H^{(2)}-z\right)^{-1}W^{(2)}\right\|<k^{-\beta } .\end{equation}
Then,
\begin{equation}
\left\|(H^{(3)}(\k)-z)^{-1}\right\|<4^3k^{2r_2'k^{2\gamma r_1}}.
\label{tik-tak*}
\end{equation}
Let us prove \eqref{March5}. Operator $\tilde H^{(2)}$ has a block structure, different blocks being separated by the $\||\cdot \||$ distance  greater than $k^{\delta }$. This means that not only the blocks themselve, but also the blocks multiplied by $W^{(2)}$ have non-zero action on orthogonals subspaces. The operator $\tilde H^{(2)}$ acts as $H_0$ ``outside" the blocks. Because of  the block structure and the estimate $|z-k^{2l}|=
k^{-2r_2'k^{2\gamma r_1}}=o(k^{-\beta})$, it suffices to prove:
\begin{equation} \label{5-0} \left\|\left(P(r_2)-\tilde P^{(2)}_j\right)\left(
H_0-k^{2l}\right)^{-1}V\right\|<\frac13k^{-\beta }, \end{equation}
\begin{equation} \label{5-}\left\|P(r_1)(\tilde H^{(2)}-z)^{-1}V(P(r_2)-P(r_1))\right\|<\frac13k^{-\beta },\end{equation}
\begin{equation} \label{5-1} \left\|P_{nr}\left(\tilde
H^{(2)}-k^{2l}\right)^{-1}V\left(P(r_2)-P_{nr}\right)\right\|<\frac13k^{-\beta
},\end{equation}
\begin{equation} \label{5-s} \left\|P_{s}\left(\tilde
H^{(2)}-k^{2l}\right)^{-1}V\left(P(r_2)-P_{s}\right)\right\|<\frac13k^{-\beta
},
 \end{equation}
\begin{equation} \label{5-2} \left\|P_{w}\left(\tilde
H^{(2)}-k^{2l}\right)^{-1}V\left(P(r_2)-P_{w}\right)\right\|<\frac13k^{-\beta
},
 \end{equation}
\begin{equation} \label{5-3} \left\|P_{g}\left(\tilde
H^{(2)}-k^{2l}\right)^{-1}V\left(P(r_2)-P_{g}\right)\right\|<\frac13k^{-\beta
},
\end{equation}
\begin{equation} \label{5-4} \left\|P_{b}\left(\tilde
H^{(2)}-k^{2l}\right)^{-1}V\left(P(r_2)-P_b\right)\right\|<\frac13k^{-\beta
}, \end{equation} By definition of $\tilde P^{(2)}_j$,
$$\left\|\left(P(r_2)-\tilde P^{(2)}_j\right)\left(
H_0-k^{2l}\right)^{-1}\right\|<k^{-2l+2+40\mu \delta}.$$ The
estimate \eqref{5-0} easily follows.

Let us prove \eqref{5-}.
By Lemma \ref{L:boundary},
$$P(r_1)V\big(P(r_2)-P(r_1)\big)=
P(r_1)^{\partial}V\big(P(r_2)-P(r_1)\big),$$ where
$P(r_1)^{\partial}$ is the projection on the boundary of $\Pi (r_1)$. Therefore, it suffice to prove:
\begin{equation} \label{5-*}\left\|P(r_1)(H^{(2)}-z)^{-1}P(r_1)^{\partial}\right\|<k^{-\beta -\delta /2},\end{equation}
the obvious relation $P(r_1)\tilde H^{(2)}=H^{(2)}$ has been taken into account. As in the proof of Theorem \ref{Thm2}, we consider $H^{(2)}$ as a perturbation
of $\tilde H^{(1)}$, $H^{(2)}=\tilde H^{(1)}+W$.
Taking into account that $\tilde H^{(1)}$ has a $k^{\delta }$-block structure and $V$ is a trigonometric polynomial, we obtain $$P(\delta )\left( \left(\tilde
H^{(1)}-z\right)^{-1}W\right)^sP^{\partial}(r_1)=0,\ \  \mbox{when  } 1\leq s\leq S,\  \    S:=\frac14 k^{r_1-\delta}.$$ Hence,
\begin{equation} \label{hence}P(r_1)(H^{(2)}-z)^{-1}P^{\partial }(r_1)=\sum _{s=0}^{S-1} P\left(\tilde
H^{(1)}-z\right)^{-1}A_*^s P^{\partial }(r_1)+P(r_1)\left(\tilde
H^{(2)}-z\right)^{-1}A_*^{S}P^{\partial }(r_1),\end{equation} where
$A_*=-P\left(\tilde H^{(1)}-z\right)^{-1}W$, $P=P(r_1)-P(\delta)$.
Considering as in the proof of Theorem \ref{Thm2},\footnote{We
replace $P(r_1)$ by $P$, this compensates  for the smallness of
$C_3$.} we obtain:   $\left\|P\left(\tilde
H^{(1)}-z\right)^{-1}\right\|<ck^{-\beta -\delta }<k^{-\beta -\delta /2}$ . It follows:
$\|A_*\|<\frac{1}{4}k^{-\beta }$. By Theorem \ref{Thm2} and the definition of
$C_3$,  $\left\|\left( \tilde
H^{(2)}-z\right)^{-1}\right\|<(\varepsilon_0^{(3)}) ^{-1}$.
Substituting the last three estimates into \eqref{hence} and taking
into account that $(\varepsilon_0^{(3)}) ^{-1}< k^{\beta S/2}$, we
obtain \eqref{5-*} and, therefore, \eqref{5-} for all $\varphi \in
\W^{(2)}_j$.

Next, we prove \eqref{5-1}. by Lemma \ref{L:boundary},
$$P_{nr}V\left(P(r_2)-P_{nr}\right)=
P_{nr}^{\partial}V\left(P(r_2)-P_{nr}\right),$$ where
$P_{nr}^{\partial}$ is the projection on the boundary of $\Pi
_{nr}$.  Therefore, it suffices to prove
\begin{equation} \label{E:nr}\left\|P_{nr}\left(\tilde
H^{(2)}-k^{2l}\right)^{-1}P_{nr}^{\partial}\right\|<k^{-\beta
-\delta /2}.\end{equation} Note that Lemma \ref{L:estnonres1} holds for
any $r_1>\delta $ (the restriction on $r_1$ is introduced later).
Therefore, the estimates \eqref{Mon3-1}--\eqref{Mon3-4} hold for
$\m\in \Omega(r_2)$. By the definition of $\Pi _{nr}$, $\varepsilon
_0>k^{-10r_1'}$ and $p_\m>k^{-5r'}$ in these estimates. It follows
(see Corollary \ref{estnonres}),
\begin{equation} \label{E:nr*}\left\|P_{nr}\left(\tilde
H^{(2)}-k^{2l}\right)^{-1}P_{nr}\right\|<k^{40r_1'}.\end{equation}  Considering as in the proof of Theorem \ref{Thm2} (see the proof of \eqref{||A||}),
we obtain \eqref{E:nr} and, hence,
\eqref{5-1} for all $\varphi \in \W^{(2)}_j$.

Next, we prove \eqref{5-s}. Denote by $\hat{H}$ the reduction of the operator $H$ onto a
particular simple cluster i.e. $\hat{H}=P_sHP_s$ where $(P_s)_{\m\m}=1$ if
$\m$ belongs to the simple cluster and $(P_s)_{\m\m}=0$ otherwise. By
Lemma \ref{L:geometric3},
\begin{equation} \label{Apr6-s}
\|(\hat{H}-k^{2l})^{-1}\|\leq ck^{2\mu r_2+r_2'},\end{equation} By
Lemma \ref{L:boundary},
\begin{equation}P_s(\hat{H}-k^{2l})^{-1}V(P(r_2)-P_s)=P_s(\hat{H}-k^{2l})^{-1}P_s^{\partial }V(P(r_2)-P_s). \label{Apr6a-s} \end{equation}
To obtain \eqref{5-s}, it is enough to show
\begin{equation}
\left\|P_s(\hat{H}-k^{2l})^{-1}P_s^{\partial }\right\|<k^{-\beta -\delta/4}.
\label{Apr5-s}
\end{equation}
We are going to construct the perturbation formula for
$P_s(\hat{H}-k^{2l})^{-1}P_s^{\partial }$. Let
$\hat{H}_0=P_{s,nr}HP_{s,nr}+(P_s-P_{s,nr})H_0$, where
$P_{s,nr}=P_sP_{nr}=P_{nr}P_{s}$ .  The operator $\hat H_0$ has $k^{\delta }$-block structure. It is analogous to the operator $\tilde H^{(1)}$ in the proof of
Theorem \ref{Thm2}. The perturbation formula for
$P_s(\hat{H}-k^{2l})^{-1}P_s^{\partial }$ has the form:
%%%%% This means that we take the $k^\delta$-"squares" around  points of
%%%%%$\MM(\varphi _{0j}, r_2)\cap \Pi _w\setminus \left(\MM^{(2)}\cup \MM(\varphi _{0j}, r_1)\right$
%%%%%and the other part of the operator is diagonal. The perturbation formula for $P_w(\hat{H}-k^2)^{-1}P_w^{\partial }$ has the form
%%%%%%As we know from above, such
%%%%%$k^\delta$-"squares" can form clusters consisting of at most 4 elements.\begin{equation}\label{dva}
\begin{equation}\label{dva-s}
\begin{split}&
P_s(\hat{H}-k^{2l})^{-1}P_s^{\partial
}=\sum_{r=0}^{R_s}P_s(\hat{H}_0-k^{2l})^{-1}\left[-W_s(\hat{H}_0-k^{2l})^{-1}\right]^rP_s^{\partial
}\cr &
+P_s(\hat{H}-k^{2l})^{-1}\left[-W_s(\hat{H}_0-k^{2l})^{-1}\right]^{R_s+1}P_s^{\partial
} ,\cr & W_s=\hat{H}-\hat{H}_0=P_sVP_s-P_{s,nr}VP_{s,nr},\  \ \  R_s=[\frac{1}{8}k^{\frac{
r_1}{2}-\delta}]-1.
\end{split}
\end{equation}
%%%%%We choose $R=[k^{\frac{\gamma r_1}{6}-\frac{\delta _0}{2}}]-1$
%%%%%with
%%%%$0<\delta'<\frac{\gamma}{6}$ to be determined later (poka ne yasno
%%%%zachem nuzhno $\delta'$).
 When $\p_{\m'}$ belongs to the boundary of the
white cluster, the $\||\cdot\||$-distance from $\p_{\m'}$ to the
point $\p_{\m}:0<p_\m< k^{-5r_1'}$ is $k^{ r_1/2}$. Notice
that $(\hat{H}_0-k^{2l})^{-1}_{\m\m'}=0$ if
$\||\p_\m-\p_{\m'}\||>8k^\delta$, since $\hat H_0$ has a $k^{\delta }$ structure.  Considering that $R_s<
\frac{1}{8}k^{\frac{ r_1}{2}-\delta }$ (so, we never reach the
central point of $\Pi_s$), we obtain that the finite sum in
\eqref{dva-s} is analytic inside $\W _j^{(2)}$ and is bounded by
$2k^{-\beta -\delta/2 }$, see \eqref{E:nr}. Moreover,
\begin{equation}\label{tri-s}
\left\|P_s\left[W_s(\hat{H}_0-k^{2l})^{-1}\right]^{R_s+1} P_s^{\partial
}\right\|       \leq k^{-(R_s+1)(\beta+\delta/2 )}.
\end{equation}
%%%%%Operator $(\hat{H}_0-k^2)^{-1}$ has not more than
%%%%%$k^{(\frac{\gamma}{6}-\delta_1)r_1}$ poles inside ${\cal O}_s$. By
%%%%%Lemma \ref{L:Pr},
%%%%% By Rouche's Theorem operator $(\hat{H}-k^2)^{-1}$
%%%%%has the same number of poles in ${\cal O}_s$. Squeezing around the
%%%%%poles with the scale $k^{10(r_2-r_1)}$ we obtain the set
%%%%%$\tilde{\cal O}_s$ which is the union of the disks of radius
%%%%%$k^{-10r_2}$. Outside $\tilde{\cal O}_s$ we have the estimate
%%%%%\begin{equation}
%%%%%\|(\hat{H}-k^2)^{-1}\|\leq
%%%%%%k^{10(r_2-r_1)k^{(\frac{\gamma}{6}-\delta_1)r_1}}.
%%%%%\end{equation}
Substituting  \eqref{Apr6-s} into \eqref{dva-s} and taking into account
\eqref{tri-s} we get
\begin{equation}\label{chetire-s}
 \left\| P_s(\hat{H}-k^{2l})^{-1}  P_s^{\partial } \right\|\leq k^{-\beta-\delta/4},
\end{equation}
when $\varphi\in\W^{(3)}$, $2r_2'<\frac18 k^{
r_1/2-\delta}\beta$.

Now, we prove \eqref{5-2}. Here and in what follows
we will often use the same notation for objects formally different
but playing similar roles in different parts of the proof. We hope
it will not lead to confusion but rather make it easier to keep the
whole construction and further inductive arguments in mind. Denote by $\hat{H}$ the reduction of the operator $H$ onto a
particular white cluster i.e. $\hat{H}=PHP$ where $P_{\m\m}=1$ if
$\m$ belongs to the white cluster and $P_{\m\m}=0$ otherwise. By
Lemma \ref{L:geometric3},
\begin{equation} \label{Apr6}
\|(\hat{H}-k^{2l})^{-1}\|\leq ck^{2\mu r_2+r_2'k^{\frac{\gamma
r_1}{6}-\delta _0r_1}}.
\end{equation}
By Lemma \ref{L:boundary},
\begin{equation}P_w(\hat{H}-k^{2l})^{-1}V(P(r_2)-P_w)=P_w(\hat{H}-k^{2l})^{-1}P_w^{\partial }V(P(r_2)-P_w). \label{Apr6a} \end{equation}
To obtain \eqref{5-2}, it is enough to show
\begin{equation}
\left\|P_w(\hat{H}-k^{2l})^{-1}P_w^{\partial }\right\|<k^{-\beta -\delta/4}.
\label{Apr5}
\end{equation}

We are going to construct the perturbation formula for
$P_w(\hat{H}-k^{2l})^{-1}P_w^{\partial }$. Let
$\hat{H}_0=P_{w,nr}HP_{w,nr}+(P_w-P_{w,nr})H_0$, where
$P_{w,nr}=P_wP_{nr}=P_{nr}P_{w}$ . The perturbation formula for
$P_w(\hat{H}-k^{2l})^{-1}P_w^{\partial }$ has the form :
\begin{equation}\label{dva}
\begin{split}&
P_w(\hat{H}-k^{2l})^{-1}P_w^{\partial
}=\sum_{r=0}^{R_w}P_w(\hat{H}_0-k^{2l})^{-1}\left[-W(\hat{H}_0-k^{2l})^{-1}\right]^rP_w^{\partial
}\cr &
+P_w(\hat{H}-k^{2l})^{-1}\left[-W(\hat{H}_0-k^{2l})^{-1}\right]^{R_w+1}P_w^{\partial
} ,\cr & W=\hat{H}-\hat{H}_0,\  \ \  R_w=[\frac{1}{8}k^{\frac{\gamma
r_1}{6}-\delta}]-1.
\end{split}
\end{equation}
%%%%%%%%\ref{dva-s} where $ R=[\frac{1}{8}k^{\frac{\gamma r_1}{6}-\delta}]-1$.
%%%%% This means that we take the $k^\delta$-"squares" around  points of
%%%%%$\MM(\varphi _{0j}, r_2)\cap \Pi _w\setminus \left(\MM^{(2)}\cup \MM(\varphi _{0j}, r_1)\right$
%%%%%and the other part of the operator is diagonal. The perturbation formula for $P_w(\hat{H}-k^2)^{-1}P_w^{\partial }$ has the form
%%%%%%As we know from above, such
%%%%%$k^\delta$-"squares" can form clusters consisting of at most 4 elements.
%%%%%\begin{equation}\label{dva}
%%%%%\begin{split}
%%%%%&
%%%%%(\hat{H}-k^2)^{-1}_{\m\m'}=\left(\sum_{r=0}^R(\hat{H}_0-k^2)^{-1}\left[W(\hat{H}_0-k^2)^{-1}\right]^r\right)_{\m\m'}\cr
%%%%%%& +\left((\hat{H}-k^2)^{-1}\left[W(\hat{H}_0-k^2)^{-1}\right]^{R+1}
%%%%\right)_{\m\m'},\cr & W=\hat{H}-\hat{H}_0,\  \
%%%%\m \in \Pi _w,\  \   \   \m'   \in  \partial \Pi _w.
%%%%%\end{split}
%%%%%\end{equation}
%%%%%\begin{equation}\label{dva}
%%%%%\begin{split}&
%%%%%P_w(\hat{H}-k^2)^{-1}P_w^{\partial }=\sum_{r=0}^RP_w(\hat{H}_0-k^2)^{-1}\left[W(\hat{H}_0-k^2)^{-1}\right]^rP_w^{\partial }\cr
%%%%%& +P_w(\hat{H}-k^2)^{-1}\left[W(\hat{H}_0-k^2)^{-1}\right]^{R+1}P_w^{\partial }
%%%%%,\cr & W=\hat{H}-\hat{H}_0,\  \
%%%%%\  R=[\frac{1}{8}k^{\frac{\gamma r_1}{6}-\delta}]-1.
%%%%%\end{split}
%%%%%\end{equation}
%%%%%We choose $R=[k^{\frac{\gamma r_1}{6}-\frac{\delta _0}{2}}]-1$
%%%%%with
%%%%$0<\delta'<\frac{\gamma}{6}$ to be determined later (poka ne yasno
%%%%zachem nuzhno $\delta'$).
 When $\p_{\m'}$ belongs to the boundary of the
white cluster, the $\||\cdot\||$-distance from $\p_{\m'}$ to the
closest  point in $\MM^{(2)}$ is $k^{\gamma r_1/6}$. Notice that
$(\hat{H}_0-k^{2l})^{-1}_{\m\m'}=0$ if
$\||\p_\m-\p_{\m'}\||>8k^\delta$. Considering that $R_w<
\frac{1}{8}k^{\frac{\gamma r_1}{6}-\delta }$ (so, we never reach the
points in $\MM^{(2)}$), we obtain that the finite sum in \eqref{dva}
is analytic inside $\W _j^{(2)}$ and is bounded by $2k^{-\beta-\delta/2}$, see
\eqref{E:nr}. Moreover,
\begin{equation}\label{tri}
\left\|P_w\left[W(\hat{H}_0-k^{2l})^{-1}\right]^{R+1} P_w^{\partial
}\right\|       \leq k^{-(R_w+1)(\beta+\delta )}.
\end{equation}
%%%%%Operator $(\hat{H}_0-k^2)^{-1}$ has not more than
%%%%%$k^{(\frac{\gamma}{6}-\delta_1)r_1}$ poles inside ${\cal O}_s$. By
%%%%%Lemma \ref{L:Pr},
%%%%% By Rouche's Theorem operator $(\hat{H}-k^2)^{-1}$
%%%%%has the same number of poles in ${\cal O}_s$. Squeezing around the
%%%%%poles with the scale $k^{10(r_2-r_1)}$ we obtain the set
%%%%%$\tilde{\cal O}_s$ which is the union of the disks of radius
%%%%%$k^{-10r_2}$. Outside $\tilde{\cal O}_s$ we have the estimate
%%%%%\begin{equation}
%%%%%\|(\hat{H}-k^2)^{-1}\|\leq
%%%%%%k^{10(r_2-r_1)k^{(\frac{\gamma}{6}-\delta_1)r_1}}.
%%%%%\end{equation}
Substituting  \eqref{Apr6} into \eqref{dva} and taking into account
\eqref{tri} we get
\begin{equation}\label{chetire}
 \left\| P_w(\hat{H}-k^{2l})^{-1}  P_w^{\partial } \right\|\leq 3k^{-\beta-\delta /2}<k^{-\beta-\delta /4},
\end{equation}
since $\varphi\in\W^{(3)}$, $2r_2'<\frac18 k^{\delta
_0r_1-\delta}\beta$.

Now, we prove \eqref{5-3}. Denote a component of the grey region  by $\Pi $ and its
boundary (see convention above) by $\partial\Pi $. Corresponding
projectors are denoted by $P$ and $P^{\partial }$
respectively. Denote by $\hat{H}$ the reduction of the operator $H$ onto a
particular grey cluster i.e. $\hat{H}=PHP$.
By Lemma \ref{L:boundary},
\begin{equation}P(\hat H-k^{2l})^{-1}V(P(r_2)-P)=P(\hat H-k^{2l})^{-1}P^{\partial }V(P(r_2)-P). \label{Apr6a'} \end{equation}
To obtain \eqref{5-3}, it is enough to show
\begin{equation}
\left\|P(\hat{H}-k^{2l})^{-1}P^{\partial }\right\|<k^{-\beta -\delta/4}.
\label{Apr5'}
\end{equation}
We are going to construct the perturbation formula for
$P(\hat{H}-k^{2l})^{-1}P^{\partial }$. Recall, that the size of the
neighborhood of grey boxes is $D=k^{\frac{\gamma r_1}{2}+2\delta
_0r_1}$.  Let $P_i$ be a projector corresponding to a white or
non-resonant cluster laying inside $\frac D2$-neighborhood of
$\partial\Pi$, the size of these clusters being much smaller than
the size of the neighborhood. For definiteness, let
$i=1,\dots,\hat{I}$. Let $P^{(int)}$ be the projector onto all
points in $\Pi$ which are at least $D/2$ away of the boundary
(internal points). Note that $P_iP^{(int)}=0$. At last, put
\begin{equation} \label{opP'}
P_0:=P-P^{(int)}-\sum_{i=1}^{\hat{I}}P_i.
\end{equation}
Denote (cf. the case of a white cluster)
\begin{equation}
\hat{H}_0:=\sum_{i=1}^{\hat{I}}P_iHP_i+P^{(int)}HP^{(int)}+H_0P_0,
\label{opH_0} \end{equation}
\begin{equation} \label{opW} W=\hat H-\hat H_0=PVP-\sum_{i=1}^{\hat{I}}P_iVP_i-P^{(int)}VP^{(int)}.\end{equation}
We are going to use perturbation arguments between $\hat{H}_0$ and
$\hat{H}$. %%%%By construction
%%%%$P^{(int)}(\hat{H}_0-k^{2l})^{-1}P^{\partial}=0$.
 Let $R$ be
the smallest natural number for which
\begin{equation}
{\cal
A}_{R}:=P^{(int)}\left[W(\hat{H}_0-k^{2l})^{-1}\right]^{R+1}P^{\partial}\not=0,\
\ \ W:=\hat{H}-\hat{H}_0. \label{r_0} \end{equation} It is proven in
Appendix 4 that $R>\frac{1}{64}
k^{(\frac\gamma2+2\delta_0)r_1-\delta}$. Therefore,
\begin{equation}\label{chetire*g}
\begin{split}
&
(\hat{H}-k^{2l})^{-1}P^{\partial}=\sum_{r=0}^{R-1}
(\hat{H}_0-k^{2l})^{-1}\left[-(P-P^{(int)})W(\hat{H}_0-k^{2l})^{-1}\right]^rP^{\partial }\cr
&
+(\hat{H}-k^{2l})^{-1}\left[-(P-P^{(int)})W(\hat{H}_0-k^{2l})^{-1}\right]^{R}P^{\partial }
.
\end{split}
\end{equation}
The first term in the RHS of \eqref{chetire*g} contains only
non-resonant and white clusters. Thus, we can use the estimates
obtained before in the case of non-resonant and white clusters (see
\eqref{5-1}, \eqref{5-2}). To estimate the second term we, first,
notice that
\begin{equation}\label{sem'}
\left\|\left[-(P-P^{(int)})W(\hat{H}_0-k^{2l})^{-1}\right]^{R}P^{\partial}\right\|\leq
k^{-\beta R}\leq k^{-\frac{\beta }{64}
k^{(\frac\gamma2+2\delta_0)r_1-\delta}}.
\end{equation}
By Lemma \ref{L:geometric3} ,
\begin{equation}\label{vosem'}
\|(\hat{H}-k^{2l})^{-1}\|\leq ck^{2\mu
r_2+r_2'k^{(\frac\gamma2+\delta_0)r_1}}.
\end{equation}
Now, considering that $2r_2'<\frac{\beta }{64}
k^{\delta_0r_1-\delta}$  and combining the estimates above, we
obtain \eqref{Apr5'} and, therefore, \eqref{5-3}.
%%%%%\begin{equation}\label{vosem''}
%%%%%%\|(\hat{H}-k^{2l})^{-1}P^{\partial}\|\leq k^{-\beta-\delta }.
%%%%%%\end{equation}
%%%%% Thus \eqref{5-3} is proven.

We prove \eqref{5-4} in the analogous way. Indeed, denote a component of the black region  by $\Pi $ and its
boundary (see convention above) by $\partial\Pi $. Corresponding
projectors are denoted by $P$ and $P^{\partial }$
respectively. Again, $\hat H:=PHP$ and,
by Lemma \ref{L:boundary},
\begin{equation}P(\hat{H}-k^{2l})^{-1}V(P(r_2)-P)=P(\hat{H}-k^{2l})^{-1}P^{\partial }V(P(r_2)-P).  \label{Apr6a'-1} \end{equation}
To obtain \eqref{5-4}, it is enough to show
\begin{equation}
\left\|P(\hat{H}-k^{2l})^{-1}P^{\partial }\right\|<k^{-\beta -\delta }.
\label{Apr5'-1}
\end{equation}
We are going to construct the perturbation formula for
$P(\hat{H}-k^{2l})^{-1}P^{\partial }$. Recall, that the size of the
neighborhood of black boxes is $D=k^{\gamma r_1+\delta _0r_1}$. Put
$\hat{H}=P HP$ (cf. the case of white  and grey clusters). Let $P_i$
be a projector corresponding to a grey, white or non-resonant
cluster laying inside $\frac D2$-neighborhood of $\partial\Pi$, the
size of these clusters being much smaller than the size of the
neighborhood. For definiteness, let $i=1,\dots,\hat{I}$. Let
$P^{(int)}$ be the projector onto all points in $\Pi$ which are at
least $D/2$ away of the boundary (internal points).  Again, we
define $P_0$, $\hat H_0$ and $W$ by formulas \eqref{opP'},
\eqref{opH_0}  and \eqref{opW}.
%%%%$$
%%%%P_0:=P-P^{(int)}-\sum_{i=1}^{\hat{I}}P_i.
%%%%$$
%%%%Note that $P_iP^{(int)}=0$. Put (cf. the case of a grey cluster)
%%%%$$
%%%%\hat{H}_0:=\sum_{i=1}^{\hat{I}}P_iHP_i+P^{(int)}HP^{(int)}+P_0H_0P_0.
%%%%$$
We are going to use perturbation arguments between $\hat{H}_0$ and
$\hat{H}$.  Let $R$ be
the smallest positive integer for which \eqref{r_0} holds in the case of a black cluster.
%%%%\begin{equation}
%%%%%{\cal
%%%%%A}_{r_0}:=P^{\partial}(\hat{H}_0-k^{2l})^{-1}\left[W(\hat{H}_0-k^{2l})^{-1}\right]^{r_0}P^{(int)}\not=0,\
%%%%\ \ W:=\hat{H}-\hat{H}_0. \label{r_0-b} \end{equation}
 It is proven
in Appendix 5 that $R>\frac{1}{64} k^{(\gamma r_1+\delta
_0r_1-\delta)}$. Next, we use \eqref{chetire*g}.
%%%%%%Therefore,
%%%%%%\begin{equation}\label{chetire*}
%%%%%%\begin{split}
%%%%%%&
%%%%%%P^{\partial}(\hat{H}-k^{2l})^{-1}=P^{\partial}\sum_{r=0}^{r_0-1}
%%%%%%(\hat{H}_0-k^{2l})^{-1}\left[W(\hat{H}_0-k^{2l})^{-1}(P-P^{(int)})\right]^r
%%%%%%\cr &
%%%%%%+P^{\partial}(\hat{H}_0-k^{2l})^{-1}\left[W(\hat{H}_0-k^{2l})^{-1}(P-P^{(int)})\right]^{r_0-1}
%%%%%%W(\hat{H}-k^{2l})^{-1},\ \ \ W=\hat{H}-\hat{H}_0.
%%%%%%\end{split}
%%%%%%\end{equation}
The first term in the RHS of \eqref{chetire*g} contains only non-resonant, white and grey
clusters. Thus, we can use the estimates \eqref{5-1}-\eqref{5-3} obtained before in the case
of non-resonant, white and grey clusters. To estimate the second term
we, first, notice that
\begin{equation}\label{sem'-b}
\left\|\left[-(P-P^{(int)})W(\hat{H}_0-k^{2l})^{-1}\right]^{R}P^{\partial}\right\|<
k^{-\beta R}<
k^{-\frac{\beta }{64} k^{\gamma r_1+\delta
_0r_1-\delta}}.
\end{equation}
By Lemma \ref{L:geometric3},
\begin{equation}\label{vosem'-b}
\|(\hat{H}-k^{2l})^{-1}\|\leq ck^{2\mu r_2+r_2'k^{\gamma r_1+3}}.
\end{equation}
Now, choosing $2r_2'<\frac{\beta }{64} k^{\delta _0r_1-\delta-3}$
and combining the estimates above we obtain \eqref{Apr5'-1} and, therefore, \eqref{5-4}.
%%%%%%\begin{equation}\label{vosem''-b}
%%%%%%\|P^{\partial}(\hat{H}-k^{2l})^{-1}\|\leq k^{-\beta -\delta }.
%%%%%%\end{equation}
%%%%%%Thus \eqref{5-4} is proven.

Estimates \eqref{5-0} -- \eqref{5-4} provide convergence of the series for the resolvent. Integrating the resolvent over the contour we get \eqref{eigenvalue-3} and \eqref{sprojector-3}.

Proof of  \eqref{Feb1a-3} is analogous to that of \eqref{Feb1a} in
Theorem \ref{Thm2}. Indeed, we consider the operator $A=
W^{(2)}\left(\tilde H^{(2)}-z\right)^{-1}$ and represent it as
$A=A_0+A_1+A_2$, where $A_0=\left(P(r_2)-\E^{(2)}({\k})\right)A
\left(P(r_2) -\E^{(2)}({\k})\right)$, $A_1=\left(P(r_2
)-\E^{(2)}({\k})\right)A \E^{(2)}({\k})$, $A_2= \E^{(2)}({\k})A
\left(P(r_2)-\E^{(2)}({\k})\right)$. Note that
$\E^{(2)}({\k})W^{(2)}\E^{(2)}({\k})=0$, because of \eqref{W2*}. We
see that $$\oint _{C_3}\left(\tilde H^{(2)}-z\right)^{-1}A_0^r
dz=0,$$ since the integrand is a holomorphic function inside $C_3$.
Therefore,
\begin{equation} \label{Feb1-3} G^{(3)}_r({\k})=\frac{(-1)^r}{2\pi i}\sum
_{j_1,...j_r=0,1,2,\ j_1^2+...+j_r^2\neq 0}\oint _{C_3}\left(\tilde
H^{(2)}-z\right)^{-1}A_{j_1}.....A_{j_r} dz.
\end{equation} At least one of indices in each term is equal to 1 or 2.
Let us show that
\begin{equation} \label{A_2-3}
\|A_2\|_1<\frac12k^{-\frac{\beta}{10} k^{r_1-\delta}-\beta }.
\end{equation} First, we notice that
\begin{align}\nonumber &\E^{(2)}W^{(2)}(P(r_2)-\E^{(2)})=\E^{(2)}W^{(2)}(P(r_2)-P(r_1))=\cr &
\E^{(2)}P^{\partial}(r_1)W^{(2)}(P(r_2)-P(r_1))=\E^{(2)}P^{\partial}(r_1)W^{(2)}\end{align}
by \eqref{W2*}.
 Hence, $A_2=\E^{(2)}P^{\partial}(r_1)A(P(r_2)-\E^{(2)})$. Using \eqref{Feb6b}, we obtain
 $\|\E^{(2)}P^{\partial}(r_1)\|<k^{-\frac{\beta}{10} k^{r_1-\delta}}$.
Considering that $\E^{(2)}$ is a one-dimensional projection, we
obtain the same estimate for $\bf S_1 $-norm. Now \eqref{A_2-3}
easily follows. Applying the same trick as in the proof of
\eqref{Feb1a} we obtain \eqref{Feb1a-3}.

Let us obtain the estimate for $g_r ^{(3)}({\k})$.
Obviously,\begin{equation} \label{Feb1'*}
g^{(3)}_r({\k})=\frac{(-1)^r}{2\pi ir}\sum _{j_1,...j_r=0,1,2,\
j_1^2+...+j_r^2\neq 0}Tr\oint _{C_3}A_{j_1}.....A_{j_r} dz.
\end{equation}
Note that each term contains both $A_1$ and $A_2$, since we compute
the trace of the integral. Using \eqref{A_2-3} and repeating
arguments from the proof of \eqref{estg2}, we obtain \eqref{estg3}.

 The estimate \eqref{Feb6a-3} follows from the fact that the biggest white, grey or black component
 has the size not greater than $k^{\gamma r_1 +3}$. Therefore the biggest block of $\tilde H^{(2)}$ not coinciding
 with $P(r_1)HP(r_1)$ has the size not greater than $k^{\gamma r_1 +3}$.

\end{proof}

It is easy to see that coefficients $g^{(3)}_r({\k})$ and operators
$G^{(3)}_r({\k})$ can be analytically extended into the complex
$k^{-r_2'-\delta}$ neighborhood of $\omega ^{(3)}$ (in fact, into
$k^{-r_2'-\delta}$-neighborhood of $\W^{(3)}$) as functions of
$\varphi $ and to the complex $(\varepsilon
^{(3)}_0k^{-2l+1-\delta})-$neighborhood of
$\varkappa=\varkappa^{(2)}(\varphi )$ as functions of $\varkappa$,
estimates \eqref{estg3}, \eqref{Feb1a-3} being preserved. Now, we
use formulae \eqref{g3}, \eqref{eigenvalue-3} to extend
$\lambda^{(3)}\left({\k}\right)=\lambda^{(3)}\left(\varkappa,\varphi\right)$
as an analytic function. Obviously, series \eqref{eigenvalue-3} is
differentiable. Using Cauchy integral and Lemma
\ref{L:derivatives-2} we get the following lemma.

\begin{lemma} \label{L:derivatives-3}Under conditions of Theorem \ref{Thm3} the following
estimates hold when $\varphi \in \omega ^{(3)}(k,\delta )$ or its
complex $k^{-r_2'-\delta}$-neighborhood and $\varkappa\in \C:$
$|\varkappa-\varkappa^{(2)}(\varphi )|<\varepsilon
^{(3)}_0k^{-2l+1-\delta}$.
\begin{equation}\label{perturbation-3c}
\lambda^{(3)}({\k})=\lambda^{(2)}({\k})+O_2\left(k^{-\frac 15 \beta
k^{r_1-\delta}-\beta }\right),
\end{equation}
\begin{equation}\label{estgder1-3k}
\frac{\partial\lambda^{(3)}}{\partial\varkappa}=
\frac{\partial\lambda^{(2)}}{\partial\varkappa} +O_2\left(k^{-\frac
15 \beta k^{r_1-\delta}-\beta  }M_1\right), \  \ \ \
M_1:=\frac{k^{2l-1+\delta}}{\varepsilon ^{(3)}_0},\end{equation}
\begin{equation}\label{estgder1-3phi}\frac{\partial\lambda^{(3)}}{\partial \varphi }=\frac{\partial\lambda^{(2)}}{\partial \varphi }+
O_2\left(k^{-\frac 15 \beta k^{r_1-\delta}-\beta+r_2'+\delta }\right),
 \end{equation}
\begin{equation}\label{estgder2-3}
\frac{\partial^2\lambda^{(3)}}{\partial\varkappa^2}=
\frac{\partial^2\lambda^{(2)}}{\partial\varkappa^2}+
O_2\left(k^{-\frac 15 \beta k^{r_1-\delta}-\beta  }M_1^2\right),
\end{equation}
\begin{equation} \label{gulf2-3}
\frac{\partial^2\lambda^{(3)}}{\partial\varkappa\partial \varphi
}=\frac{\partial^2\lambda^{(2)}}{\partial\varkappa\partial \varphi
}+ O_2\left(k^{-\frac 15 \beta k^{r_1-\delta}-\beta+r_2'+\delta
}M_1\right),
\end{equation}
\begin{equation} \label{gulf3-3}
\frac{\partial^2\lambda^{(3)}}{\partial\varphi
^2}=\frac{\partial^2\lambda^{(2)}}{\partial\varphi ^2}+
O_2\left(k^{-\frac 15 \beta k^{r_1-\delta}-\beta+2r_2'+2\delta
}\right).
\end{equation}\end{lemma}

\begin{corollary} \label{"O"} All ``$O_2$"-s on the right hand sides of \eqref{perturbation-3c}-\eqref{gulf3-3} can be written as $O_1\left(k^{-\frac {1}{10} \beta k^{r_1-\delta}}\right)$.
\end{corollary}

\subsection{\label{IS3}Isoenergetic Surface for  Operator $H^{(3)}$}

\begin{lemma}\label{ldk-3} \begin{enumerate}
\item For every $\lambda :=k^{2l}$,  $k>k_*$, and $\varphi $ in the real  $\frac{1}{2} k^{-r_2'-\delta }$-neighborhood
of $\omega^{(3)}(k,\delta, \tau )$, there is a unique
$\varkappa^{(3)}(\lambda, \varphi )$ in the interval
$$I_2:=[\varkappa^{(2)}(\lambda, \varphi )-\varepsilon
^{(3)}_0k^{-2l+1-\delta},\varkappa^{(2)}(\lambda, \varphi
)+\varepsilon ^{(3)}_0k^{-2l+1-\delta}],$$ such that
    \begin{equation}\label{2.70-3}
    \lambda^{(3)} \left(\k
^{(3)}(\lambda ,\varphi )\right)=\lambda ,\ \ \k ^{(3)}(\lambda
,\varphi ):=\varkappa^{(3)}(\lambda ,\varphi )\vec \nu(\varphi).
    \end{equation}
\item  Furthermore, there exists an analytic in $ \varphi $ continuation  of
$\varkappa^{(3)}(\lambda ,\varphi )$ to the complex  $\frac{1}{2}
k^{-r_2'-\delta }$-neighborhood of $\omega^{(3)}(k,\delta, \tau )$
such that $\lambda^{(3)} (\k ^{(3)}(\lambda, \varphi ))=\lambda $.
Function $\varkappa^{(3)}(\lambda, \varphi )$ can be represented as
$\varkappa^{(3)}(\lambda, \varphi )=\varkappa^{(2)}(\lambda, \varphi
)+h^{(3)}(\lambda, \varphi )$, where
\begin{equation}\label{dk0-3} |h^{(3)}(\varphi )|=O_1\left(k^{-\frac 15 \beta k^{r_1-\delta}-\beta -2l+1
}\right),
\end{equation}
\begin{equation}\label{dk-3}
\frac{\partial{h}^{(3)}}{\partial\varphi}= O_2\left(k^{-\frac 15 \beta
k^{r_1-\delta}-\beta -2l+1 +r_2'+\delta }\right),\ \ \ \ \
\frac{\partial^2{h}^{(3)}}{\partial\varphi^2}= O_4\left(k^{-\frac 15
\beta k^{r_1-\delta}-\beta -2l+1 +2r_2'+2\delta }\right).
\end{equation} \end{enumerate}\end{lemma}
\begin{proof}  The proof is completely analogous to that of Lemma \ref{ldk}, estimates \eqref{perturbation-3c} --\eqref{gulf3-3} being used. \end{proof}

%%%%%It follows from
%%%%%\eqref{perturbation-2c} and Rouch\'{e}'s theorem that for any
%%%%%It$\varphi$ in $k^{-r_1'-\delta}$-neighborhood of $\omega ^{(2)}$
%%%%%Itthere exists unique value of ${k}^{(2)}$ such that
%%%%%$|{k}^{(2)}(\varphi )-k^{(1)}(\varphi)|<k^{-4r'_1-2l+1-\delta}$ and
%%%%%$\lambda^{(2)}\left({k}^{(2)},\varphi\right)=\lambda_0:=k^{2l}$.
%%%%%Actually,
%%%%%\begin{equation} |{k}^{(2)}(\varphi )-k^{(1)}(\varphi)|<k^{-2k^{\delta }Q^{-1}}k^{-4l+5+48\mu \delta
%%%%%}.\label{kappa2} \end{equation} Then it follows from
%%%%%\eqref{estgder1-2} and  implicit function theorem that
%%%%%${k}^{(2)}(\varphi)$ is locally analytic. Combined with uniqueness
%%%%%this implies global analyticity. Considering (\ref{kappa2}), it is
%%%%%easy to show that perturbation series  for $\lambda^{(2)}
%%%%

Let us consider the set of points in $\R^2$ given by the formula:
$\k=\k^{(3)} (\varphi), \ \ \varphi \in \omega ^{(3)}(k,\delta, \tau
)$. By Lemma \ref{ldk-3} this set of points is a slight distortion
of ${\cal D}_{2}$. All the points of this curve satisfy the equation
$\lambda^{(3)}(\k ^{(3)}(\varphi ))=k^{2l}$. We call it isoenergetic
surface of the operator $H^{(3)}$ and denote by ${\cal D}_{3}$.

\subsection{Preparation for Step IV \label{S:4}}
\subsubsection{Properties of the Quasiperiodic Lattice. Continuation}\label{Quasiperiodicgeomcont}
Let \begin{equation}\label{triind}
\SS ^{(2)}(k,\xi):=\{\k\in \R^{2}:\
\|(H^{(2)}(\k)-k^{2l})^{-1}\|>k^{\xi}\}.
\end{equation}
The main purpose of this section is to estimate the number of
 points $\k _0+\p_{\m}$, $\||\p_{\m}\||< k^{r_2}$ in $\SS ^{(2)}(k,\xi)$, $\k _0$ being fixed. In fact, we prove a more subtle result, see
Lemma~\ref{norm2/3}.

 We consider $\p_\m=2\pi(\s_1+\alpha\s_2)$ with
integer vectors $\s_j$ such that $|\s_j|\leq 4k^{r_2}$. We repeat
the arguments from the beginning of Section~\ref{geomIII}. Namely,
let $(q,p)\in\Z^2$ be a pair such that $0<q\leq 4k^{r_2}$ and
\begin{equation}\label{qind}
|\alpha q+p|\leq 16k^{-r_2}.
\end{equation}
We choose a pair $(p,q)$ which gives the best approximation. In
particular, $p$ and $q$ are mutually simple. Put
$\epsilon_q:=\alpha+\frac{p}{q}$. We have
\begin{equation}k^{-2r_2\mu}\leq|\epsilon_q|\leq 16q^{-1}k^{-r_2}.\label{epsilon_qind}\end{equation}
 We  write $\s_2$ in the form\begin{equation}\s_2=q\s_2'+\s_2'' \label{s1ind}
\end{equation}
with  integer vectors $\s_2'$ and $\s_2''$, $0\leq (\s_2'')_j< q$
for $j=1,2$. Hence, $|(\s_2')_j|\leq 4k^{r_2}/q+1$. It follows
$$
(2\pi)^{-1}\p_\m=(\s_1-p\s_2')+(-\frac{p}{q}\s_2''+\epsilon_q\s_2'')+\epsilon_q
q\s_2'.
$$
Denote $\s:=\s_1-p\s_2'$. Then $|\s|\leq 8k^{r_2}$. The number of
different vectors $\tilde{\s}:=-\frac{p}{q}\s_2''+\epsilon_q\s_2''$
is not greater than $(2q)^2$. For each fixed pair $\tilde \s,\ \s$
we obtain a lattice parameterized by $\s_2'$. We call this lattice a
cluster corresponding to given $\tilde \s,\ \s$. Each cluster,
obviously, is a square lattice with the step $\epsilon _qq$. It
contains no more than $\left(9k^{r_2}q^{-1}\right)^2$ elements,
since $|(\s_2')_j|\leq 4k^{r_2}q^{-1}+1$, $j=1,2$. The size of each
cluster is less than $5|\epsilon _q|k^{r_2}$. As before we have the
following statements.

     \begin{lemma}\label{Lattice-1ind}Suppose that $\epsilon _q$ satisfies the
     inequality
     \begin{equation}|\epsilon_q|\leq \frac{1}{64}q^{-1}k^{-r_2}.\label{epsilon_q'ind}\end{equation}
     Then, the size of each cluster is less that $\frac{1}{8q}$. The distance between clusters is greater than
     $\frac{1}{2q}$. \end{lemma}

     \begin{lemma}\label{Lattice-2ind} The number of vectors $\p_\m$,
      satisfying the inequalities $\||\p_{\m}\||<2k^{r_2}$,
     $p_{\m}<|\epsilon _q|qk^{r_2/3}$, does not exceed
     $k^{2r_2/3}$.\end{lemma}

     \begin{lemma} \label{Lattice-3ind} Suppose $q$ in the inequality
     \eqref{q} satisfies the estimate $q>k^{2r_2/3}$. Then, the
     number of vectors  $\p_\m$,
     $\||\p_{\m}\||<2k^{r_2}$, satisfying the inequality
     $p_{\m}<k^{-2r_2/3}$ does not exceed
     $2^{12}\cdot k^{2r_2/3}$.\end{lemma}

We consider the matrix $H^{(2)}(\k)=P(\gamma r_1)H(\k )P(\gamma
r_1)$ where $\k \in \R^2$, $P(\gamma r_1)$ is the orthogonal
projection corresponding to $\Omega (\gamma r_1)$. \footnote{It is a
slight abuse of notations, since $H^{(2)}$ in Step II was defined
for $\gamma =1$.} We construct the block structure in  $H^{(2)}(\k)$
analogous to that in Step II.  The difference is that now we
consider any $\k \in \R^2$, not only $\k$ being close to
$\k^{(1)}(\varphi )$. Indeed,  we call $\m \in \Omega (\gamma r_1)$
non-resonant if  (cf. \eqref{resonance1})
\begin{equation}
\left||\k+\p_{\m }|^{2}-k^{2}\right|>k^{-40\mu\delta}.
\label{Aug29a}
\end{equation}
Obviously, this estimate is stable in the $k^{-41\mu \delta
-1}$-neighborhood of a given $\k$. Hence, the definition of a
non-resonant $\m$ is stable in this neighborhood up to a multiplier
$1+o(1)$ in the r.h.s. of \eqref{Aug29a}.  Around each resonant $\m$
we construct $k^\delta$-boxes/clusters (see \eqref{defP}). Let
$P_{\m}$ be the projection on the $k^{\delta }$-cluster containing
$\m $. If
\begin{equation}
\left\|(P_{\m}(H(\k)-k^{2l})P_{\m})^{-1}\right\|<k^{4\gamma r_1'},
\label{Aug29b} \end{equation} then we call the $k^{\delta }$-cluster
effectively non-resonant (cf. \eqref{Mon3a**}) for a given $\k$.
%%%%%% \footnote{Formally a non-resonant set is defined geometrically by  Definition \ref{D:3.17}. Every non-resonant set in the sense of the geometric definition is effectively %%%%non-resonant by Lemma \ref{L:geometric2}.}
 Note, that the above
estimate and, therefore, the definition of an effectively
non-resonant $k^{\delta}$-cluster is  stable in the $k^{-4\gamma
r_1'-2l+1-\delta }$-neighborhood of a given $\k$. The
$k^{\delta}$-clusters, where \eqref{Aug29b} is not valid, are called
effectively resonant $k^{\delta}$-clusters.  Thus, we have
constructed a block structure in $H^{(2)}(\k)$, which is stable  in
the $k^{-4\gamma r_1'-2l+1-\delta}$-neighborhood of a given $\k$.
\begin{definition} \label{D-J}We denote   by
$J(\k)$ the number of the effectively resonant $k^\delta$-clusters
in $H^{(2)}(\k)$ for a given $\k$.
%%%%%In order words, $J_{\m_0}(\k)$ is the number of points of the
%%%%%set  ${\MM}^{(2)}(\k,r_3,4\gamma r_1')$ in $\Pi_{\m_0}$, the set
%%%%%${\MM}^{(2)}$ being defined by (\ref{M^2}) ( {\em somewhat sloppy,
%%%%%since ${\MM}^{(2)}$ is defined by discs around poles, not by the
%%%%%%%%%%inequality for the exponent, but basically it is the same for
%%%%%$k^{\delta }$ boxes}).
%%%%%Obviously, $J_{\m_0}(\k)=J(\k+\p_{\m_0})$.
 Further (with a slight abuse of notations) we
consider $J(\k)$ to be constant in the $2k^{-4\gamma
r_1'-2l+1-\delta}$-neighborhood of a given $\k$. \end{definition}
%%%%Note that by the Geometric
%%%%Lemmas \ref{L:number of points-1} and \ref{L:Resnumberofpoints} the
%%%%%number of resonant squares $J_{\m_0}(\k)$ does not exceed
%%%%$Ck^{\frac23 \gamma r_1+3}$. (po otnosheniyu k
%%%%%$\varepsilon_0=k^{-5\gamma r_1'}$?)
%%%%%%%%Let us fix $k^{2l}$ and consider the curves $D_{\m_0}(\k,
%%%%%%%%k^{2l})=0$ in $\R^{2}$.
%%To shorten notations we drop $\m_0$ and $k^{2l}$.
%%%%%%%%Naturally $\k=(k_1,k_2)$. It is proven in Appendix 7
%%%%follows from Lemma~\ref{elpieces}

%%%%%%%%that
%%%%%%%%the set $D_{\m_0}(\k)=0$ can be split into $k^{17l\gamma r_1}$ or less
%%%%%%%%elementary pieces (in the sense of Definition~\ref{elementary}).

Let  $\k=\a \tau_1+\b,\  \ \ |\a|=1,\  |\b|<4k^{\gamma r_{1}}$. We
consider $H^{(2)}(\k)$ as a function of $\tau _1$ in the complex $
2k^{-\rho _1}$-neighborhood of zero, $\rho _1=4\gamma
r_1'+2l-1+\delta $.
\begin{lemma} \label{May24-1}The resolvent $( H^{(2)}(\k)-k^{2l})^{-1}$ has no more than $8J(\b)$ poles  in the the complex $ 2k^{-\rho _1}$-neighborhood of zero. It satisfies the following estimate
in the complex $k^{-\rho _1}$-neighborhood of zero:
\begin{equation}\label{Sept3a}
\|(H^{(2)}(\k)-k^{2l})^{-1}\|<k^{16\rho _{1}}\left(\frac{k^{ -\rho
_{1}}}{\varepsilon _0}\right)^{8J(\k)},
\end{equation}
where $\varepsilon _0=\min \{k^{-2\rho _1},\varepsilon\}$,
$\varepsilon $ being the distance to the nearest pole.
\end{lemma}
\begin{proof}
 Recall (Definition \ref{D-J}) that
$J(\k)$ may be considered to be constant in $2k^{-\rho
_1}$-neighborhood  of $\b$. Hence, $J(\k)=J(\b)$ for such $\k$-s.
Let us consider the collection of all  $k^{\delta}$-clusters $P_{\m
}H(\k)P_{\m }$ for $H^{(2)}(\k)$, $\tau _1$ being in $2k^{-\rho _1
}$-neighborhood of zero. Note that the collection is the same for
all such $\k$. We construct the corresponding block operator $\tilde
H^{(1)}(\k)$:
$$\tilde H^{(1)}(\k)=\sum P_{\m}HP_{\m}+H_0(I-\sum P _{\m }).$$  If a $k^{\delta }$-cluster $P_{\m}H(\k )P_{\m}$ is effectively non-resonant, then  its resolvent, obviously,  has no poles $\tau _1$ in the $2k^{-\rho _1 }$-neighborhood of $\tau _{1}=0$.
Considering that a $k^{\delta }$-cluster contains no more than 4 squares (Appendix 2), we obtain, that the resolvent of
each effectively resonant $k^{\delta }$-cluster has no more than 8
poles $\tau _{1j}$ in the $2k^{-40\mu\delta -\delta}$-neighborhood
of $\tau _{1}=0$. Indeed, the relation opposite to \eqref{Aug29a}
can hold for no more than four different $\m$-s. Each function
$|\k+\p_{\m }|_{\R}^{2}-k^{2}$ is a quadratic polynomial with
respect to $\tau _1$. It is easy to see that  $\left||\k+\p_{\m
}|_{\R}^{2l}-k^{2l}\right|>k^{2l-2-80\mu\delta -4\delta }$ when
$\tau _1\not \in \D_0$, $\D_0$ being   the $k^{-40\mu \delta
-2\delta }$-neighborhood of the the roots $\tau _{1j}$ of the
polynomials. Obviously, $\D_0$ consists of at most 8 discs. We
consider only those connected components of $\D_{0}$ which are
inside $2k^{-40\mu \delta -\delta }$ disk around $\tau _1=0$
(components of $\D _{0}$ are much smaller than the size of the
disk). The perturbation series for the resolvent of $P_{\m}
H(\k)P_{\m}$ with respect to $ P_{\m}H_0(\k)P_{\m}$ $(V=0)$
converges on the boundaries of these components and the following
estimate holds there:
 $$
\|\left( P_{\m}(H(\k)-k^{2l})P_{\m}\right)^{-1}\|<k^{-2l+2+80\mu
\delta +5\delta }.
$$
 Hence, the resolvents of $P_{\m} H_0(\k)P_{\m}$ and $ P_{\m}H(\k)P_{\m}$ have the same number of poles inside each component of $\D_0$. This means that the resolvent of
each effectively resonant cluster $ P_{\m}H(\k)P_{\m}$ has no more
than 8 poles $\tau _{1j}$ in the $k^{-40\mu \delta-\delta
}$-neighborhood of $\tau _{1}=0$. By the maximum principle,
 $$
\|\left( P_{\m}(H(\k)-k^{2l})P_{\m}\right)^{-1}\|<k^{-2l+2+80\mu
\delta +5\delta}\left(\frac{16k^{-40\mu \delta -2\delta}}{k^{-2\rho
_1}}\right)^8<\frac{1}{4}k^{16\rho _1}
$$
at the distance greater than $k^{-2\rho _1}$ from the poles.
Therefore, resolvent $\left(\tilde H^{(1)}(\k)-k^{2l}\right)^{-1}$
has no more than $8J(\k)$ poles $\tau _{1j}$ in the complex
$k^{-40\mu \delta-\delta }$-neighborhood of $\tau _{1}=0$ and
satisfies the estimate \begin{equation} \label{May25} \|(\tilde
H^{(1)}(\k)-k^{2l})^{-1}\|<\frac14 k^{16\rho _1 }
\end{equation}
at the distance greater than $k^{-2\rho _1 }$ from the poles.
 Let us consider the union of
$k^{-2\rho _1}$ neighborhoods of these poles. It may consist from
several connected components. We are interested only in those
intersecting with the $2k^{-\rho _1}$ disk around $\tau _1=0$. We
denote their union by by ${\D}$. Using a rough estimate
$J(\k)<k^{4\gamma r_1}$ gives that $\D$ belongs to the $3k^{-\rho
_1}$-neighborhood of zero. Thus, \eqref{May25} holds outside $\D$.
%%%%%%%%%%%%Hence, the number of poles $\tau_{1j}$ of $({\cal
%%%%%%%%%%H}_{\m_0}(\k)-k^{2l})^{-1}$ in $k^{-5\gamma r_1'}$-neighborhood of
%%%%%%%%%%$\l$ does not exceed  $2J$. Let us consider $k^{-10\gamma
%%%%%%%%%%r_1'}$-neighborhood of such poles.
 Considering as before  (see the proof of Theorem \ref{Thm2}, \eqref{step2raz}--\eqref{step25}),
 %%%%%%(see {\em Remark } at the end of the proof of Lemma \ref{L:curves-2}, formulas \eqref{Aug10a}-\eqref{Aug10c})
we can show that the perturbation series for the resolvent
$(H^{(2)}(\k)-k^{2l})^{-1}$ with respect to $(\tilde
H^{(1)}(\k)-k^{2l})^{-1}$ converges on the boundary of $\D$ and
$$
\|(H^{(2)}(\k)-k^{2l})^{-1}\|<k^{16\rho _1}
$$
outside $\D$, the resolvent has no more than $8J(\k)$ poles in $\D$.
Using again the maximum principle we obtain \eqref{Sept3a}.\end{proof}
%%%%%%%%%\begin{equation}\label{Sept3a}
%%%%%%%%%\\|(H^{(2)}(\k)-k^{2l})^{-1}\|<Ck^{80\gamma }k^{(\zeta
%%%%%%%%%\-10\gamma r_1')J(\k)}
%%%%%%%%%\\end{equation}
%%%%%%%%%\outside the discs of the radius $k^{-\zeta}$ around the poles of the
%%%%%%%%%\resolvent when $\zeta >10\gamma r_1'$.

%%%%%%%%% Next, we introduce
%%%%%%%%%\begin{equation}\label{triind}
%%%%%%%%%\SS (k,\xi):=\{\k\in \R^{2}:\
%%%%%%%%%\|(H^{(2)}(\k)-k^{2l})^{-1}\|>k^{\xi}\}.
%%%%%%%%%\end{equation}
Note that each connected component of $\SS ^{(2)}(k,\xi)$, see \eqref{triind}, is
bounded  by the curves $D(\k, k^{2l}\pm k^{-\xi})=0$, where $ D(\k,
\lambda)=\hbox{det}\ (H^{(2)}(\k)-\lambda ).$
%%%%%%%%%%=\hbox{det}\,({\cal H}(\k+\p_{\m_0})-\lambda).
%%%%%%Obviously, $\SS_{\m_0}(\k, \xi)=\SS(\k+\p_{\m_0},\xi)$ and  $D_{\m_0}(\k, \lambda )=D(\k+\p_{\m_0},\lambda )$.
 %%%%%%%%%%%%We can split $\SS_{\m_0}(k,\xi)$ into ``elementary" components, each being adjacent to an elementary piece of the boundary. The number of such components does not exceed $Ck^{65\gamma r_1}$.
\begin{lemma}\label{L:curves-2ind}  Let $\l$ be a segment of a straight line in $\R^{2}$,
\begin{equation}\l:=\{\k=\a \tau_1+\b,\ \tau _1\in
(0,\eta)\}, \ \ |\a|=1,\  |\b|<4k^{\gamma r_1},\  \ 0<\eta <
k^{-5\gamma r_1'}. \label{segment}\end{equation} Suppose both ends
of $\l$ belong to a connected component of $\SS ^{(2)}(k,\xi)$. If $\xi $
is sufficiently large, namely, $\xi\geq 30J(\b)\ln
_k\frac{1}{\eta}$, then, there is an inner part $\l'$ of the
segment,
 which is not in $\SS^{(2)}(k,\xi)$.
 %%%%%%%%%%%%Moreover, there is a point $(\tau _{1*},\tau _{2*})$ in $\l'$ such that a
 %%%%%%%%%%%% boundary curve $D_{\m_0}(\k, k^{2l}\pm k^{-\xi})=0$ of $\SS_{\m_0}(k,\xi)$ and the segment $\l$ are parallel.
 \end{lemma}
 \begin{corollary} \label{C:curves-2ind} Let $\k\in \SS ^{(2)}(k,\xi)$ and
 ${\xi }/30J(\k)>10\gamma r_1'$. Then the distance from $\k$ to the boundary of
 $\SS^{(2)}(k,\xi)$ is less than $k^{-\xi/30J(\k)}$. \end{corollary}
 {\em Proof of the corollary.}  Let us consider a segment of the length $\eta =k^{-\xi/30J(\k)}$ starting at $\k $.
 By the statement of the lemma it intersects a boundary $D(\k, k^{2l}\pm k^{-\xi})=0$.

 \begin{proof}
%%%%%%%%%%Let us consider the complex $k^{-5\gamma r_1'}$-neighborhood of $\l$.
%%%%%%%%%%%%We construct the analog of the block structure (48?).
%%%%%%%%Let $\k_0 =(\tau_{1,0}, \beta _1 \tau _{1,0}+\beta _2)$.
 Choose
$\varepsilon =\eta ^2$ in \eqref{Sept3a}. Using the hypothesis of
the lemma, we obtain that the right-hand side of \eqref{Sept3a} is
less than $k^{{\xi }}$ outside the discs.  Let us estimate the total
size (sum of the sizes) of the discs. Indeed, the size of each disc
is $2\varepsilon $ and the number of discs is, obviously, less
$16k^{4\gamma r_1}$. Therefore, the total size admits the estimate
from above: $32\varepsilon k^{4\gamma r_1} <<\eta$, since $\eta
<k^{-5\gamma r_1'}$. This means there is a part $\l'$ of $\l$
outside these discs. By \eqref{Sept3a}, this part is
 outside $\SS ^{(2)}(k,\xi)$, when $\xi $ is as described in the statement of the lemma.
 %%%%%%%%%%Since, obviously, the distance
%%%%%%%%%between $D=0$ and $\l$ attains its maximum in $\l'$ and the curve
%%%%%%%%%%%and the segment are parallel at this point, the lemma is proven.
\end{proof}

Let $\k  _0\in \R^2$ be fixed and ${\cal N}(k,r_2,\k  _0,J_0)$ be the
following subset of the lattice $\k _0+\p_{\n}$, $\n \in \Omega (r_2)$:
$${\cal N}(k,r_2,\k _0,J_0)=\left\{\k _0+\p_{\n}:\n \in \Omega (r_2):\ J(\k  _0+\p_{\n})\leq J_0\right\},$$
$J$ being defined by Definition \ref{D-J}. Thus, ${\cal N}$ includes
only such $\n$ that the surrounding $k^{\gamma r_1}$- block contains
less than $J_0 $ of effectively resonant $k^{\delta }$-clusters. Let
$N(k,r_2,\k _0,J_0, \xi)$ be the number of points $\k _0+\p_{\n}$ in $\SS ^{(2)}
(k,\xi )\cap {\cal N}(k,r_2,\k _0,J_0)$.
%%%%%%%%Let $\xi >0$ and $N(r_2,\k,J_0, \xi)$ be the number of points $\m_0$
%%%%%%%%in ${\cal N}(r_2,\k,J_0)$ satisfying the estimate
%%%%%%%%\begin{equation}\label{0eqnorm2/3}
%%%%%%%%\|({\cal H}_{\m_0}(\k)-k^{2l})^{-1}\|>k^{\xi}.
%%%%%%%%\end{equation}
\begin{lemma}\label{norm2/3}
%%%%%%%%Let $N(\k)$, $\k \in \R^3$, be the number of $\m_0$ with $\||\p_{\m_0}\||\leq k^{r_2}$ such that
%%%%%%%%\begin{equation}\label{0eqnorm2/3}
%%%%%%%%\|({\cal H}_{\m_0}(\k)-k^{2l})^{-1}\|>k^{k^\xi},
%%%%%%%%\end{equation}
If  $r_2>10\gamma r_1'$ and $\xi>60\mu r_2J_0$, then
\begin{equation}\label{eqnorm2/3}
N(k,r_2,\k _0,J_0, \xi)\leq k^{\frac{2}{3}r_2+43l\gamma r_1}.
\end{equation}
\end{lemma}
\begin{proof} Let us call a subset $\tilde \SS $ of $\SS ^{(2)}(k,\xi )$ elementary if it can be described by a formula of the type:
$$ \tilde \SS:=\{\k: a<\varkappa_1<b, f_1(\varkappa_1)<\varkappa_2<f_2(\varkappa_1)\},$$
where the curves $\varkappa_2=f_i(\varkappa_1)$, $i=1,2$, belong to
the boundary of $\SS ^{(2)}(k,\xi )$,   have the lengths less than 1,
functions $f_i(\varkappa_1)$ are monotone, continuously
differentiable and have no inflection points. We assume that the
boundaries $\varkappa_2=f_i(\varkappa_1)$ are parameterized by
$\varkappa_1$ for definiteness. The set where
$\varkappa_1=f_i(\varkappa_2)$, $a<\varkappa_2<b$, is completely
analogous.

Next, we show that the number of points in $\tilde \SS \cap {\cal
N}(k,r_2,\k _0,J_0)$ does not exceed $2^{14}k^{\frac{2}{3}r_2}$.
Indeed, let us consider a segment $\p_{\n-\n'}$ between
     two points $\k _0 +\p_{\n}$ and $\k _0 +\p_{\n'}$ in $\tilde \SS$.
     Obviously, $\||\p_{\n-\n'}\||<2k^{r_2}$ and $p_{\n-\n'}>k^{-\mu
     r_2}$. The direction of the
     segment cannot be parallel to the axis $\varkappa_2$ by Corollary \ref{C:curves-2ind}.
     %%%%%%%%Indeed, suppose it is not so. Then, by the previous lemma, there is a part of  $\p_{\m-\m'}$ which is not in $\SS_{\m_0}(k,\xi)$. Therefore, the ends of $\p_{\m-\m'}$ are in the different elementary components of the %%%%%% set $\SS_{\m_0}(k,\xi)$.
     We
     enumerate the points $ \k _0 +\p_{\n}\in \tilde \SS \cap
     {\cal N}(k,r_2,\k _0,J_0)$ in the order of the increasing first
    coordinate  and connect subsequent points by  segments. Consider all segments with the length greater or equal to
     $\frac{1}{64}k^{-\frac{2r_2}{3}}$. The number of such segments does not
     exceed $128k^{\frac{2r_2}{3}}$, since they are much longer than the width of $\tilde \SS $ (Corollary \ref{C:curves-2ind}). It remains to estimate the
     number of segments with the length less than
     $\frac{1}{64}k^{-\frac{2r_2}{3}}$.

      First, we prove that no more than two
     segments $\p_{\n_1-\n'_1}$, $\p_{\n_2-\n'_2}$ can be parallel to
     each other, if they are in the same elementary component $\tilde \SS $. Indeed, both ends of $\p_{\n_1-\n'_1}$ are in
     $\tilde \SS$. By the previous lemma there is a piece of the segment which is not in $\tilde \SS$ (we notice that now we use the lemma for $k^{-\mu r_2}<\eta <\frac{1}{64}k^{-\frac{2r_2}{3}}$). Hence, the segment intersects one of the curves $\varkappa_2=f_i(\varkappa_1)$ twice. It follows, that there is a point on the curve,
     where the curve is parallel to the segment. Suppose another segment $\p_{\n_2-\n'_2}$  intersects the same curve. Then, there is a point on the curve,
     where the curve is parallel to the second segment. Since the curve is concave, it can not be the case. Therefore, $\p_{\n_2-\n'_2}$  intersects another  curve. It follows that no more than two
     segments $\p_{\n_1-\n'_1}$, $\p_{\n_2-\n'_2}$ can be parallel to
     each other, if they are in the same elementary component $\tilde \SS $.

     To finish the proof of the lemma  we consider two cases. Suppose $q$ in the inequality
     \eqref{qind} satisfies the estimate $q>k^{2r_2/3}$. Then, by Lemma \ref{Lattice-3ind}, the
     number of vectors  $\p_\n$,
     $\||\p_{\n}\||<2k^{r_2}$, satisfying the inequality
     $p_{\n}<\frac{1}{64}k^{-2r_2/3}$ does not exceed
     $2^{12}k^{2r_2/3}$. Since each of them can be used only twice, the
     total number of short segments  does not exceed $2^{13}k^{2r_2/3}$.

     Let  $q\leq k^{2r_2/3}$.
     If $|\epsilon _q|>\frac{1}{64}q^{-1}k^{-r_2}$. Then, obviously,
     $\frac{1}{64}k^{-2r_2/3}<|\epsilon _q|qk^{r_2/3}$. Applying Lemma
     \ref{Lattice-2ind}, we obtain that the number of segments with the
     length less than $\frac{1}{64}k^{-2r_2/3}$ is less than
     $k^{2r_2/3}$. Since each of them can be used only twice, the
     total number of short segments  does not exceed $2k^{2r_2/3}$.
     It remains to consider the case $q\leq
     k^{2r_2/3}$, $|\epsilon _q|\leq \frac{1}{64}q^{-1}k^{-r_2}$. By
     Lemma \ref{Lattice-1ind}, clusters are well separated. Considering
     that the distance between clusters is greater than
     $\frac{1}{2q}$ and the size of each cluster is less than
     $\frac{1}{8q}$, we obtain that no more than $8q$ clusters can intersect $\tilde \SS $. Indeed, the distance between two clusters is greater than $\frac{1}{2q}$.
    By Corollary  \ref{C:curves-2ind},  the set $\tilde \SS $ belongs to the $k^{-{\xi }/30J_0}$-neighborhood of
    each curve $\varkappa_2=f_i(\varkappa_1)$, $i=1,2$. Using the hypothesis of the lemma we easily get that the size of the neighborhood is $o(q^{-1})$.
    If a cluster intersects $ \tilde \SS$, its $\frac{1}{4q}$-neighborhood intersects both curves $\varkappa_2=f_i(\varkappa_1)$, $i=1,2$. Since the distance between clusters is greater than
     $\frac{1}{2q}$, the distance along the curve between its intersection with $\frac{1}{4q}$-neighborhoods of different clusters is  greater than $\frac{1}{4q}$. Considering that the lengths of the curves is
     less than 1, we obtain that no more than $8q$ clusters
     can intersect  $\tilde \SS$.
     Next, the segments with the length less than
     $\frac{1}{2}k^{-2r_2/3}$ cannot connect different clusters, since the
     distance between clusters is greater than $\frac{1}{2q}\geq
     \frac{1}{2}k^{-2r_2/3}$. Therefore, any segment of the length
     less than $\frac{1}{2}k^{-2r_2/3}$ is inside one
     cluster.   The part of the shorter curve inside the clusters has the length $L_{in}$ which is less
     than the double size of a cluster $10|\epsilon _q|k^{r_2}$ (the curve is concave) multiplied by
     the number of clusters $8q$, i.e., $L_{in}<80|\epsilon
     _q|qk^{r_2}$. If we consider the segments with the length greater
     than $|\epsilon _q|qk^{r_2/3}$, then the number of such segments
     is less than $L_{in}/|\epsilon _q|qk^{r_2/3}$, i.e., it is less
     than $80k^{2r_2/3}$. By Lemma \ref{Lattice-2ind}, the total
     number of segments of the length less than $|\epsilon
     _q|qk^{r_2/3}$ is less than $k^{2r_2/3}$. Each of them can be used only twice. Thus, the total
     number of segments is less than $162k^{2r_2/3}$.

     We proved that the number of segments in $ \tilde \SS$ does not exceed
     $2^{14}k^{2r_2/3}$. Therefore, the number of points in
     $\tilde \SS \cap {\cal N}(k,r_2,\k _0,J_0)$ does not exceed
     $2^{14}k^{\frac{2}{3}r_2}+1$. Considering that
     $k^{\gamma lr_1}>2^{15}$, we obtain that the number of points
     inside $\tilde \SS \cap {\cal N}(k,r_2,\k _0,J_0)$ does not
     exceed $k^{\frac{2}{3}r_2+\gamma l r_1}$.

     If we show that  $\SS ^{(2)}(k,\xi )$ is the union of no more than $k^{42l\gamma r_1}$ elementary components $\tilde \SS$, then  estimate
     \eqref{eqnorm2/3} easily follows. Indeed, let us consider the boundary of $\SS ^{(2)}(k,\xi)$. It is described by  curves
$D(\k , k^{2l}\pm k^{-\xi})=0$, $\k \in \R^2$. We break each curve
into elementary components as described in Appendix 6. By Lemma
\ref{elpieces} the number of such pieces is less than $k^{17l\gamma
r_1}$. With each elementary piece of the boundary we associate the
part of the adjacent connected component of $\SS ^{(2)}(k ,\xi) $,
which is in the   $k^{-\xi/30J_0}$-neighborhood of the elementary
piece. By Corollary \ref{C:curves-2ind}, every point in $\SS
^{(2)}(k,\xi)$ belongs to such a component,
%%%%%%%%%%%%Recall, that each
%%%%%%%%component has the length not greater than $1$.
%%%%%%%%Correspondingly, we
%%%%%%%%split $\SS_{\m_0}(k,\xi)$ into no more than $ck^{17l\gamma r_1}$
%%%%%%%%components, each being adjacent to an  elementary component
%%%%%%%%of a curve $D_{\m_0}(\k , k^{2l}\pm k^{-\xi})=0$ which is at the distance less  than $k^{-\xi/30J_{\m_0}(\k)}$, see Corollary\ref{C:curves-2ind}.
some  components  overlapping. Let us consider one of these
components $\hat \SS $. By construction, it is adjacent to a
boundary elementary component, which can be described in the form
$\varkappa_1=f_1(\varkappa_2)$ or $\varkappa_2=f_1(\varkappa_1)$.
Let us assume for definiteness that it is described by the formula
$\varkappa_2=f_1(\varkappa_1)$. By Corollary \ref{C:curves-2ind},
there is another boundary (described by
$\varkappa_2=f_2(\varkappa_1)$) of $\hat \SS $ in the
$k^{-\xi/30J_0}$-neighborhood of $\varkappa_2=f_1(\varkappa_1)$. It
also can be split into  no more than $k^{17l\gamma r_1}$ elementary
components. Further, each elementary component contains no more than
$k^{8l\gamma r_1}$  points $\k: D(\k , k^{2l}+ k^{-\xi})=D(\k ,
k^{2l}- k^{-\xi})=0$, unless the last equality is an identity on
this component (Bezout Theorem). We use these points to break each
elementary component into at most $k^{8l\gamma r_1}$ parts.
Correspondingly, we split the set $\hat \SS$ by lines
$\varkappa_1=C$ into at most $k^{25l\gamma r_1}$ components $\tilde
\SS $. The second boundary of $\tilde \SS $ also can be
parameterized by $\varkappa_2$, since $D_{\varkappa_2}\neq 0$ on an
elementary component of the boundary. By the definition of an
elementary component of the boundary (Appendix 6), both functions
$\varkappa_2=f_i(\varkappa_1)$ are monotone, continuously
differentiable and don't have inflection points, the length of the
corresponding curves being less than 1. Moreover, neither boundary
contains intersections with other pieces of the boundary of $\SS
^{(2)}(k,\xi )$. Thus, $\SS ^{(2)}(k,\xi )$ is the union of at most
$k^{42l\gamma r_1}$ elementary components $\tilde \SS$, each being
bounded by lines $\varkappa_i=a,b$ and elementary pieces of the
boundary of $\SS  ^{(2)}(k,\xi )$, which do not intersect with other
pieces of the boundary of $\SS ^{(2)}(k,\xi )$.
\end{proof}

\subsubsection{Model Operator for Step IV}
Let $r_3>r_2$.
%%%%%%%%%%%% Let us consider $\Omega(r_3)$ as before.
We repeat for $r_3$ the construction from the
subsection~\ref{MOforStep3}, which was done for arbitrary $r_2>r_1$.
Let $\m \in \Omega (r_3)$. We denote the $k^{\gamma r_1}$-component
containing $\m$ by $\tilde \Pi (\m)$ and the corresponding projector
by $\tilde P(\m)$. For $\m$ belonging to the same $k^{\gamma
r_1}$-component, $\tilde \Pi (\m)$ and  $\tilde P(\m)$ are the same.
 Put
\begin{equation}\label{M^3} {\MM}^{(3)}:={\MM}^{(3)}(\varphi _0, r_3)=\{\m\in \MM^{(2)}(\varphi _0, r_3)\cup \Omega _{s}^{(2)}(r_3)
:\ \varphi_0\in{\cal O}_\m^{(3)}(r_2',1)\},\end{equation} where $\Omega _{s}^{(2)}(r_3)$ is the extension of $\Omega _s^{(2)}(r_2)$ to $\Omega (r_3)$,
\begin{equation} \label{se}\Omega _{s}^{(2)}(r_3)=\{\m\in \Omega (r_3),\ 0<p_\m\leq
k^{-5r_1'}\},\end{equation}
${\cal O}_\m^{(3)}(r_2',\tau)$ is the union of the disks of the
radius $\tau k^{-r_2'}$ with the centers at poles of the resolvent
$(\tilde P(\m)(H(\k^{(2)}(\varphi ))-k^{2l}I)\tilde P(\m))^{-1}$ in the $k^{-44r_1'-2l-\delta }$-neighborhood of $\varphi _0$.
(Here $\MM^{(2)}(\varphi _0, r_3)$ is defined as in \eqref{M^2} with
$r_3$ instead of $r_2$). For $\m$ belonging to the same $k^{\gamma
r_1}$-component, the sets ${\cal O}_\m^{(3)}(r_2',\tau)$ are
identical. We say that $\m \in {\MM}^{(3)}$ is $k^{\gamma
r_1}$-resonant. The corresponding $k^{\gamma r_1}$-clusters we call
resonant too.

 Let $\varphi_0\in \omega^{(3)} (k,
\delta ,1)$. By
construction of the non-resonant set $\omega^{(3)} (k, \delta ,1)$,
we have ${\MM}^{(3)}\cap \Omega (r_2)=\emptyset $.

Further we use the property of the set $\MM^{(3)}$ formulated in the
next lemma which is an analogue of the Lemma~\ref{L:2/3-1}.

\begin{lemma}\label{L:2/3-1ind} Let $r_2'>2k^{(\gamma +\delta_0)10^{-4}r_1-2\delta}$.\footnote{ We
also notice that this condition is consistent with the restriction
$r_2'<\frac{\beta}{128}k^{\delta_0 r_1-\delta-3}$ in \eqref{r_2'}.} Let  $1/20<\gamma '<20$, $\m _0\in\Omega
(r_3)$ and $\Pi _{\m_0}$ be the $k^{\gamma
'r_2}$-neighborhood (in $\||\cdot\||$-norm) of $\m_0$. Then the set
$\Pi _{\m _0}$ contains less than $k^{\frac 23 \gamma'r_2 +50l\gamma r_1}$
elements of $\MM^{(3)}$.
\end{lemma}
%\marginpar{Zdesj $\m _0 \in \Omega (r_3)$. V lemme 7.7 eto nechto drugoe ($\m_0$ v lemme 7.7 skoree to ge chto $\m$ zdesj. Nado menyatj kak-to}

\begin{proof}
If $\m\in\MM^{(3)}$, then there is a $\varphi_*$ such that
$|\varphi_*-\varphi_0|<k^{-r_2'}$ and
$$
\hbox{det}\,\left(\tilde P(\m)(H(\k^{(2)}(\varphi_*))-k^{2l}I)\tilde P(\m)\right)=0,
$$
where $\tilde P(\m)$ is the projection corresponding to the $k^{\gamma r_1}$-cluster $\tilde \Pi (\m)$, which includes $\m$. The cluster $\tilde \Pi (\m)$ can be simple, white, grey or black.
Since $\varphi _0$ is close to $\varphi_*$, perturbation arguments give:
\begin{equation}\label{1ind}
\left\|(\tilde P(\m)(H(\k^{(2)}(\varphi_0))-k^{2l}I)\tilde P(\m))^{-1}\right\|\geq \frac{1}{4l}k^{\xi },\ \ \ \xi \geq r_2'-2l+1.
\end{equation}
 We will apply
Lemma~\ref{norm2/3} to $\tilde \Pi(\m) $ with $\xi=r_2'-2l+1$ in
order to prove the lemma in hand in the same way we proved
Lemma~\ref{L:2/3-1}, using Lemmas \ref{L:number of points-1}, \ref{L:Resnumberofpoints}.  There are some technical complications though. Here is a detailed proof.

 We start with considering simple boxes $\tilde \Pi (\m)$,  $\m\in \MM^{(3)}\cap \Omega _{s}(r_3)$. Each box has the $\||\cdot \||$-size $2k^{r_1/2}$ and contains no other than $\m$ elements of $\MM^{(2)}(\varphi _0, r_3)\cup \Omega _s(r_3)$.  Indeed,
$\k^{(2)}(\varphi_0)$ satisfies the conditions of Lemma
\ref{L:geometric2}. This means that the $k^{\delta }$-cluster around
each $\q$: $\||\p_\q\||<k^{r_{1}}$ is non-resonant. Since
$\k=\k^{(2)}(\varphi_0)+\p_\m$ is a small perturbation of
$\k^{(2)}(\varphi_0)$, the $k^{\delta }$-box around each $\m+\q$:
$\||\p_\q\||<k^{r_{1}}$ is non-resonant too. This means $\m+\q \not
\in \MM^{(2)}$. Further, $\m+\q \not \in  \Omega _s(r_3)$ by
\eqref{below}, since $\m \in  \Omega _s(r_3)$ and
$\||\p_\q\||<k^{r_1}$. Thus, $\m+\q \not \in  \MM^{(2)}(\varphi _0,
r_3)\cup\Omega _s(r_3)$. Next, we apply Lemma \ref{norm2/3} with
$\k _0=\k^{(2)}(\varphi_0)+\p_{\m_0}$, $J_0=1$, $\xi =k^{r_2'-2l+1}$,
to conclude that the number of simple
 boxes $\tilde \Pi (\m)$, $\m\in \MM^{(3)}\cap \Omega _{s}(r_3)$ does not exceed $k^{\frac 23 \gamma'r_2 +43l\gamma r_1}$. Indeed, we rewrite $\k^{(2)}(\varphi_0) +\p_{\m}$ in the form:
 $\k^{(2)}(\varphi_0)+\p_{\m} =\k  _0+\p_{\n}$,  $\n=\m -\m_0\in \Omega(\gamma 'r_2)$.
 %Obviously, $\k$ is the same for all simple boxes and $\n \in \Omega(\gamma 'r_2)$.
 By \eqref{triind},
 $\k _0 +\p_{\n} \in \SS ^{(2)}(k,\xi )$ (the operator in formula \eqref{triind} having the size
$2k^{\frac{1}{2}r_1}$ and $\xi =r_2'-2l+1$, see \eqref{1ind}).
Since, $\tilde \Pi (\m)$ is simple, $\k _0 +\p_{\n}\in {\cal
N}(k,\gamma'r_2,\k _0,1)$ (here, $\gamma $ is taken to be equal to
$1/2$ in the definition of $\SS ^{(2)}(k,\xi)$). Thus, $\k _0 +\p_{\n}\in
\SS^{(2)}(k,\xi )\cap {\cal N}(k,\gamma'r_2,\k _0,1)$. By  Lemma
\ref{norm2/3}, the number of such $\k _0 +\p_{\n}$ does not exceed
 $k^{2\gamma'r_2/3+43l\frac12 r_1}$. Therefore, the number of
 $\MM^{(3)}$ elements in simple boxes also does not exceed
 $k^{2\gamma'r_2/3+43l\frac12 r_1}$.

 Next, let us consider white clusters $\tilde \Pi (\m)$, such that
 $\xi \geq k^{\frac 16 \gamma r_1-2\delta }$. Generally speaking,
 $\tilde \Pi (\m)$ has a shape (in $\Z^4$) more complicated than a
 rectangular. However, each such cluster can be put in a box of the
 size $3k^{\gamma r_1/2+2\delta _0 }$, the box containing less than
 $k^{\frac 16 \gamma r_1 -\delta _0 r_1}$ elements of $\MM ^{(2)}$
 and the box resolvent satisfying \eqref{1ind} with
$\xi =k^{\frac 16 \gamma r_1-2\delta }$ (Lemma \ref{propw}).
Applying Lemma~\ref{norm2/3} to such boxes
($\k _0=\k^{(2)}(\varphi_0)+\p_{\m_0}$, $J_0=k^{\frac 16 \gamma r_1
-\delta _0 r_1}$, $\xi =k^{\frac 16 \gamma r_1-2\delta }$), we
obtain that the number of $\Pi _{\m _0}$ points $\m$ in  such boxes
does not exceed $k^{2\gamma'r_2/3+43l\gamma r_1}$.  Similarly, we
can treat grey boxes when $\xi \geq k^{\frac 12 \gamma r_1+2\delta
_0r_1-2\delta }$ (Lemma \ref{propg}), black boxes when $\xi \geq
k^{\gamma r_1+\delta _0r_1-2\delta }$ (Lemma \ref{propb}).
 However, in some cases $\xi $ does not satisfy the previous estimates from below. For such $\xi $ a somewhat more complicated construction is needed.
%%%%%First, the conditions of Lemma~\ref{norm2/3} do not hold if $\tilde
%%%%%\Pi(\m)$ contains more than certain number of $\m \in {\MM}^{(3)}$.
%%%%%%%%%%Second, Lemma~\ref{norm2/3} is proven for a block of a simple shape
%%%%%$|\|\m-\m_0\||<k^{\gamma r_1}$ while the shape of the cluster
%%%%%$\tilde \Pi(\m)$ is not so simple. The latter problem is discussed
%%%%%in the Appendix 9. Here, we concentrate on some technical tuning to
%%%%%address the former one.
Indeed, let us consider $(\tilde P(H(\k^{(2)}(\varphi_0))-k^{2l}I)\tilde
P)^{-1}$ for $\tilde \Pi$ being white, grey or black cluster
containing a point(s) of ${\MM}^{(3)}$. A cluster $\tilde \Pi$
consists of blocks with the minimal size $k^{\gamma r_1/6}$. Let us
create a substructure inside $\tilde \Pi$. Namely, we construct
 white, grey and black clusters corresponding to a smaller
$\gamma$ which we denote by $\tilde \gamma $,
$\tilde{\gamma}=10^{-4}\gamma$. Note, that there are no simple small clusters inside $\tilde \Pi$, since $\tilde \Pi$ is not simple. The size of these new clusters is much
smaller than $k^{\gamma r_1/6}$. However, they have  properties
analogous to those of the bigger clusters ($\gamma $). These new clusters we call
subclusters.
We assert that at least  one subcluster satisfies one of the following estimates (depending on whether this subcluster is white, grey or black):

\begin{equation}\label{estwind}
\left\|(P_{w,sub}(H(\k^{(2)}(\varphi_0))-k^{2l}I)P_{w,sub})^{-1}\right\|>k^{k^{\frac{\tilde{\gamma}r_1}{6}-2\delta}},
\end{equation}

\begin{equation}\label{estgind}
\left\|(P_{g,sub}(H(\k^{(2)}(\varphi_0))-k^{2l}I)P_{g,sub})^{-1}\right\|>
k^{k^{(\frac{\tilde{\gamma}}{2}+2\tilde{\delta}_0)r_1-2\delta}},
\end{equation}

\begin{equation}\label{estbind}
\left\|(P_{b,sub}(H(\k^{(2)}(\varphi_0))-k^{2l}I)P_{b,sub})^{-1}\right\|>
k^{k^{(\tilde{\gamma}+\tilde{\delta}_0)r_1-2\delta}},
\end{equation}
%%%%%%%%%%\begin{equation}\label{estsind}
%%%%%%%%%%\\left\|(P_{s,sub}(H(\k^{(2)}(\varphi_0))-k^{2l}I)P_{s,sub})^{-1}\right\|>k^{k^{\frac{\tilde{\gamma}r_1}{2}-2\delta}},
%%%%%%%%%%\\end{equation}
where $\tilde{\delta}_0=\tilde{\gamma}/100$ (cf. definition of
$\delta_0$). Indeed, if all subclusters satisfy the inequalities
opposite to the inequalities above, then the perturbation series for the resolvent of the
bigger cluster  ($\gamma $) (with respect to the block operator consisting of subclusters) converges, see the proof of Theorem \ref{Thm3}, in particular the proof of \eqref{5-s} -- \eqref{5-4}.
%%%%%%%(mozhet soslat'sya na rasstoyanie mezhdu clusterami, chto i daet
%%%%%%%shodimost' pertubation series? A to v etoy teoreme sovsem drugie
%%%%%%%ocenki ispol'zuyutsya...)
Hence, we have
$$
\left\|(\tilde P(H(\k^{(2)}(\varphi_0))-k^{2l}I)\tilde
P)^{-1}\right\|\leq
k^{k^{(\tilde{\gamma}+\tilde{\delta}_0)r_1-2\delta}},
$$
which contradicts to \eqref{1ind} under the hypothesis of the lemma
$r_2'>2k^{(\gamma +\delta_0)10^{-4}r_1-2\delta}$.

Now, let us consider a resonant $k^{\gamma r_1}$-cluster $\tilde
\Pi$, see \eqref{1ind}, and the substructure inside. Note that each
subcluster satisfying \eqref{estwind}-\eqref{estbind} can be treated
the same way we treated the bigger clusters for large $\xi $.
Namely, let us  consider all $k^{\gamma r_1}$-clusters $\tilde \Pi $
for which there exists a white subcluster satisfying
\eqref{estwind}. By Lemma \ref{propg} each such subcluster can be
put in a box of the size $3k^{\tilde \gamma r_1/2+2\tilde \delta _0
}$, the box resolvent satisfying \eqref{estwind}. Such box has less
than $k^{(\frac{\tilde{\gamma}}{6}-\tilde{\delta}_0)r_1}$ points $\m
$ of $ {\MM }^{(2)}$. Now, applying Lemma~\ref{norm2/3} with
$\k _0=\k^{(2)}(\varphi_0)+\p_{\m_0}$,
$J_0=k^{(\frac{\tilde{\gamma}}{6}-\tilde{\delta}_0)r_1}$,
$\xi=k^{\frac{\tilde{\gamma}r_1}{6}-2\delta}$, we obtain that the
number of  points $\m$ in white subclusters  \eqref{estwind}  does
not exceed $k^{\frac{2\gamma'r_2}{3}+43l\tilde \gamma r_1}$. Here we
notice that condition of Lemma~\ref{norm2/3} holds, since
$r_2<k^{\tilde{\delta}_0r_1-3\delta}$ by \eqref{r_2}. It follows
that the number of $k^{\gamma r_1}$-clusters $\tilde \Pi (\m)$,
containing at least one white subcluster \eqref{estwind}, does not
exceed $k^{\frac{2\gamma'r_2}{3}+43l\tilde \gamma r_1}$.

Next, we consider all $k^{\gamma r_1}$-clusters $\tilde \Pi (\m )$ for which there
exists a grey subcluster, satisfying \eqref{estgind}, but no white subclusters satisfying \eqref{estwind} . Applying
Lemma \ref{propg} and Lemma~\ref{norm2/3} with
$J_0=k^{(\frac{\tilde{\gamma}}{2}+\tilde{\delta}_0)r_1}$ and $\xi=k^{(\frac{\tilde{\gamma}}{2}+2\tilde{\delta}_0)r_1-2\delta}$, we obtain that the number of such $k^{\gamma
r_1}$-clusters $\tilde \Pi (\m )$ in $\Pi _{\m _0}$ does not exceed
$k^{\frac{2\gamma'r_2}{3}+43l\gamma r_1}$.

Similarly, applying Lemma \ref{propb} and Lemma~\ref{norm2/3} with
$\xi=k^{(\tilde{\gamma}+\tilde{\delta}_0)r_1-2\delta}$ and
$J_0 =ck^{\tilde{\gamma}r_1+3}$, we obtain
that the number of \ $k^{\gamma r_1}$-clusters $\tilde \Pi (\m )$, containing a
black subcluster \eqref{estbind} (and no grey or white subclusters, satisfying   \eqref{estwind},  \eqref{estgind}), does not exceed
$k^{\frac{2\gamma'r_2}{3}+43l\gamma r_1}$. Here, we also used
$r_2<k^{\tilde{\delta}_0r_1-3-3\delta}$.

%%%%At last, the same estimate holds for the number of $k^{\gamma
%%%%%r_1}$-clusters containing a singular simple subcluster. In this
%%%%case, we apply Lemma~\ref{norm2/3} with
%%%%%$\xi=k^{\frac{\tilde{\gamma}r_1}{2}-2\delta}$ and $J_0=1$.

Combining these estimates, we see that the number of clusters $\tilde \Pi $, containing at least one point of $\MM ^{(3)}$ does not exceed $k^{\frac{2\gamma'r_2}{3}+43l\gamma r_1}$. Taking into account that each
$k^{\gamma r_1}$-cluster has a size not greater than
$k^{\frac{3\gamma r_1}{2}+3}$ and, hence,  contains less than
$k^{6\gamma r_1+12}$ elements, we obtain that the total number of
elements of $\MM^{(3)}$ in $\Pi _{\m _0}$, does not exceed
$k^{\frac{2\gamma'r_2}{3}+50l\gamma r_1}$.

\end{proof}

We continue with constructing $k^{\gamma r_1}$-clusters in $\Omega
(r_3)$, $r_3>r_2$, the same way we did it for $\Omega (r_2)$ in
Section \ref{MOforStep3}. We call a $k^{\gamma r_1}$-cluster
resonant if it contains $\m \in {\MM}^{(3)}$, see \eqref{M^3}. Next,
we repeat the construction after Lemma \ref{L:2/3-1}. More
precisely, let us split $\Omega (r_3)\setminus \Omega (r_2)$ into
$k^{\gamma r_2}$-boxes, $\gamma =\frac{1}{5}$.

\begin{enumerate} \item {\em Simple region.} \label{simple-2} Let $\Omega _s^{(3)}(r_3)\subset \Omega _{s}^{(2)}(r_3)$ be defined by the formula:
\begin{equation} \label{Omega-s}\Omega _{s}^{(3)}(r_3)=\{\m\in \Omega (r_3),\ 0<p_\m\leq k^{- r_2'k^{2\gamma
r_1}}\}.\end{equation}
It is easy to see  that $\Omega _s^{(3)}(r_3)\subset \MM (\varphi _0,r_3)$, since $p_{\m}$ is small, see \eqref{M}, \eqref{resonance1}. Next, if  $\m\in \Omega_s^{(3)}(r_3)$, then there are
 no other elements of $\MM (\varphi _0,r_3)$ in the $k^{\delta }$-box  around  $\m$.  Further,
$\m$ itself can belong
or do not belong to $\MM^{(2)}(\varphi _0, r_3)$,  but there are
 no other elements of
$\MM^{(2)}(\varphi _0, r_3)$ in the $k^{r_1}$-box  around such $\m$.
The proof of these facts is analogous to that in Step III, see
``Simple region", page \pageref{simple}. Next, if $\m\in\Omega
_s^{(3)}(r_3)$, then  there are no other elements of $\Omega
_s^{(3)}(r_3)$ in the surrounding $\||\cdot\||$-box of the size
$k^{r_2}$, see \eqref{below}. Last, $\m$ can belong or do not belong
to $\MM^{(3)}(\varphi _0, r_3)$, but there are no other elements
from $\MM^{(3)}(\varphi _0, r_3)$ in the $k^{r_2}$-box around such
$\m$. Indeed, $\k^{(2)}(\varphi_0)$ satisfies the conditions of
Lemma \ref{L:geometric3}. This means that the $k^{\gamma
r_1}$-cluster around each $\q$: $0<\||\p_\q\||<k^{r_2}$ is
non-resonant. Since $\k=\k^{(2)}(\varphi_0)+\p_\m$ is a small
perturbation of $\k^{(2)}(\varphi_0)$, the $k^{\gamma r_1}$-box
around each $\m+\q$: $0<\||\p_\q\||<k^{r_2}$ is non-resonant too.
This means $\m+\q \not \in \MM^{(3)}(\varphi _0, r_3)$.

For each $\m \in \Omega _s^{(3)}(r_3)$ we consider its $k^{r_2/2}$-neighborhood.
The union of such boxes we call the simple region and denote it by
$\Pi _s(r_3)$. The corresponding projection is $P_s(r_3)$. Note that the
distance from the simple region to the nearest point of $\MM^{(3)}$
is greater than $\frac12 k^{r_2}$.
\item {\em Black, grey and white regions} are
defined in the same way as in the construction after Lemma
\ref{L:2/3-1} with $r_3$ instead of $r_2$, $r_2$ instead of $r_1$, $\MM^{(3)}$ instead of $\MM^{(2)}$ and the restriction $p_\m >k^{- r_2'k^{2\gamma
r_1}}$ instead of $p_\m >k^{- 5r_1'}$. We continue to use
notation $P_b, P_g, P_g', P_w, P_w'$ and $\Pi_b, \Pi_g, \Pi_g',
\Pi_w, \Pi_w'$. Sometimes, where it can lead to confusion we will
write $P_b(r_3)$ etc. to distinguish these objects from the ones
introduced in Step II.

\item {\em Non-resonant region.} \label{NRR} Now, the non-resonant region consists of two parts: the simpler part which was non-resonant
already in the previous step and the part which is new for the
current step. Namely, first we consider $k^{\delta}$-neighborhoods
of all points in the set $\MM(r_3, \varphi _0)\setminus
\left(\MM(r_2, \varphi _0)\cup
\MM^{(2)}(r_3,\varphi_0)\cup \Omega ^{(2)}_s(r_3)\right)$. The union of this
neighborhoods we denote $\Pi _{nr,\delta}$. The corresponding
projection is $P_{nr,\delta}$. These $k^{\delta }$-clusters can be treated by means of the second step. We also consider
 all points in the set $\MM^{(2)}(r_3,
\varphi _0)\cup \Omega ^{(2)}_s(r_3)\setminus \left(\MM ^{(2)}(r_2, \varphi _0)\cup
\MM^{(3)}(r_3,\varphi_0)\cup \Omega _{s}^{(3)}(r_3)\right)$.
%%%%%%%%% where$ \Omega _{ss}(r_3)=\{\m \in \Omega (r_3),\ p_\m\leq k^{- r_2'k^{2\gamma
%%%%r_1}}\}.$
 We construct simple, white, grey and black clusters around them exactly as in preparation to Step III. The union of these
clusters we denote $\Pi _{nr,r_1}$. The corresponding
projection is $P_{nr,r_1}$. The set $\Pi _{nr}:=\Pi
_{nr,\delta}\cup\Pi _{nr,r_1}$ is called the non-resonant set with
$P_{nr}$ being the corresponding projection. The part of the
non-resonant region which is outside
$\Pi_s\cup\Pi_b\cup\Pi_g\cup\Pi_w$, we denote $\Pi _{nr}'$ and the
corresponding projection by $P_{nr}'$.
\end{enumerate}

We put as before
\begin{equation}P_r(r_3):=P_s(r_3)+P_b(r_3)+P_g'(r_3)+P_w'(r_3),\ \ \ P^{(3)}:=P_r(r_3)+P_{nr}'(r_3)+P(r_2). \label{P(3)} \end{equation}
We also continue to use the similar agreement in the notation which
we set in Step II. We just note that now we use $r_2$ rather than
$r_1$ to establish equivalence between the boxes.

We continue construction from Step II. Repeating the arguments from
the proofs of Lemmas~\ref{L:black},~\ref{L:grey},~\ref{L:white} with
obvious changes (in particular, using Lemma~\ref{L:2/3-1ind} instead
of Lemma~\ref{L:2/3-1}) we obtain the following results.
%\marginpar{check coeff. in lemmas}
\begin{lemma} \label{L:blackind} \begin{enumerate} \item
Each $\Pi _b^j$ contains no more than $k^{\gamma r_2/2-\delta
_0r_2+150l\gamma r_1}$ black boxes.
\item The size of $\Pi _b^j$ in $\||\cdot \||$ norm is less than
$k^{3\gamma r_2/2+150l\gamma r_1}$.
\item Each $\Pi _b^j$ contains no more than $k^{\gamma r_2+150l\gamma r_1}$ elements of
$\MM^{(3)}$. Moreover, any box of $\||\cdot \||$-size $k^{3\gamma
r_2/2+150l\gamma r_1}$ containing $\Pi _b^j$ has no more than
$k^{\gamma r_2+150l\gamma r_1}$ elements of $\MM^{(3)}$ inside.
\end{enumerate}
\end{lemma}

\begin{lemma}\label{L:greyind} \begin{enumerate} \item
Each $\Pi _g^j$ contains no more than $k^{\gamma r_2/3+2\delta
_0r_2}$ grey boxes.
\item The size of $\Pi _g^j$ in $\||\cdot \||$ norm is less than
$k^{5\gamma r_2/6+4\delta _0r_2}$.
\item Each $\Pi _g^j$ contains no more than $k^{\gamma r_2/2+\delta _0r_2}$ elements of
$\MM^{(3)}$.\end{enumerate}
\end{lemma}

\begin{lemma}\label{L:whiteind} \begin{enumerate} \item The size of $\Pi _w^j$ in $\||\cdot \||$ norm is less than
$k^{\gamma r_2/3-\delta _0r_2}$.
\item Each $\Pi _w^j$ contains no more
than $k^{\gamma r_2/6-\delta _0r_2}$ points of $\MM^{(3)}$.
\end{enumerate}
\end{lemma}

The construction of the rest of Section~\ref{MOforStep3} stays unchanged. Let us
introduce corresponding notation, formulate the results and provide
some comments.

Next lemmas are the analogues of
Lemmas~\ref{L:Pnr},~\ref{L:Pr},~\ref{L:Ps}.

\begin{lemma}\label{L:Pnrn}Let $\varphi _0\in \omega^{(3)}(k,\delta ,\tau )$,
$|\varphi-\varphi _0|<k^{-k^{r_1}}$. Then,
\begin{equation}\label{Pnrn}
\left\|\Bigl(P_{nr}\bigl(H(\k^{(3)}(\varphi
))-k^{2l}I\bigr)P_{nr}\Bigr)^{-1}\right\|<k^{r_2'k^{2\gamma
r_1}}k^{r_2'}\leq k^{k^{3\gamma r_1}}.
\end{equation} \end{lemma}
\begin{proof} Construction in Section \ref{S:3} is made for an arbitrary large $r_2$. Here we repeat it for $r_3$. We use Lemma \ref{L:Pnr} for $\Pi _{nr, \delta }$ Lemma \ref{L:Pr} for white, grey and black clusters ($\varepsilon _0=k^{-r_2'}$).  We also use Lemma \ref{L:Ps} ($p_{\m}>k^{-r_2'k^{2\gamma r_1}}$, $\varepsilon _0=k^{-r_2'}$), for simple clusters in $\Pi _{nr, r_1}$.  We also use \eqref{dk0-3}. All together the estimates for the clusters resolvents yield \eqref{Pnrn}.  The estimate \eqref{Pnrn} is stable when $|\varphi-\varphi _0|<k^{-k^{r_1}}$, since
$k^{-k^{r_1}+2l }=o(k^{-k^{3\gamma r_1}})$.
\end{proof}

\begin{lemma}\label{L:Prn} Let $\varphi _0\in \omega^{(3)}(k,\delta ,\tau )$,
%%%%%%%%$P_{r}^{(2)j_2}=P_{r}^{(2)j_2}(\varphi _0, r_2)$, $j_2=1,...,J_2$,
and $|\varphi-\varphi _0|<k^{-k^{r_1}}$, $i=1,2,3$. Then,
\begin{enumerate}
\item The number of poles of the resolvent $\Bigl(P_i\bigl(H(\k^{(3)}(\varphi
))-k^{2l}I\bigr)P_i\Bigr)^{-1}$ in the disc $|\varphi-\varphi
_0|<k^{-k^{r_1}}$ is no greater than $N_i^{(2)}$, where $N_1^{(2)}=k^{\gamma
r_2+150l\gamma r_1}$, $N_2^{(2)}=k^{\gamma r_2/2+\delta _0r_2}$,
$N_3^{(2)}=k^{\gamma r_2/6-\delta _0r_2}$.
\item Let $\varepsilon$ be the distance to the nearest pole of
the resolvent in ${\cal W}^{(3)}$ and
$\varepsilon_0=\min\{\varepsilon,\,k^{-r_2'}\}$. Then the following
estimates hold:

\begin{align}\label{Pr-1n}
\left\|\Bigl(P_i\bigl(H(\k^{(3)}(\varphi
))-k^{2l}I\bigr)P_i\Bigr)^{-1}\right\|<k^{2r_2'k^{2\gamma
r_1}}k^{r_2'}\left(\frac{k^{-r_2'}}{\varepsilon
_0}\right)^{N_{i}^{(2)}}\leq \cr k^{k^{3\gamma
r_1}}\left(\frac{k^{-r_2'}}{\varepsilon _0}\right)^{N_{i}^{(2)}},
\end{align}
\begin{align}\label{Pr-2n}
\left\|\Bigl(P_i\bigl(H(\k^{(3)}(\varphi
))-k^{2l}I\bigr)P_i\Bigr)^{-1}\right\|_1<k^{2r_2'k^{2\gamma
r_1}}k^{r_2'+8\gamma r_2}\left(\frac{k^{-r_2'}}{\varepsilon
_0}\right)^{N_{i}^{(2)}}\leq \cr k^{k^{3\gamma
r_1}}\left(\frac{k^{-r_2'}}{\varepsilon _0}\right)^{N_{i}^{(2)}}.
\end{align}
\end{enumerate}
\end{lemma}
\begin{proof} The proof of this lemma is analogous to that of Lemma \ref{L:Pr} up to the replacement of ${\MM}^{(2)}$ by ${\MM}^{(3)}$, $\OO^{(2)}_{\m}$ by $\OO^{(3)}_{\m}$, and the shift of indices: $\delta $ to $r_1$, $r_1$ to
$r_2$, etc. We apply Lemmas \ref{L:blackind}--\ref{L:whiteind} instead of \ref{L:black}--\ref{L:white} and Lemmas \ref{L:Pr}, \ref{L:Ps} with $\varepsilon _0=k^{-r_2'}$ and $p_\m>k^{-r_2'k^{2\gamma r_1}}$ instead of Lemma
\ref{L:estnonres1}. We also note that $N_i^{(1)}<k^{2\gamma r_1}$ in \eqref{Pr-1}, \eqref{Pr-2}. \end{proof}

\begin{lemma}\label{L:Psn} Let $\varphi _0\in \omega^{(3)}(k,\delta ,\tau )$. Then, the operator
$\Bigl(P_s^j\bigl(H(\k^{(3)}(\varphi
))-k^{2l}I\bigr)P_s^j\Bigr)^{-1}$ has no more than one pole in the
disk $|\varphi-\varphi _0|<k^{-k^{r_1}}$. Moreover,
\begin{equation}\label{Ps-1n}
\left\|\Bigl(P_s^j\bigl(H(\k^{(3)}(\varphi
))-k^{2l}I\bigr)P_s^j\Bigr)^{-1}\right\|<\frac{8k^{-2l+1}}
{p_\m\varepsilon _0},
\end{equation}
\begin{equation}\label{Ps-2n}
\left\|\Bigl(P_s^j\bigl(H(\k^{(3)}(\varphi
))-k^{2l}I\bigr)P_s^j\Bigr)^{-1}\right\|_1<\frac{8k^{-2l+1+
4r_2}}{p_\m\varepsilon_0},
\end{equation}
$\varepsilon _0=\min\{\varepsilon,\,k^{-r_2'}\}$, where
$\varepsilon$ is the distance to the pole of the operator.
\end{lemma}

Note that $p_{\m}>k^{-\mu r_3}$ when $\m \in \Omega (r_3)$. The
analogues of Lemma~\ref{L:boundary} and Corollary~\ref{C:PHP-2} also
hold.

\subsubsection{Resonant and Nonresonant Sets for Step IV \label{GSIV}}

We divide $[0,2\pi )$ into $[2\pi k^{k^{r_1}}]+1$ intervals
$\Delta_m^{(3)}$ with the length not bigger than $k^{-k^{r_1}}$. If a
particular interval belongs to $\OO^{(3)}$ we ignore it; otherwise,
let $\varphi_0(m)\not\in\OO^{(3)}$ be a point inside the $\Delta_m^{(3)}$.
Let
\begin{equation}\W_m^{(3)}=\{\varphi \in \W^{(3)}:\ | \varphi -\varphi
_0(m)|<4k^{-k^{r_1}}\}. \label{W2mind} \end{equation} Clearly,
neighboring sets $\W_m^{(3)}$  overlap (because of the multiplier 4
in the inequality), they cover $\hat \W^{(3)}$ , which is
%%%%%$\omega^{(2)}$ and its $k^{-44r_1'-2l-\delta
%%%%%}$-neighborhood, i.e
the restriction of $\W^{(3)}$ to the $2k^{-k^{r_1}}$-neighborhood of
$[0,2\pi )$. For each $\varphi \in \hat \W^{(3)}$ there is an $m$
such that $|\varphi -\varphi _{0}(m)|<4k^{-k^{r_1}}$. We consider
the poles of the resolvent  $\left(P^{(3)}
(H(\k^{(3)}(\varphi))-k^{2l})P^{(3)}\right)^{-1}$ in $\hat
\W_m^{(3)}$ and denote them by $\varphi^{(3)} _{mj}$, $j=1,...,M_m$.
As before, the resolvent has a block structure. The number of blocks
clearly cannot exceed the number of elements in $\Omega (r_3)$, i.e.
$k^{4r_3}$. Using the estimates for the number of poles for each
block, the estimate being provided by Lemma \ref{L:Prn}, Part 1, we
can roughly estimate the number of poles of the resolvent by
$k^{4r_3+r_2}$. Next, let $r_3'>k^{r_1}$ and $\OO^{(4)}_{mj}$ be the
disc of the radius $k^{-r_3'}$ around $\varphi ^{(3)}_{mj}$.
\begin{definition} The set
\begin{equation}\OO^{(4)}=\cup _{mj}\OO^{(4)}_{mj} \label{O4}
\end{equation}
we call the forth resonant set. The set
\begin{equation}\W^{(4)}= \hat\W^{(3)}\setminus \OO^{(4)}\label{W4}
\end{equation}
is called the forth non-resonant set. The set
\begin{equation}\omega^{(4)}= \W^{(4)}\cap [0,2\pi) \label{w4}
\end{equation}
is called the forth real non-resonant set. \end{definition} The
following statements can be proven in the same way as
Lemmas~\ref{L:geometric3}, \ref{4.16} and \ref{estnonres0-1}.
\begin{lemma}\label{L:geometric4n}Let  $r_3'>\mu r_3>k^{r_1}$, $\varphi \in \W^{(4)}$, $\varphi
_0(m)$ corresponds  to an interval $\Delta _m^{(3)}$ containing $\Re
\varphi $. Let $\Pi $ be one of the components $\Pi _s^j(\varphi
_0(m))$, $\Pi _b^j(\varphi _0(m))$, $\Pi _g^j(\varphi _0(m))$, $\Pi
_w^j(\varphi _0(m))$ and
 $P(\Pi )$ be the projection corresponding to $\Pi $. Let also
 $\varkappa \in \C: |\varkappa-\varkappa^{(3)}(\varphi )|<k^{-r_3'k^{2\gamma
r_2}}$. Then,
\begin{equation} \label{March3-2n4} \left\|\left(P(\Pi )
\left(H\big(\k(\varphi )\big)-k^{2l}I\right)P(\Pi
)\right)^{-1}\right\|<k^{2\mu r_3+r_3'N^{(2)}},\end{equation}
\begin{equation} \label{March3-3n4}
\left\|\left(P(\Pi )\left(H\big(\k(\varphi
)\big)-k^{2l}I\right)P(\Pi )\right)^{-1}\right\|_1<k^{(2\mu+1)
r_3+r_3'N^{(2)}},\end{equation} $N^{(2)}$ corresponding to the color
of $\Pi $ ($N^{(2)}=1,\ k^{\gamma r_2+150l\gamma r_1},\ k^{\gamma
r_2/2+\delta _0r_2},\ k^{\gamma r_2/6-\delta _0r_2}$ for simple,
black, grey and white clusters, correspondingly).
\end{lemma}
By total size of the set $\OO^{(4)}$ we mean the sum of the sizes of
its connected components.
\begin{lemma}\label{7.10} Let $r_3'\geq (\mu+10)r_3$, $r_3>k^{r_1}$. Then, the size of each connected component of
 $\OO^{(4)}$ is less
than $32k^{4r_3-r_3'}$. The total size of $\OO^{(4)}$ is less than
$k^{-r_3'/2}$.
\end{lemma}

%%%%%%%{\it We say that $\varphi\in I_j$ is a resonant point of the second
%%%%%%%kind if its distance from the nearest pole of the operator
%%%%%%%$(P(H(\k_j)-k^2)P)^{-1}$, $\k_j=k(\cos \varphi_0(j), \sin
%%%%%%%\varphi_0(j))$  is less than $k^{-r_1'}$.} The set of all resonant
%%%%%%%points of the second kind we denote by $\OO^{(2)}$. Notice that
%%%%%%%$$meas\,(\OO^{(2)}\cap[0,2\pi))\leq(2\pi k^{2+\delta(40\mu+1)}+1)ck^{4r_1}k^{-r_1'}\leq k^{-r_1}.$$
%%%%%%%{\it Here and in what follows we assume that $r_1'>5r_1+2$.} We also
%%%%%%%put $\W^{(2)}:=\W ^{(1)}\setminus\OO^{(2)}$,
%%%%%%%$\W^{(2)}_j:=\W^{(1)}\cap\{\varphi:\
%%%%%%%|\varphi-\varphi_0(j)|<k^{-2-\delta(40\mu+1)}\}$.

\begin{lemma}\label{estnonres0-1n} Let $\varphi\in\W^{(3)}$ and
$C_4$ be the circle $|z-k^{2l}|=k^{-2r_3'k^{2\gamma r_2}}$. Then
$$
\left\|\left(P(r_2)(H(\k^{(3)}(\varphi))-z)P(r_2)\right)^{-1}\right\|\leq
4^3k^{2r_3'k^{2\gamma r_2}}. $$ \end{lemma} We prove this lemma
using \eqref{tik-tak*}.

\section{STEP IV}

\subsection{Operator $H^{(4)}$. Perturbation Formulas}
Let $P(r_3)$ be an orthogonal projector onto $\Omega(r_3):=\{\m:\
|\|\p_\m\||\leq k^{r_3}\}$ and $H^{(4)}=P(r_3)HP(r_3) $. From now on,
we assume \begin{equation} \label{r_2IV}k^{r_1}<r_3<k^{\gamma
10^{-7}r_2},\ \ \ \ \ k^{2\gamma 10^{-4}r_2}<r_3'<k^{\delta_0
r_2/2}.\end{equation} We consider $H^{(4)}(\k^{(3)}(\varphi ))$ as a
perturbation of $\tilde H^{(3)}(\k^{(3)}(\varphi ))$:
$$\tilde H^{(3)}:=\tilde
P_j^{(3)}H\tilde P_j^{(3)}+\left(P(r_3)-\tilde
P_j^{(3)}\right)H_0,$$ where $H=H(\k^{(3)}(\varphi ))$,
$H_0=H_0(\k^{(3)}(\varphi ))$ and $\tilde P_j^{(3)}$ is the
projection $P^{(3)}$ corresponding to $\varphi _{0}(j)$ in the
interval $\Delta _j^{(3)}$ containing $\varphi $, see \eqref{P(3)}.  Note that the operator
$\tilde H^{(3)}$ has a block structure, the block $\tilde
P_j^{(3)}H\tilde P_j^{(3)}$ being composed of smaller blocks
$P_iHP_i$, $i=0,...,5$.
%%%%%\begin{equation}
%%%%%\\begin{split}
%%%%%\&\tilde H^{(2)}=P(r_1)H(\k^{(2)}:=
%%%%%\))P(r_1)+P_{j,s}^{(2)}H(\k^{(2)}(\varphi
%%%%%\))P_{j,s}^{(2)}+P_{j,b}^{(2)}H(\k^{(2)}(\varphi ))P_{j,b}^{(2)}\cr &
%%%%%\+P_{j,g}^{'(2)}H(\k^{(2)}(\varphi
%%%%%\))P_{j,g}^{'(2)}+P_{j,w}^{'(2)}H(\k^{(2)}(\varphi
%%%%%\))P_{j,w}^{'(2)}+P_{j,nr}^{'(2)}H(\k^{(2)}(\varphi
%%%%%\))P_{j,nr}^{'(2)}\cr &
%%%%%\+\left(P(r_2)-\tilde
%%%%%\P_j^{(2)}\right)H_0(\k^{(2)}(\varphi ))\left(P(r_2)-\tilde
%%%%%\P_j^{(2)}\right),
%%%%%\\end{split}
%%%%%\\end{equation}
By analogy with \eqref{W2*}--\eqref{G3},
\begin{equation}W^{(3)}=H^{(4)}-\tilde H^{(3)}=P(r_3)VP(r_3)-\tilde P_j^{(3)}V\tilde P_j^{(3)}, \label{W2IV}\end{equation}
\begin{equation}\label{g3IV} g^{(4)}_r({\k}):=\frac{(-1)^r}{2\pi
ir}\hbox{Tr}\oint_{C_4}\left(W^{(3)}(\tilde
H^{(3)}({\k})-zI)^{-1}\right)^rdz,
\end{equation} \begin{equation}\label{G3IV}
G^{(4)}_r({\k}):=\frac{(-1)^{r+1}}{2\pi i}\oint_{C_4}(\tilde
H^{(3)}({\k})-zI)^{-1}\left(W^{(3)}(\tilde
H^{(3)}({\k})-zI)^{-1}\right)^rdz,
\end{equation}
where $C_4$ is the circle $|z-k^{2l}|=\varepsilon _0^{(4)}$,
$\varepsilon _0^{(4)}=k^{-2r_3'k^{2\gamma r_2}}.$

The proof of the following statements is analogous to the one in the
previous step (see Theorem~\ref{Thm3}, Corollary~\ref{corthm3} and
Lemma~\ref{L:derivatives-3}) up to the replacement of  $r_3$ by $r_4$, $r_2$ by $r_3$, $r_1$ by $r_2$, etc.

\begin{theorem} \label{Thm3IV} Suppose  $k>k_*$, $\varphi $ is in
the real  $k^{-r_3'-\delta }$-neighborhood of $\omega
^{(4)}(k,\delta,\tau )$ and $\varkappa\in\R$,
$|\varkappa-\varkappa^{(3)}(\varphi )|\leq \varepsilon
^{(4)}_0k^{-2l+1-\delta }$, $\k=\varkappa(\cos \varphi ,\sin \varphi
)$. Then,  there exists a single eigenvalue of $H^{(4)}({\k})$ in
the interval $\varepsilon _4( k,\delta,\tau )=\left(
k^{2l}-\varepsilon _0^{(4)}, k^{2l}+\varepsilon _0^{(4)}\right)$. It
is given by the absolutely converging series:
\begin{equation}\label{eigenvalue-3IV}\lambda^{(4)}({\k})=\lambda^{(3)}({\k})+
\sum\limits_{r=2}^\infty g^{(4)}_r({\k}).\end{equation} For
coefficients $g^{(4)}_r({\k})$ the following estimates hold:
\begin{equation}\label{estg3IV} |g^{(4)}_r({\k})|<k^{-\frac{\beta}{5}
k^{r_2-r_1 }-\beta (r-1)}.
\end{equation}
The corresponding spectral projection is given by the series:
\begin{equation}\label{sprojector-3IV}
\E ^{(4)}({\k})=\E^{(3)}({\k})+\sum\limits_{r=1}^\infty
G^{(4)}_r({\k}), \end{equation} $\E^{(3)}({\k})$ being the spectral
projection of $H^{(3)}(\k)$. The operators $G^{(4)}_r({\k})$ satisfy
the estimates:
\begin{equation}
\label{Feb1a-3IV}
\left\|G^{(4)}_r({\k})\right\|_1<k^{-\frac{\beta}{10} k^{r_2-r_1 }
-\beta r},
\end{equation}
\begin{equation}G^{(4)}_r({\k})_{\s\s'}=0,\ \mbox{when}\ \ \
2rk^{\gamma r_2+150l\gamma
r_1}+3k^{r_2}<\||\p_\s\||+\||\p_{\s'}\||.\label{Feb6a-3IV}
\end{equation}
\end{theorem}

\begin{corollary}\label{corthm3IV} For the perturbed eigenvalue and its spectral
projection the following estimates hold:
 \begin{equation}\label{perturbation-3IV}
\lambda^{(4)}({\k})=\lambda^{(3)}({\k})+ O_2\left(k^{-\frac15 \beta
k^{r_2-r_1 }-\beta }\right),
\end{equation}
\begin{equation}\label{perturbation*-3IV}
\left\|\E^{(4)}({\k})-\E^{(3)}({\k})\right\|_1<2k^{-\frac{\beta}{10}
k^{r_2-r_1 }-\beta }.
\end{equation}
\begin{equation}
\left|\E^{(4)}({\k})_{\s\s'}\right|<k^{-d^{(4)}(\s,\s')},\ \
\mbox{when}\ \||\p_\s\||>4k^{r_2} \mbox{\ or }
\||\p_{\s'}\||>4k^{r_2 },\label{Feb6b-3IV}
\end{equation}
$$d^{(4)}(\s,\s')=\frac18(\||\p_\s\||+\||\p_{\s'}\||)k^{-\gamma r_2-150l\gamma r_1}\beta +\frac{1}{10}\beta
k^{r_2-r_1 }.$$
\end{corollary}

\begin{lemma} \label{L:derivatives-3IV}Under conditions of Theorem \ref{Thm3IV} the following
estimates hold when $\varphi \in \omega ^{(4)}(k,\delta )$ or its
complex $k^{-r_3'-\delta}$-neighborhood and $\varkappa\in \C:$
$|\varkappa-\varkappa^{(3)}(\varphi )|<\varepsilon
^{(4)}_0k^{-2l+1-\delta}$.
\begin{equation}\label{perturbation-3cIV}
\lambda^{(4)}({\k})=\lambda^{(3)}({\k})+O_2\left(k^{-\frac 15 \beta
k^{r_2-r_1}-\beta }\right),
\end{equation}
\begin{equation}\label{estgder1-3kIV}
\frac{\partial\lambda^{(4)}}{\partial\varkappa}=\frac{\partial\lambda^{(3)}}{\partial\varkappa}
+O_2\left(k^{-\frac 15 \beta k^{r_2-r_1}-\beta  }M_2\right), \  \ \
\ M_2:=\frac{k^{2l-1+\delta}}{\varepsilon ^{(4)}_0},\end{equation}
\begin{equation}\label{estgder1-3phiIV}\frac{\partial\lambda^{(4)}}{\partial \varphi }=\frac{\partial\lambda^{(3)}}{\partial \varphi }+
O_2\left(k^{-\frac 15 \beta k^{r_2-r_1}-\beta+r_3'+\delta }\right),
 \end{equation}
\begin{equation}\label{estgder2-3IV}
\frac{\partial^2\lambda^{(4)}}{\partial\varkappa^2}=
\frac{\partial^2\lambda^{(3)}}{\partial\varkappa^2}+
O_2\left(k^{-\frac 15 \beta k^{r_2-r_1}-\beta  }M_2^2\right),
\end{equation}
\begin{equation} \label{gulf2-3IV}
\frac{\partial^2\lambda^{(4)}}{\partial\varkappa\partial \varphi
}=\frac{\partial^2\lambda^{(3)}}{\partial\varkappa\partial \varphi
}+ O_2\left(k^{-\frac 15 \beta k^{r_2-r_1}-\beta+r_3'+\delta
}M_2\right),
\end{equation}
\begin{equation} \label{gulf3-3IV}
\frac{\partial^2\lambda^{(4)}}{\partial\varphi
^2}=\frac{\partial^2\lambda^{(3)}}{\partial\varphi
^2}+O_2\left(k^{-\frac 15 \beta k^{r_2-r_1}-\beta+2r_3'+2\delta
}\right).
\end{equation}\end{lemma}

\begin{corollary} \label{"O"*} All ``$O_2$"-s on the right hand sides of \eqref{perturbation-3cIV}-\eqref{gulf3-3IV} can be written as $O_1\left(k^{-\frac {1}{10} \beta k^{r_2-r_1}}\right)$.
\end{corollary}

\begin{remark}
In the proof of Theorem~\ref{Thm3IV} and similar statements in every
further step of the induction we obtain the estimate of the form
\eqref{March5}. It is important to notice that the right hand side
of these estimates is always exactly $k^{-\beta}$. It can't become
better since it comes from the estimate of the free resolvent on the
set of points satisfying $\left||\k+\p_\m|^{2l}_\R-k^{2l}\right|\geq
k^{2l-2-40\mu\delta}$. What changes is the first term in
the perturbation series, see  e.g. \eqref{estg3}, \eqref{Feb1a-3} vs \eqref{estg3IV}, \eqref{Feb1a-3IV}.
%%%%%All perturbation series
%%%%%converge exponentially faster as we go further in the induction and
%%%%%in particular, the estimates of the form \eqref{March5} stay uniform
%%%%%during induction process.
\end{remark}

\subsection{\label{IS3IV}Isoenergetic Surface for Operator $H^{(4)}$}

The following statement is an analogue of Lemma~\ref{ldk-3}.

\begin{lemma}\label{ldk-3IV} \begin{enumerate}
\item For every $\lambda :=k^{2l}$,  $k>k_*$, and $\varphi $ in the real  $\frac{1}{2} k^{-r_3'-\delta }$-neighborhood
of $\omega^{(4)}(k,\delta, \tau )$,$\ $ there is a unique
$\varkappa^{(4)}(\lambda, \varphi )$ in the interval
$I_3:=[\varkappa^{(3)}(\lambda, \varphi )-\varepsilon
^{(4)}_0k^{-2l+1-\delta},\varkappa^{(3)}(\lambda, \varphi
)+\varepsilon ^{(4)}_0k^{-2l+1-\delta}]$, such that
    \begin{equation}\label{2.70-3IV}
    \lambda^{(4)} \left(\k
^{(4)}(\lambda ,\varphi )\right)=\lambda ,\ \ \k ^{(4)}(\lambda
,\varphi ):=\varkappa^{(4)}(\lambda ,\varphi )\vec \nu(\varphi).
    \end{equation}
\item  Furthermore, there exists an analytic in $ \varphi $ continuation  of
$\varkappa^{(4)}(\lambda ,\varphi )$ to the complex  $\frac{1}{2}
k^{-r_3'-\delta }$-neighborhood of $\omega^{(4)}(k,\delta, \tau )$
such that $\lambda^{(4)} (\k ^{(4)}(\lambda, \varphi ))=\lambda $.
Function $\varkappa^{(4)}(\lambda, \varphi )$ can be represented as
$\varkappa^{(4)}(\lambda, \varphi )=\varkappa^{(3)}(\lambda, \varphi
)+h^{(4)}(\lambda, \varphi )$, where
\begin{equation}\label{dk0-3IV} |h^{(4)}(\varphi )|=O_1\left(k^{-\frac 15 \beta k^{r_2-r_1}-\beta -2l+1
}\right),
\end{equation}
\begin{equation}\label{dk-3IV}
\frac{\partial{h}^{(4)}}{\partial\varphi}= O_2\left(k^{-\frac 15 \beta
k^{r_2-r_1}-\beta -2l+1 +r_3'+\delta }\right),\ \ \ \ \
\frac{\partial^2{h}^{(4)}}{\partial\varphi^2}= O_4\left(k^{-\frac 15
\beta k^{r_2-r_1}-\beta -2l+1 +2r_3'+2\delta }\right).
\end{equation} \end{enumerate}\end{lemma}

Let us consider the set of points in $\R^2$ given by the formula:
$\k=\k^{(4)} (\varphi), \ \ \varphi \in \omega ^{(4)}(k,\delta, \tau
)$. By Lemma \ref{ldk-3IV} this set of points is a slight distortion
of ${\cal D}_{3}$. All the points of this curve satisfy the equation
$\lambda^{(4)}(\k ^{(4)}(\varphi ))=k^{2l}$. We call it isoenergetic
surface of the operator $H^{(4)}$ and denote by ${\cal D}_{4}$.

\section{Induction}

\subsection{Inductive formulas for $r_n$} Now, we are ready to introduce the induction. In fact, STEP IV has
been the first inductive step. Here, for the sake of convenience, we
reformulate the main statements from the previous step in terms of
$r_n$, $n\geq3$, and provide necessary comments. First, we choose
\begin{equation}\label{indrn}
k^{r_{n-2}}<r_n<k^{\gamma10^{-7}r_{n-1}},\ \ \
k^{2\gamma10^{-4}r_{n-1}}<r_n'<k^{\delta_0r_{n-1}/2},\ \ \ n\geq3.
\end{equation}
\subsection{Preparation for Step $n+1$, $n\geq 4$}
\subsubsection{Properties of the Quasiperiodic Lattice. Induction}\label{Lattice-Induction}
Here we prove the inductive version of the results from
Section~\ref{Quasiperiodicgeomcont}. We consider
$\p_\m=2\pi(\s_1+\alpha\s_2)$ with integer vectors $\s_j$ such that
$|\s_j|\leq 4k^{r_{n-1}}$. We repeat the arguments from the
beginning of Section~\ref{geomIII}. Namely, let $(q,p)\in\Z^2$ be a
pair such that $0<q\leq 4k^{r_{n-1}}$ and
\begin{equation}\label{qind-last}
|\alpha q+p|\leq 16k^{-r_{n-1}}.
\end{equation}
We choose a pair $(p,q)$ which gives the best approximation. In
particular, $p$ and $q$ are mutually simple. Put
$\epsilon_q:=\alpha+\frac{p}{q}$. We have
\begin{equation}k^{-2r_{n-1}\mu}\leq|\epsilon_q|\leq 16q^{-1}k^{-r_{n-1}}.\label{epsilon_qind-last}\end{equation}
The analogs of Lemmas \ref{Lattice-1ind}--\ref{Lattice-3ind} hold
with $n-1$ instead of $2$.

We consider the matrix $H^{(n-1)}(\k)=P(\gamma r_{n-2})H(\k )(\gamma
r_{n-2})$ where $\k \in \R^2$, $P(\gamma r_{n-2})$ is the orthogonal
projection corresponding to $\Omega (\gamma r_{n-2})$. \footnote{It
is a slight abuse of notations, since $H^{(n-1)}$ in Step $n-1$ was
defined for $\gamma =1$.} We construct the block structure in
$H^{(n-1)}(\k)$ analogous to that in Step $n-1$.  The difference is
that now we consider any $\k \in \R^2$, not only $\k$ being close to
$\k^{(n-2)}(\varphi )$. Correspondingly, we define non-resonant $\m$
not in terms of $\varphi $, but in more general terms of
inequalities providing convergence of perturbation series. Indeed,
we call $\m \in \Omega (\gamma r_{n-2})$ non-resonant if  (cf.
\eqref{resonance1})
\begin{equation}
\left||\k+\p_{\m }|^{2}-k^{2}\right|>k^{-40\mu\delta}.
\label{Aug29a-last}
\end{equation}
Obviously, this estimate is stable in the $k^{-41\mu \delta
-1}$-neighborhood of a given $\k$. Hence, the definition of a
non-resonant $\m$ is stable in this neighborhood up to a multiplier
$1+o(1)$ in the r.h.s. of \eqref{Aug29a-last}.  Around each resonant
$\m$ we construct $k^\delta$-boxes/clusters (see \eqref{defP}). Let
$P_{\m}$ be the projection on the $k^{\delta }$-cluster containing
$\m $. If
\begin{equation}
\left\|(P_{\m}(H(\k)-k^{2l})P_{\m})^{-1}\right\|<k^{4\gamma r_1'}
\label{Aug29b-last} \end{equation} (cf. \eqref{Mon3a**}), then we
call the $k^{\delta }$-cluster effectively non-resonant  for a given
$\k$.
%%%%%% \footnote{Formally a non-resonant set is defined geometrically by  Definition \ref{D:3.17}. Every non-resonant set in the sense of the geometric definition is effectively %%%%non-resonant by Lemma \ref{L:geometric2}.}
 Note, that the above
estimate and, therefore, the definition of an effectively
non-resonant $k^{\delta}$-cluster is  stable in the $k^{-4\gamma
r_1'-2l+1-\delta }$-neighborhood of a given $\k$. The
$k^{\delta}$-clusters, where \eqref{Aug29b-last} is not valid, are
called  effectively resonant $k^{\delta}$-clusters.  Around each
effectively resonant $k^{\delta }$-cluster, we construct $k^{\gamma
r_1}$-clusters. We sort these clusters into four types: simple,
white and black clusters as in Section~\ref{MOforStep3}, using the
term ``$\m $ is effectively resonant" instead of ``$\m \in \MM
^{(2)}$". There is no need to consider a special case of simple
clusters here. Note that Lemmas \ref{L:black} -- \ref{L:white} are
valid for an arbitrary $\k$, since they are based on Lemmas
\ref{L:number of points-1} \ref{L:Resnumberofpoints} proven for an
arbitrary $\k$. Be analogy with \eqref{March3-2}, a $k^{\gamma
r_1}$-cluster is called effectively non-resonant if
%%%%%%%%%%Using Lemmas \ref{L:number of points-1}}, \ref{L:resnumberofpoints}, we prove the analog of \ref{.} being valid. The cluster of the size $k^{\gamma r_1}$ is called effectively non-resonant if
\begin{equation}
\left\|(P_{\m}(H(\k)-k^{2l})P_{\m})^{-1}\right\|<k^{2\mu
r_2+r_2'N^{(1)}_i}, \label{Aug29c-last} \end{equation} where
$N^{(1)}_i$ corresponds to the color of a $k^{\gamma r_1}$-cluster,
$N^{(1)}_i=k^{\gamma r_1+3}, k^{\gamma r_1/2+\delta _0r_1}$ or
$k^{\gamma r_1/6-\delta _0r_1}$. If $n=4$ we stop here. If $n>4$, we
surround effectively resonant $k^{\gamma r_1}$-clusters by blocks of
the next size, etc.  The analogs of Lemmas \ref{L:black} --
\ref{L:white} are valid, see Lemmas \ref{L:blackind} --
\ref{L:whiteind}, \ref{L:blackindlast} -- \ref{L:whiteindlast}.
Eventually, the $k^{\gamma r_{n-3}}$-cluster  is effectively
non-resonant if
\begin{equation}
\left\|(P_{\m}(H(\k)-k^{2l})P_{\m})^{-1}\right\|<k^{\mu
r_{n-2}+r_{n-2}'N^{(n-3)}_i}, \label{Aug29d-last} \end{equation}
where $N^{(n-3)}_i$ is $N^{(n-3)}_i= k^{\gamma r_{n-3}+150l\gamma
r_{n-4} }, k^{\gamma r_{n-3}/2+\delta _0r_{n-3} }, k^{\gamma
r_{n-3}/6-\delta _0r_{n-3} }$, depending on the color of the cluster
(cf. \eqref{Aug29c-last}, \eqref{March3-2n4}). Further we put
$150l\gamma r_{0}=3$. This will make \eqref{Aug29c-last} to be a
special case of \eqref{Aug29d-last} ($n=4$). Thus, we have
constructed a block structure in $H^{(n-1)}(\k)$, which is stable in
the $k^{-\rho _{n-2}}$-neighborhood of a given $\k$, where $\rho
_1=4\gamma r_1'+2l-1+\delta $ and $$\rho _{n-2}=\mu
r_{n-2}+r_{n-2}'k^{\gamma r_{n-3}+150\gamma r_{n-4} }+2l-1+\delta,\
\mbox{when } n\geq 4.$$ It is not difficult to show that $\rho
_{n-2}<r_{n-1}$.
\begin{definition} \label{D-J-last}We denote   by
$J(\k)$ the number of the effectively resonant $k^{\gamma
r_{n-3}}$-clusters  in $H^{(n-1)}(\k)$ for a given $\k$.
%%%%%In order words, $J_{\m_0}(\k)$ is the number of points of the
%%%%%set  ${\MM}^{(2)}(\k,r_3,4\gamma r_1')$ in $\Pi_{\m_0}$, the set
%%%%%${\MM}^{(2)}$ being defined by (\ref{M^2}) ( {\em somewhat sloppy,
%%%%%since ${\MM}^{(2)}$ is defined by discs around poles, not by the
%%%%%%%%%%inequality for the exponent, but basically it is the same for
%%%%%$k^{\delta }$ boxes}).
%%%%%Obviously, $J_{\m_0}(\k)=J(\k+\p_{\m_0})$.
 Further (with a slight abuse of notations) we
consider $J(\k)$ to be constant in the $k^{-\rho
_{n-2}}$-neighborhood of a given $\k$. \end{definition}
%%%%Note that by the Geometric
%%%%Lemmas \ref{L:number of points-1} and \ref{L:Resnumberofpoints} the
%%%%%number of resonant squares $J_{\m_0}(\k)$ does not exceed
%%%%$Ck^{\frac23 \gamma r_1+3}$. (po otnosheniyu k
%%%%%$\varepsilon_0=k^{-5\gamma r_1'}$?)
%%%%%%%%Let us fix $k^{2l}$ and consider the curves $D_{\m_0}(\k,
%%%%%%%%k^{2l})=0$ in $\R^{2}$.
%%To shorten notations we drop $\m_0$ and $k^{2l}$.
%%%%%%%%Naturally $\k=(k_1,k_2)$. It is proven in Appendix 7
%%%%follows from Lemma~\ref{elpieces}
%%%%%%%%that
%%%%%%%%the set $D_{\m_0}(\k)=0$ can be split into $k^{17l\gamma r_1}$ or less
%%%%%%%%elementary pieces (in the sense of Definition~\ref{elementary}).
Let  $\k=\a \tau_1+\b,\  \ \ |\a|=1,\  |\b|<4k^{\gamma r_{n-2}}$. We
consider $H^{(n-1)}(\k)$ as a function of $\tau _1$ in the complex $
k^{-\rho _{n-2}}$-neighbothood of zero.
\begin{lemma} \label{May24-2}The resolvent $( H^{(n-1)}(\k)-k^{2l})^{-1}$ has no more than $k^{2\gamma r_{n-3}}J(\b)$ poles $\tau _{1j}$ in the the complex $ 2k^{-\rho _{n-2}}$-neighborhood of zero. It satisfies the following estimate in the the complex $k^{-\rho _{n-2}}$-neighborhood of zero.:
\begin{equation}\label{Sept3a-last}
\|(H^{(n-1)}(\k)-k^{2l})^{-1}\|<k^{\rho _{n-2}k^{2\gamma
r_{n-3}}}\left(\frac{k^{ -\rho _{n-2}}}{\varepsilon
_0}\right)^{J(\k)k^{2\gamma r_{n-3}}},
\end{equation}
where $\varepsilon _0=\min \{k^{-2\rho _{n-2}},\varepsilon\}$,
$\varepsilon $ being the distance to the nearest pole $\tau _{1,j}$.
\end{lemma}
\begin{proof} The lemma is proved by induction. For $n=3$, see Lemma \ref{May24-1}. Let us consider the case $n\geq 4$.  Recall (Definition \ref{D-J-last}) that
$J(\k)$ may be considered to be constant in $2k^{-\rho
_{n-2}}$-neighborhood  of $\tau _1=0$. Hence, $J(\k)=J(\b)$ for such
$\k$-s.
%%%%%%%%%%% By analogy with the Step II for $\k=(\tau_{10}, \alpha \tau _{10}+\beta )$ we call

Let us consider the collection of all  $k^{\delta }$,...,$k^{\gamma
r_{n-3}}$-clusters $P_{\m }$ for $H^{(n-1)}(\k)$. Note that the
collection is the same for all such $\k$. We construct the
corresponding block operator $\tilde H^{(n-2)}(\k)$:
$$\tilde H^{(n-2)}(\k)=\sum P_{\m}HP_{\m}+H_0(I-\sum P _{\m }).$$  If a $k^{\gamma r_{n-3} }$-cluster $P_{\m}H(\k )P_{\m}$ is effectively non-resonant, then  its resolvent, obviously,  has no poles $\tau _1$ in the $2k^{- \rho _{n-2} }$-neighborhood of $\tau _{1}=0$. The resolvent of each effectively resonant $k^{\gamma r_{n-3} }$-cluster $P_{\m}H(\k )P_{\m}$ has no more than $N_i^{(n-3)}k^{2\gamma r_{n-4}}$  ($k^{2\gamma r_{0}}$ is taken to be equal to $8$ for $n=4$) poles $\tau _{1j}$   in the
$k^{-\rho _{n-3}}$-neighborhood of $\tau _{1}=0$. It follows from
this lemma for the previous step and also
Lemmas~\ref{L:black}--\ref{L:white},
\ref{L:blackind}--\ref{L:whiteind} and
\ref{L:blackindlast}-\ref{L:whiteindlast} for previous steps, which
give the estimates for $J(\k)$ in the previous steps, based on the
color of clusters. Let us consider the union of $k^{-2\rho _{n-2}}$
neighborhoods of these poles and denote it by ${\D}_{\m}$. By this
lemma for $n-1$, instead of $n$, each $k^{\gamma r_{n-3}}$ cluster
satisfies the estimate
$$
\|(P_{\m} (H^{(n-1)}(\k)-k^{2l})P_{\m})^{-1}\|<k^{\rho
_{n-3}k^{2\gamma r_{n-4}}}k^{2\rho _{n-2}N_i^{(n-3)}k^{2\gamma
r_{n-4}}}
$$
outside $\D_{\m}$, $N_i^{(n-3)}$ corresponding the color of the
cluster. Note that $\max _i N_i^{(n-3)}=N_1^{(n-3)}<k^{\gamma
r_{n-3} +150l\gamma r_{n-4}}$.
%%%%%%%%%%%%Hence, the number of poles $\tau_{1j}$ of $({\cal
%%%%%%%%%%H}_{\m_0}(\k)-k^{2l})^{-1}$ in $k^{-5\gamma r_1'}$-neighborhood of
%%%%%%%%%%$\l$ does not exceed  $2J$. Let us consider $k^{-10\gamma
%%%%%%%%%%r_1'}$-neighborhood of such poles.
Therefore, the resolvent $\left(\tilde
H^{(n-2)}(\k)-k^{2l}\right)^{-1}$ has no more than
$J(\k)N_1^{(n-3)}k^{2\gamma r_{n-4}}$ poles $\tau _{1j}$ in the
complex $k^{-\rho _{n-3}}$-neighborhood of $\tau _{1}=0$. Let $\D
=\cup _{\m}\D_{\m}$, the union being taken over all $\m$
corresponding to all  resonant clusters. The number of $\m$-s in the
union, obviously, does not exceed $k^{4\gamma r _{n-2}}$, which is
the number of different $\m$ in $H^{(n-1)}(\k)$. Therefore, the size
of each connected component of $\D$ is less than $k^{-2\rho
_{n-2}}k^{4\gamma r _{n-2}}=o\left(k^{-\rho _{n-2}}\right)$. We are
interested only in those components of $\D$, which are completely in
the disk of the radius $2k^{-\rho _{n-2}}$ around $\tau _1=0$.
 Considering as before  \footnote{see the proof of Theorem \ref{Thm3} with  $r_{n-2}$ instead of $r_2$, $r_{n-3}$ instead of $r_1$ and $k^{\gamma r_{n-3} +150l\gamma r_{n-4}}$ instead of $k^{\gamma r_1+3}$, when one considers black clusters.},
 %%%%%%(see {\em Remark } at the end of the proof of Lemma \ref{L:curves-2}, formulas \eqref{Aug10a}-\eqref{Aug10c})
we can show that the perturbation series for the resolvent
$(H^{(n-1)}(\k)-k^{2l})^{-1}$ with respect to $(\tilde
H^{(n-2)}(\k)-k^{2l})^{-1}$ converges on the boundary of $\D$. The
resolvents have the same number of poles inside each component of
$\D$. Hence, $(H^{(n-1)}(\k)-k^{2l})^{-1}$ has no more than
$J(\k)N_1^{(n-3)}k^{2\gamma r_{n-4}}$ poles in $\D$. It is easy to
see that $J(\k)N_1^{(n-3)}k^{2\gamma r_{n-4}}<J(\k)k^{2\gamma
r_{n-3}}$. The resolvent satisfies the following estimate outside
$\D$:
$$
\|(H^{(n-1)}(\k)-k^{2l})^{-1}\|<k^{\rho _{n-3}k^{2\gamma
r_{n-4}}}k^{2\rho _{n-2}N_1^{(n-3)}k^{2\gamma r_{n-4}}}<k^{\rho
_{n-2}k^{2\gamma r_{n-3}}}.
$$
 Using the maximum principle we obtain \eqref{Sept3a-last}.
 \end{proof}

 Next, we introduce
\begin{equation}\label{triind-last}
\SS ^{(n-1)}(k,\xi):=\{\k\in \R^{2}:\
\|(H^{(n-1)}(\k)-k^{2l})^{-1}\|>k^{\xi}\}.
\end{equation}
It is easy to see that each connected component of $\SS^{(n-1)}(k,\xi)$ is
bounded  by the curves $D(\k, k^{2l}\pm k^{-\xi})=0$, where $ D(\k,
\lambda)=\hbox{det}\,(H^{(n-1)}(\k)-\lambda ).$
%%%%%%%%%%=\hbox{det}\,({\cal H}(\k+\p_{\m_0})-\lambda).
%%%%%%Obviously, $\SS_{\m_0}(\k, \xi)=\SS(\k+\p_{\m_0},\xi)$ and  $D_{\m_0}(\k, \lambda )=D(\k+\p_{\m_0},\lambda )$.
 %%%%%%%%%%%%We can split $\SS_{\m_0}(k,\xi)$ into ``elementary" components, each being adjacent to an elementary piece of the boundary. The number of such components does not exceed $Ck^{65\gamma r_1}$.
\begin{lemma}\label{L:curves-2ind-last}  Let $\l$ be a segment of a straight line in $\R^{2}$,
\begin{equation}\l:=\left\{\k=\a \tau_1+\b,\ \tau _1\in
(0,\eta)\}, \ \ |\a|=1,\  |\b|<4k^{\gamma r_{n-2}},\  \ 0<\eta <
k^{-\rho _{n-2}}\right\}. \label{segment-last}\end{equation} Suppose
both ends of $\l$ belong to a connected component of $\SS ^{(n-1)}(k,\xi)$.
If $\xi $ is sufficiently large, namely, $\xi\geq 2k^{2\gamma
r_{n-3}}J(\b)\log _k\frac{1}{\eta}$, then, there is an inner part
$\l'$ of the segment,
 which is not in $\SS^{(n-1)}(k,\xi)$.
 %%%%%%%%%%%%Moreover, there is a point $(\tau _{1*},\tau _{2*})$ in $\l'$ such that a
 %%%%%%%%%%%% boundary curve $D_{\m_0}(\k, k^{2l}\pm k^{-\xi})=0$ of $\SS_{\m_0}(k,\xi)$ and the segment $\l$ are parallel.
 \end{lemma}
 \begin{corollary} \label{C:curves-2ind-last} Let $\k\in \SS^{(n-1)}(k,\xi)$ and
 ${\xi }>k^{2\gamma r_{n-3}}J(\k )\rho _{n-2}$. Then the distance from $\k$ to the boundary of
 $\SS^{(n-1)}(k,\xi)$ is less than $k^{-\tilde \xi}$, $\tilde \xi =\xi k^{-2\gamma r_{n-3} }J(\k )^{-1}$. \end{corollary}
 {\em Proof of the corollary.}  Let us consider a segment of the length $\eta =k^{-\tilde \xi }$ starting at $\k $.
 By the statement of the lemma it intersects a boundary $D(\k, k^{2l}\pm k^{-\xi})=0$.

 \begin{proof}
%%%%%%%%%%Let us consider the complex $k^{-5\gamma r_1'}$-neighborhood of $\l$.
%%%%%%%%%%%%We construct the analog of the block structure (48?).
%%%%%%%%Let $\k_0 =(\tau_{1,0}, \beta _1 \tau _{1,0}+\beta _2)$.
 Choose $\varepsilon =\eta ^2$. Using the hypothesis of the lemma, we
obtain that the right-hand side of \eqref{Sept3a-last} is less than
$k^{{\xi }}$ outside the discs.  Let us estimate the total size (sum
of the sizes) of the discs. Indeed, the size of each disc is $2\eta
^2$ and the number of discs is, obviously, less  $16k^{4\gamma
r_{n-2}}$. Therefore, the total size admits the estimate from above:
$32\eta ^2 k^{4\gamma r_{n-2}}=o(\eta)$, since $\eta <k^{-\rho
_{n-2}}$. This means there is a part $\l'$ of $\l$ outside these
discs. By \eqref{Sept3a-last}, this part is
 outside $\SS ^{(n-1)}(k,\xi)$, when $\xi $ is as described in the statement of the lemma.
 %%%%%%%%%%Since, obviously, the distance
%%%%%%%%%between $D=0$ and $\l$ attains its maximum in $\l'$ and the curve
%%%%%%%%%%%and the segment are parallel at this point, the lemma is proven.
\end{proof}
%Zmechu, chto dopolnitel'noe neravenstvo na $\xi$ nichgo ne meyaet, tak kak nizhe (Lemma 9.7) poyavlyaetsya neravenstvo  %$\xi>20\mu r_{n+1}J_0$, kotoroe kruche, chem  ${\xi }>k^{\gamma r_{n-1} +150l\gamma r_{n-2}+1}J(\k)5\gamma r_n'$.

Let $\k _0 \in \R^2$ be fixed and ${\cal N}(k,r_{n-1},\k _0,J_0)$ be the
following subset of the lattice $\k _0+\p_{\n}$, $\n \in \Omega
(r_{n-1})$:
$${\cal N}(k,r_{n-1},\k _0,J_0)=\left\{\k _0+\p_{\n}:\n \in \Omega (r_{n-1}):\
J(\k _0 +\p_{\n})\leq J_0\right\},$$ $J$ being defined by Definition
\ref{D-J-last}. Thus, ${\cal N}$ includes only such $\n$ that the
surrounding $k^{\gamma r_{n-2}}$- block contains less than $J_0 $ of
effectively resonant points. Let $N(k,r_{n-1},\k _0,J_0, \xi)$ be the
number of points $\k _0+\p_{\n}$ in $\SS ^{(n-1)}(k,\xi )\cap {\cal
N}(k,r_{n-1},\k _0,J_0)$.
%%%%%%%%Let $\xi >0$ and $N(r_{n-1},\k,J_0, \xi)$ be the number of points $\m_0$
%%%%%%%%in ${\cal N}(r_{n-1},\k,J_0)$ satisfying the estimate
%%%%%%%%\begin{equation}\label{0eqnorm2/3}
%%%%%%%%\|({\cal H}_{\m_0}(\k)-k^{2l})^{-1}\|>k^{\xi}.
%%%%%%%%\end{equation}
\begin{lemma}\label{norm2/3-last}
%%%%%%%%Let $N(\k)$, $\k \in \R^3$, be the number of $\m_0$ with $\||\p_{\m_0}\||\leq k^{r_{n-1}}$ such that
%%%%%%%%\begin{equation}\label{0eqnorm2/3}
%%%%%%%%\|({\cal H}_{\m_0}(\k)-k^{2l})^{-1}\|>k^{k^\xi},
%%%%%%%%\end{equation}
If   $\xi>4\mu r_{n-1}J_0k^{2\gamma r_{n-3}}$, then
\begin{equation}\label{eqnorm2/3-last}
N(k,r_{n-1},\k _0,J_0, \xi)\leq k^{\frac{2}{3}r_{n-1}+43l\gamma
r_{n-2}}.
\end{equation}
\end{lemma}
\begin{proof} The proof of the lemma is completely analogous to that of \ref{norm2/3} up to replacement of 2 by $n-1$. Instead of  Corollary \ref{C:curves-2ind} we use Corollary \ref{C:curves-2ind-last} and the inequality $\rho _{n-2}<r_{n-1}$.\end{proof}

\subsubsection{Model Operator for Step $n+1$}
%%%%%%%%%%%% Let us consider $\Omega(r_3)$ as before.
We repeat for $r_n$ the construction from the
subsection~\ref{MOforStep3}, which was done for arbitrary large $r_{n-1}$.
We start with introducing a new notation by analogy with \eqref{se} and \eqref{Omega-s}:
\begin{equation} \Omega _s^{(j)}(r_n)=\{\m \in \Omega (r_n),\ 0<p_{\m}<k^{-r_{j-1}'k^{2\gamma r_{j-2}}}\}, \ \ j\geq 2,\label{Omega-j} \end{equation}
where $k^{2\gamma r_{j-2}}$ is  taken to be just $5$ when $j=2$. Note that $ \Omega _s^{(j+1)}\subset  \Omega _s^{(j)}$ and  $\Omega _s^{(j)}=\emptyset $ when $j>n$.
Next, let $\m \in \Omega (r_n)$. We denote the $k^{\gamma r_{n-2}}$-component
containing $\m$ by $\tilde \Pi (\m)$ and the corresponding projector
by $\tilde P(\m)$. For $\m$ belonging to the same $k^{\gamma
r_{n-2}}$-component, $\tilde \Pi (\m)$ and  $\tilde P(\m)$ are the same.
We define ${\MM}^{(n)}$ by the recurrent formula, which starts with ${\MM}^{(3)}$, see \eqref{M^3}:
\begin{equation}\label{M^n} {\MM}^{(n)}:={\MM}^{(n)}(\varphi _0, r_n)=\{\m\in \MM^{(n-1)}(\varphi _0, r_n)\cup \Omega _s^{(n-1)}(r_n)
:\ \varphi_0\in{\cal O}_\m^{(n)}(r_{n-1}',1)\},\end{equation} where
${\cal O}_\m^{(n)}(r_{n-1}',\tau)$ is the union of the disks of the
radius $\tau k^{-r_{n-1}'}$ with the centers at poles of the resolvent
$(\tilde P(\m)(H(\k^{({n-1})}(\varphi ))-k^{2l}I)\tilde P(\m))^{-1}$ in the $k^{-44r_{n-2}'-2l-\delta }$-neighborhood of $\varphi _0$.
For $\m$ belonging to the same $k^{\gamma
r_{n-1}}$-component, the sets ${\cal O}_\m^{(n)}(r_{n-1}',\tau)$ are
identical. We say that $\m \in {\MM}^{(n)}$ is $k^{\gamma
r_{n-2}}$-resonant. The corresponding $k^{\gamma r_{n-2}}$-clusters we call
resonant too.

 Let $\varphi_0\in \omega^{(n)} (k,
\delta ,1)$. By
construction of the non-resonant set $\omega^{(n)} (k, \delta ,1)$,
we have ${\MM}^{(n)}\cap \Omega (r_{n-1})=\emptyset $.

Further we use the property of the set $\MM^{(n)}$ formulated in the
next lemma which is an analogue of the Lemmas~\ref{L:2/3-1}, \ref{L:2/3-1ind}.

\begin{lemma}\label{L:2/3-1ind-last} Let $r_{n-1}'>2k^{(\gamma +\delta_0)10^{-4}r_{n-2}-2\delta}$.\footnote{ We
also notice that this condition is consistent with the restriction
\eqref{indrn}.} Let  $1/20<\gamma '<20$, $\m _0\in\Omega
(r_n)$ and $\Pi _{\m_0}$ be the $k^{\gamma
'r_{n-1}}$-neighborhood (in $\||\cdot\||$-norm) of $\m_0$. Then the set
$\Pi _{\m _0}$ contains less than $k^{\frac 23 \gamma'r_{n-1} +50l\gamma r_{n-2}}$
elements of $\MM^{(n)}$.
\end{lemma}
%\marginpar{Zdesj $\m _0 \in \Omega (r_3)$. V lemme 7.7 eto nechto drugoe ($\m_0$ v lemme 7.7 skoree to ge chto $\m$ zdesj. Nado menyatj kak-to}

\begin{proof} The proof is similar to that of Lemma \ref{L:2/3-1ind} up to the replacement of $3$ by $n$. Instead of Lemma \ref{L:geometric2} we use Lemmas \ref{L:geometric3}(n=4)
\ref{L:geometric4n} (n=5) and \ref{L:geometric4nlast} with  $n-2$ instead of $n$ when $n>5$. We also use Lemma \ref{norm2/3-last} instead of Lemma \ref{norm2/3}. We use \eqref{indrn} to show that the hypothesis
of Lemma \ref{norm2/3-last} holds. In particular, we use the inequality $r_{n-1}'>>4\mu  r_{n-1}k ^{2\gamma r_{n-3}}$, following from \eqref{indrn}. \end{proof}
%\marginpar{proof is in Lemma 9.9 geomind4}

We continue with constructing $k^{\gamma r_{n-2}}$-clusters in
$\Omega (r_n)$, $r_n>r_{n-1}$, the same way we did it for $\Omega
(r_2)$ in Section \ref{MOforStep3}. We call a $k^{\gamma
r_{n-2}}$-cluster resonant if it contains $\m \in {\MM}^{(n)}$, see
\eqref{M^3}, \eqref{M^n}. Next, we repeat the construction after
Lemma \ref{L:2/3-1} up to the replacement of $r_1$ by $r_{n-1}$ and
$\delta $ be $\gamma r_{n-2}$. Indeed, let us split $\Omega
(r_n)\setminus \Omega (r_{n-1})$ into $k^{\gamma r_{n-1}}$-boxes,
$\gamma =\frac{1}{5}$.

First, let's consider  $\m \in \Omega _s^{(n)}(r_n)$. As before (see
``Simple region", page \pageref{simple}) one can prove that $\Omega
_s^{(n)}(r_n)\subset \MM (r_n)$; there are no other elements of $\MM
(r_n)$ in the $k^{\delta}$-box around $\m$; $\m$ itself can belong
or do not belong to $\MM^{(j)}(r_n)$, but there are
 no other elements of
$\MM^{(j)}(r_n)$ in the $k^{r_{j-1}}$-box around such $\m$,
$j=2,\dots,n$; and there are no other elements of $\Omega
_s^{(n)}(r_n)$ in the $k^{r_{n-1}}$-box around $\m$.
%By \eqref{below}, \eqref{indrn}, if such
%element exists it is unique in the box of the size $k^{ r_{n-1}}$.
%Moreover, this element itself can belong or do not belong to
%$\MM^{(n)}$ but there are no other elements from $\MM^{(n)}$ in such
%$k^{ r_{n-1}}$-box. Indeed, let us consider $\k=\k^{(n-1)}+\p_\m$.
%It is easy to show that $\k$ satisfies the conditions of Lemmas
%\ref{L:geometric4n}, \ref{L:geometric4nlast}. It follows that $\m
%+\q \not \in \MM^{(n)}$ when $0<\||\p_\q\||<k^{r_n}$.

For each $\m \in \Omega _s^{(n)}(r_n)$ we consider its $k^{
r_{n-1}/2}$-neighborhood in $\||\cdot \||$ norm. The union of such
boxes we call the simple region and denote $\Pi _s$. The
corresponding projection is $P_s$.

Now, consider all other boxes (all elements $\p_\m$ there satisfy
$p_\m>k^{- r_{n-1}'k^{2\gamma r_{n-2}}}$). We call a box black if it
together with its neighbors contains more than $k^{\gamma
r_{n-1}/2+\delta _0r_{n-1}}$ elements of $\MM ^{(n)}$,
$\delta_0=\gamma /100$. Let us consider "black" boxes together with
their $k^{\gamma r_{n-1}+\delta _0r_{n-1}}$-neighborhoods and call
this the black region. We denote the black region by $\Pi _b$. The
corresponding projector is $P_{b}$. By white boxes we mean
$k^{\gamma r_{n-1}}$-boxes which together with its neighbors contain
no more than $k^{\gamma r_{n-1}/2+\delta_0r_{n-1}}$ elements of $\MM
^{(n)}$. Every white box we split into "small" boxes of the size
$k^{\gamma r_{n-1}/2+2\delta_0r_{n-1}}$. We call a small box "grey"
if it together with its neighbors contains more than $k^{\gamma
r_{n-1}/6-\delta_0r_{n-1}}$ elements of $\MM ^{(n)}$. Grey small
boxes together with its $k^{\gamma
r_{n-1}/2+2\delta_0r_{n-1}}$-neighborhoods we call the grey region.
The notation for this region is $\Pi _g$. The corresponding
projector is $P_{g}$. The part of the grey region which is outside
the black region, we denote $\Pi _g'$ and the corresponding
projection by $P_g'$. By a white small box we call a small box which
has no more than $k^{\gamma r_{n-1}/6-\delta_0r_{n-1}}$ elements of
$\MM ^{(n)}$. In each small white box we consider $k^{\gamma
r_{n-1}/6}$-boxes around each point of $\MM^{(n)}$. The union of
such $k^{\gamma r_{n-1}/6}$-boxes we call the white region and
denote $\Pi _w$. The corresponding projection is $P_w$. The part of
the white region which is outside the black and grey regions, we
denote $\Pi _w'$ and the corresponding projection by $P_w'$.
%%%%%%%%%%% We also consider
%%%%%%%%%%%$k^{\gamma r_1}$-neighborhoods of the set $\MM(r_3, \varphi
%%%%%%%%%%%_0)\setminus \left(\MM(r_2, \varphi _0)\cup \MM^{(3)}\right)$. The
%%%%%%%%%%%union of this neighborhoods we call the non-resonance region $\Pi
%%%%%%%%%%%_{nr}$. The corresponding projection is $P_{nr}$.  The part of the
%%%%%%%%%%%nonresonant region which is outside
%%%%%%%%%%%$\Pi_s\cup\Pi_b\cup\Pi_g\cup\Pi_w$, we denote $\Pi _{nr}'$ and the
%%%%%%%%%%%corresponding projection by $P_{nr}'$. Let
%%%%%%%%%%%$$P_r:=P_s+P_b+P_g'+P_w',\ \ \ P^{(2)}:=P_r+P_{nr}'$$
%%%%%%%%%%%index $r$ standing for "resonant".
%%%%%%%%%%%The model operator $\tilde H^{(2)}$ for the third step we define by
%%%%%%%%%%%the formula:
%%%%%%%%%%%We split $k^{r_3}$-box into $k^{\gamma r_2}$-boxes and repeat all
%%%%%%%%%%%constructions from Section 3 of white, grey, black and simple
%%%%%%%%%%%$k^{\gamma r_1}$-regions with obvious changes.

We put as before
$$P_r^{(n)}:=P_s^{(n)}+P_b^{(n)}+P_g^{(n)'}+P_w^{(n)'}.$$
The construction of the non-resonant region is the inductive extension of that for Step IV, see Section \ref{S:4}, page \pageref{NRR}. Indeed, we start with construction of $k^{\delta }$ clusters in $\Omega (r_n)$. Those of them, who are resonant, we extend to $k^{\gamma r_1}$ clusters, those of them, which are resonant we extend to  $k^{\gamma r_2}$ clusters, and so on until we reach the size $k^{\gamma r_{n-2}}$. On each step we construct a colored structure (simple, black, grey, white). If $k^{\gamma r_j}$-cluster happens to intersect $k^{\gamma r_{j+1}}$-cluster, we consider it to be a part of $k^{\gamma r_{j+1}}$-cluster. Thus, $k^{\gamma r_{j}}$-clusters are built around the points of
${\MM}^{(j+1)}(r_n,\varphi _0)\cup \Omega _{s}^{(j+1)}(r_n)\setminus \left({\MM}^{(j+1)}(r_{n-1},\varphi _0)\cup {\MM}^{(j+2)}(r_n,\varphi _0)\cup \Omega _{s}^{(j+2)}(r_n)\right)$.
 The set of all other non-resonant
$k^{\gamma r_{j}}$-clusters we denote by
$\Pi_{nr,r_{j}}^{(n)'}$. Then
$$
\Pi_{nr}^{(n)}:=\cup_{j=0}^{n-2}\Pi_{nr,r_j}^{(n)},$$
 Those $\Pi_{nr,r_j}^{(n)}$, which intersect with
$\Pi_r^{(n)}$ we attach to $\Pi_r^{(n)}$ just slightly abusing the
notation (cf. Section \ref{S:4}). The part of $\Pi_{nr,r_j}^{(n)}$ which does not intersect with $\Pi_r^{(n)}$ we denote by $\Pi_{nr,r_j}^{(n)'}$. Correspondingly, the part of $\Pi_{nr}^{(n)}$ which does not intersect $\Pi_{r}^{(n)}$ is denoted by $\Pi_{nr}^{(n)'}$.
%%%%%Let us provide a little bit more details about region
%%%%%L We consider all non-resonant $k^{\gamma
%%%%%Lr_{j}}$-clusters, $j=0,..., n-2$, $r_0=\delta $ (cf. with non-resonant $k^\delta$-clusters and $k^{r_1}$,
%%%%%Lconsidered in the preparation for Step IV).
Further, \begin{equation} \label{P(n)}
P^{(n)}:=P^{(n)}_r+P^{(n)'}_{nr}+P(r_{n-1}).
\end{equation}
We continue construction from Section \ref{S:3}. Repeating the arguments
from the proofs of Lemmas~\ref{L:black},~\ref{L:grey},~\ref{L:white}
with obvious changes (in particular, using Lemma~\ref{L:2/3-1ind-last}
instead of Lemmas~\ref{L:2/3-1}, \ref{L:2/3-1ind}) we obtain the following results.
(Here and in what follows we will omit superscript $(n)$ when it
does not lead to a confusion.)

\begin{lemma} \label{L:blackindlast} \begin{enumerate} \item
Each $\Pi _b^j$ contains no more than $k^{\gamma r_{n-1}/2-\delta
_0r_{n-1}+150l\gamma r_{n-2}}$ black boxes.
\item The size of $\Pi _b^j$ in $\||\cdot \||$ norm is less than
$k^{3\gamma r_{n-1}/2+150l\gamma r_{n-2}}$.
\item Each $\Pi _b^j$ contains no more than $k^{\gamma r_{n-1}+150l\gamma r_{n-2}}$ elements of
$\MM^{(n)}$. Moreover, any box of $\||\cdot \||$-size $k^{3\gamma
r_{n-1}/2+150l\gamma r_{n-2}}$ containing $\Pi _b^j$ has no more
than $k^{\gamma r_{n-1}+150l\gamma r_{n-2}}$ elements of
$\MM^{(n)}$ inside.
\end{enumerate}
\end{lemma}

\begin{lemma}\label{L:greyindlast} \begin{enumerate} \item
Each $\Pi _g^j$ contains no more than $k^{\gamma r_{n-1}/3+2\delta
_0r_{n-1}}$ grey boxes.
\item The size of $\Pi _g^j$ in $\||\cdot \||$ norm is less than
$k^{5\gamma r_{n-1}/6+4\delta _0r_{n-1}}$.
\item Each $\Pi _g^j$ contains no more than $k^{\gamma r_{n-1}/2+\delta _0r_{n-1}}$ elements of
$\MM^{(n)}$.\end{enumerate}
\end{lemma}

\begin{lemma}\label{L:whiteindlast} \begin{enumerate} \item The size of $\Pi _w^j$ in $\||\cdot \||$ norm is less than
$k^{\gamma r_{n-1}/3-\delta _0r_{n-1}}$.
\item Each $\Pi _w^j$ contains no more
than $k^{\gamma r_{n-1}/6-\delta _0r_{n-1}}$ points of $\MM^{(n)}$.
\end{enumerate}
\end{lemma}

The construction of the rest of Section~\ref{MOforStep3} stays unchanged. Let us
introduce corresponding notation, formulate the results and provide
some comments.

Next lemmas are the analogues of
Lemmas~\ref{L:Pnr},~\ref{L:Pr},~\ref{L:Ps}.

\begin{lemma}\label{L:Pnrnlast}Let $\varphi _0\in \omega^{(n)}(k,\delta ,\tau )$,
$|\varphi-\varphi _0|<k^{-k^{r_{n-2}}}$. Then,
\begin{equation}\label{Pnrnlast}
\left\|\Bigl(P_{nr}\bigl(H(\k^{(n)}(\varphi
))-k^{2l}I\bigr)P_{nr}\Bigr)^{-1}\right\|<k^{r_{n-1}'k^{2\gamma
r_{n-2}}}k^{r_{n-1}'}\leq k^{k^{3\gamma r_{n-2}}}.
\end{equation} \end{lemma}

\begin{lemma}\label{L:Prnlast} Let $\varphi _0\in \omega^{(n)}(k,\delta ,\tau )$,
%%%%%%%%$P_{r}^{(2)j_2}=P_{r}^{(2)j_2}(\varphi _0, r_2)$, $j_2=1,...,J_2$,
and $|\varphi-\varphi _0|<k^{-k^{r_{n-2}}}$, $i=1,2,3$. Then,
\begin{enumerate}
\item The number of poles of the resolvent $\Bigl(P_i\bigl(H(\k^{(n)}(\varphi
))-k^{2l}I\bigr)P_i\Bigr)^{-1}$ in the disc $|\varphi-\varphi
_0|<k^{-k^{r_{n-2}}}$ is no greater than $N_i^{(n-1)}$, where $N_1^{(n-1)}=k^{\gamma
r_{n-1}+150l\gamma r_{n-2}}$, $N_2^{(n-1)}=k^{\gamma r_{n-1}/2+\delta
_0r_{n-1}}$, $N_3^{(n-1)}=k^{\gamma r_{n-1}/6-\delta _0r_{n-1}}$.
\item Let $\varepsilon$ be the distance to the nearest pole of
the resolvent in ${\cal W}^{(n)}$ and let
$\varepsilon_0=\min\{\varepsilon,\,k^{-r_{n-1}'}\}$. Then the
following estimates hold:
\begin{equation}\label{Pr-1nlast}
\begin{split}
\left\|\Bigl(P_i\bigl(H(\k^{(n)}(\varphi
))-k^{2l}I\bigr)P_i\Bigr)^{-1}\right\|<k^{2r_{n-1}'k^{2\gamma
r_{n-2}}}k^{r_{n-1}'}\left(\frac{k^{-r_{n-1}'}}{\varepsilon
_0}\right)^{N_{i}^{(n-1)}}\leq \cr k^{k^{3\gamma
r_{n-2}}}\left(\frac{k^{-r_{n-1}'}}{\varepsilon _0}\right)^{N_{i}^{(n-1)}},
\end{split}
\end{equation}
\begin{equation}\label{Pr-2nlast}
\begin{split}
\left\|\Bigl(P_i\bigl(H(\k^{(n)}(\varphi
))-k^{2l}I\bigr)P_i\Bigr)^{-1}\right\|_1<k^{2r_{n-1}'k^{2\gamma
r_{n-2}}}k^{r_{n-1}'+8\gamma
r_{n-1}}\left(\frac{k^{-r_{n-1}'}}{\varepsilon
_0}\right)^{N_{i}^{(n-1)}}\leq \cr k^{k^{3\gamma
r_{n-2}}}\left(\frac{k^{-r_{n-1}'}}{\varepsilon _0}\right)^{N_{i}^{(n-1)}}.
\end{split}
\end{equation}
\end{enumerate}
\end{lemma}
\begin{proof} The proof of this lemma is analogous to that of Lemma \ref{L:Pr} up to the replacement of ${\MM}^{(2)}$ by ${\MM}^{(n)}$, $\OO^{(2)}_{\m}$ by $\OO^{(n)}_{\m}$, and the shift of indices: $\delta $ to $r_{n-2}$, $r_1$ to
$r_{n-1}$, etc. We apply Lemmas \ref{L:blackindlast}--\ref{L:whiteindlast} instead of \ref{L:black}--\ref{L:white}. We apply Lemmas \ref{L:Prn}, \ref{L:Psn} with $\varepsilon _0=k^{-r_{3}'}$ and $p_\m>k^{-r_{3}'k^{2\gamma r_{2}}}$ instead of Lemma \ref{L:Pr}, \ref{L:Ps} for $n=4$ and Lemmas \ref{L:Prnlast}, \ref{L:Psnlast} with  inductively (with $n-1$ instead of $n$ and $\varepsilon _0=k^{-r'_{n-1}}$, $p_\m>k^{-r'_{n-1}k^{2\gamma r'_{n-2}}}$) for further steps.
\end{proof}

\begin{lemma}\label{L:Psnlast} Let $\varphi _0\in \omega^{(n)}(k,\delta ,\tau )$. Then, the operator
$\Bigl(P_s^j\bigl(H(\k^{(n)}(\varphi
))-k^{2l}I\bigr)P_s^j\Bigr)^{-1}$ has no more than one pole in the
disk $|\varphi-\varphi _0|<k^{-k^{r_{n-2}}}$. Moreover,
\begin{equation}\label{Ps-1nlast}
\left\|\Bigl(P_s^j\bigl(H(\k^{(n)}(\varphi
))-k^{2l}I\bigr)P_s^j\Bigr)^{-1}\right\|<\frac{8k^{-2l+1}}
{p_\m\varepsilon _0},
\end{equation}
\begin{equation}\label{Ps-2nlast}
\left\|\Bigl(P_s^j\bigl(H(\k^{(n)}(\varphi
))-k^{2l}I\bigr)P_s^j\Bigr)^{-1}\right\|_1<\frac{8k^{-2l+1+
4r_{n-1}}}{p_\m\varepsilon_0},
\end{equation}
$\varepsilon _0=\min\{\varepsilon,\,k^{-r_{n-1}'}\}$, where
$\varepsilon$ is the distance to the pole of the operator.
\end{lemma}

Note that $p_{\m}>k^{-2\mu r_n}$ when $\m \in \Omega (r_n)$. The
analogues of Lemma~\ref{L:boundary} and Corollary~\ref{C:PHP-2} also
hold.

\subsubsection{Resonant and Nonresonant Sets for Step $n+1$ \label{GSlast}}

We divide $[0,2\pi )$ into $[2\pi k^{k^{r_{n-2}}}]+1$ intervals
$\Delta_m^{(n)}$ with the length not bigger than $k^{-k^{r_{n-2}}}$.
If a particular interval belongs to $\OO^{(n)}$ we ignore it;
otherwise, let $\varphi_0(m)\not\in\OO^{(n)}$ be a point inside the
$\Delta_m^{(n)}$. Let
\begin{equation}\W_m^{(n)}=\{\varphi \in \W^{(n)}:\ | \varphi -\varphi
_0(m)|<4k^{-k^{r_{n-2}}}\}. \label{W2mind-last} \end{equation} Clearly,
neighboring sets $\W_m^{(n)}$  overlap (because of the multiplier 4
in the inequality), they cover $\hat \W^{(n)}$ , which is
%%%%%$\omega^{(2)}$ and its $k^{-44r_1'-2l-\delta
%%%%%}$-neighborhood, i.e
the restriction of $\W^{(n)}$ to the
$2k^{-k^{r_{n-2}}}$-neighborhood of $[0,2\pi )$. For each $\varphi
\in \hat \W^{(n)}$ there is an $m$ such that $|\varphi -\varphi
_{0}(m)|<4k^{-k^{r_{n-2}}}$. We consider the poles of the resolvent
$\left(P^{(n)} (H(\k^{(n)}(\varphi))-k^{2l})P^{(n)}\right)^{-1}$ in
$\hat \W_m^{(n)}$ and denote them by $\varphi^{(n)} _{mj}$,
$j=1,...,M_m$. As before, the resolvent has a block structure. The
number of blocks clearly cannot exceed the number of elements in
$\Omega (r_n)$, i.e. $k^{4r_n}$. Using the estimates for the number
of poles for each block, the estimate being provided by Lemma
\ref{L:Prnlast} Part 1, we can roughly estimate the number of poles
of the resolvent by $k^{4r_n+r_{n-1}}$. Next, let  $\OO^{(n+1)}_{mj}$ be the disc of the radius $k^{-r_n'}$ around
$\varphi ^{(n)}_{mj}$.
\begin{definition} The set
\begin{equation}\OO^{(n+1)}=\cup _{mj}\OO^{(n+1)}_{mj} \label{Olast}
\end{equation}
we call the $n+1$-th resonant set. The set
\begin{equation}\W^{(n+1)}= \W^{(n)}\setminus \OO^{(n+1)}\label{Wlast}
\end{equation}
is called the $n+1$-th non-resonant set. The set
\begin{equation}\omega^{(n+1)}= \W^{(n+1)}\cap [0,2\pi) \label{wlast}
\end{equation}
is called the $n+1$-th real non-resonant set. \end{definition}
The
following statements can be proven in the same way as
Lemmas~\ref{L:geometric3}, \ref{4.16} and \ref{estnonres0-1}.
\begin{lemma}\label{L:geometric4nlast} Let  $r_n'>\mu r_n>k^{r_{n-2}}$ \footnote{These inequalities follow from \eqref{indrn}.}, $\varphi \in \W^{(n+1)}$, $\varphi
_0(m)$ corresponds  to an interval $\Delta _m^{(n)}$ containing $\Re
\varphi $. Let $\Pi $ be one of the components $\Pi _s^j(\varphi
_0(m))$, $\Pi _b^j(\varphi _0(m))$, $\Pi _g^j(\varphi _0(m))$, $\Pi
_w^j(\varphi _0(m))$ and
 $P(\Pi )$ be the projection corresponding to $\Pi $. Let also
 $\varkappa \in \C: |\varkappa-\varkappa^{(n)}(\varphi )|<k^{-r_n'k^{2\gamma
r_{n-1}}}$. Then,
\begin{equation} \label{March3-2nlast} \left\|\left(P(\Pi )
\left(H\big(\k(\varphi )\big)-k^{2l}I\right)P(\Pi
)\right)^{-1}\right\|<k^{2\mu r_n+r_n'N^{(n-1)}},\end{equation}
\begin{equation} \label{March3-3nlast}
\left\|\left(P(\Pi )\left(H\big(\k(\varphi
)\big)-k^{2l}I\right)P(\Pi )\right)^{-1}\right\|_1<k^{(2\mu+1)
r_n+r_n'N^{(n-1)}},\end{equation}$N^{(n-1)}$ corresponding to the
color of $\Pi $ ($N^{(n-1)}=1$, $k^{\gamma r_{n-1}+150l\gamma
r_{n-2}}$, $k^{\gamma r_{n-1}/2+\delta _0r_{n-1}}$, $k^{\gamma
r_{n-1}/6-\delta _0r_{n-1}}$ for simple, black, grey and white
clusters, correspondingly).
\end{lemma}
\begin{proof} The lemma follows from Lemmas \ref{L:Prnlast}, \ref{L:Psnlast} and the definition of $\W^{(n+1)}$.\end{proof}
%\begin{proof} For $\tilde \k=\k^{(n)}(\varphi )$ the lemma follows
%immediately from the definition of $\W^{(n+1)}$ and Lemmas
%\ref{L:Prnlast}, \ref{L:Psnlast}. It is easy to see that estimates
%\eqref{March3-2nlast} and \eqref{March3-3nlast} are stable with
%respect to perturbation of $k^{(n)}$ of order $k^{-r_n'k^{2\gamma
%r_{n-1}}}$. This stability ensure \eqref{March3-2nlast} and
%\eqref{March3-3nlast}.
%\end{proof}

 By total size of the set $\OO^{(n+1)}$ we mean the sum of
the sizes of its connected components.
\begin{lemma}\label{10.9} Let $r_n'\geq (\mu+10)r_n$, $r_n>k^{r_{n-2}}$. Then, the size of each connected component of
 $\OO^{(n+1)}$ is less
than $32k^{4r_n-r_n'}$. The total size of $\OO^{(n+1)}$ is less than
$k^{-r_n'/2}$.
\end{lemma}
%\begin{proof} Indeed, each set $\W_j^{(n)}$ contains no more than
%$ck^{4r_n}$ discs $\OO_{jm}^{(n+1)}$. Therefore,  the size of
%$\OO^{(n+1)}\cap \W_j^{(n)}$ is less than $ck^{-r_n'+4r_n}$.
%Considering that $ck^{-r_n'+4r_n}$ is much smaller that the length
%of $\Delta _j$, we obtain that there is no connected components
%which go across the whole set $\W_j^{(n)}$ and the size of each
%connected component of $\OO^{(n+1)}$ is less than $ck^{4r_n-r_n'}$.
%Considering that the number of intervals $\Delta _j$ is less than
%$ck^{k^{r_{n-2}}}$, we obtain the required estimate for the total
%size of $\OO^{(n+1)}$.
%\end{proof}

%%%%%%%{\it We say that $\varphi\in I_j$ is a resonant point of the second
%%%%%%%kind if its distance from the nearest pole of the operator
%%%%%%%$(P(H(\k_j)-k^2)P)^{-1}$, $\k_j=k(\cos \varphi_0(j), \sin
%%%%%%%\varphi_0(j))$  is less than $k^{-r_1'}$.} The set of all resonant
%%%%%%%points of the second kind we denote by $\OO^{(2)}$. Notice that
%%%%%%%$$meas\,(\OO^{(2)}\cap[0,2\pi))\leq(2\pi k^{2+\delta(40\mu+1)}+1)ck^{4r_1}k^{-r_1'}\leq k^{-r_1}.$$
%%%%%%%{\it Here and in what follows we assume that $r_1'>5r_1+2$.} We also
%%%%%%%put $\W^{(2)}:=\W ^{(1)}\setminus\OO^{(2)}$,
%%%%%%%$\W^{(2)}_j:=\W^{(1)}\cap\{\varphi:\
%%%%%%%|\varphi-\varphi_0(j)|<k^{-2-\delta(40\mu+1)}\}$.

\begin{lemma}\label{estnonres0-1last} Let $\varphi\in\W^{(n)}$ and
$C_{n+1}$ be the circle $|z-k^{2l}|=k^{-2r_n'k^{2\gamma r_{n-1}}}$.
Then
$$
\left\|\left(P(r_{n-1})(H(\k^{(n)}(\varphi))-z)P(r_{n-1})\right)^{-1}\right\|\leq
4^nk^{2r_n'k^{2\gamma r_{n-1}}}. $$ \end{lemma}
%\begin{proof} Indeed, by Theorem \ref{Thm3last}, the operator $P(r_{n-1})H(\k^{(n)}(\varphi))P(r_{n-1})$  has a single
%eigenvalue in the interval $\varepsilon _n( k,\delta,\tau )=\left(
%k^{2l}-k^{-2r_{n-1}'k^{2\gamma r_{n-2}}},
%k^{2l}+k^{-2r_{n-1}'k^{2\gamma r_{n-2}}}\right)$. By definition of
%$\k^{(n)}(\varphi )$ this eigenvalue is equal to $k^{2l}$. The lemma
%follows immediately. \end{proof}

\subsection{Operator $H^{(n+1)}$. Perturbation Formulas}
Let $P(r_{n})$ be an orthogonal projector onto
$\Omega(r_{n}):=\{\m:\ |\|\p_\m\||\leq k^{r_{n}}\}$ and
$H^{(n+1)}=P(r_{n})HP(r_{n}) $. We consider
$H^{(n+1)}(\k^{(n)}(\varphi ))$ as a perturbation of
\begin{equation}
\begin{split}
&\tilde H^{(n)}=\tilde
P_j^{(n)}H\tilde
P_j^{(n)}+\left(P(r_n)-\tilde
P_j^{(n)}\right)H_0,
\end{split}
%%%%%%%%%P(r_{n-1})H(\k^{(n)}(\varphi
%%%%%%%%%))P(r_{n-1})+P_{j,s}^{(n)}H(\k^{(n)}(\varphi
%%%%%%%%%))P_{j,s}^{(n)}+P_{j,b}^{(n)}H(\k^{(n)}(\varphi ))P_{j,b}^{(n)}\cr &
%%%%%%%%%+P_{j,g}^{'(n)}H(\k^{(n)}(\varphi
%%%%%%%%%))P_{j,g}^{'(n)}+P_{j,w}^{'(n)}H(\k^{(n)}(\varphi
%%%%%%%%%))P_{j,w}^{'(n)}+P_{j,nr}^{'(n)}H(\k^{(n)}(\varphi
%%%%%%%%%))P_{j,nr}^{'(n)}\cr & +
\end{equation}
where $H=H(\k^{(n)}(\varphi ))$, $H_0=H_0(\k^{(n)}(\varphi ))$, and $\tilde
P_j^{(n)}$
%%%%%%%%%%%%%%%%=P_{j,s}^{(n)}+P_{j,b}^{(n)}+P_{j,g}^{'(n)}+P_{j,w}^{'(n)}+P_{j,nr}^{'(n)}+P(r_{n-1})$
is the projection $P^{(n)}$, see \eqref{P(n)}, corresponding to $\varphi _{0}(j)$ in
the interval $\Delta _j^{(n)}$ containing $\varphi $. Let
\begin{equation}W^{(n)}=H^{(n+1)}-\tilde H^{(n)}=P(r_{n})VP(r_{n})-\tilde
P_j^{(n)}V\tilde
P_j^{(n)}, \label{W2last}\end{equation}
\begin{equation}\label{g3last} g^{(n+1)}_r({\k}):=\frac{(-1)^r}{2\pi
ir}\hbox{Tr}\oint_{C_{n+1}}\left(W^{(n)}(\tilde
H^{(n)}({\k})-zI)^{-1}\right)^rdz,
\end{equation} \begin{equation}\label{G3last}
G^{(n)}_r({\k}):=\frac{(-1)^{r+1}}{2\pi i}\oint_{C_{n+1}}(\tilde
H^{(n)}({\k})-zI)^{-1}\left(W^{(n)}(\tilde
H^{(n)}({\k})-zI)^{-1}\right)^rdz,
\end{equation}
where $C_{n+1}$ is the circle $|z-k^{2l}|=\varepsilon _0^{(n+1)}$,
$\varepsilon _0^{(n+1)}=k^{-2r_n'k^{2\gamma r_{n-1}}}.$

Recall that $\beta:=2l-2-41\mu\delta$. The proof of the following
statements is analogous to the one in Step III (see
Theorem~\ref{Thm3}, Corollary~\ref{corthm3},
Lemma~\ref{L:derivatives-3} and Lemma~\ref{ldk-3}).

\begin{theorem} \label{Thm3last} Suppose $k>k_*$, $\varphi $ is in
the real  $k^{-r_n'-\delta }$-neighborhood of $\omega
^{(n+1)}(k,\delta,\tau )$ and $\varkappa\in\R$,
$|\varkappa-\varkappa^{(n)}(\varphi )|\leq \varepsilon
^{(n+1)}_0k^{-2l+1-\delta }$, $\k=\varkappa(\cos \varphi ,\sin
\varphi )$.  Then, there exists a single eigenvalue of
$H^{(n+1)}({\k})$ in the interval $\varepsilon _{n+1}( k,\delta,\tau
)=\left( k^{2l}-\varepsilon _0^{(n+1)}, k^{2l}+\varepsilon
_0^{(n+1)}\right)$. It is given by the absolutely converging series
series:
\begin{equation}\label{eigenvalue-3last}
\lambda^{(n+1)}({\k})=\lambda^{(n)}({\k})+ \sum\limits_{r=2}^\infty
g^{(n+1)}_r({\k}).\end{equation} For coefficients
$g^{(n+1)}_r({\k})$ the following estimates hold:
\begin{equation}\label{estg3last} |g^{(n+1)}_r({\k})|<
k^{-\frac{\beta}{5}
k^{r_{n-1}-r_{n-2} }-\beta (r-1)}.
\end{equation}
The corresponding spectral projection is given by the series:
\begin{equation}\label{sprojector-3last}
\E ^{(n+1)}({\k})=\E^{(n)}({\k})+\sum\limits_{r=1}^\infty
G^{(n+1)}_r({\k}), \end{equation} $\E^{(n)}({\k})$ being the
spectral projection of $H^{(n)}$. The operators $G^{(n+1)}_r({\k})$
satisfy the estimates:
\begin{equation}
\label{Feb1a-3last}
\left\|G^{(n+1)}_r({\k})\right\|_1<k^{-\frac{\beta}{10}
k^{r_{n-1}-r_{n-2} } -\beta r},
\end{equation}
\begin{equation}G^{(n+1)}_r({\k})_{\s\s'}=0,\ \mbox{when}\ \ \
2rk^{\gamma r_{n-1}+150l\gamma
r_{n-2}}+3k^{r_{n-1}}<\||\p_\s\||+\||\p_{\s'}\||.\label{Feb6a-3last}
\end{equation}
\end{theorem}
\begin{corollary}\label{corthm3last} For the perturbed eigenvalue and its spectral
projection the following estimates hold:
 \begin{equation}\label{perturbation-3last}
\lambda^{(n+1)}({\k})=\lambda^{(n)}({\k})+ O_2\left(k^{-\frac15
\beta k^{r_{n-1}-r_{n-2} }-\beta }\right),
\end{equation}
\begin{equation}\label{perturbation*-3last}
\left\|\E^{(n+1)}({\k})-\E^{(n)}({\k})\right\|_1<2k^{-\frac{\beta}{10}
k^{r_{n-1}-r_{n-2} }-\beta }.
\end{equation}
\begin{equation}
\left|\E^{(n+1)}({\k})_{\s\s'}\right|<k^{-d^{(n+1)}(\s,\s')},\ \
\mbox{when}\ \||\p_\s\||>4k^{r_{n-1}} \mbox{\ or }
\||\p_{\s'}\||>4k^{r_{n-1} },\label{Feb6b-3last}
\end{equation}
$$d^{(n+1)}(\s,\s')=\frac18(\||\p_\s\||+\||\p_{\s'}\||)k^{-\gamma r_{n-1}-150l\gamma
r_{n-2}}\beta +\frac{1}{10}\beta k^{r_{n-1}-r_{n-2} }.$$
\end{corollary}

\begin{lemma} \label{L:derivatives-3last}Under conditions of Theorem \ref{Thm3last} the following
estimates hold when $\varphi \in \omega ^{(n+1)}(k,\delta )$ or its
complex $k^{-r_n'-\delta}$ neighborhood and $\varkappa\in \C$,
$|\varkappa-\varkappa^{(n)}(\varphi
)|<\varepsilon_0^{(n+1)}k^{-2l+1-\delta}$.
\begin{equation}\label{perturbation-3cIVlast}
\lambda^{(n+1)}({\k})=\lambda^{(n)}({\k})+O_2\left(k^{-\frac 15
\beta k^{r_{n-1}-r_{n-2}}-\beta }\right),
\end{equation}
\begin{equation}\label{estgder1-3kIVlast}
\frac{\partial\lambda^{(n+1)}}{\partial\varkappa}=
\frac{\partial\lambda^{(n)}}{\partial\varkappa} +O_2\left(k^{-\frac
15 \beta k^{r_{n-1}-r_{n-2}}-\beta }M_{n-1}\right), \ \ \ \
M_{n-1}:=\frac{k^{2l-1+\delta}}{\varepsilon
^{(n+1)}_0},\end{equation}
\begin{equation}\label{estgder1-3phiIVlast}\frac{\partial\lambda^{(n+1)}}{\partial \varphi }=\frac{\partial\lambda^{(n)}}{\partial \varphi }+
O_2\left(k^{-\frac 15 \beta k^{r_{n-1}-r_{n-2}}-\beta+r_n'+\delta
}\right),
 \end{equation}
\begin{equation}\label{estgder2-3IVlast}
\frac{\partial^2\lambda^{(n+1)}}{\partial\varkappa^2}=
\frac{\partial^2\lambda^{(n)}}{\partial\varkappa^2}+
O_2\left(k^{-\frac 15 \beta k^{r_{n-1}-r_{n-2}}-\beta
}M_{n-1}^2\right),
\end{equation}
\begin{equation} \label{gulf2-3IVlast}
\frac{\partial^2\lambda^{(n+1)}}{\partial\varkappa\partial \varphi
}=\frac{\partial^2\lambda^{(n)}}{\partial\varkappa\partial \varphi
}+ O_2\left(k^{-\frac 15 \beta k^{r_{n-1}-r_{n-2}}-\beta+r_n'+\delta
}M_{n-1}\right),
\end{equation}
\begin{equation} \label{gulf3-3IVlast}
\frac{\partial^2\lambda^{(n+1)}}{\partial\varphi
^2}=\frac{\partial^2\lambda^{(n)}}{\partial\varphi
^2}+O_2\left(k^{-\frac 15 \beta
k^{r_{n-1}-r_{n-2}}-\beta+2r_n'+2\delta }\right).
\end{equation}\end{lemma}

\begin{corollary} \label{"O"last} All ``$O_2$"-s on the right hand sides of \eqref{perturbation-3cIVlast}-\eqref{gulf3-3IVlast}
can be written as $O_1\left(k^{-\frac {1}{10} \beta
k^{r_{n-1}-r_{n-2}}}\right)$.
\end{corollary}

\subsection{\label{IS3IVlast}Isoenergetic Surface for Operator $H^{(n+1)}$}

The following statement is an analogue of Lemma~\ref{ldk-3}.

\begin{lemma}\label{ldk-3IVlast} \begin{enumerate}
\item For every $\lambda :=k^{2l}$,  $k>k_*$, and $\varphi $ in the real  $\frac{1}{2} k^{-r_n'-\delta }$-neighborhood
of $\omega^{(n+1)}(k,\delta, \tau )$ , there is a unique
$\varkappa^{(n+1)}(\lambda, \varphi )$ in the interval
$I_n:=[\varkappa^{(n)}(\lambda, \varphi )-\varepsilon
^{(n+1)}_0k^{-2l+1-\delta},\varkappa^{(n)}(\lambda, \varphi
)+\varepsilon ^{(n+1)}_0k^{-2l+1-\delta}]$, such that
    \begin{equation}\label{2.70-3IVlast}
    \lambda^{(n+1)} \left(\k
^{(n+1)}(\lambda ,\varphi )\right)=\lambda ,\ \ \k ^{(n+1)}(\lambda
,\varphi ):=\varkappa^{(n+1)}(\lambda ,\varphi )\vec \nu(\varphi).
    \end{equation}
\item  Furthermore, there exists an analytic in $ \varphi $ continuation  of
$\varkappa^{(n+1)}(\lambda ,\varphi )$ to the complex  $\frac{1}{2}
k^{-r_n'-\delta }$-neighborhood of $\omega^{(n+1)}(k,\delta, \tau )$
such that $\lambda^{(n+1)} (\k ^{(n+1)}(\lambda, \varphi ))=\lambda
$. Function $\varkappa^{(n+1)}(\lambda, \varphi )$ can be
represented as $\varkappa^{(n+1)}(\lambda, \varphi
)=\varkappa^{(n)}(\lambda, \varphi )+h^{(n+1)}(\lambda, \varphi )$,
where
\begin{equation}\label{dk0-3IVlast} |h^{(n+1)}(\varphi )|=O_1\left(k^{-\frac 15 \beta k^{r_{n-1}-r_{n-2}}-\beta -2l+1
}\right),
\end{equation}
\begin{equation}\label{dk-3IVlast}
\frac{\partial{h}^{(n+1)}}{\partial\varphi}= O_2\left(k^{-\frac 15
\beta k^{r_{n-1}-r_{n-2}}-\beta -2l+1 +r_n'+\delta }\right),
\end{equation}
\begin{equation}
\frac{\partial^2{h}^{(n+1)}}{\partial\varphi^2}= O_4\left(k^{-\frac
15 \beta k^{r_{n-1}-r_{n-2}}-\beta -2l+1 +2r_n'+2\delta }\right).
\end{equation} \end{enumerate}\end{lemma}

Let us consider the set of points in $\R^2$ given by the formula:
$\k=\k^{(n+1)} (\varphi), \ \ \varphi \in \omega ^{(n+1)}(k,\delta,
\tau )$. By Lemma \ref{ldk-3IVlast} this set of points is a slight
distortion of ${\cal D}_{n}$. All the points of this curve satisfy
the equation $\lambda^{(n+1)}(\k^{(n+1)}(\varphi ))=k^{2l}$. We call
it isoenergetic surface of the operator $H^{(n+1)}$ and denote by
${\cal D}_{n+1}$.

\section{Isoenergetic Sets. Generalized Eigenfunctions of $H$}
\subsection{Construction of Limit-Isoenergetic Set} At every step $n$ we constructed  a set
$\B _n(\lambda)$, ${\cal B}_{n}(\lambda)\subset {\cal
B}_{n-1}(\lambda)\subset S_1$,  and a function
$\varkappa^{(n)}(\lambda,\vec \nu)$, $\vec \nu \in {\cal
B}_n(\lambda)$, with the following properties. The set ${\cal
D}_{n}(\lambda )$ of vectors $\k=\varkappa^{(n)}(\lambda ,\vec
\nu)\vec \nu$,
   $\vec \nu \in {\cal B}_{n}(\lambda )$,
    is a slightly distorted circle with holes,  see Fig.\ref{F:1}, Fig.\ref{F:2}, formula (\ref{Dn})
    and Lemmas \ref{ldk}, \ref{ldk-2},
\ref{ldk-3}, \ref{ldk-3IV}, \ref{ldk-3IVlast}. For any $\k^{(n)}(\lambda,\vec \nu)\in {\cal D}_{n}(\lambda )$ there is a
single eigenvalue of
 $H^{(n)}(\k^{(n)})$
equal to $\lambda $ and  given by the perturbation series.
%%%%%%%\footnote{The operator $H^ {(n)}(\k )$ is defined for
%%%%%%%every $\k \in \R^2$ as explained in Remark
%%%%%%%\ref{R:May20}, page \pageref{R:May20}. The perturbation series is
%%%%%%%given by a formula analogous to (\ref{3.66}), which coincides with
%%%%%%%(\ref{3.15}) up to a shift of indices corresponding to the
%%%%%%%parallel shift of $\k $ into $K_n$.}
Let
    ${\cal B}_{\infty}(\lambda)=\bigcap_{n=1}^{\infty}{\cal B}_n(\lambda).$
Since ${\cal B}_{n+1} \subset {\cal B}_n$ for every $n$, ${\cal
B}_{\infty}(\lambda)$ is a unit circle with infinite number of
holes, more and more holes of smaller and smaller size appearing at
each step. \begin{lemma} \label{L:Dec9} The length of ${\cal
B}_{\infty}(\lambda)$ satisfies estimate (\ref{B}) with $\gamma
_4=37\mu\delta $.
\end{lemma}
\begin{proof}
 Using   Lemmas \ref{L:G1} (part 3),  \ref{L:O2size}, \ref{4.16},
 \ref{7.10} and \ref{10.9} and considering that $r_n>>37\delta \mu $, we easily
conclude that $L\left({\cal B}_n\right)=\left(2\pi+O(k^{-37\mu
\delta })\right)$, $k=\lambda ^{1/2l}$ uniformly in $n$. Since
${\cal B}_n$ is a decreasing sequence of sets, (\ref{B}) holds.
\end{proof} Let us consider
    $\varkappa_{\infty}(\lambda, \vec \nu)=\lim_{n \to \infty}\varkappa^{(n)}(\lambda,\vec \nu),\quad
    \vec \nu \in {\cal B}_{\infty}(\lambda ).$
    \begin{lemma} The limit $\varkappa_{\infty}(\lambda,
    \vec \nu)$ exists for any $\vec \nu \in {\cal B}_{\infty}(\lambda
    )$ and the following estimates hold:
    \begin{align}\label{6.1}& \left|\varkappa_{\infty}(\lambda,
\vec \nu)-\lambda^{1/2l}\right|<ck^{-4l+1+(80\mu +6)\delta},
     \cr &\left|\varkappa_{\infty}(\lambda, \vec \nu)-\varkappa
    ^{(1)}(\lambda,\vec \nu)\right|<ck^{-2k^{\delta }Q^{-1}}k^{-4l+5+48\delta },\cr &\left|\varkappa_{\infty}(\lambda, \vec \nu)-\varkappa
    ^{(n)}(\lambda,\vec \nu)\right|<3k^{-\frac15\beta k^{r _{n-1}-r_{n-2}}}
    ,\ \ \ n\geq 2.
    \end{align}
    \end{lemma}
    \begin{corollary}\label{Dec18}
    For every $\vec \nu \in {\cal B}_{\infty}(\lambda)$ estimate (\ref{h}) holds, where\\ $\gamma
_5=(4l-1-(80\mu +6)\delta)/2l
    >0$.
    \end{corollary}
    The lemma easily follows from
  the estimates (\ref{dk0}),  (\ref{dk0-2}), (\ref{dk0-3}), \eqref{dk0-3IV} and (\ref{dk0-3IVlast}).

    Estimates (\ref{dk}), (\ref{dk-2}) (\ref{dk-3}), \eqref{dk-3IV} and (\ref{dk-3IVlast}) justify convergence of the
    series $\sum_{n=1}^{\infty}
    \frac{\partial h_n}{\partial \varphi },$ and hence,
    of the sequence $\frac{\partial \varkappa^{(n)}}{\partial \varphi }.$
    We denote the limit of this sequence by $\frac{\partial \varkappa_{\infty}}{\partial \varphi }.$
    \begin{lemma} The  estimate (\ref{Dec9a}) with $\gamma
_{11}=(4l-1-(120\mu+7)\delta)/2l >0,$ holds for any $\vec \nu \in
    {\cal B}_{\infty}(\lambda)$.
   \end{lemma}
We define ${\cal D}_{\infty}(\lambda )$ by (\ref{Dinfty}). Clearly,
${\cal D}_{\infty}(\lambda )$ is a slightly distorted circle of
radius $k$ with infinite number of holes. We can assign a tangent
vector $\frac{\partial \varkappa_{\infty }}{\partial \varphi }\vec
\nu +\varkappa_{\infty } \vec \mu $, $\vec \mu =(-\sin \varphi ,\cos
\varphi )$ to the curve ${\cal D}_{\infty}(\lambda )$, this tangent
vector being the limit of corresponding tangent vectors for  curves
${\cal D}_{n}(\lambda )$ at points $\k^{(n)}(\lambda ,\vec \nu )$ as
$n\to \infty $.

Next we show that ${\cal D}_{\infty}(\lambda )$ is an isoenergetic
curve for $H$. Namely for every $\k \in {\cal
D}_{\infty}(\lambda )$ there is a generalized eigenfunction $\Psi
_{\infty }(\k ,\x)$ such that $H\Psi _{\infty }=\lambda
\Psi _{\infty }$.

\subsection{Generalized Eigenfunctions of H}

 We show that for
    every  $\k $ in
a set $${\cal G} _{\infty }=\cup _{\lambda
>\lambda _*}{\cal D}_{\infty}(\lambda ),\ \ \lambda _*= k _*^{2l} ,$$ there is a solution
$\Psi _{\infty }(\k , \x)$ of the equation for
eigenfunctions:
    \begin{equation} (-\Delta)^{l}\Psi _{\infty}(\k , \x)+V(\x)\Psi _{\infty }(\k ,
    \x)=\lambda _{\infty}(\k )\Psi _{\infty }(\k , \x),
    \label{6.2.1}
    \end{equation}
which can be represented in the form
    \begin{equation}
    \Psi _{\infty }(\k , \x)=e^{i\langle \k , \x
    \rangle}\Bigl(1+u_{\infty}(\k , \x)\Bigr),\ \ \ \
    \bigl\|u_{\infty}(\k , \x))\bigr\| _{L_{\infty }(\R^2)}=O(|\k |^{-\gamma_1}),
   \label{6.2.1a}
    \end{equation}
where $u_{\infty}(\k , \x)$ is a quasi-periodic
function,
 $\gamma _1=2l-1-45\mu \delta >0$; the eigenvalue $\lambda _{\infty}(\k )$ satisfies the asymptotic
formula:
 \begin{equation}\lambda _{\infty}(\k )=|\k |^{2l}+O(|\k |^{-\gamma _2}), \ \ \ \gamma _2=2l-(80\mu +6)\delta
>0.\label{6.2.4}
\end{equation}
We also show that the set $\cal{G} _{\infty }$ satisfies
(\ref{full}).

 In fact,   by (\ref{6.1}), any $\k \in {\cal D}_{\infty}(\lambda
)$ belongs to the $k^{-\frac15\beta k^{r
_{n-1}-r_{n-2}}}$-neighborhood of ${\cal D}_n(\lambda ),\ \ n\geq3$.
Let us consider spectral projectors $\E^{(n)}$, each of them being
defined in a finite dimensional space of sequences with indices in
$\Omega (r_{n-1})$, $r_0:=\delta $. We extend each of them to the
whole space $\ell^{2}(\Z^4)$ by putting it to be zero into the
orthogonal complement of $\ell^{2}\left(\Omega (r_{n-1})\right)$.
This way they all act in space $\ell^{2}(\Z^4)$.  Applying the
perturbation formulae proved in the previous sections (see
Corollaries \ref{corthm1}, \ref{corthm2}, \ref{corthm3},
\ref{corthm3IV}, \ref{corthm3last}), we obtain the following
inequalities:
    \begin{equation}
    \begin{split}&\bigl\|\E^{(1)}(\k)-{\E}_0(\k)\bigr\|_1<ck^{-\gamma_0},
    \cr &\bigl\|\E^{(2)}(\k)-{\E}^{(1)}(\k)\bigr\|_1<ck^{-k^\delta Q^{-1}+2-\gamma_0},\ \
    \gamma_0:=2l-1-44\mu\delta,\
    \cr &\bigl\|\E^{(n)}(\k)-{\E}^{(n-1)}(\k)\bigr\|_1<
    2k^{-\frac1{10}\beta k^{r_{n-2}-r_{n-3}} -\beta}, \quad n \geq 3,\label{6.2.2}
    \end{split}
    \end{equation}
    \begin{equation}
    \begin{split}&\bigl|\lambda ^{(1)}(\k)-|\k |^{2l}
     \bigr|
     <ck^{-\gamma _2}
    , \ \ \bigl|\lambda ^{(2)}(\k)-\lambda ^{(1)}(\k)
     \bigr|
     <ck^{-2k^\delta Q^{-1}+4-\gamma _2},
     \cr &\bigl|\lambda ^{(n)}(\k )-\lambda ^{(n-1)}(\k )\bigr|<
     2k^{-\frac15\beta k^{r_{n-2}-{r_{n-3}}}-\beta },
     \quad n \geq 3,
    \label{6.2.3}
    \end{split}
    \end{equation}
where  $\lambda ^{(n+1)}(\k )$ is the eigenvalue
corresponding to $\E^{(n+1)}(\k)$, $\
{\E}_0(\k)$ corresponds to $V=0$.

 \begin{remark} \label{R:Dec9} We see from (\ref{6.1}), that any $\k
\in {\cal D}_{\infty}(\lambda )$ belongs to the $k^{-\frac15\beta
k^{r _{n-1}-r_{n-2}}}$-neighborhood of ${\cal D}_n(\lambda ),\
n\geq3$. Applying perturbation formulae for $n$-th step, we easily
obtain that  there is an
 eigenvalue  $\lambda^{(n)}(\k )$ of $H^{(n)}(\k )$
satisfying the estimate $\lambda^{(n)}(\k )=\lambda
+\delta _n$, $\delta _n=o(1)$ as $n\to \infty $, the eigenvalue
$\lambda^{(n)}(\k )$
     being given by a perturbation
series of the type (\ref{eigenvalue-3last}). Hence, for every
$\k \in {\cal D}_{\infty}(\lambda)$ there is a limit:
\begin{equation} \lim _{n\to \infty }\lambda^{(n)}(\k
)=\lambda.\label{6.2}
\end{equation}
\end{remark}
Let $\v^{(n)}$ be a unit vector corresponding to  the projection
${\E}^{(n)}(\k )$, ${\E}^{(n)}(\k)=(\cdot ,\v^{(n)})\v^{(n)}$,
$\v^{(n)}=\{ v^{(n)}_{\s}\}_{\s \in \Z^4}\in\ell^2(\Z^4)$. By
construction, $v^{(n)}_{\s}=0$ when $\s\not \in \Omega (r_{n-1})$.
Let us consider the linear combination of exponents corresponding to
this vector:
$$\Psi _n(\k,\x)=\sum _{\s \in \Omega (r_{n-1})} v^{(n)}_{\s}e^{i\langle\k+\p_{\s},\x\rangle}.$$
\begin{lemma}\label{prelimit} Function $\Psi _n(\k,\x),\ n\geq4,$ satisfies the equation:
$$(-\Delta)^l \Psi _n(\k ,\x)+V(\x)\Psi _n(\k,\x)=
\lambda _n(\k)\Psi _n(\k,\x)+g_n(\k,\x),$$
 the vector $\g _n$ of the Fourier coefficients of $g_n(\k,\x)$ satisfying the estimate:
\begin{equation}\|\g_n\|_{\ell
^1(\Z^4)}<k^{-k^{\frac 12 r_{n-1}}}.
%%%%%%%%k^{-\frac1{10}\beta k^{r_{n-2}-r_{n-3}} +2r_{n-1}}.
\label{Aug13-2} \end{equation} Coefficients $(\g_n)_{\s}$
can differ from zero only when  $k^{r_{n-1}}<\||\p_\s\||\leq k^{r_{n-1}}+Q$. Function
$g_n(\k,\x)$ obeys the estimate:
\begin{equation}\label{g_n}\|g_n\|_{L_{\infty }(\R^2 )}<k^{-
k^{\frac 12 r_{n-1}}}.\end{equation} \end{lemma}
\begin{proof} Let $P(r_{n-1})$ be the projection in $\ell
^2(\Z^4)$ on the subspace corresponding to $\Omega (r_{n-1})$.
By  construction, $P(r_{n-1})\v^{(n)}=\v^{(n)}$ and
$$H_0\v^{(n)}+P(r_{n-1})VP(r_{n-1})\v^{(n)}=\lambda _n(\k )\v^{(n)}.$$
Since $V$ is a trigonometric polynomial,
$$(I-P(r_{n-1}))VP(r_{n-1})=(I-P(r_{n-1}))VP_{\partial }(r_{n-1}),$$
where $P_{\partial }(r_{n-1})_{\m \m}=1$ only if $\m$ is in the $Q$-vicinity of the boundary. Using \eqref{Feb6b-3last} with $n$ instead of $n+1$, we obtain:
 $\|P_{\partial
}(r_{n-1}){\E}^{(n)}\|<k^{- k^{r_{n-1}(1-\gamma)}}$ and, hence,
$\|P_{\partial }(r_{n-1}){\v^{(n)}}\|<k^{- k^{r_{n-1}(1-\gamma)}}$.
Therefore, $\|(I-P(r_{n-1}))VP(r_{n-1})\v^{(n)}\|<\|V\|k^{-
k^{r_{n-1}(1-\gamma)} }$. It follows that $H(\k)\v^{(n)}=\lambda
_n(\k)\v^{(n)}+\g_{n}$, where $\|\g_{n}\|_{\ell ^2(\Z^4)}<\|V\|k^{-
k^{ r_{n-1}(1-\gamma)} }$. Note that elements $(\g_n)_{\s}$ are
equal to zero when $ \||\p_\s\||\leq k^{r_{n-1}}$ or $ \||\p_\s\||>
k^{r_{n-1}}+Q$. Therefore, \eqref{Aug13-2} holds. Estimate
\eqref{g_n} follows.\end{proof}
\begin{lemma} Functions $\Psi _n(\k,\x)$ satisfy the inequalities:
\begin{equation}\label{psi1}
    \|\Psi _{1}-\Psi _0\|_{L_{\infty }(\R^2)}<ck^{-\gamma _0+2\delta },\ \ \
    \|(-\Delta)^l \Psi _{1}-(-\Delta)^l \Psi _0\|_{L_{\infty }(\R^2)}<ck^{-\gamma _0+2\delta +2l},\ \ \
    \Psi _0(\x)=e^{i \langle \k ,\x
    \rangle},
    \end{equation}
    \begin{equation}\label{psi2}
    \begin{split}
&\|\Psi _{2}-\Psi _1\|_{L_{\infty
}(\R^2)}<ck^{-k^{\delta}Q^{-1}+2-\gamma _0+2r_1
    },\cr &\|(-\Delta)^l \Psi _{2}-(-\Delta)^l \Psi _1\|_{L_{\infty
}(\R^2)}<ck^{-k^{\delta}Q^{-1}+2-\gamma _0+(2+2l)r_1
    },
    \end{split}
    \end{equation}
    \begin{equation}
    \begin{split}
    &\|\Psi _{n}-\Psi _{n-1}\|_{L_{\infty}(\R^2)}
    < k^{-\frac1{10}\beta k^{r_{n-2}-r_{n-3}}+2r_{n-1} }, \ \cr &\|(-\Delta)^l \Psi _{n}-(-\Delta)^l \Psi _{n-1}\|_{L_{\infty}(\R^2)}
    < (2\pi )^{2l}k^{-\frac1{10}\beta k^{r_{n-2}-r_{n-3}}+(2+2l)r_{n-1} },\ \ n \geq 3. \label{Dec10}
\end{split}
    \end{equation}
    \end{lemma}
        \begin{corollary} \label{C:Psi}All functions $\Psi _{n}$, $n=0,1,...,$ obey the estimate $ \|\Psi _{n}\|_{L_{\infty
}(\R^2)}<1+Ck^{-\gamma _0+2\delta }$ uniformly in $n$. \end{corollary}
    \begin{proof} Using \eqref{6.2.2} and considering that
    $v^{(n)}_{\s}$ are equal to zero when $\||\p_\s \||>k^{r_{n-1}}$, we obtain
    \begin{equation}
    \begin{split}&\|\v^{(1)}-\v^{(0)}\|_{\ell ^{1}(\Z^4)}<ck^{-\gamma _0+2\delta },\
    \ \|\v^{(2)}-\v^{(1)}\|_{\ell ^{1}(\Z^4)}<ck^{-k^{\delta}Q^{-1}+2-\gamma _0+2r_1
    },\cr &
    \|\v^{(n)}-\v^{(n-1)}\|_{\ell ^{1}(\Z^4)}<3k^{-\frac1{10}\beta k^{r_{n-2}-r_{n-3}}+2r_{n-1} -\beta},\ \ n\geq3.
    \end{split}\label{Cauchy-1} \end{equation}
    \begin{equation}
    \begin{split}&\|H_0(\v^{(1)}-\v^{(0)})\|_{\ell ^{1}(\Z^4)}<ck^{-\gamma _0+2\delta +2l},\ \
    \|H_0(\v^{(2)}-\v^{(1)})\|_{\ell ^{1}(\Z^4)}<ck^{-k^{\delta}Q^{-1}+2-\gamma _0+(2+2l)r_1
    },\cr &\|H_0(\v^{(n)}-\v^{(n-1)})\|_{\ell ^{1}(\Z^4)}<(2\pi )^{2l}k^{-\frac1{10}\beta k^{r_{n-2}-r_{n-3}}+(2+2l)r_{n-1} -\beta},\ \
    n\geq3.\label{Cauchy-2}\end{split} \end{equation}
Now \eqref{psi1} -- \eqref{Dec10} easily follow. \end{proof}

\begin{theorem} \label{T:Dec10}
For every  $\lambda >k_*^{2l}$ and $\k \in {\cal
D}_{\infty}(\lambda)$ the sequence of functions $\Psi _n(\k,\x)$
converges  in $L_{\infty}(\R^2)$ and $W_{2,loc}^{2l}(\R^2)$. The
limit function $\Psi _{\infty }(\k,\x):=\lim_{n\to \infty }
    \Psi _n(\k,\x)$, is a quasi-periodic function:
  \begin{equation} \label{quasi-periodic}
  \Psi _{\infty }(\k,\x)=\sum _{\s \in \Z^4}  (v_{\infty})_{\s}e^{i\langle\k +\p_\s,\x\rangle}, \end{equation}
  where $\v_{\infty }=\{(v_\infty)_\s\}_{\s\in\Z^4}\in \ell^1(\Z^4)$ and $\|\v_{\infty }\|_{\ell^2 (\Z^4)}=1$.
  The function $\Psi _{\infty }(\k,\x)$   satisfies the equation
    \begin{equation}\label{6.7}
     (-\Delta)^{l}\Psi _{\infty }(\k, \x)+V(\x)\Psi _{\infty }(\k,
    \x)= \lambda \Psi _{\infty }(\k, \x).
    \end{equation}
 It can be represented in the form
   \begin{equation}\label{6.4}
    \Psi _{\infty }(\k,\x)=e^{i\langle \k, \x
    \rangle}\bigl(1+u_{\infty}(\k, \x)\bigr),
    \end{equation}
where $u_{\infty}(\k, \x)$ is a quasi-periodic
function:
   \begin{equation}\label{6.5}
    u_{\infty}(\k, \x)=\sum_{n=1}^{\infty} u_n(\k,
    \x),
    \end{equation}
 \begin{equation}u_n(\k, \x)=\sum _{\s \in \Omega (r_{n-1})}
 (v^{(n)}_{\s}-v^{(n-1)}_{\s})e^{i\langle\p_\s,\x\rangle},\ \ \  r_0:=\delta ,\label{u_n} \end{equation}
 functions $u_n$ satisfying the estimates:
\begin{equation}
    \|u_1\|_{L_{\infty}(\R^2)} <ck^{-\gamma_0+2\delta }, \ \
    \|u_2\|_{L_{\infty}(\R^2)} <ck^{-k^{\delta}Q^{-1}+2-\gamma _0+2r_1
    },\label{6.6a}
    \end{equation}
    \begin{equation}\label{6.6}
    \|{u}_n\|_{L_{\infty}(\R^2)} <k^{-\frac1{10}\beta k^{r_{n-2}-r_{n-3}}+2r_{n-1} }, \ \ \ n\geq
    3 .\end{equation}
\end{theorem}
\begin{corollary} Function $u_{\infty}(\k, \x)$ obeys the
 estimate (\ref{6.2.1a}).
\end{corollary}

\begin{proof} Using \eqref{Cauchy-1},\eqref{Cauchy-2}, we obtain that the sequence $\v^{(n)}$ has the limit in $\ell ^1(\Z^4)$. We denote this limit by $\v_{\infty }$. Since, vectors $\v^{(n)}$ are normalized in $\ell^ {2}(\Z ^4)$,
\begin{equation}
\|\v_{\infty }\|_{\ell ^2(\Z^4)}=1.\label{norm} \end{equation} By
\eqref{Dec10}, we obtain that $\Psi _n(\k,\x)$ is a Cauchy sequence
in $L_{\infty }(\R^2)$ and $W_{2,loc}^{2l}(\R^2)$. Let $\Psi _{\infty }(\k
,\x)=\lim_{n \to \infty}\Psi _n(\k, \x).$ This limit is
defined pointwise uniformly in $\x$ and in $W_{2,loc}^{2l}(\R^2)$.
Noting also that $\lim \lambda _n(\k )=\lambda $, and
taking into account Lemma~\ref{prelimit} we obtain that \eqref{6.7}
holds.

 Let us show that $\Psi _{\infty }$ is a
quasi-periodic function. Obviously,
    $$ \Psi _{\infty } =\Psi _0+\sum_{n=1}^{\infty}(\Psi_{n}-\Psi
    _{n-1}),
$$ the series converging in $L_{\infty}(\R^2)$ by (\ref{Dec10}).
Introducing $u_{n}:=e^{-i \langle \k ,\x\rangle}(\Psi_{n}-\Psi
_{n-1}),$ we arrive at (\ref{6.4}), (\ref{6.5}). Note that
 $u_n$ has a form \eqref{u_n} Estimates (\ref{6.6a}), (\ref{6.6})  follow from (\ref{psi1}), (\ref{Dec10}).
 \end{proof}

\begin{theorem} Formulae (\ref{6.2.1}), (\ref{6.2.1a}) and (\ref{6.2.4}) hold for every $\k \in \cal{G} _{\infty }$.
The set $\cal{G} _{\infty }$ is Lebesgue measurable and satisfies
  (\ref{full})  with $\gamma _3=37\mu\delta $.
\end{theorem}
\begin{proof}
By Theorem~\ref{T:Dec10}, (\ref{6.2.1}), (\ref{6.2.1a}) hold, where
$\lambda _{\infty}(\k )=\lambda $ for $\k \in {\cal
D}_{\infty}(\lambda )$. Using (\ref{h}), which is proven in
Corollary~\ref{Dec18}, with $\varkappa_{\infty }=|\k |$, we easily
obtain (\ref{6.2.4}). It remains to prove (\ref{full}). Let us
consider a small region $U_n(\lambda _0)$ around an isoenergetic
surface ${\cal D}_n(\lambda _0)$, $\lambda _0>k_*^{2l}$. Namely,
$U_n(\lambda _0)=\cup _{|\lambda -\lambda _0|<\varepsilon
_0^{(n+1)}} {\cal D}_{n}(\lambda )$, $k=\lambda _0^{1/2l}$. By
Theorem \ref{Thm3last} the construction of the $n$-th non-resonant
set is stable in $\varepsilon _0^{(n)}$ -neighborhood of $\lambda
_0$. Therefore, in fact, we can (and for the sake of convenience
will) assume that the sets $\omega^{(n)}(\lambda)$ are chosen to be
equal to $\omega^{(n)}(\lambda_0)$ for
$|\lambda-\lambda_0|<\varepsilon _0^{(n)}$ . Thus, $U_n(\lambda _0)$
is an open set (a distorted ring with holes) and $|U_n(\lambda
_0)|=l^{-1}\varepsilon _0^{(n)} k^{-2l+1}2\pi k
\left(1+O(k^{-37\mu\delta})\right)$. It easily follows from
\eqref{dk1} and \eqref{dk0-3IVlast} that $U_{n+1}\subset U_n$.
Definition of ${\cal D}_{\infty }(\lambda _0)$ yield: ${\cal
D}_{\infty }(\lambda _0)=\cap _{n=1}^{\infty }U_n(\lambda _0)$.
Hence, ${\cal G}_{\infty }=\cap _{n=1}^{\infty }{\cal G}_n$, where
 \begin{equation}{\cal G}_n=\cup _{\lambda
>\lambda _*}{\cal D}_n(\lambda ), \ \ \lambda _*=k_*^{2l}.\label{Gn} \end{equation} Considering that $U_{n+1}\subset U_n$
for every $\lambda _0>\lambda _*$, we obtain  ${\cal G}_{n+1}\subset
{\cal G}_n$. Hence, $\left|{\cal G} _{\infty }\cap
        \bf B_R\right|=\lim _{n\to \infty }\left|{\cal G} _{n}\cap
        \bf B_R\right|$. Calculating the volume of the region
        $\cup_{\lambda_*<\lambda<R^{2l}}U_{n}(\lambda)$, we easily conclude $\left|{\cal G} _{n}\cap
        \bf B_R\right|=|{\bf B_R}|\left(1+O(R^{-37\mu\delta})\right)$ uniformly in $n$. Thus, we have
        obtained (\ref{full}) with $\gamma _3=37\mu\delta $.
\end{proof}

\begin{theorem}[Bethe-Sommerfeld Conjecture]
The spectrum of  operator $H$ contains a semi-axis.
\end{theorem}
\begin{proof}
    The theorem immediately follows from the fact that the equation \eqref{6.7} has a bounded solution $\Psi _{\infty }(\k ,\x)$ for every  sufficiently large $\lambda $.
\end{proof}
%%%%%%%%%%%%\begin{lemma} Let $\k ,\k ' \in {\cal G}_{\infty }$. Then
%%%%%%%%%%%%$$\lim _{R\to \infty }\int _{\bf B_R}\Psi _{\infty }(\k ,\x)\overline{\Psi _{\infty }(\k ,\x)}d\x={\cal \delta }(\k -\k ').$$
%%%%%%%%%%%%\end{lemma}

\section{Proof of Absolute Continuity of the Spectrum}\label{chapt8}
The proof is somewhat analogous to that for the case limit-periodic
potentials  \cite{KL2}. We will just refer to  \cite{KL2} in some places.
\subsection{Operators $E_n(\mathcal{G}_n')$,
$\mathcal{G}_n'\subset \mathcal{G}_n$}

Let us consider the open sets $\mathcal{G}_n$ given by (\ref{Gn}).
 There is a family of  eigenfunctions $\Psi _n(\k ,\x)$, $\k \in \mathcal{G}_n$, of the operator
    $H^{(n)}$, which are described by the perturbation formulas
    (\ref{na}), (\ref{na-n}). Let, $\mathcal{G}_n'\subset \mathcal{G}_n$, where $\mathcal{G}_n'$ is Lebesgue measurable and bounded.
    Let
    %%%%%%%%%% us consider $E_n\left( \mathcal{G}'_n\right): L_2(\R^2)\to
%%%%%%%%%%L_2(\R^2)$
    \begin{equation} E_n\left( \mathcal{G}'_n\right)F=\frac{1}{4\pi ^2}\int
    _{ \mathcal{G}'_n}\bigl( F,\Psi _n(\k )\bigr) \Psi _n(\k) d\k \label{s}
    \end{equation}
    for any $F\in C_0^{\infty}(\R^2)$, here and below $\bigl( \cdot ,\cdot \bigr)$
    is the canonical scalar product in $L_2(\R^2)$, i.e.,
    $$\bigl( F,\Psi _n(\k )\bigr)=\int _{\R^2}F(\x)\overline{\Psi _n(\k ,\x)}d\x.$$
    We will show that $ E_n\left( \mathcal{G}'_n\right)$ is almost a projector in $L_2(\R^2)$ in the sense: $ E_n\left( \mathcal{G}'_n\right)= E_n^*\left( \mathcal{G}'_n\right)$, $ E_n^2\left( \mathcal{G}'_n\right)= E_n\left( \mathcal{G}'_n\right)+o(1)$, where $o(1)$ is in the class of bounded operators as $n\to \infty $.
First, we note that \eqref{s} can be rewritten in the form:
    \begin{equation} E_n\left(\mathcal{G}'_n\right)=S_n\left(\mathcal{G}'_n\right)T_n \left(
    \mathcal{G}'_n\right), \label{ST}
    \end{equation}
    %%%5555L_2(\R^2)
    $$T_n:L_2(\R^2) \to L_2\left(  \mathcal{G}'_n\right), \ \
    \ \ S_n:L_2\left( \mathcal{G}'_n\right)\to L_2(\R^2),$$
    \begin{equation}
    T_nF=\frac{1}{2\pi }\bigl( F,\Psi _n(\k )\bigr) \mbox{\ \ for any $F\in C_0^{\infty}(\R^2)$},
    \label{eq2}
    \end{equation}
    $T_nF$ being in $L_{\infty }\left(  \mathcal{G}'_n\right)$, and,
    \begin{equation}S_nf = \frac{1}{2\pi }\int _{  \mathcal{G}'_n}f (\k)\Psi _n(\k ,\x)d\k  \mbox{\ \ for any $f \in L_{\infty }\left(
    \mathcal{G}'_n\right)$.} \label{ev}
    \end{equation}
    Note that $S_nf \in L_2(\R^2)$, since $\Psi _n$ is a finite combination of exponentials $e^{i\langle\k +\p_{\q},\x\rangle}$.
%%%%%%%%%It is
%%%%%%%%%    easy to show that $T_nF\in L_{\infty }(\mathcal G_n)$, when $F\in C_0^{\infty }(\R^2)$.
%%%%%%%%%    Hence $E_n\left( \mathcal{G}'_n\right)$ can be described by formula
%%%%%%%%%    (\ref{s}) for $F\in C_0^{\infty }(\R^2)$.
 %%%%%%%By \cite{6r}, $\|T_nF\|_{L_2\left( \mathcal{G}'_n\right)}\leq \|F\|_{L_2(\R^2)}$ and
 %%%%%%$\|S_nf \|_{L_2(\R^2)}\leq \|f \|_{L_2\left(
 %%%%%%\mathcal{G}'_n\right)}$.
  %%%%%%   Hence, $T_n$, $S_n$ can be extended by
  %%%%%%   continuity from $C_0^{\infty }(\R^2)$, $L_{\infty }\left(  \mathcal{G}'_n\right)$
  %%%%%%   to $L_2(\R^2)$ and $L_2\left( \mathcal{G}'_n\right)$,
   %%%%%%  respectively. Thus, the operator $E_n\left( \mathcal{G}'_n\right)$ is
   %%%%%%  described by (\ref{ST}) in the whole space $L_2(\R^2)$.

   \begin{lemma}\label{9.1} Let $\mathcal{G}'_n$ be  bounded and $f(\k), g(\k)\in L_{\infty }\left(
    \mathcal{G}'_n\right)$. Then,
    \begin{equation} \left( S_n f, S_n g \right)_{L_2(\R^2)}=_{n\to \infty}
    \left(f, g \right)_{L_2(\mathcal{G}'_n)}+o(1)\|f\|_{L_2(\mathcal{G}'_n)}\|g\|_{L^2(\mathcal{G}'_n)}. \label{May14a-12}\end{equation}
    where $o(1)$ goes to zero uniformly in $f$, $g$  and $\mathcal{G}'_n$ as $n\to \infty $; namely, $|o(1)|<\xi_*^{-{r_{n-3}(\xi_*)}}$, where  $\xi _*=\inf _{\vec \xi \in \mathcal{G}'_n}|\vec \xi |$.
\end{lemma}
\begin{corollary} The following relation holds:
     \begin{equation} \left|\left( S_n f, S_n g \right)_{L_2(\R^2)}\right|<(1+o(1))
     \left\|f\right\|_{L_{\infty}(\mathcal{G}'_n)}\left\|g\right\|_{L_{\infty}(\mathcal{G}'_n)}|\mathcal{G}_n'| \label{Vienna-1},
     \end{equation}
     where $|\mathcal{G}_n'|$ is the Lebesgue measure of $\mathcal{G}_n'$.
     \end{corollary}
\begin{corollary} The operator $S_n$ is bounded and $\|S_n\|=_{n\to \infty} 1+o(1)$. \end{corollary}
\begin{proof}
     The function $\Psi _n(\k,\x)$  can be
    represented as a  combination of plane waves:
    \begin{equation}\Psi _n(\k,\x)=
    \sum _{\m\in \Z^2}v_{\m}^{(n)}(\k)
    \exp\{ i\langle \k+\p _{\m},\x\rangle\},\label{++}
    \end{equation}
where $v_{\m}^{(n)}(\k)$ are Fourier coefficients. By construction,
$v_{\m}^{(n)}(\k)=0$, when $\m \not \in \Omega (r_{n-1}) $. Let
$\v^{(n)}(\k)$ be the vector in $\ell^2(\Z^2)$ with  components
equal to $v_{\m}^{(n)}(\k)$. Note that the size of $\Omega
(r_{n-1})$ depend on $\varkappa=|\k|$; to stress this fact we will
use here the notations $\Omega (r_{n-1},\varkappa)$ and
$r_{n-1}(\varkappa)$.
%%%%%%%The characteristic function of $\Omega (r_{n-1},k)$ we denote by $\chi(r_{n-1},k, \m)$, where $\m$ is an independent variable.
The
Fourier transform  $\widehat \Psi_n$ is a combination of $\delta
$-functions:
    $$\widehat \Psi _n(\k,\vec \xi)=2\pi \sum _{\m \in \Z^2}v_{\m}^{(n)}(\k)
    \delta \bigl(\vec \xi +\k+\p_{\m}\bigr)$$
    %%%2\pi \sum _{\m \in \Omega (r_{n-1},k)}c_{\m}^{(n)}(\k)
   %%% \delta \bigl(\vec \xi +\k+\p_{\m}\bigr)=$$
From this, we easily compute  the Fourier
    transform of $(S_nf )(\x)$:
    $$ (\widehat{S_nf})(\vec \xi)=\sum _{\m \in \Z^2}v_{\m}^{(n)}\bigl(-\vec \xi-
    \p_{\m}\bigr)f \bigl(-\vec \xi-
    \p_{\m}\bigr)\chi \bigl({\cal G}_{n}',-\vec \xi-
    \p_{\m}\bigr),$$
where $\chi ({\cal G}_{n}',\cdot ) $ is the
characteristic function of ${\cal G}_{n}'$. Note that $v_{\m}^{(n)}\bigl(-\vec \xi-
    \p_{\m}\bigr)\chi \bigl({\cal G}_{n}',-\vec \xi-
    \p_{\m}\bigr)$ can differ from zero only when $\m \in \Omega (r_{n-1}, |\vec \xi+
    \p_{\m}|)\subset  \Omega (r_{n-1}, \xi _{**})$, $\xi _{**}=\sup _{\vec\xi \in {\cal G}_{n}'}|\vec\xi |$.
    By Parseval's identity,
    $$
    \left(S_nf,S_ng\right)_{L_2(\R^2)}=\left(\widehat{S_nf},\widehat{S_ng}\right)_{L_2(\R^2)}=$$ $$\int _{\R^2 }\sum _{\m , \m'\in \Z^2} T_{\m,\m'}(\vec \xi )f \bigl(-\vec \xi-
    \p_{\m}\bigr)\bar g \bigl(-\vec \xi-
    \p_{\m'}\bigr)\chi \bigl({\cal G}_{n}',-\vec \xi-
    \p_{\m}\bigr)\chi \bigl({\cal G}_{n}',-\vec \xi-
    \p_{\m'}\bigr)d\vec \xi ,$$ $$
   %%%%  \sum _{\m , \m'\in \Omega (r_{n-1})}\int _{\R^2 }
   T_{\m,\m'}(\vec \xi ):=v_{\m}^{(n)}\bigl(-\vec \xi-
    \p_{\m}\bigr)\overline{v_{\m '}^{(n)}}\bigl(-\vec \xi-
    \p_{\m '}\bigr).$$
    Note that, in fact, the summation here is over the finite set $\m,\m' \in \Omega (r_{n-1}, \xi _{**})$. Hence we can exchange summation and integration in the above formula.
    Next, shifting the variable $\vec \xi+
    \p_{\m} \to \vec \xi $, denoting $\m'-\m$ by $\m''$ and considering that $\langle\v^{(n)},\v^{(n)}\rangle=1$, we obtain:
    $$\left(\widehat{S_nf},\widehat{S_ng}\right)_{L_2(\R^2)}=\left(f,g\right)_{L_2({\cal G}_{n}')}+$$
    %%%% $$
   %%%%  \sum _{\m, \m+\m''\in \Omega (r_{n-1})}\int _{\R^2 } c_{\m}^{(n)}\bigl(-\vec \xi\bigr)\overline{c_{\m +\m''}^{(n)}}\bigl(-\vec \xi-
   %%%%  \p_{\m ''}\bigr)f \bigl(-\vec \xi\bigr)\bar g\bigl(-\vec \xi-
   %%%%  \p_{\m''}\bigr)\chi \bigl({\cal G}_{n}',-\vec \xi \bigr)\chi \bigl({\cal G}_{n}',-\vec \xi-
   %%%%  \p_{\m''}\bigr)\ d\vec \xi =$$
    \begin{equation}\sum _{\m''\in \Z^2\setminus\{\bf 0\}}\int _{\R^2 }B_{\m''}(\vec \xi )f \bigl(-\vec \xi\bigr)\bar g \bigl(-\vec \xi-
    \p_{\m''}\bigr)\chi \bigl({\cal G}_{n}',-\vec \xi \bigr)\chi \bigl({\cal G}_{n}',-\vec \xi-
    \p_{\m''}\bigr)\ d\vec \xi , \label{May14-12}
    \end{equation}
     $$B_{\m''}(\vec \xi )= \sum _{\m\in \Z^2}v_{\m}^{(n)}\bigl(-\vec \xi\bigr)\overline{v_{\m +\m''}^{(n)}}\bigl(-\vec \xi-
    \p_{\m ''}\bigr).$$
    Obviously,
    \begin{equation}B_{\m ''}=\langle\v^{(n)}(-\vec \xi ), \v^{(n)}_*(-\vec \xi-\p_{\m''})\rangle ,\label{Aug11-2} \end{equation}
    where $\v_*^{(n)}\bigl(-\vec \xi-\p_{\m''}\bigr)$ the ``shifted" eigenvector:
    $\v_*^{(n)}\bigl(-\vec \xi-\p_{\m''}\bigr)$: $\left(\v_*^{(n)}\bigl(-\vec \xi-\p_{\m''}\bigr)\right)_{\m}=v^{(n)}_{\m+\m''}\bigl(-\vec \xi-\p_{\m''}\bigr).$
    %%%%, $B_{\bf 0}(\vec \xi )=1$
   %%%%%%%% Obviously, $B_{\m''}(\vec \xi )$ differs from zero only for a finite number of $\m''$.
   To obtain \eqref{May14a-12}, it is enough to prove two estimates:
   \begin{equation} \label{Aug9-2}\sum _{\||\p_{\m''}\||>\xi_*^{ r_{n-3}(\xi _*)}}\ \ \sup _{\vec \xi \in {\cal G}_{n}'}\left| B_{\m''}(\vec \xi )\right| <\frac 12 \xi _*^{- r_{n-3}(\xi _*)}, \end{equation}
\begin{equation} \label{Aug9-2*}\sum _{0<\||\p_{\m''}\||\leq {\xi _*^{ r_{n-3}(\xi _*)}}}\ \ \sup _{\vec \xi \in {\cal G}_{n}'}\left| B_{\m''}(\vec \xi )\right| <\frac 12 \xi _*^{- r_{n-3}(\xi _*)}. \end{equation}
To prove \eqref{Aug9-2} we first check that
    \begin{equation}
    \sup _{\vec \xi \in {\cal G}_{n}'}\left| B_{\m''}(\vec \xi )\right| <\||\p_{\m''}\||^{-8}
    \ \ \ \mbox{when} \ \  \||\p_{\m''}\||>\xi _*^{ r_{n-3}(\xi _*)}.\label{Aug8-1}\end{equation}
    Indeed, for every $\m''$ we break ${\cal G}_{n}'$ into several parts, partition being dependent on $\m''$:
    $ {\cal G}_{n}'=\cup _{s,s'=0}^{n}{\cal G}_{ss'},$
\begin{equation}\label{Aug8-2}
\begin{split}&
    {\cal G}_{ss'}=\cr &
    \left\{\vec \xi \in {\cal G}_{n}':\ |\vec \xi|^{r_{s-1}(|\vec \xi | )}\leq \frac 12
    \||\p_{\m''}\||<\gamma _s |\vec \xi|^{r_{s}(|\vec \xi | )
    }\right\}\cap\cr &
    \left\{\vec \xi \in {\cal G}_{n}':\ |\vec \xi+\p_{\m''}|^{r_{s'-1}(|\vec \xi +\p_{\m''}|)}\leq \frac 12
    \||\p_{\m''}\||< \gamma _{s'} |\vec \xi+\p_{\m''}|^{r_{s'}(|\vec \xi +\p_{\m''}|)}\right\},
\end{split}
\end{equation}
     where $r_{-1}:=0$, $r_0=\delta $, $\gamma _s=1$ when $s<n$, $\gamma _n=\infty $. To prove \eqref{Aug8-1}, it is enough to show
     \begin{equation}
    \sup _{\vec \xi \in {\cal G}_{ss'}}\left| B_{\m''}(\vec \xi )\right| <\||\p_{\m''}\||^{-8}\label{Aug8-3}\end{equation}
    for all $s,s'$. Assume $s,s'=n$. It follows from \eqref{Aug8-2} that for any $\m \in \Z^4$ either
     $v_{\m }^{(n)}(\vec \xi )$ or  $v_{\m +\m''}^{(n)}(\vec \xi +\p_{\m''})$ is zero.   Hence, $\langle\v^{(n)}(-\vec \xi ), \v^{(n)}_*(-\vec \xi-\p_{\m''})\rangle =0$, i.e.,
    $B_{\m''}(\vec \xi )=0$. Next, let $0<s<n$, $s'=n$. By \eqref{perturbation*-3last},
    \begin{equation}\|\v^{(n)}(\vec \xi )-\v^{(s)}(\vec \xi )\|<|\vec \xi |^{-|\vec \xi |^{\frac 12r_{s-1}(|\vec \xi |)}}. \label{Aug8-4} \end{equation} It follows from the definition of
    ${\cal G}_{sn}$ that $\langle \v^{(s)}(\vec \xi ),\v^{(n)}_*(\vec \xi +\p_{\m''})\rangle =0$. Therefore,
    \begin{equation}
    \left| B_{\m''}(\vec \xi )\right| \leq \|\v^{(n)}(\vec \xi )-\v^{(s)}(\vec \xi )\|<|\vec \xi |^{-|\vec \xi |^{\frac 12 r_{s-1}(|\vec \xi |)}}\ \ \mbox{when  } \vec \xi \in {\cal G}_{sn}.
    \label{Aug8-5} \end{equation}
    Using \eqref{Aug13-1}, \eqref{r_2} and  \eqref{indrn}, we obtain
    $\left| B_{\m''}(\vec \xi )\right| \leq |\vec \xi |^{-10 r_{s}(|\vec \xi| )}$. Considering again the definition of
    ${\cal G}_{sn}$ , we get
    \eqref{Aug8-3}. Next, we consider ${\cal G}_{0n}$. By \eqref{perturbation*},
    $\v^{(1)}=\v^{(0)}+O(|\vec \xi |^{-2l+1+44\mu \delta })$, where $v^{(0)}_{\m}=\delta _{\m,{\bf 0}}$. By \eqref{Aug8-2},
  $\langle \v^{(0)}(\vec \xi ),\v^{(n)}(\vec \xi +\p_{\m''})\rangle =0$. Hence,  $\left| B_{\m''}(\vec \xi )\right| \leq C|\vec \xi |^{-2l+1+44\mu \delta }$. Using again the definition of ${\cal G}_{0n}$ and the inequality $2l-1-44\mu \delta >8\delta $, we obtain \eqref{Aug8-3}. The case $s'<n$ is considered in the analogous way. Thus, \eqref{Aug8-3} is proved. Summarizing \eqref{Aug8-3} over $\m''$, we obtain \eqref{Aug9-2}.

    Suppose $0<\||\p_{\m''}\||\leq \xi _*^{r_{n-3}(\xi _*)}$. Let us estimate $B_{\m''}(\vec \xi )$. Assume for definiteness that $|\vec \xi +\p_{\m''}|\leq |\vec \xi |$. The case of the opposite inequality is analogous up to the change of the notation $\vec \xi \to \vec \xi +\p_{-\m''}$, since $B_{\m''}(\vec \xi )=B_{-\m''}(\vec \xi +\p_{\m''})$.
    By \eqref{++},
\begin{equation}H^{(n)}(-\vec \xi )\v^{(n)}\bigl(-\vec \xi\bigr)=\lambda ^{(n)}\bigl(-\vec \xi\bigr)\v^{(n)}\bigl(-\vec \xi).
%%%%%%%%%%%\ \ \ \v^{(n)}\bigl(-\vec \xi\bigr)=\left\{ v^{(n)}_{\m}\bigl(-\vec \xi\bigr)\right\}_{\m\in \Omega (r_{n-1}, |\vec \xi |)}.
\label{Vienna-4+} \end{equation}
The analogous relation holds for $\v^{(n)}\bigl(-\vec \xi-
    \p_{\m ''}\bigr) $ up to the replacement of $H^{(n)}(-\vec \xi )$ by $ H^{(n)}\left(-\vec \xi - \p_{\m ''}\right)$ and $\lambda ^{(n)}\bigl(-\vec \xi\bigr)$ by $\lambda ^{(n)}\bigl(-\vec \xi-\p_{\m''}\bigr)$:
      \begin{equation}H^{(n)}(-\vec \xi -\p_{\m''})\v^{(n)}\bigl(-\vec \xi-\p_{\m''}\bigr)=\lambda ^{(n)}\bigl(-\vec \xi-\p_{\m''}\bigr)\v^{(n)}\bigl(-\vec \xi-\p_{\m''}).\label{Vienna-4"+} \end{equation}
 Note that $ H^{(n)}\left(-\vec \xi - \p_{\m ''}\right)$ up to the shift of indices by $-\m''$ is equivalent to
    the operator $P_{\m''}H(-\vec \xi )P_{\m''}$, where $P_{\m''}$ is the projection onto the box of the size
    $|\vec \xi +\p_{\m ''}|^{r_{n-1}(|\vec \xi +\p_{\m ''}|)}$ around $-\m''$. Using  the shifted eigenvector
    $\v_*^{(n)}\bigl(-\vec \xi-\p_{\m''}\bigr)$,
     we can rewrite \eqref{Vienna-4"+}  in the form:
    \begin{equation}P_{\m ''}H(-\vec \xi )P_{\m ''}\v_*^{(n)}\bigl(-\vec \xi-\p_{\m''}\bigr)=
    \lambda ^{(n)}\bigl(-\vec \xi-\p_{\m''}\bigr)\v_*^{(n)}\bigl(-\vec \xi-\p_{\m''}),\label{Vienna-4'+} \end{equation}
    where $P_{\m ''}\v_*^{(n)}=\v_*^{(n)}$.
    By \eqref{Aug13-2},
    \begin{equation}H(-\vec \xi )\v^{(n)}\bigl(-\vec \xi\bigr)=\lambda ^{(n)}\bigl(-\vec \xi\bigr)\v^{(n)}\bigl(-\vec \xi)+O\left(|\xi |^{-
|\xi |^{\frac{1}{2}r_{n-1}(|\vec \xi
|)}}\right).\label{Vienna-4BHM+}
\end{equation} Similarly,
\begin{equation}\begin{split}& H(-\vec \xi )\v_*^{(n)}\bigl(-\vec
\xi-\p_{\m''}\bigr)=\cr &
    \lambda ^{(n)}\bigl(-\vec \xi-\p_{\m''}\bigr)\v_*^{(n)}\bigl(-\vec \xi-\p_{\m''})+O\left(|\xi +\p_{\m''}|^{-
|\xi +\p_{\m''}|^{\frac{1}{2}r_{n-1}(|\vec \xi+\p_{\m''}
|)}}\right).\end{split}\label{Aug10-2+}
\end{equation} Assume first $|\vec \xi +\p_{\m''}|\leq \frac 12 |\vec \xi |$. Then
$|\lambda ^{(n)}\bigl(-\vec \xi\bigr)-\lambda ^{(n)}\bigl(-\vec \xi
-\p_{\m''}\bigr)|> \frac{1}{2}|\vec \xi |^{2l}$. Using
\eqref{Aug11-2}, \eqref{Vienna-4BHM+} and
     \eqref{Aug10-2+}, we obtain:
     \begin{equation} \label{Aug11-1}B_{\m ''}=O\left(\xi_*^{-
\xi _*^{\frac{1}{2}r_{n-1}(\xi _*) }}\right)\ \ \mbox{when  } |\vec
\xi +\p_{\m''}| \leq \frac 12 |\vec \xi |. \end{equation}
     Similar, but somewhat more subtle considerations are required when $|\vec \xi +\p_{\m''}|>\frac 12 |\vec \xi |$. We start with introducing a parameter $s$. We will use it to
      cut $\Omega (r_{n-1}, |\vec \xi |)$ to approximately the same size as $\Omega (r_{n-1}, |\vec \xi +\p_{\m''}|)$.
  If the boxes are of approximately  the same size, then $s=n-1$. Indeed, for each $\vec \xi $ one of the following relations holds:
\begin{equation}
  \label{Aug10-5}|\vec \xi |^{r_{s-1}(|\vec \xi |)}\leq
 |\vec \xi +\p_{\m''} |^{r_{n-1}(|\vec \xi +\p_{\m''} |)}
 <|\vec \xi |^{r_{s}(|\vec \xi |)},
  \end{equation}
  where $1\leq s\leq n-1$ and $s$ is defined by $\m ''$ and $\vec \xi $. Note that $s<n-1$ when $\Omega (r_{n-1}, |\vec \xi |)$ is essentially bigger than $\Omega (r_{n-1}, |\vec \xi +\p_{\m''}|)$.
  Using the second inequality in \eqref{Aug10-5} and \eqref{indrn}, we get
  \begin{equation}|\vec \xi +\p_{\m''} |^{r_{n-3}(|\vec \xi +\p_{\m''} |)}<\frac18 |\vec \xi |^{r_{s-1}(|\vec \xi |)}. \label{Aug11-3} \end{equation}
  Let  $P_{{\bf 0}s}$ be the projecting corresponding to $\Omega (r_s, |\vec \xi |)$.
 %%%%%%%%%%%   Using the definition of the operator $H^{(n)}$, we rewrite \eqref{Vienna-4} in the form:
 %%%%%%%%%%%     \begin{equation}P_{\bf 0}H(-\vec \xi )P_{\bf 0}\v^{(n)}\bigl(-\vec \xi\bigr)=\lambda ^{(n)}\bigl(-\vec \xi\bigr)\v^{(n)}\bigl(-\vec \xi),\label{Vienna-4*} \end{equation}
 %%%%%%%%%%%     where $P_{\bf 0}$ is the projection $P_{\m''}$ for $\m''=0$.
 By \eqref{perturbation*-3last} with $s$ instead of $n$, \begin{equation}
 (I-P_{{\bf 0}s})\v^{(n)}(-\vec \xi )=O\left(|\vec \xi | ^{-
|\vec \xi | ^{\frac 12 r_{s-1}({|\vec \xi |})}}\right).\label{Spb-1+} \end{equation}
    Let us prove the analogous estimate for $\v^{(n)}_*$:
   \begin{equation} \label{Spb-1}(I-P_{{\bf 0}s})\v^{(n)}_*(-\vec \xi-
    \p_{\m ''})=O\left(|\vec \xi | ^{-
|\vec \xi | ^{\frac 12 r_{s-1}({|\vec \xi
|})}}\right).\end{equation} Indeed, if $(P_{{\bf 0}s})_{\m \m}=0$,
then $\||\p_\m\||> |\vec \xi |^{r_{s-1}(|\vec \xi|)}$.  Using
\eqref{Aug11-3} and  the bound on $\||\p_{\m ''}\||$, we obtain
$\||\p_{\m+\m''}\||> \frac 12 |\vec \xi |^{r_{s-1}(|\vec \xi|)}$.
Using \eqref{perturbation*-3last}, \eqref{Feb6b-3last}, we obtain
\eqref{Spb-1}. From \eqref{Vienna-4BHM+},\eqref{Aug10-2+},
considering that $\|P_{{\bf 0}s}H\|= O\left( |\xi |^{2lr_{s-1}(|\vec
\xi|)}\right)$ and using \eqref{Spb-1+},\eqref{Spb-1}, we get:
\begin{equation}H^{(s)}(-\vec \xi )P_{{\bf 0}s}\v^{(n)}\bigl(-\vec \xi\bigr)=\lambda ^{(n)}\bigl(-\vec \xi\bigr)P_{{\bf 0}s}\v^{(n)}\bigl(-\vec \xi)+O\left(|\xi |^{-
|\xi |^{\frac{1}{4}r_{s-1}}}\right),\label{Vienna-4BHM} \end{equation}
%%%%%%%%%%%%It follows from \eqref{Vienna-4"}, \eqref{Vienna-4*} and \eqref{Spb-1} that
 %%%%%$H^{(n)}\left(-\vec \xi - \p_{\m ''}\right)$:
 %%%%%%%%   $$\left(H^{(n)}(-\vec \xi - \p_{\m ''})-H^{(n)}(\vec \xi )\right)\v^{(n)}(-\vec \xi)=O\left(\xi ^{-
%%%%%%%%%%\xi ^{\frac{1}{4}r_{n-1}}}\right).$$
%%%%the components of $\v^{(n)}(-\vec \xi-
%%%    \p_{\m ''})$ being defined by zero outside the $\xi ^{r_{n-1}}$-box around $\m''$.
%%%%where the components of $\v^{(n)}(-\vec \xi-
%%%%    \p_{\m ''})$ are also shifted: $\left(\v^{(n)}(-\vec \xi-
%%%%    \p_{\m ''})\right)_{\m}=c_{\m +\m''}^{(n)}\bigl(-\vec \xi-
%%%%    \p_{\m ''}\bigr)$.
     \begin{equation} \label{Vienna-5}H^{(s)}(-\vec \xi )P_{{\bf 0}s}\v^{(n)}_*\bigl(-\vec \xi-
    \p_{\m ''}\bigr)=\lambda ^{(n)}\bigl(-\vec \xi-
    \p_{\m ''}\bigr)P_{{\bf 0}s}\v^{(n)}_*\bigl(-\vec \xi-
    \p_{\m ''}\bigr)+O\left(|\vec \xi | ^{-
|\vec \xi | ^{\frac {1}{4}r_{s-1}}}
%%%%%%%+O\left(|\xi +\p_{\m''}|^{-
%%%%%%%%%%|\xi +\p_{\m''}|^{\frac{1}{4}r_{s-1}}}\right)
%%%%%O\left(\xi ^{-
%%%%%%%%%%\xi ^{\frac{1}{4}r_{n-1}}}
\right).
\end{equation}
Next, by Theorem \ref{Thm3last} for step $s$, $\left|\lambda
^{(n)}\bigl(-\vec \xi \bigr)-\lambda ^{(n)}\bigl(-\vec \xi-
    \p_{\m ''}\bigr)\right|>\varepsilon _0^{(s)}/2$, where $\varepsilon _0^{(s)}=|\vec \xi | ^{-2r_{s-1}'|\vec \xi |^{2\gamma r_{s-2}}}.$
    Indeed, $\v^{(n)}(-\vec \xi )$ and $\v^{(n)}_*\bigl(-\vec \xi-
    \p_{\m ''}\bigr)$ are almost orthogonal since they are concentrated around $\m={\bf 0}$ and $\m=\m''\not={\bf 0}$ respectively; thus
    $\lambda ^{(n)}\bigl(-\vec \xi-
    \p_{\m ''}\bigr)$ must be outside of the interval described in Theorem \ref{Thm3last}, while $\lambda ^{(n)}\bigl(-\vec \xi \bigr)$
    is inside twice shorter interval. Now, using \eqref{Vienna-4BHM} and \eqref{Vienna-5}, we obtain:
    $$\langle P_{{\bf 0}s}\v^{(n)}(-\vec \xi ), P_{{\bf 0}s}\v^{(n)}_*(-\vec \xi-\p_{\m''})\rangle= O\left(|\vec \xi | ^{-
|\vec \xi |^{\frac{1}{4}r_{s-1}}}\right)
%%%%%%%%%|\xi +\p_{\m''}|^{-
%%%%%%%%|\xi +\p_{\m''}|^{\frac{1}{4}r_{n-1}}}
|\vec \xi | ^{2r_{s-1}'|\vec \xi |^{2\gamma r_{s-2}}}=O\left(|\vec \xi |^{-
|\vec \xi |^{\frac{1}{8}r_{s-1} }}\right),$$
%%%%%%%%%%=O\left(\xi^{-
%%%%%%%%%%=\xi^{\frac{1}{8}r_{n-3} }}\right),$$
see \eqref{indrn}.  Using one more time \eqref{Spb-1+},
\eqref{Spb-1}, and considering \eqref{Aug11-3}, we obtain $B_{\m
''}=O\left( \xi _*^{-\xi _*^{\frac18 r_{n-3}(\xi _*) }}\right)$  for
the case  $|\vec \xi +\p_{\m''}|>\frac 12 |\vec \xi |$. Using this
estimate together with \eqref{Aug11-1} and considering that the
number of $\m''$ satisfying $0<\||\p_{\m''}\||\leq \xi
_*^{r_{n-3}(\xi _*)}$ does not exceed $16\xi _*^{4r_{n-3}(\xi _*)}$,
we obtain \eqref{Aug9-2*}. Substituting the estimates for $B_{\m
''}$ into \eqref{May14-12}, we obtain  \eqref{May14a-12}.

%%%%%%%%%If ${\cal G}_n'$ is not bounded, then the formula \eqref{May14a-12} holds for all $f,g$  with finite support, $o(1)$ being independent of $f,g$. By continuity the formula can be extended to all $f,g$ in $L^2({\cal %%%%%%%%%G}_n')$.\\

%%%%%%%%\item Now \eqref{Vienna-1} easily follows from \eqref{May14-12} and above estimates for functions $B_{\m''}(\vec \xi )$.
%%%%%%%%%%If $|\vec \xi+
%%%%%%%%%%    \p_{\m''}|\leq \frac 14 |\vec \xi|$, then $\left|\lambda ^{(n)}\bigl(-\vec \xi \bigr)-\lambda ^{(n)}\bigl(-\vec \xi-
%%%%%%%%%%    \p_{\m ''}\bigr)\right|> \frac 14 |\vec \xi|^2$. Therefore, $B_{\m ''}=O\left(|\xi +\p_{\m''}|^{-
%%%%%%%%%%|\xi +\p_{\m''}|^{\frac{1}{4}r_{n-1}}}\right)=O\left(\xi _*^{-
%%%%%%%%%%\xi _*^{\frac{1}{4}r_{n-1}}}\right)$. It follows that
%%%%%%%%%%\begin{equation} \label{Aug10-3}\sum _{\||\m''\||\leq \xi_*^{\xi _*^{\frac 12 r_{n-1}(\xi _*)}}}\ \ \sup _{\vec \xi \in {\cal G}_{n}'}\left| B_{\m''}(\vec \xi )\right| <C\xi_*^{-\xi _*^{\frac 12 r_{n-1}(\xi _*)}}. \end{equation}
%%%%%%%%%%Adding \eqref{Aug-2} and \eqref{aug10-3}, we obtain the analogous estimate for the whole sum. Now, \eqref{May14a-12} easily follows.
\end{proof}

It is easy to see that $T_n\subset S_n^*$. Therefore, $\|T_n\|\leq 1+o(1)$ and can be extended to the whole space $L_2({\cal G}_n)$. We still denote the extended operator by $T_n$, $T_n=S_n^*$. Therefore, $E_n$ is a self-adjoint operator.
\begin{lemma} \label{L:10.4}  Let $\mathcal{G}_n'\subset \mathcal{G}_n''\subset \mathcal{G}_n$. The following relation holds as $n\to \infty $:
\begin{equation} E_n(\mathcal{G}_n')E_n(\mathcal{G}_n'')=E_n(\mathcal{G}_n')+o(1), \label{Vienna-6} \end{equation}
where $o(1)$ is taken in the space of bounded operators and uniform in $\mathcal{G}_n'$, $\mathcal{G}_n''$. \end{lemma}
\begin{corollary} \label{C:10.4-1} $E_n(\mathcal{G}_n'')E_n(\mathcal{G}_n')=E_n(\mathcal{G}_n')+o(1).$\end{corollary}
This corollary is valid, since $E_n$ is selfajoint.
\begin{corollary} \label{C:10.4-2}  $E_n^2(\mathcal{G}_n')=E_n(\mathcal{G}_n')+o(1)$ for any $\mathcal{G}_n'\subset \mathcal{G}_n$. \end{corollary}

\begin{proof}  Let $I_n(\mathcal{G}_n')$ be the projection from $L_2(\mathcal{G}_n'')$ to $L_2(\mathcal{G}_n')$. It is easy to see that $T_n(\mathcal{G}_n')=I_n(\mathcal{G}_n')T_n(\mathcal{G}_n'')$. Hence, $T_n(\mathcal{G}_n')S_n(\mathcal{G}_n'')=I_n(\mathcal{G}_n')T_n(\mathcal{G}_n'')S_n(\mathcal{G}_n'')$. By \eqref{May14a-12} for set $\mathcal{G}_n''$,
$T_n(\mathcal{G}_n'')S_n(\mathcal{G}_n'')=id(\mathcal{G}_n'')+o(1)$, where $id(\mathcal{G}_n'')$ is the identity in $L_2(\mathcal{G}_n'')$. It immediately follows $T_n(\mathcal{G}_n')S_n(\mathcal{G}_n'')=I_n(\mathcal{G}_n')+o(1)$. Substituting the last relation into the formula $ E_n(\mathcal{G}_n')E_n(\mathcal{G}_n'')= S_n(\mathcal{G}_n')T_n(\mathcal{G}_n')
 S_n(\mathcal{G}_n'')T_n(\mathcal{G}_n'')$, we obtain \eqref{Vienna-6}.
\end{proof}

Let \begin{equation}
\mathcal{G}_{n, \lambda}=\{ \k \in {\mathcal{G}}_n:
\lambda ^{(n)}(\k) < \lambda\}. \label{d} \end{equation}
 This set is Lebesgue measurable, since ${\mathcal{G}}_n $ is
open and $\lambda ^{(n)}(\k)$ is continuous on $
{\mathcal{G}}_n$.

\begin{lemma}\label{L:abs.6}
$\left|{\mathcal{G}}_{n,\lambda+\varepsilon} \setminus
{\mathcal{G}}_{n,\lambda}\right| \leq 2\pi \lambda ^{-1+\frac 1l}\varepsilon $ when $0\leq
\varepsilon \leq 1$. \end{lemma}
The proof is based on Lemma \ref{ldk-3IVlast} and completely analogous to that of Lemma 2 in \cite{KL2}.

 By (\ref{s}),
$E_n\left({\mathcal{G}}_{n,\lambda+\varepsilon}\right)-E_n\left({\mathcal{G}}_{n,\lambda}\right)
=E_n\left({\mathcal{G}}_{n,\lambda+\varepsilon}\setminus
    {\mathcal{G}}_{n,\lambda}\right)$. Let us obtain an estimate for
    this projection.
\begin{lemma}\label{L:abs.7} For any $F \in C_0^{\infty}(\R^2)$ and
$0\leq \varepsilon \leq 1$, \begin{equation}
\left\|\bigl(E_n({\mathcal{G}}_{n,\lambda+\varepsilon})-E_n({\mathcal{G}}_{n,\lambda})\bigr)F\right\|^2_{L_2(\R^2)}
 \leq C( F)\lambda ^{-1+\frac 1l} \epsilon , \label{tootoo1}
 \end{equation}
 where $C(F)$ is uniform with respect to $n$ and $\lambda$.
\end{lemma}
\begin{proof} Let $\mathcal{G}_n'=\mathcal{G}_{n,\lambda+\varepsilon}\setminus \mathcal{G}_{n,\lambda}$. Using the definition \eqref{ST} of $E_n$ and formula \eqref{Vienna-1}  with $f=g=T_nF$, we obtain
\begin{equation}
\|E_n(\mathcal{G}_n')F\|^2_{L_2(\R^2)}<(1+o(1))\|T_nF\|^2_{L_{\infty}(\mathcal{G}_n')}|\mathcal{G}_n'|. \label{Vienna-7} \end{equation}
Using \eqref{eq2} and Corollary \ref{C:Psi}  we easily get $ \|T_nF\|_{L_{\infty}(\mathcal{G}_n')}<2\|F\|_{L_1(\R^2)}$. Substituting this estimate into \eqref{Vienna-7} and using Lemma \ref{L:abs.6}, we obtain \eqref{tootoo1}.

\end{proof}

\subsection{Sets ${\mathcal{G}}_{\infty}$ and
${\mathcal{G}}_{\infty ,\lambda }$} \label{S:8.1}
%%%%%The sets
%%%%%$\mathcal{G}_{\infty}$, $\mathcal{G}_{n}$ are given by
%%%%%\eqref{Ginfty}, \eqref{Gn}.
%%%%%%%By definition,
%%%%%%%$$\mathcal{G}_{\infty}=\bigcup_{\lambda > \lambda_*}\mathcal D _{\infty}(\lambda) \text{ and }
%%%%%%%\mathcal{G}_{n}=\bigcup _{\lambda >\lambda_*} \mathcal D_n(\lambda).$$
%%%%% As it was shown in the proof of Theorem \ref{Thm:6.10},
By construction,
$
    \mathcal{G}_{n+1}\subset \mathcal{G}_n,$
    %%%%%%% \label{ac2}
   %%%%%%%%% \end{equation}
%%%%%%%%\begin{equation}\label{tt}
    $\mathcal{G}_{\infty}=\bigcap_{n=1}^{\infty}{\mathcal{G}}_n. $
%%%%%%%%%%%    \end{equation}
Therefore, the perturbation formulas for $\lambda ^{(n)}(\k )$ and $\Psi _n(\k )$ hold in
$\mathcal{G}_{\infty}$ for all $n$.
%%%%%Moreover, coordinates
%%%%%$(\lambda_n, \varphi )$ can be used in $\mathcal{G}_{\infty}$  for
%%%%%every $n$.
 Let
     \begin{equation}
     \mathcal{G}_{\infty, \lambda }=\left\{\k \in
    \mathcal{G}_{\infty }: \lambda _{\infty }(\k )<\lambda
    \right\}. \label{dd}
    \end{equation}
The function $\lambda _{\infty }(\k )$ is a Lebesgue
measurable function, since it is a limit of the sequence of
measurable functions. Hence, the set  $\mathcal{G}_{\infty, \lambda
}$ is measurable.

\begin{lemma}\label{add6*} The measure of the symmetric difference of
two sets $\mathcal{G}_{\infty, \lambda }$ and $\mathcal{G}_{n,
\lambda}$ converges
 to zero as $n \to
\infty$ uniformly in $\lambda$ in every bounded interval:
    $$\lim _{n\to \infty }\left|\mathcal{G}_{\infty, \lambda }\Delta \mathcal{G}_{n, \lambda
    }\right|=0.$$
 \end{lemma}
 The proof is completely analogous to the proof of Lemma 4 in \cite{KL2}.

     \subsection{Projections $E(\mathcal{G}_{\infty , \lambda })$}

In this section, we show that the operators
$E_n(\mathcal{G}_{\infty , \lambda })$ have a strong limit
$E_{\infty }(\mathcal{G}_{\infty , \lambda })$ in $L_2(\R^2)$ as $n$
tends to infinity. The operator $E_{\infty }(\mathcal{G}_{\infty ,
\lambda })$ is a spectral projection of $H$. It can be represented
in the form $E_{\infty }(\mathcal{G}_{\infty , \lambda })=S_\infty
T_{\infty }$, where $S_{\infty }$ and $T_{\infty }$ are strong
limits of $S_n(\mathcal{G}_{\infty , \lambda })$ and
$T_n(\mathcal{G}_{\infty , \lambda })$, respectively.  For any $F\in
C_0^{\infty }(\R^2)$, we show: \begin{equation} E_{\infty }\left(
\mathcal{G}_{\infty , \lambda }\right)F=\frac{1}{4\pi ^2}\int
    _{ \mathcal{G}_{\infty , \lambda }}\bigl( F,\Psi _{\infty }(\k )\bigr) \Psi _{\infty }(\k,\x)\,d\k ,\label{s1u}
    \end{equation}
    \begin{equation} HE_{\infty }\left(
\mathcal{G}_{\infty , \lambda }\right)F=\frac{1}{4\pi ^2}\int
    _{ \mathcal{G}_{\infty , \lambda }}\lambda _{\infty }(\k )\bigl( F,\Psi _{\infty }(\k )\bigr) \Psi _{\infty }(\k,\x)\,d\k .\label{s1uu}
    \end{equation}
   %%%%%%% $\bigl( F,\Psi _{\infty }(\vec
 %%%%%%%   \varkappa )\bigr) \Psi _{\infty }$ being an integral analogous
 %%%%%%%%   to the dot product in $L_2(R^2)$:
 %%%%%%%%    \begin{equation}\bigl( F,\Psi _{\infty}(\vec
%%%%%%%%     \varkappa )\bigr)=\int _{R^2}F(x)\overline{\Psi _{\infty }(\vec
%%%%%%%%     \varkappa ,x)}dx. \label{eq1}
%%%%%%%%     \end{equation}
Using properties of $E_{\infty }\left( \mathcal{G}_{\infty , \lambda
}\right)$, we  prove absolute continuity of the branch of the
spectrum corresponding to functions $\Psi _{\infty }(\k )$.

We consider the sequence of operators $S_n(\mathcal{G}_{\infty ,
\lambda })$ which are given by (\ref{ev}) with $\mathcal
G_n'=\mathcal{G}_{\infty , \lambda }$.
\begin{lemma} We have
\begin{equation} \label{Vienna-2}\left\|(S_n(\mathcal{G}_{\infty , \lambda})-S_{n-1}(\mathcal{G}_{\infty , \lambda}))
f\right\|_{L_2(\R^2)}< C\|f\|_{L_2(\mathcal{G}_{\infty , \lambda
})}\xi_*^{-\frac14\xi _*^{r_{n-2}(\xi_*)}}.
\end{equation}
\end{lemma}
\begin{proof} Considering as in the proof of Lemma \ref{9.1}, we obtain $$
    \left\|(S_n-S_{n-1})f\right\|_{L_2(\R^2)}^2=\left\|\widehat{S_nf}-\widehat{S_{n-1}f}\right\|_{L_2(\R^2)}^2=$$
  \begin{equation}     \int _{\R^2 }\sum _{\m''}\tilde B_{\m''}(\vec \xi )f \bigl(-\vec \xi\bigr)\bar f \bigl(-\vec \xi-
    \p_{\m''}\bigr)\chi \bigl(\mathcal{G}_{\infty , \lambda },-\vec \xi \bigr)\chi \bigl(\mathcal{G}_{\infty , \lambda },-\vec \xi-
    \p_{\m''}\bigr)\ d\vec \xi , \label{Vienna-3} \end{equation}
    \begin{align}\nonumber &\tilde B_{\m''}(\vec \xi )=\cr & \sum _{\m \in \Omega (r_{n-1},|\vec \xi |):\,\m+\m''\in \Omega (r_{n-1},|\vec \xi +\p_{\m''}|)}\left(v_{\m}^{(n)}-v_{\m}^{(n-1)}\right)\bigl(-\vec \xi\bigr)\overline{\left(v_{\m +\m''}^{(n)}-v_{\m +\m''}^{(n-1)}\right)}\bigl(-\vec \xi-
    \p_{\m ''}\bigr).\end{align}
    Assume for definiteness that $|\vec \xi +\p_{\m''}|\leq |\vec \xi |$.
    If $\||\p_{\m''}\||>2|\vec \xi |^{r_{n-1}(|\vec \xi |) }$, then $\tilde B_{\m''}(\vec \xi )=0$.
    Let $\vec \xi :\||\p_{\m''}\||\leq 2|\vec \xi |^{r_{n-1}(|\vec \xi |) }$.
   Using \eqref{perturbation*-3last} with $n$ instead of $n+1$, we easily obtain:
    \begin{equation}\left|\tilde B_{\m''}(\vec \xi )\right|=O\left(|\vec \xi |^{-
|\vec \xi |^{\frac{1}{2}r_{n-2}(|\vec \xi |) }}\right) O\left(|\vec
\xi +\p_{\m''}|^{- |\vec \xi +\p_{\m''}|^{\frac{1}{2}r_{n-2}(|\vec
\xi +\p_{\m''}|) }}\right).\label{Aug13-3}
\end{equation} Considering  \eqref{indrn} with $n-1$ instead of and
$n$ and taking into account that $\||\p_{\m''}\||\leq 2|\vec \xi
|^{r_{n-1}(|\vec \xi |) }$, we easily get:
$$\left|\tilde B_{\m''}(\vec \xi )\right|=\||\p_{\m''}|\|^{-8} O\left(\xi _*^{-
\xi _*^{\frac{1}{2}r_{n-2}(\xi_*) }}\right).$$
 Summarizing the last estimate for $\m ''\neq {\bf 0}$ and using \eqref{Aug13-3} for $\m''={\bf 0}$,
 we arrive at \eqref{Vienna-2}.
 \end{proof}

 By \eqref{Vienna-2}, the sequence of operators $S_n(\mathcal{G}_{\infty ,
\lambda })$ is a Cauchy sequence in the space of bounded operators.
We denote its limit by $S_{\infty }(\mathcal{G}_{\infty , \lambda
    })$. Note that the convergence of $S_n(\mathcal{G}_{\infty , \lambda
    })$ to $S_{\infty }(\mathcal{G}_{\infty , \lambda
    }) $  is uniform in $\lambda $ when $\lambda >\lambda _*$.
%%%    $$S_n(\mathcal{G}_{\infty , \lambda
%%%     }):\ L_2(\mathcal{G}_{\infty , \lambda
%%%     })\to L_2(\R^2).$$ We prove that the sequence has a strong limit and
%%%     describe its properties.
\begin{lemma} \label{Lem2}  The operator $S_{\infty }(\mathcal{G}_{\infty , \lambda
    })$
    %%%%satisfies $\|S_{\infty }\|=1$ and
    can be described by the
    formula
    \begin{equation}
    (S_{\infty }f) (\x)= \frac{1}{2\pi }\int _{\mathcal{G}_{\infty , \lambda
    }}f (\k)\Psi _{\infty }(\k ,\x) d\k  \label{ev1}
    \end{equation}
for any $f \in L_{\infty }\left( \mathcal{G}_{\infty , \lambda
}\right)$.
\end{lemma}
\begin{proof} From Theorem \ref{T:Dec10} it follows that for every $f \in L_{\infty }\left( \mathcal{G}_{\infty , \lambda
}\right)$
\begin{equation}
     \lim _{n\to \infty }\int _{\mathcal{G}_{\infty , \lambda
    }}f (\k)\Psi _{n}(\k ,\x) d\k =\int _{\mathcal{G}_{\infty , \lambda
    }}f (\k)\Psi _{\infty }(\k ,\x) d\k
 \label{Vienna-9}
    \end{equation}
    for all $\x$. Hence, \eqref{ev1} holds. \end{proof}

Now we consider the sequence of operators $T_n(\mathcal{G}_{\infty ,
\lambda
    })$ which are given by (\ref{eq2}) and act from $L_2(\R^2)$ to $L_2(\mathcal{G}_{\infty , \lambda
    })$. Since, $T_n=S_n^*$, the sequence has a limit  $T_{\infty }$ in the class of bounded operators,  $T_{\infty }=S^*_{\infty }$.
     Note that the convergence of $T_n(\mathcal{G}_{\infty , \lambda
    })$ to $T_{\infty }(\mathcal{G}_{\infty , \lambda
    }) $  is uniform in $\lambda $ when $\lambda >\lambda _*$.

\begin{lemma} \label{Lem1} The operator $T_{\infty }(\mathcal{G}_{\infty , \lambda
    })$  can be described by the
    formula $T_{\infty }(\mathcal{G}_{\infty , \lambda
    })F=\frac{1}{2\pi }\bigl( F,\Psi _{\infty }(\k )\bigr) $
for any $F\in C_0^{\infty }(\R^2)$.
\end{lemma}
\begin{proof} The lemma easily follows from  Theorem \ref{T:Dec10} and formula \eqref{eq2}. \end{proof}

\begin{lemma}\label{May8} Operators
$E_n(\mathcal{G}_{\infty , \lambda })$ have a limit
$E_{\infty }(\mathcal{G}_{\infty , \lambda })$ in the class of bounded operators in $L_2(\R^2)$, the
convergence being uniform for $\lambda >\lambda _*$. The
operator $E_{\infty }(\mathcal{G}_{\infty , \lambda })$ is a
projection. For any $F\in C_0^{\infty }(\R^2)$ it is given by
(\ref{s1u}).
%%%%% and formula (\ref{s1uu}) holds.
\end{lemma}
\begin{proof} The lemma immediately follows from convergence of sequences $S_n$, $T_n$ and Lemmas \ref{L:10.4},  \ref{Lem2}, \ref{Lem1}. \end{proof}

\begin{lemma}\label{onemore} There is a strong limit
$E_\infty(\mathcal{G}_{\infty})$ of the projections $E_\infty
(\mathcal{G}_{\infty,\lambda })$ as $\lambda $ goes to infinity.
\end{lemma} \begin{corollary}\label{onemore1} The operator $E_{\infty
}(\mathcal{G}_{\infty})$ is a projection. \end{corollary}
\begin{proof} It can be  easily seen from \eqref{s1u} that the sequence of $E_{\infty}
(\mathcal{G}_{\infty,\lambda })$ is monotonuos in $\lambda $. It is well known that a monotone sequence of projectors has a limit.\end{proof}

The proofs of the next two lemmas are completely analogous to the proofs of Lemmas 10, 11 in \cite{KL2}.

\begin{lemma}\label{L:abs.9} Projections
$E_\infty(\mathcal{G}_{\infty,\lambda })$, $\lambda \in \R$, and
$E_\infty(\mathcal{G}_{\infty})$ reduce the operator $H$.
\end{lemma}

\begin{lemma} The family of projections
$E_\infty(\mathcal{G}_{{\infty},\lambda} )$ is the resolution of the
identity of the operator $HE_\infty(\mathcal{G}_{\infty})$ acting in
$E_\infty(\mathcal{G}_{\infty})L_2(\R^2)$. \end{lemma}

\begin{lemma} Formula \eqref{s1uu} holds, when $F\in C_0^{\infty}(\R^2)$. \end{lemma}
\begin{proof} By the previous lemma, $E_{\infty }\left(
\mathcal{G}_{\infty , \lambda }\right)F \in D(H)$. It is easy to see that the r.h.s. of \eqref{s1u} can be differentiated with respect to $\x$ under the integral sign. Now, considering
\eqref{6.7}, we get \eqref{s1uu}.

\end{proof}

\subsection{Proof of Absolute Continuity}

Now we show that the  branch of spectrum (semi-axis) corresponding
to $\mathcal G_{\infty }$  is absolutely continuous.

\begin{thm}\label{T:abs} For any $F\in C_0^{\infty }(\R^2)$ and
$0\leq \varepsilon \leq 1$,
    \begin{equation}
    \left| \left(E_\infty(\mathcal{G}_{\infty,\lambda+\varepsilon})F,F\right)-
    \left(E_\infty(\mathcal{G}_{\infty,\lambda})F,F
    \right) \right| \leq C_F\lambda ^{-1+\frac 1l}\varepsilon .\label{May10*}
    \end{equation}

\end{thm} \begin{corollary} The spectrum of the operator
$HE_\infty(\mathcal G_{\infty})$ is absolutely continuous.
\end{corollary}
\begin{proof} By formula (\ref{s1u}),
    $$ | \left(E_\infty(\mathcal{G}_{\infty,\lambda+\varepsilon})F,F\right)-\left(E(\mathcal{G}_{\infty,\lambda})F,F
    \right) | \leq C_F\left| \mathcal{G}_{\infty , \lambda +\varepsilon
    }\setminus \mathcal{G}_{\infty , \lambda } \right| .$$
Applying Lemmas \ref{L:abs.6} and \ref{add6*}, we immediately get
(\ref{May10*}).

\end{proof}

\section{Appendices}

\subsection{Appendix 1. Proof of Lemma \ref{L:3.1}}
\begin{proof}
\begin{enumerate}
\item The case $p_\m>4k$. From \eqref{Jan23a} it immediately follows
that $|\Im \varphi _{\m}^{\pm }|>(\cosh )^{-1}2>1$. Hence, $\W_0\cap
\OO _\m (k, \tau )=\emptyset $.

Further we use the Taylor series for $|\vec k(\varphi
)+\p_\m|_\R^2-k^{2}$ near its zeros: Noting that
  \begin{equation}
  |\vec k(\varphi)+\p_\m|_\R^2-k^2 = 2kp_\m\cos (\varphi-\varphi_\m)+p_\m^2
  \label{08.20}
  \end{equation}
 and recalling that $\varphi_\m^\pm$ are the solutions of
 $|\vec k(\varphi)+\p_\m|_\R^2 = k^2 $, we see:
  \begin{equation}\label{E:sin-cos}
  \cos (\varphi_\m^\pm-\varphi_\m) = -\frac{p_\m}{2k},\quad
   |\sin (\varphi_\m^\pm-\varphi_\m)| = \sqrt{\left|1-\frac{p_\m^2}{4k^2}\right|}.
  \end{equation}
 Expanding (\ref{08.20}) around $\varphi_\m^\pm$, we get:
  \begin{equation}\label{E:first-factor}
  \begin{split}
  &|\vec k(\varphi)+\p_\m|_\R^2-k^2
  =\cr & -2k p_\m \sin(\varphi_\m^\pm-\varphi_\m)r_\m\left(1+O(r_\m^2)\right)
   + k p_\m
   \cos(\varphi_\m^\pm-\varphi_\m)r_\m^2\left(1+O(r_\m^2)\right),
  \end{split}
  \end{equation}
where $r_\m=|\varphi -\varphi _\m^{\pm}|.$

 \item In the second case we put $r_\m = \frac{\tau k^{-1-40\mu \delta}}{p_\m  \sqrt{\left|1-\frac{p_\m^2}{4k^2}\right|}}(1+o(1))$
 when $k^{-1-39\mu \delta } < p_\m < 4k  \text{ and }
 \left|1-\frac{p_\m^2}{4k^2}\right|>\tau k^{-2-40\mu \delta}$. Substituting $r_\m$ into \eqref{E:first-factor}, we get that
 the modulus of the first term is $2\tau k^{-40\mu \delta }(1+o(1))$ and that of the second term is
 $\frac{\tau ^2k^{-2-80\delta }}{2\left|1-\frac{p_\m^2}{4k^2}\right|}(1+o(1))$.
 Using the condition $\left|1-\frac{p_\m^2}{4k^2}\right|>\tau k^{-2-40\mu \delta}$, one can
  easily see that the former is at least twice greater than the latter. Thus,  we get
  \begin{equation}
  \left||\vec k(\varphi)+\p_\m|_\R^2-k^2\right|>\tau k^{-40\mu \delta }\ \  \mbox{when  }|\varphi -\varphi _\m^{\pm }|=r_\m.
  \label{Lan23b} \end{equation}
  Now, the maximum principle yield that this inequality holds everywhere
  outside the discs $\cup_{\pm,j\in\Z}(\Phi^{\pm}_{\m}+2\pi
j)$. Hence, $\OO_{\m}\subset \cup_{\pm,j\in\Z}(\Phi^{\pm}_{\m}+2\pi
j)$.

\item In the third case we put $r_\m = 32\tau k^{-1-20\mu \delta }(1+o(1))$ and $\left|1-\frac{p_\m^2}{4k^2}\right| <
 \tau k^{-2-40\mu \delta }$. This time the modulus of the second term in \eqref{E:first-factor} is $64\cdot32\tau ^2 k^{-40\mu \delta}\left(1+o(1)\right)$ and
 that of the first is smaller than $128\tau ^{3/2} k^{-40\mu \delta}(1+o(1))$. Therefore we again have
 \eqref{Lan23b} and $\OO_{\m}\subset \cup_{\pm,j\in\Z}(\Phi^{\pm}_{\m}+2\pi
j)$.\end{enumerate}\end{proof}

\subsection{Appendix 2. Proof of Lemma \ref{L:4}}
\begin{proof} The proof is by contradiction. By definition of $\MM_2$, $J\geq 1$.
Suppose  $J\geq 4$. We  arrive to contradiction in several steps.
\begin{enumerate}
\item  By
definition of $\MM'(\varphi _0)$, $\varphi _0\in \cap _{j=0}^J \OO_{\m_j}(k,1)$.
This means:
$$\left||\vec k+\p_{\m _j}|^2-k^2\right|\leq k^{-40\mu\delta},\ j=0,...J$$It
follows:
\begin{equation}\left|2\left(\vec k(\varphi
)+\p_{\m _0},\p_{\q_j}\right)+p_{\q_j}^2\right|<c k^{-40\mu
\delta},\ \q_j=\m_j-\m_{0},\ \ \ j=1,...J,\label{I5}\end{equation}
where $\||\p_{\q_j}\||\leq jk^{\delta}$. In this part of the proof
we show that no two vectors $\p_{\q_j}$, $j=1,2,3,4$ are colinear.
Indeed, suppose $\p_{\q_1}$ and $\p_{\q_2}$ are colinear. Let ${\vec \nu }$ be a unit vector in the direction of $\p_{\q_1}$. Then the
directional vector of $\p_{\q_2}$ is $\pm \vec \nu $, where $\pm $
means $+$ or $-$. Inequality \eqref{I5} together with the estimate \eqref{below} for
$p_{\q_j}$ yields:
$$2(\vec k(\varphi
)+\p_{\m _0}, \vec \nu )+p_{\q_1}=O( k^{-39\mu \delta}),\ \ \pm
2(\vec k(\varphi )+\p_{\m _0}, \vec \nu )+p_{\q_2}=O( k^{-39\mu
\delta}).$$ Hence, \begin{equation} p_{\q_1}\mp p_{\q_2}=O(
k^{-39\mu \delta}).\label{I8} \end{equation} Note that
$\||\p_{\q_1}\mp \p_{\q_2} \||<8k^{\delta }$. Therefore, $p_{\q_1\mp
\q_2} >ck^{-\mu \delta}.$ Considering that  $p_{\q_1\mp
\q_2}=|p_{\q_1}\mp p_{\q_2}|$,
  we arrive to contradiction with
\eqref{I8}. Hence, no two vectors $\p_{\q_j}$, $j=1,2,3,4$ are
colinear.

\item We represent every $\p_{\q_j}$ in the form:
$\p_{\q_j}/2\pi=\s_j+\alpha\s_j'$, $\s_j,\s_j'\in \Z^2,
|\s_j|,|\s_j'|<8k^{\delta}$. Let us show that
\begin{equation}[\s_1,\s_1']\neq 0 \mbox{ when } J>1,\label{NE1}
\end{equation}  $[\a,\b ]=a_1b_2-a_2b_1.$  Indeed,
suppose that $[\s_1,\s_1']=0$. Then, $\s_1,\s_1'$ are colinear
integer vectors. Hence, there are integers $n,m$ and an integer
vector $\s_1''$, such that $\s_1=n\s_1''$ and $\s_1'=m\s_1''$,
$0<|n|+|m|<8k^{\delta }$. Therefore, $\p_{\q_1}/2\pi=(n+\alpha m)
\s_1''$. Hence, $[\p_{\q_1}, \p_{\q_2}](2\pi)^{-2}=(n +\alpha
m)\left([\s_1'',\s_2]+\alpha [\s_1'',\s_2']\right).$ It follows from
\eqref{I5} that the angle (modulo $\pi $) between $\p_{\q_1},\p_{\q_2}$ is less than
$(p_{\q_1}+p_{\q_2})k^{-1}$. Hence,
 \begin{equation} \label{I10} [\p_{\q_1},\p_{\q_2}]=O(k^{-1+3\delta }).\end{equation} Therefore,
\begin{equation} \label{I9}n +\alpha m=O(k^{-1/2+3\delta /2})\ \ \mbox{ or }\ \
[\s_1'',\s_2]+\alpha [\s_1'',\s_2']=O(k^{-1/2+3\delta /2
}).\end{equation} The first relation is impossible since (with the
same proof as for \eqref{below}) $|n +\alpha m|>8C_\varepsilon
k^{-(\mu-1 +\epsilon)\delta }$ for any $\epsilon >0$ and
$(2\mu+3)\delta<1$. Similarly, the second relation \eqref{I9} is
possible if and only if $[\s_1'',\s_2']= [\s_1'',\s_2]=0$. Therefore
$\p_{\q_2}$ is colinear to $\s_1''$, i.e. to $\p_{\q_1}$. This means
$\p_{\q_1},\p_{\q_2}$ are colinear. It cannot be the case, as we
proved before. Thus, we have arrived to \eqref{NE1}.

\item Let us consider \eqref{I10}. Substituting $\p_{\q_j}=\s_j
+\alpha \s_j'$, $j=1,2$, we obtain:
\begin{equation}
n_1+\alpha p_1+\alpha^2m_1 =O(k^{-1+3\delta }), \label{I6a}
\end{equation}
where $n_1,p_1,m_1$ are integers, $n_1=[\s_1,\s_2]$,
$p_1=[\s_1,\s_2']+[\s_1',\s_2]$, $m_1=[\s_1',\s_2']$. Obviously,
$n_1,p_1,m_1=O(k^{2\delta })$. Note that $n_1^2+p_1^2+m_1^2\neq 0$,
since otherwise vectors $\p_{\q_1}, \p_{\q_2}$ are
 colinear. Next, $m_1\neq 0$, since otherwise $O(k^{-1+3\delta })>c(\epsilon )k^{-2(\mu-1 +\epsilon)\delta }$
for any $\epsilon >0$. This cannot be true for our choice of $\delta $.

Suppose there is
another triple $(n_2,p_2,m_2)$, $0<n_2^2+p_2^2+m_2^2<10^3k^{4\delta }$,
such that \eqref{I6a} holds. Namely,
\begin{equation}
n_2+\alpha p_2+\alpha^2m_2 =O(k^{-1+3\delta }). \label{I6'}
\end{equation}
The goal of this part is to show that such $(n_2,p_2,m_2)$ is, in
fact, a multiple of $(n_1,p_1,m_1)$. Indeed, $m_2\neq 0$ for the
same reason as $m_1\neq 0$. Excluding $\alpha^2$ from \eqref{I6} and
\eqref{I6'}, we get
$$(n_1m_2-n_2m_1)+\alpha(p_1m_2-p_2m_1)=O(k^{-1+5\delta }).$$
It follows that $p_1m_2-p_2m_1=0$, $n_1m_2-n_2m_1=0$, since
otherwise $O(k^{-1+5\delta })>c(\epsilon )k^{-4(\mu-1
+\epsilon)\delta }$ for any $\epsilon >0$. Thus, $(n_2,p_2,m_2)$  is
colinear to $(n_1,p_1,m_1)$ and
$(n_2,p_2,m_2)=\frac{r}{s}(n_1,p_1,m_1)$, where $r,s$ are integers,
$s\neq 0$ and $r,s=O(k^{2\delta })$.

\item  In this part we show that \begin{equation}
\p_{\q_3}=\frac{r}{s}\p_{\q_2}+\frac{t}{s}\p_{\q_1}, \label{E2}
\end{equation} where $s,r,t$ are nonzero integers,
$s,r,t=O(k^{3\delta})$. Indeed, let us  consider the relation
$[\p_{\q_1}, \p_{\q_3}]=O(k^{-1+3\delta })$. This relation follows
from \eqref{I5} the same way as \eqref{I10}. Substituting
$\p_{\q_j}=\s_j +\alpha \s_j'$, $j=1,3$, we obtain:
\begin{equation}
n_2+\alpha p_2+\alpha^2m_2 =O(k^{-1+3\delta }), \label{I6}
\end{equation}
where $n_2,p_2,m_2$ are integers, $n_2=[\s_1,\s_3]$,
$p_2=[\s_1,\s_3']+[\s_1',\s_3]$, $m_2=[\s_1',\s_3']$. Note that
$n_2^2+p_2^2+m_2^2\neq 0$, since otherwise vectors $\p_{\q_1},
\p_{\q_3}$ are
 colinear. Therefore, by part 3, $(n_2,p_2,m_2)=\frac{r}{s}(n_1,p_1,m_1)$,
where $r,s$ are integers, $s,r\neq 0, s,r=O(k^{2\delta })$. This
yields:
\begin{equation}s[\s_1,\s_3]=r[\s_1,\s_2], \label{I7a} \end{equation}
\begin{equation}s\left([\s_1,\s_3']+[\s_1',\s_3]\right)=r\left([\s_1,\s_2']+[\s_1',\s_2]\right),  \label{I7b} \end{equation}
\begin{equation}s[\s_1',\s_3']=r[\s_1',\s_2'].\label{I7c} \end{equation}
 Note that $|\s_1|, |\s_1'| \neq
0$, since $[\s_1,\s_1'] \neq 0$ by part 2. It follows from
\eqref{I7a} and \eqref{I7c} that $s\s_3-r\s_2$ is colinear to $\s_1$
and $s\s_3'-r\s_2'$ is colinear to $\s_1'$. Hence
\begin{equation}s\s_3=r\s_2+t\s_1,\ \ \ s\s_3'=r\s_2'+t'\s_1',\label{E3} \end{equation}
$t,t'$ being rational numbers. Substituting these expressions for
$s\s_3$, $s\s_3'$ into \eqref{I7b} and simplifying, we obtain
$(t-t')[\s_1',\s_1]=0$. Considering that $[\s_1',\s_1]\neq 0$ yields
$t=t'$. It easily follows from \eqref{E3} that $t\s_1$ is an integer
vector and $t\s_1=O(k^{3\delta })$ . Hence, $t$ is a rational number
and the denominator of $t$ is less than $|\s_1|$, i.e. it is
$O(k^{\delta })$. Multiplying both sides of \eqref{E3} by the
denominator of $t$, we rewrite \eqref{E3} with all integers $r,s,t$
such that $r,s,t=O(k^{3\delta })$. We already showed that $s,r\neq
0$. Note that $t\neq 0$ too, since otherwise $\p_{\q_3}$ and
$\p_{\q_2}$ are colinear.

\item  In this part by the way of contradiction we show that $\p_{\q_4}$ does not exist.
Indeed, suppose it does. Excluding $(\vec k+\p_{\m_0},\p_{\q_1})$
and $(\vec k+\p_{\m_0},\p_{\q_2})$ from relations \eqref{I5} for
$\p_{\q_1} $, $\p_{\q_2} $ and
$\p_{\q_3}=\frac{r}{s}\p_{\q_2}+\frac{t}{s}\p_{\q_1} $, we obtain:
\begin{equation} \label{Est} p_{\q_1}^2\frac{t}{s}\left(\frac{t}{s}-1\right)+2p_{\q_1}p_{\q_2}\frac{t}{s}\frac{r}{s}+p_{\q_2}^2\frac{r}{s}\left(\frac{r}{s}-1\right)=
O(k^{-40\mu\delta +3\delta }).\end{equation} Considering
as in Part 4,  we obtain $\p_{\q_4}=\frac{\hat r}{\hat s}\p_{\q_2}+\frac{\hat t}{\hat s}\p_{\q_1}$
and
\begin{equation}\label{Est1}p_{\q_1}^2\frac{\hat t}{\hat s}\left(\frac{\hat t}{\hat
s}-1\right)+2p_{\q_1}p_{\q_2}\frac{\hat t}{\hat s}\frac{\hat r}{\hat
s}+p_{\q_2}^2\frac{\hat r}{\hat s}\left(\frac{\hat r}{\hat
s}-1\right)= O(k^{-40\mu\delta
+3\delta}).\end{equation}\begin{enumerate}
\item Assume first $(t-s)^2+(\hat t-\hat s)^2\neq 0$. We multiply
both parts of \eqref{Est} by $s^2$ and both parts of \eqref{Est1} by
$\hat s^2$.
 Excluding the terms containing $p_{\q_1}^2$ from the last
two relations and using the estimates $s,\hat s, r,\hat r, t,\hat
t=O(k^{3\delta })$, we obtain:
$$p_{\q_1}p_{\q_2}S=Rp_{\q_2}^{2}+O(k^{-40\mu\delta +15\delta
}),$$ where
$$ S=2\hat t\hat r t(t-s)-2tr\hat t (\hat t-\hat s),\ \ S=O(k^{12\delta }).$$
%%%%%%%%%$$S=\frac{2\hat t}{\hat s}\frac{\hat r}{\hat s}\frac{
%%%%%%%%%t}{s}\left(\frac{t}{ s}-1\right)- \frac{2t}{s}\frac{ r}{s}\frac{\hat
%%%%%%%%%t}{\hat s}\left(\frac{\hat t}{\hat s}-1\right),\ \ \ S=O(k^{12\delta
%%%%%%%%%}),$$
$$R=\hat t r (\hat t -\hat s) (r-s)-t \hat r ( t - s) (\hat r-\hat
s),\ \ R=O(k^{12\delta }).$$
%%%%%%%%%$$R=\frac{\hat t}{\hat s}\frac{r}{s}\left(\frac{\hat t}{ \hat s}-1\right)\left(\frac{ r}{s}-1\right)-
%%%%%%%%%\frac{t}{s}\frac{\hat r}{\hat s}\left(\frac{t}{
%%%%%%%%%s}-1\right)\left(\frac{\hat r}{\hat s}-1\right),\ \ R=O(k^{12\delta
%%%%%%%%%}).
%%%%%%%%%$$
It follows $p_{\q_1}S-p_{\q_2}R=O(k^{-39\mu\delta +15\delta })$.
\begin{enumerate}\item Suppose $R^2+S^2 \neq 0$. Considering that the angle (modulo $\pi $) between
$\p_{\q_1 }$ and $\p_{q_2}$ is less than $k^{-1+\delta }$ and
\eqref{below} for $p_{\q_2}$, we obtain $|\p_{\q_1}S\pm \p_{\q_2}R
|=O(k^{-39\mu\delta +15\delta })$, where $\pm $ means $+$ or $-$.
The last relation yields $p_{S\q_1\pm R\q_2}=O(k^{-39\mu\delta
+15\delta })$. However, it follows from estimates for $S$ and $R$
that we have $\||\p_{S\q_1\pm R\q_2}\||<ck^{13\delta }$. Hence, by
\eqref{below} $p_{S\q_1\pm R\q_2}>ck^{-13\delta \mu }$. We have
arrived to the contradiction. \item Now we check the case
$R^2+S^2=0$. It was shown in Part 4 that $s,\hat s, t,\hat t, r,\hat
r\neq 0$. The equation $S=0$ yields
 \begin{equation} \label{E1}\hat r (t-s)=r(\hat t-\hat s),\end{equation}
 both parts are nonzero by the assumption $(t-s)^2+(\hat t-\hat s)^2\neq
 0.$
  Next, $R=0$ yields
 \begin{equation} \label{E2'}\hat t(r-s)=t(\hat r-\hat s).\end{equation} Let us consider \eqref{E1}, \eqref{E2'} as a
linear   system with respect to $s,\hat s$. If the determinant of
this system $\hat r t-r\hat t$ is zero, then it follows $r/\hat
r=t/\hat t=s/\hat s$. This means that $\p_{\q_3}$ and $\p_{\q_4}$
are equal. This contradicts to our initial assumption. Suppose that
the determinant is not zero. Solving \eqref{E1}, \eqref{E2'} with
respect to $s$ and $\hat s$, we get: $s=t+r, \hat s=\hat t+\hat r$.
Substituting $s=t+r$ into \eqref{Est}, we easily obtain:
$\frac{rt}{s^2}\left(p_{\q_1}-p_{\q_2}\right)^2=O(k^{-40\mu \delta
+3\delta })$. Considering that $s^2=O(k^{6\delta})$ and the angle (modulo $\pi $)
between
 $\p_{\q_1},\p_{\q_2}$ is $O(k^{-1+\delta })$, we obtain $|\p_{\q_1}\pm \p_{\q_2}|=O(k^{-20\mu \delta +5\delta })$, where $\pm $ means $+$ or $-$. Obviously, $p_{\q_1-\q_2}=|\p_{\q_1}-\p_{\q_2}|$. Hence, $p_{\q_1-\q_2}=O(k^{-20\mu \delta +5\delta })$.
 This cannot be the case, because $\||\p_{\q_1-\q_2}\||<3k^{\delta
 }$ and thus, $|p_{\q_1-\q_2}|\geq ck^{-\mu\delta}$.\end{enumerate}
\item It remains to consider the case $(t-s)^2+(\hat t-\hat s)^2=0$.
If $(r-s)^2+(\hat r-\hat s)^2\neq 0$, we exclude $p_{\q_2}^2$ from
\eqref{Est}, \eqref{Est1} by  and make considerations similar to the
case $(t-s)^2+(\hat t-\hat s)^2\neq 0$. If $(t-s)^2+(\hat t-\hat
s)^2+(r-s)^2+(\hat r-\hat s)^2=0$, then \eqref{Est}, \eqref{Est1}
give: $p_{\q_1}p_{\q_2}=O(k^{-40\mu \delta +3\delta })$. It
contradict to the inequalities \eqref{below} for
$p_{\q_1},p_{\q_2}$.
   Thus, $\p_{\q_4}$ does not exist.
\end{enumerate} \end{enumerate}\end{proof}

\subsection{Appendix 3}\label{A:5} \begin{lemma}\label{L: Appendix
1} The equation \begin{equation}\lambda^{(1)} \big(\k^{(1)}(\varphi
)+\p_\m\big)=k^{2l}+\varepsilon _0,\ \ \ 0<p_\m\leq4k^\delta,\
|\varepsilon _0|\leq p_\m k^{\delta },\label{25a}
\end{equation} has no more than two solutions
$\varphi^\pm(\varepsilon_0) $ in $\tilde{\cal W}^{(1)}(k,\frac18)\cap \OO
_\m(k,\frac{1}{2})$. They satisfy the estimates:
$$\big|\varphi^{\pm }(\varepsilon_0)-\varphi _\m ^{\pm
}\big|<k^{-2l+1+2\delta }.$$\end{lemma}
 \begin{proof} Let $\varphi \in \tilde{\cal W}^{(1)}(k,\frac14)\cap \OO
_\m(k,\frac{1}{2})$. The equation (\ref{25a}) is equivalent to
$$\lambda^{(1)}(\y(\varphi))=\lambda ^{(1)}(\y(\varphi)-\p
_\m)+\varepsilon_0,\ \ \y(\varphi) =\k^{(1)}(\varphi)+\p _\m.$$ We
use perturbation formula (\ref{eigenvalue}): $$|\y
(\varphi)|_{\R}^{2l}+f_1(\y (\varphi))
    =|\y (\varphi)-\p _\m|_{\R}^{2l}+f_1(\y (\varphi)-\p_\m)+\varepsilon_0,$$
    where $f_1$ is the series in the right-hand side of (\ref{eigenvalue}). This equation can be rewritten as
    \begin{equation}\label{3.7.1.5}
    \Bigl(2(\k^{(1)} (\varphi),\p _\m)_{\R}+p_\m^2\Bigr)(\Bigl(|\y
    (\varphi)|_{\R}^{2l-2}+....+|\y
    (\varphi)-\p_\m|_{\R}^{2l-2}\Bigr)
    +f_1(\y (\varphi))
    -f_1(\y (\varphi)-\p _\m)=\varepsilon_0.
    \end{equation}
 Using the
notation $\p_\m=p_\m(\cos \varphi_\m,\sin \varphi_\m)$, dividing
both sides of the equation (\ref{3.7.1.5}) by $2p_\m k(\Bigl(|\y
    (\varphi)|_{\R}^{2l-2}+....+|\y
    (\varphi)-\p_\m|_{\R}^{2l-2}\Bigr)$, and considering
that $\y(\varphi )=\k^{(1)}(\varphi)+\p _\m=
    (k+h^{(1)})\vec \nu +\p_\m$, we obtain:
    \begin{equation}\label{3.7.1.6}
    \cos (\varphi-\varphi_\m)+\dfrac{p_\m}{2k}-\varepsilon_0g_1(\varphi)+g_2(\varphi)+g_3(\varphi)=0,
    \end{equation}
where $g_1(\varphi)
    =(2p_\m k)^{-1}\Bigl(|\y
    (\varphi)|_{\R}^{2l-2}+....+|\y
    (\varphi)-\p_\m|_{\R}^{2l-2}\Bigr)^{-1}$
    and
    $$g_2(\varphi)=\dfrac{( \h^{(1)}(\varphi),\p_\m)}{p_\m k} , \  \  g_3(\varphi)=\Bigl(f_1(\y(\varphi))
    -f_1(\y(\varphi)-\p_\m)\Bigr)g_1(\varphi),\ \ \h^{(1)}(\varphi)=h^{(1)}(\varphi)\vec \nu.$$
  Obviously $g_1=O\left(p_\m^{-1}k^{-2l+1}\right)$.
 Using Lemma \ref{ldk}  and
considering that $0<p_\m\leq 4k^{\delta }$, we easily obtain:
    $$
   |g_2(\varphi)|= \left|\frac{(\h^{(1)}(\varphi),\p_\m)}{p_\m k}\right|\leq
    \frac{2h^{(1)}}{k}=O(k^{-4l+(80\mu +6)\delta}).
    $$
Let us show $g_3(\varphi)=O(k^{-2l+1+\delta})$.  If $p_\m\geq
k^{-2l+\delta (80\mu +6)}$, then the estimate easily follows from
\eqref{perturbation} and the estimate for $g_1$. Let $p_\m
<k^{-2l+\delta (80\mu +6)}$. It can be easily shown that the series
$f_1(\y)$, $\nabla f_1(\y)$ converge for all $\y$:
$|\y-\k^{(1)}(\varphi)|_{\C^2}<k^{-1-\delta (40\mu +1)}$ and
holomorphic with respect to $y_1$, $y_2$. Using
\eqref{perturbation}, we get $\nabla f_1(y) =k^{-2l+1+\delta (120\mu
+7)}$. Hence,
    $$ \left|f_1(\y(\varphi))-f_1(\y(\varphi)-\p_\m)\right| \leq \sup |\nabla
    f_1|p_\m=o(p_\m ),$$
and therefore, $g_3(\varphi)=O(k^{-2l+1+\delta})$.
   %%% $ \left|\bigl(f_1(\y(\varphi))-f_1(\y(\varphi)-\p_\m
 %%%%%   )\bigr)g_1(\varphi)\right| =o(k^{2l-1}).$
%%%%%    (k^{-4l+1+(120\mu +7)\delta})
Since  $|\varepsilon_0|<p_\m k^{\delta}$, we obtain
    $\varepsilon_0g_1(\varphi)=O(k^{-2l+1+\delta }).$ Thus,
    \begin{equation}
    g_2(\varphi)+ g_3(\varphi)
    -\varepsilon_0g_1(\varphi)=O(k^{-2l+1+\delta}) \ \ \mbox{when}\ \ \varphi \in \tilde{\cal W}^{(1)} (k,\frac14)\cap \OO
_\m(k,\frac{1}{2}).\label{Nov14}
    \end{equation}
    By definition $\varphi_\m^{\pm }$ satisfy the equation $\cos
    (\varphi-\varphi_\m)+\dfrac{p_\m}{2k}$=0.

    Suppose both $\varphi_\m^{\pm
    }$ are in $\tilde{\cal W}^{(1)}(k,\frac{3}{16})$. We draw two circles $C_{\pm}$ centered at
    $\varphi_\m^{\pm }
$ with the radius $k^{-2l+1+2\delta }$. They are both inside
$\tilde{\cal W}^{(1)}(k,\frac14)\cap \OO _\m(k,\frac{1}{2})$, the
perturbation series converging and the estimate (\ref{Nov14}) holds.
For any $\varphi$ on $C_{\pm}$, $|\varphi-\varphi_\m ^{\pm
}|=k^{-2l+1+2\delta}$ and, therefore,
    $|\cos(\varphi-\varphi_\m)+\frac{p_\m}{2k}|>\frac{1}{2}k^{-2l+1+2\delta }>|g_2(\varphi)+g_3(\varphi)-\varepsilon_0g_1(\varphi)|$
    for any $\varphi \in C_{\pm}$.
By Rouch\'{e}'s Theorem, there is only one solution of the equation
(\ref{3.7.1.6}) inside each $C_{\pm}$. Obviously, \eqref{3.7.1.6}
does not have solutions in $\tilde{\cal W}^{(1)}(k,\frac14)\cap \OO
_\m(k,\frac{1}{2})$ outside $C_{\pm}$.

If both $\varphi_\m^{\pm }$ are not in $\tilde{\cal W}^{(1)}(k,\frac{3}{16})$,
then their distance to $\tilde{\cal W}^{(1)}(k,\frac{1}{8})$ is at least
$\frac{1}{16}k^{-40\mu \delta}$, hence
$|\cos(\varphi-\varphi_\m)+\frac{p_\m}{2k}|>\frac{1}{4}k^{-2l+1+2\delta}$
in $\tilde{\cal W}^{(1)}(k,\frac18)$. Therefore,  equation (\ref{3.7.1.6})
has no solution in $\tilde{\cal W}^{(1)}(k,\frac18)\cap \OO
_\m(k,\frac{1}{2})$. The case, when only one $\varphi_\m^{\pm }$ is
not in $\tilde{\cal W}^{(1)}(k,\frac{3}{16})$ is the obvious combination of the
two previous situations.
 Thus, there are at most two solutions in $\tilde{\cal W}^{(1)}(k,\frac18)\cap \OO
_\m(k,\frac{1}{2})$ and
    $|\varphi^{\pm }({\varepsilon _0})-\varphi_\m ^{\pm })|<k^{-2l+1+2\delta}.$
\end{proof}

    \begin{lemma} \label{L:Appendix 2} For any $\varphi \in \tilde{\cal W}^{(1)}(k,\frac14)\cap \OO _\m(k,1)$ satisfying the estimate $\big|\varphi
-\varphi _\m ^{\pm  }\big|<k^{-\delta },$
\begin{equation}\frac{\partial }{\partial \varphi }\lambda^{(1)}
\big(\k^{(1)}(\varphi )+\p_\m\big)=\pm 2lp_\m k^{2l-1}(1+o(1)),
\label{26a}
\end{equation} \end{lemma} \begin{proof} First, assume $\varphi $ is real. Let $\y(\varphi )= \k^{(1)}(\varphi )+\p_\m$. Using the perturbation
formula (\ref{eigenvalue}) and Lemma \ref{ldk}, we obtain:
    \begin{align}\label{3.7.1.7}
    \lefteqn{
    \frac{\partial}{\partial \varphi}\lambda^{(1)}\bigl(\y(\varphi)\bigr)
    =\frac{\partial}{\partial \varphi}\left[\lambda^{(1)}\bigl(\y(\varphi)\bigr)-k^{2l}\right]=
    \frac{\partial}{\partial \varphi}\left[\lambda^{(1)}\bigl(\y(\varphi)\bigr)-\lambda^{(1)}\bigl(\y(\varphi)-\p_\m\bigr)\right]=
    }& \notag \\
    &\hspace{1cm}
    \left ( \nabla_{\y}\lambda^{(1)}\bigl(\y(\varphi)\bigr)-\nabla_{\y}\lambda^{(1)}
    \bigl(\y(\varphi)-\p_\m\bigr),\frac{\partial}{\partial \varphi}\y(\varphi)\right )
    _\R=
    \notag \\
    &\hspace{1cm}
    \left ( \nabla \bigl|\y(\varphi)\bigr|_{\R}^{2l}-\nabla \bigl|\y(\varphi)-\p_\m\bigr|_{\R}^{2l},(k+h^{(1)})\vec t+(h^{(1)})^{\prime}\vec \nu\right )
_\R +\notag\\
    &\hspace{2cm}
    \left ( \nabla f_1\bigl(\y(\varphi)\bigr)-\nabla f_1\bigl(\y(\varphi)-\p_\m\bigr),(k+h^{(1)})\vec t+(h^{(1)})^{\prime}\vec \nu\right ) _\R,
    \end{align}
where $\vec \nu=(\cos \varphi, \sin \varphi)$ and $\vec t=\vec \nu^{\prime}=(-\sin
\varphi,\cos \varphi)$, $f_1$ is the series in the right-hand side
of (\ref{eigenvalue}). Note that
    $$
    \nabla \bigl|\y(\varphi)\bigr|_{\R}^{2l}-\nabla
    \bigl|\y(\varphi)-\p_\m\bigr|_{\R}^{2l}=$$
   $$
   2l\bigl|\y(\varphi)\bigr|_{\R}^{2l-2}\y(\varphi)-2l\bigl|\y(\varphi)-\p_\m\bigr|_{\R}^{2l-2}(\y(\varphi)-\p_\m)=$$
\begin{equation} 2l\bigl|\y(\varphi)\bigr|_{\R}^{2l-2}\p_\m+2l
\Bigl(\bigl|\y(\varphi)\bigr|_{\R}^{2l-2}-\bigl|\y(\varphi)-\p_\m\bigr|_{\R}^{2l-2}\Bigr)(\y(\varphi)-\p_\m).
 \label{3.7.1.8}   \end{equation}
Substituting (\ref{3.7.1.8}) into (\ref{3.7.1.7}), we get $
    \frac{\partial}{\partial
    \varphi}\lambda^{(1)}\bigl(\y(\varphi)\bigr)=T_1+T_2+T_3,$
    \begin{align*}
    T_1 &
    =2l \bigl|\y(\varphi)\bigr|_{\R}^{2l-2}\left ( \p_\m,(k+h^{(1)})\vec t+(h^{(1)})^{\prime} \vec \nu \right ) _\R ,\\
T_2 &
=\Bigl(\bigl|\y(\varphi)\bigr|_{\R}^{2l-2}-\bigl|\y(\varphi)-\p_\m\bigr|_{\R}^{2l-2}\Bigr)\left(\y(\varphi)-\p_\m,
(k+h^{(1)})\vec t+(h^{(1)})^{\prime} \vec \nu \right) _\R,\\
    T_3 &
=\left ( \nabla f_1\bigl( \y(\varphi) \bigr) - \nabla
    f_1\bigl( \y(\varphi)-\p_\m\bigr),(k+h^{(1)})\vec t+(h^{(1)})^{\prime}\vec \nu\right ) _\R .
    \end{align*}
We see that $\varphi$ is close to $\varphi_\m \pm \pi /2$, since
$|\varphi -\varphi _\m^{\pm}|<k^{-\delta }$ by the hypothesis of the
lemma and $\varphi _\m^{\pm}=\varphi_\m \pm \pi /2 +O(k^{-1+\delta
})$ when $p_\m<4k^{\delta }$. Now we readily obtain:
    $( \p_\m,\vec \nu) _\R=o(p_\m),$
 $( \p_\m,\vec t) _\R=\pm p_\m(1+o(1))$.  Using also estimates \eqref{dk0}, \eqref{dk} for $h^{(1)}$, we get
  $T_1=\pm 2lp_\m k^{2l-1}(1+o(1))$.  Note that $\y(\varphi)-\p_\m=\k^{(1)}(\varphi )$ and, hence, it is orthogonal to
  $\tau (\varphi)$.
  Using this fact, we simplify the expression for $T_2$:
  $$T_2
=\Bigl(\bigl|\y(\varphi)\bigr|_{\R}^{2l-2}-\bigl|\y(\varphi)-\p_\m\bigr|_{\R}^{2l-2}\Bigr)\left(\y(\varphi)-\p_\m,
(h^{(1)})^{\prime} \vec \nu \right) _{\R}.$$ Using \eqref{dk} for
$(h^{(1)})^{\prime}$, we obtain $T_2=o(p_\m k^{2l-1})$.
Let us estimate $T_3$.
 It can be easily shown that the series $f_1(\y)$, $\nabla f_1(\y)$, $D^2 f_1(\y)$ converge for all $\y$: $|\y-\k_1(\varphi)|_{\C^2}<k^{-1-\delta (40\mu +1)}$ or $|\y+\p_\m-\k_1(\varphi)|_{\C^2}<k^{-1-\delta (40\mu +1)}$, the series being holomorphic with respect to $y_1$, $y_2$ in these neighborhoods.
Using \eqref{perturbation}, we get $\nabla f_1(\y) =k^{-2l+1+\delta
(120\mu +7)}$, $D^2 f_1(\y) =k^{-2l+2+\delta (160\mu +8)}$. Let
$p_\m\geq \frac 12 k^{-1-\delta (40\mu +1)}$. Then, using the
estimate for $\nabla f_1(\y) $, we easily obtain
$T_3=k^{-2l+2+\delta (120\mu +7)}=o(k^{2l-1}p_\m)$. Let $p_\m< \frac
12 k^{-1-\delta (40\mu +1)}$. Then, using the estimate for the
second derivative in the direction of $\p_\m$, we get $\Bigl|\nabla
f_1\bigl(\y(\varphi)\bigr)-\nabla
f_1\bigl(\y(\varphi)-\p_\m\bigr)\Bigr|=O(p_\m
k^{-2l+2+(160\mu+8)\delta}).$
    Therefore,
$T_3=O\left(p_\m k^{-2l+3+(160\mu +8)\delta}\right)$. Thus,
$T_3=o(k^{2l-1}p_\m)$ for all $\p_\m$. Adding the estimates for
$T_1,T_2,T_3$, we get (\ref{26a}).

Since all formulas can be analytically extended to the area  of non-real $\varphi $, the estimates being preserved, (\ref{26a}) holds for any $\varphi \in \tilde{\cal W}^{(1)}(k,\frac14)\cap \OO _\m(k,1)$. \end{proof}

 \begin{lemma}\label{L:3.7.1} Let $\tilde {\cal O}_{\m,\varepsilon}^\pm$ be the open discs of
 the radius $\varepsilon $ centered at $\varphi ^{\pm }(0)$. For any
$ \varphi \in \tilde{\cal W}^{(1)}(k,\frac18)\cap {\cal
O}_{\m}(k,\frac{1}{2})$, $\varphi \not \in \tilde {\cal
O}_{\m,\varepsilon}^\pm$, and $0\leq \varepsilon<k^{-2l+1+\delta}$,
    \begin{equation}\label{3.7.1.10}
    |\lambda^{(1)}(\y(\varphi))-k^{2l}|\geq k^{2l-1}
    p_\m\varepsilon.
    \end{equation}
\end{lemma} \begin{proof} Suppose (\ref{3.7.1.10}) does not hold for
some $ \varphi \in \tilde{\cal W}^{(1)}(k,\frac18)\cap {\cal
O}_{\m}(k,\frac{1}{2})$, $\varphi \not \in \tilde {\cal
O}_{\m,\varepsilon}^\pm$. This means that $\varphi$ satisfies
equation (\ref{25a}) with some $\varepsilon _0$: $|\varepsilon _0|
<k^{2l-1}p_\m\varepsilon\,(<p_\m k^\delta)$. By Lemma \ref{L: Appendix 1} , $\varphi$ could
be either $\varphi^{+}(\varepsilon_0)$ or
$\varphi^{-}(\varepsilon_0)$. Without loss of generality, assume
$\varphi =\varphi^{+}(\varepsilon_0)$. By Lemma \ref{L: Appendix 1},
$|\varphi^{+}(\varepsilon_0)-\varphi _{\m}^{\pm}|<k^{-2l+1+2\delta
}$ for $\varphi _{\m}^{+}$ or $\varphi _{\m}^{-}$. Obviously,
$k^{-2l+1+2\delta}$ neighborhood of $\varphi^{+}(\varepsilon_0)$
satisfies conditions of Lemma \ref{L:Appendix 2}.
     Using (\ref{26a}) and Rouche's theorem in  the $k^{-2l+1+2\delta }$-neighborhood of
$\varphi^{+}(\varepsilon_0)$, we obtain that there is a point
$\varphi^{+}(0)$ in this neighborhood and   $\left|\varphi^{+}(\varepsilon_0)-\varphi^{+}(0)
    \right|<\varepsilon$, i.e., $\varphi \in \tilde {\cal O}_{\m,\varepsilon}$. This contradicts
    the hypothesis of the lemma.
    \end{proof}

\begin{lemma} \label{L:July5} If $0<p_\m\leq 4k^{\delta }$ and
$\varphi \in \tilde{\cal W}^{(1)}(k,\frac18)\cap {\cal
O}_\m(k,\frac{1}{2})$, then
\begin{equation} \left \|(\lambda ^{(1)}(\y(\varphi
))-k^{2l})\left(P_\m(H(\k^{(1)}(\varphi
))-k^{2l})P_\m\right)^{-1}\right\|\leq 8, \ \ \ \ \y(\varphi ):=
\k^{(1)}(\varphi )+\p _\m. \label{July3a}
\end{equation} \end{lemma} \begin{proof} Let $\tilde C_1$ be the
circle in $\C$ of the radius $k^{2l-1-80\mu \delta }$ centered at
$z=|\y(\varphi )|_{\R}^{2l}$. Using \eqref{metka1} and \eqref{metka2}, we easily get:
$$\left||\y(\varphi )+\p_\q|_{\R}^{2l}-z\right|\gtrsim
\frac{1}{2}k^{2l-1-40\mu \delta },\ \ \ \mbox{when}\ \ \q\in
\Omega(\delta),\ \ z\in \tilde C_1.$$ Therefore,
\begin{equation}\left\|\left(P_\m(H_0(\k^{(1)}(\varphi ))-z)P_\m\right)^{-1}\right\|\lesssim 2k^{-2l+1+40\mu \delta
},\label{2.8} \end{equation}
\begin{equation}\left\|\left(P_\m(H_0(\k^{(1)}(\varphi ))-z)P_\m\right)^{-1}\right\|_1\lesssim 2k^{-2l+1+40\delta(\mu
+1)}.\label{2.8'} \end{equation} Next, by
 (\ref{determinants}),
 \begin{equation}\Bigl|\det \frac{P_\m(H(\k^{(1)}(\varphi ))-z)P_\m}{P_\m(H_0(\k^{(1)}(\varphi ))-z)P_\m}-1 \Bigr|<
 2 \|V\|\left \|\left(P_\m(H_0(\k^{(1)}(\varphi ))-z)P_\m\right)^{-1}\right\|_1 \label{r1}
 \end{equation}
 for every $z$ on the contour $\tilde C_1$.
 Using the estimate (\ref{2.8'}), we obtain  that
 the right-hand part of (\ref{r1}) is less than 1.
 Applying  Rouch\'{e}'s theorem, we conclude that the
 determinant has the same number of zeros and poles inside $\tilde C_1$.
 Considering that the resolvent $\left(P_\m(H_0(\k^{(1)}(\varphi ))-z)P_\m\right)^{-1}$ has a single pole,
 $z=|\y(\varphi )|^{2l}_\R$,
 we obtain that  $\left(P_\m(H(\k^{(1)}(\varphi ))-z)P_\m\right)^{-1}$ has a single pole inside $\tilde C_1$ too.
  Obviously, the pole is in the point
 $z=\lambda
^{(1)}(\y(\varphi ))$. Therefore $$(\lambda ^{(1)}(\y(\varphi
))-z)\left(P_\m(H(\k^{(1)}(\varphi ))-z)P_\m\right)^{-1}$$ is
a holomorphic function of $z$
 inside $\tilde C_1$.

 Let $z\in \tilde C_1$. Using (\ref{perturbation}), we easily obtain:
 $|\lambda ^{(1)}(\y(\varphi))-z|\leq 2k^{2l-1-40\mu \delta }.$ From (\ref{2.8}) and
Hilbert identity it follows that
\begin{equation}\left\|\left(P_\m(H(\k^{(1)}(\varphi ))-z)P_\m\right)^{-1}\right\|\leq 4k^{-2l+1+40\mu \delta }, \
\mbox{when }z\in \tilde C_1.\label{2.8"} \end{equation} Multiplying
the last two estimates, and using maximum principle we get
\begin{equation} \left \|(\lambda ^{(1)}(\y(\varphi
))-z)\left(P_\m(H(\k^{(1)}(\varphi
))-z)P_\m\right)^{-1}\right\|\leq8, \ \ \ \ z\in
\overline{\hbox{Int}\,\tilde C_1}. \label{July3a'}
\end{equation} Note that $z:=k^{2l}\in\overline{\hbox{Int}\,\tilde
C_1}$. Indeed, by \eqref{metka1},
%%%%%since $\varphi \in {\cal O}_\m(k,\frac{1}{2})$,
 $||\y(\varphi )|_{\R}^{2l}-k^{2l}|<k^{2l-2-40\mu
\delta }$. Substituting $z=k^{2l}$ in the last estimate, we get
\eqref{July3a}. \end{proof}

\subsection{Appendix 4}

\begin{lemma}\label{L:r_0} Let $R$ be the smallest positive integer for which \eqref{r_0} holds.
%%%%%%%%%$$
%%%%%%%%%{\cal
%%%%%%%%%A}_{r_0}:=P^{\partial}(\hat{H}_0-k^2)^{-1}\left[W(\hat{H}_0-k^2)^{-1}\right]^{r_0}P\not=0,\
%%%%%%%%%\ \ W:=\hat{H}-\hat{H}_0.
%%%%%%%%%$$
We have $R>\frac{1}{64} k^{(\frac\gamma2+2\delta_0)r_1-\delta}$.
\end{lemma}

\begin{proof}
Notice that
$$
{\cal A}_{R}=\sum_{i_1,\dots,i_{R}=0}^{\hat{I}}{\cal
A}_{i_1,\dots,i_{R}},
$$
where
$$
{\cal
A}_{i_1,\dots,i_{R}}:=P^{(int)}W(\hat{H}_0-k^{2l})^{-1}\left[\prod_{k=1}^{R}
\left(P_{i_k}W(\hat{H}_0-k^{2l})^{-1}\right)\right]P^{\partial}.
$$
Here we used that $R$ is the smallest positive integer for which
${\cal A}_{R}\not=0$.

In some sense everything is defined by the case where all $i_k$ are
equal to zero. But to include impurities of non-resonant and white
clusters we need additional construction. Consider a particular
${\cal A}_{i_1,\dots,i_{R}}$. For the sequence $i_1,\dots,i_{R}$ we
take a subsequence of {\it all} non-zero indices
$i_{k_1},\dots,i_{k_s}$, $1\leq k_1<\dots<k_s\leq R$ (this sequence
can be empty). Now we construct a subsequence ${j_1},\dots,{j_p}$
($p\leq s$) of non-repeating indices as follows. We choose
${j_1}=i_{k_1}$. If $i_{k_1}$ is not equal to any other $i_{k_t}$,
$t=2,\dots,s$, then ${j_2}=i_{k_2}$. If there is one or more
$i_{k_t}$ equal to $i_{k_1}$ then we denote the segment between the
first and the last $i_{k_1}$ as $I_1$. The next term after $I_1$ we
choose to be an ${j_2}$ etc. Thus (with a slight abuse of the
notation) we have
$$
i_{k_1},\dots,i_{k_s}=I_1, I_2,\dots,I_p,\ \ \ p\leq s.
$$
Now, the initial sequence $i_1,\dots,i_{R}$ can be represented as
$I_1^0,\tilde{I}_1,I_2^0,\tilde{I}_2,\dots,\tilde{I}_p,I_{p+1}^0$.
Here each $I_j^0$ is a sequence of only zeros (it can be empty) and
$\tilde{I}_j$ is $I_j$ with possibly some zeros inside. Put
$$
P_{j_1}BP_{j_1}:=P_{j_1}W(\hat{H}_0-k^{2l})^{-1}P_{i'}W(\hat{H}_0-k^{2l})^{-1}\dots(\hat{H}_0-k^{2l})^{-1}P_{j_1}.
$$
Here for all internal projectors $P_{i'}$ we have either $i'\in I_1$
or $i'=0$. We notice that $P_{j_1}BP_{j_1}$ has a block form and
$\|P_{j_1}BP_{j_1}\|\leq k^{-\beta}$. We can represent now ${\cal
A}_{i_1,\dots,i_{R}}$ in the following form
\begin{equation} \label{Apr12}
\begin{split}
&{\cal
A}_{i_1,\dots,i_{R}}=P^{(int)}W(\hat{H}_0-k^{2l})^{-1}\left[\prod_{k=1}^{p}
\left(P_{0}W(\hat{H}_0-k^{2l})^{-1}\right)^{s_k}P_{j_k}BP_{j_k}W(\hat{H}_0-k^{2l})^{-1}\right]\cr
&
\times\left(P_{0}W(\hat{H}_0-k^{2l})^{-1}\right)^{s_{p+1}}P^{\partial}.
\end{split}
\end{equation}
Here $s_k$ is the number of elements in $I_k^0$, $s_k \geq 0$; $j_k$
is a non-zero index corresponding to $I_{k}$. Obviously,
$\left(W(\hat{H}_0-k^{2l})^{-1}\right)_{\m\m'}=0$ if
$\||\cdot\||$-distance between the cluster containing $\p_\m$ and
the cluster containing $\p_{\m'}$ is greater than $k^\delta$ (here,
as usual, we consider the points in the range of $P_0$ as $1\times1$
clusters). Next, if $P_{j_k}$ is the projection on a non-resonant
cluster, then
$$
\left(P_{j_k}BP_{j_k}\right)_{\m\m'}=0,\ \ \ \hbox{for}\
\||\p_{\m-\m'}\||>8k^\delta,
$$
since a non-resonant cluster has the size not greater than $8k^{\delta }$.
%%%%%onto the region where at most two $k^\delta$-"squares" can form a
%%%%%cluster, we obviously have
Let $p'$ be the number of non-resonant projections in the sequence
$\{P_{j_k}\}_{k=1}^p$. Hence, $p-p'$ is the number of white
clusters. The operator ${\cal A}_{i_1,\dots,i_{R}}$ can be non-zero
only if
\begin{equation}\label{odin*}
\frac D2 \leq k^\delta\sum_{k=1}^{p+1}(s_k+1)+8k^{\delta }p'
+\sum_{m=1}^{p-p'} d_m,
\end{equation}
where $d_m$ is the size of a white cluster $P_{j_{k_m}}$. Next, we
prove that
\begin{equation}\label{dva*}
\sum_{k=1}^{p+1}(s_k+1)+p' \geq
\frac{1}{32}k^{(\frac\gamma2+2\delta_0)r_1-\delta}.
\end{equation}
Assume that \eqref{dva*} does not hold. Then, by \eqref{odin*}
\begin{equation}\label{tri*}
\sum_{m=1}^{p-p'} d_m\geq \frac 14 k^{(\frac\gamma2+2\delta_0)r_1},
\end{equation}
since $D=k^{(\frac\gamma2+2\delta_0)r_1}$.
Obviously,
$$d_m\leq n_mk^{\frac{\gamma r_1}{6}},$$
where $n_m$ is the number of $\MM^{(2)}$ points in the white cluster
$\Pi _{j_{k_m}},\ m=1,\dots,p-p'$. Let $\ell$ be the size of a
minimal box containing all white clusters  $\Pi _{j_{k_m}}$.  It is
easy to see that
$$\ell\leq k^{\delta }\sum_{k=1}^{p+1}(s_k+1)+8k^{\delta }p' +\sum_{m=1}^{p-p'} d_m\leq
\frac 14 k^{(\frac\gamma2+2\delta_0)r_1}+\sum_{m=1}^{p-p'}
d_m\leq2\sum_{m=1}^{p-p'} d_m.$$ Here we also used \eqref{tri*} and
the inequality opposite to \eqref{dva*}. By Lemma \ref{L:2/3-1}
\begin{equation}\label{vot1}
\sum_{m=1}^{p-p'} n_m\leq C\ell^{2/3}k.
\end{equation}
Combining the last three inequalities and solving for
$\sum_{m=1}^{p-p'} d_m $, we obtain: $\sum_{m=1}^{p-p'}
d_m<k^{\frac{\gamma r_1}{2}+3}$.
%%%At the same time, obviously, $\ell$ is bounded by the size of our
%%%%neighborhood ($k^{(\frac\gamma2+\delta_0+\delta_2)r_1}$)
%At the same time, small clusters together with $\sum_{k=1}^{p+1}s_k$
%other "squares" form the chain connecting $P^{\partial}$ and $P$.
%Thus, taking into account \eqref{tri*} and opposite to \eqref{dva*}
%we obtain
%\begin{equation}\label{vot2}
%\ell\leq\sum_{k=1}^p
%d_k+3k^\delta\sum_{k=1}^{p+1}s_k\leq4\sum_{k=1}^p d_k.
%\end{equation}
%Combining \eqref{vot1} and \eqref{vot2} we get
%$$
%\ell\leq Ck^{\gamma r_1/2}.
%$$
%Substituting this into \eqref{vot1} we obtain $\sum_{k=1}^p d_k\leq
This contradicts to \eqref{tri*}. Thus, we proved \eqref{dva*}.
Using the obvious inequality $\sum _{k=1}^{p+1}(s_k+1)+p'\leq R+1$
proves the lemma.

\end{proof}
\subsection{Appendix 5}
%%% Let
%%%%$r_0$ is the smallest positive integer for which ${\cal A}_{r_0}$
%%%%defined by \eqref{r_0} is not equal to zero.
\begin{lemma}\label{L:r_0-b} Let $R$ be the smallest positive integer for which \eqref{r_0} holds in the case of a black cluster.
We have $R>\frac{1}{64} k^{\gamma r_1+\delta _0r_1-\delta}$.
%%%%In the case of a black component the estimate
%%%%%%%$r_0>\frac{1}{32} k^{\gamma r_1+\delta _0-\delta}$ holds.
\end{lemma}
\begin{proof} We again use formula \eqref{Apr12}, where $P_{j_k}$ are projections on non-resonant, white and grey components in a component of a black region.
Assume first that all components $\Pi_{j_k}$ can be placed in one
$\||\cdot \||$ box of the size $4k^{\gamma r_1+\delta _0r_1}$.
Obviously,
\begin{equation} \label{Apr12-1}\frac{1}{2}k^{\gamma r_1+\delta _0r_1}\leq k^{\delta }\sum (s_k+1) +8k^{\delta }p'+\sum  _m d_m^w +\sum _{\tilde m} d_{\tilde m}^g,\end{equation}
where $p'$ is the number of non-resonant components in the black component,
$\sum  _m d_m^w$ and $\sum  _{\tilde m} d_{\tilde m}^g$ are the total lengths of white and grey components
 in the black component. Let us prove first that
 $k^{\delta }\sum (s_k+1) +8k^{\delta }p'+\sum  _m d_m^w >\frac{1}{4}k^{\gamma r_1+\delta _0r_1}$. Suppose that it is not so.
 Then, by \eqref{Apr12-1},  $\sum _{\tilde m} d_{\tilde m}^g>\frac{1}{4}k^{\gamma r_1+\delta _0r_1}$.
The $4k^{\gamma r_1+\delta _0r_1}$-box containing all $\Pi_{j_k}$,
consists of no more than $4^4k^{4\delta_0r_1}$ boxes of the size
$k^{\gamma r_1}$. Since all $\Pi _{j_k}$ are in white $k^{\gamma
r_1}$ -boxes, the total number of points of $\MM^{(2)}$ in these
white boxes does not exceed $ck^{\frac{1}{2}\gamma r_1+\delta
_0r_1}\cdot k^{4\delta _0r_1}$. Since each grey box contains no less
than $k^{\frac{1}{6}\gamma r_1-\delta _0r_1}$ points, the total
number of grey boxes is less than $ck^{\frac{1}{2}\gamma r_1+\delta
_0r_1}\cdot k^{4\delta _0r_1}\cdot k^{-\frac{1}{6}\gamma r_1+\delta
_0r_1}=ck^{\frac{1}{3}\gamma r_1+6\delta _0r_1}$. Therefore, the
total size of the grey region is less than $ck^{\frac{1}{3}\gamma
r_1+6\delta _0r_1}\cdot k^{\frac{1}{2}\gamma r_1+2\delta
_0r_1}=ck^{\frac{5}{6}\gamma r_1+8\delta _0r_1}$. Since $\delta
_0<\frac{1}{48}\gamma$, it is much less than $\frac{1}{4}k^{\gamma
r_1+\delta _0r_1}$. We have arrived to the contradiction with the
assumption $\sum _{\tilde k} d_{\tilde k}^g>\frac{1}{4}k^{\gamma
r_1+\delta _0r_1}$.  Therefore, $k^{\delta }\sum (s_k+1) +8k^{\delta
}p'+\sum  _k d_k^w > \frac{1}{4}k^{\gamma r_1+\delta _0r_1}$.
Considering again that the total number of $\MM^{(2)}$ points in the
white boxes of the $4k^{\gamma r_1+\delta _0}$-box does not exceed
$ck^{\frac{\gamma r_1}{2}+5\delta _0r_1}$, we obtain $\sum _k
d_k^w<ck^{\frac{\gamma r_1}{2}+5\delta _0r_1}\cdot 4k^{\frac{\gamma
r_1}{6}}<\frac{1}{20}k^{\gamma r_1+\delta _0r_1}$. It follows
$k^{\delta }\sum (s_k+1) +8k^{\delta }p'
>\frac{1}{5}k^{\gamma r_1+\delta _0r_1}$. Hence, $R+1\geq\sum
(s_k+1) +p'
>\frac{1}{40}k^{\gamma r_1+\delta _0r_1-\delta}$.

Assume that we cannot put all the components $\Pi_{j_k}$  in one
$\||\cdot \||$-box of the size $4k^{\gamma r_1+\delta _0r_1}$. Let
us consider the box of this size around   $\Pi_{j_1}$. Let $K$ be a
number such that all $\Pi_{j_k}$, $k=1,...K$ are in the box and
$\Pi_{j_{K+1}}$ is not. Then, instead of \eqref{Apr12} we consider
just its piece \begin{equation} \label{Apr12-2}
\begin{split}&
P_{j_1}W(\hat{H}_0-k^{2l})^{-1}\left[\prod_{k=2}^{K}
\left(P_{0}W(\hat{H}_0-k^{2l})^{-1}\right)^{s_k}P_{j_k}BP_{j_k}W(\hat{H}_0-k^{2l})^{-1}\right]
 \cr &
 \times\left(P_{0}W(\hat{H}_0-k^{2l})^{-1}\right)^{s_{K+1}}P_{j_{K+1}}.
\end{split}
\end{equation}
Further considerations are the same as in the previous case since by
construction the distance between $\Pi_{j_1}$ and $\Pi_{j_{K+1}}$ is
at least $\frac12 k^{\gamma r_1+\delta _0r_1}$.

\end{proof}

\subsection{Appendix 6. On Application of Bezout Theorem}

Let $D(\k, \lambda )$ be the determinant of the truncated operator
$H(\k)-\lambda $ of the size $k^{r_*}$, $r_*\geq 1$. Obviously, $D$ is the
polynomial of the degree $k^{4lr_*}$ with respect to
$\varkappa_1,\varkappa_2$ and {\it a line is not a solution of the
equation} $D(\k, \lambda )=0$. Let $\lambda $ be fixed, $\lambda
=k^{2l}$.
\begin{definition}\label{elementary} We call a piece of $D(\k, \lambda )=0$ elementary, if

1) it can be parameterized by $\varkappa_1$:
$\varkappa_2=\varkappa_2(\varkappa_1)$ with
$|\varkappa_2'(\varkappa_1)|\leq1$ or by $\varkappa_2$:
$\varkappa_1=\varkappa_1(\varkappa_2)$ with
$|\varkappa_1'(\varkappa_2)|\leq1$;

2) function $\varkappa_1=\varkappa_1(\varkappa_2)$ (or
$\varkappa_2=\varkappa_2(\varkappa_1)$) is monotone and continuously
differentiable;

3) it does not have inflection points inside;

4) it has a length not greater than $1$.
\end{definition}
We will show that the curve $D(\k, \lambda )=0$ can be split into elementary pieces  and estimate the number of such pieces. In
the proof we will apply several times the following statement (which
is a simplified version of Bezout Theorem).
\begin{theorem}\label{bezout}
Let $P$ and $Q$ be two plane real-valued polynomials of degree $p$
and $q$ respectively. If $P$ and $Q$ do not contain common factors
then the total number of points satisfying $P(\k)=0=Q(\k)$ (i.e.
number of points of intersection) does not exceed $pq$.
\end{theorem}
We have
\begin{lemma} \label{elpieces} The set $D(\k,\lambda )=0$ can be split into $k^{17l r_*}$ or less elementary pieces.
\end{lemma}
\begin{proof}
%For every fixed $k_2$, the number of $k_1$ satisfying the equation $D(\k)=0$ is less than $k^{4\gamma r_1}$,
%since $D(\k)$ is a nonzero polynomial of the degree $k^{4\gamma
%r_1}$ in $k_1$ for a fixed $k_2$.
First, $D(\k ,\lambda )$ can be represented as a product of simple (i.e.
irreducible) factors (counting multiplicity). The total number of
factors is less than $k^{4lr_*}$ which is also the bound for their
total degree. We consider one of such simple factors $P$ and denote
by $p$ its degree (note that we do ignore the multiplicity of the
factor). Let us consider the points
\begin{equation}\label{odinindapp}
P(\k)=0,\ \ \frac{\partial}{\partial \varkappa_2}P(\k)=0.
\end{equation}
Since $P$ is irreducible and $\frac{\partial}{\partial
\varkappa_2}P(\k)$ has degree less than $p$ (we also notice that
$\frac{\partial}{\partial \varkappa_2}P(\k)$ is not identically zero
since $D(\k,\lambda)=0$ does not contain lines) they do not have
common factors. Thus, the number of such points $\k$ does not exceed
$p(p-1)$. Next, by the same reasons the number of points
\begin{equation}\label{odinindapp-1}
P(\k)=0,\ \ \frac{\partial}{\partial \varkappa_1}P(\k)=0
\end{equation}
 does not exceed $p(p-1)$ and the number of points
\begin{equation}\label{dvaindapp}
P(\k)=0,\ \ \frac{\partial P}{\partial
\varkappa_2}(\k)=\pm\frac{\partial P}{\partial \varkappa_1}(\k)
\end{equation}
does not exceed $2p(p-1)$. We split each previous piece by such
points. Thus, we have at most $4p(p-1)+1$ pieces, each end
satisfying \eqref{odinindapp} or \eqref{dvaindapp}. The sign of $
(\frac{\partial}{\partial \varkappa_2}P)^2-(\frac{\partial}{\partial
\varkappa_1}P)^2$ is constant on each piece, i.e. the piece admits
parametrization as in the property 2 of Definition~\ref{elementary}.
Making parametrization by $\varkappa_1$ or $\varkappa_2$, depending
on the sign, we obtain that the length of a piece does not exceed
$\sqrt{2}\cdot 4k^{r_*}$ (obviously, $|\varkappa_j|<2k^{r_*}$).
Therefore the total length of the curve $P=0$ does not exceed
$18p^2k^{r_*}$. Next, for each piece where $\frac{\partial}{\partial
\varkappa_2}P(\k)\not=0$ inflection points of $P(\k)=0$ are
described by the system
\begin{equation}\label{starindapp}
P=0,\ \ P_{\varkappa_2\varkappa_2}(P_{\varkappa_1})^2-
2P_{\varkappa_1\varkappa_2}P_{\varkappa_1}P_{\varkappa_2}+
P_{\varkappa_1\varkappa_1}(P_{\varkappa_2})^2=0.
\end{equation}
Again, since $P$ is irreducible and no line is a solution, we have no
common factors here and can apply Bezout Theorem. The number of
points satisfying \eqref{starindapp} does not exceed
$p(2(p-1)+(p-2))=3p^2-4p$. Therefore, we have at most $12p^4$ pieces
with the ends satisfying \eqref{odinindapp} or \eqref{dvaindapp} or
\eqref{starindapp}. At last, we split each of these concave pieces
into pieces with the length not greater than $1$. Considering that
the total length of $P(\k)=0$ is less that $18p^2k^{r_*}$, we obtain that
the total number of elementary pieces does not exceed $18p^2k^{r_*}+12p^4$.
Taking the sum over all simple factors of $D$ we prove the lemma.
\end{proof}

%%%%%%%%%%%%\subsection{Appendix 8. Proof of the estimate $M_j<k^{4r_1+8\delta }$ in Section \ref{GSII}. \label{App8}}
%%%%%%%%%%%%The operator $(P(\varphi _j^0)(H(\k^{(1)}(\varphi)-k^{2l})P(\varphi
%%%%%%%%%%%%_j^0))^{-1}$  has a pole only if  $$Det \Big( P(\varphi _j^0)(H(\k^{(1)}(\varphi)-k^{2l})P(\varphi
%%%%%%%%%%%%_j^0)\Big)=0.$$ Since the operator has a block structure this means that $Det \Big( P_{\m _0}(H(\k^{(1)}(\varphi))-k^{2l})P_{\m _0}\Big)=0,$ $\m _0\neq 0$,
%%%%%%%%%%%%$P_{\m _0}$ being the projector on one of the $k^{\delta}$ clusters in $\Omega (r_1)$. The number of such blocks, obviously, does not exceed the number of elements in $\Omega (r_1)$, i.e. $16k^{4r_1}$.
%%%%%%%%%%%%Let $\k_{jm}=\k^{(1)}(\varphi _{jm})$. It satisfies one of the the equations   $Det \Big( P_{\m _0}(H(\k _{jm})-k^{2l})P_{\m _0}\Big)=0,$ $\m_0\in \Omega (r_1)$, the determinants being polynomials of the degree $ck^{4\delta }$.
%%%%%%%%%%%%By definition of $\k^{(1)}$, $\k _{jm}$ also satisfies the equation $Det \Big( P_{0}(H(\k _{jm})-k^{2l})P_{0}\Big)=0.$ By Bezout theorem the number of vectors $\k_{jm}$ and therefore, of $\varphi _{jm}$ does not exceed $k^{8\delta }$
%%%%%%%%%%%%for a fixed $\m_0$. Considering the number of blocks, we obtain that the total number $M_j$ of points $\varphi _{jm}$ is less than $k^{4r_1+8\delta }$.

\subsection{Appendix 7. On the Proof of Geometric Lemmas Allowing to Deal with Clusters instead of Boxes \label{App9}}

In the proof of Lemma~\ref{norm2/3} it is important that we deal
with the same curve generated by the determinant and just change the
argument $\k$. At the same time, a priori we have the estimates for
the resolvent of the operator reduced onto a particular cluster. The
form of clusters can vary which formally changes the projector and
thus the determinant and the curve. Here we explain how to deal with
this situation. We will show that every cluster (white, grey or
black) can be embedded into a box of the fixed size (depending on
the color of the cluster) such that the estimate for the resolvent
on this box is  essentially the same as for the embedded cluster. We
also notice that the estimate for the number of points of
$\MM^{(2)}$ inside these boxes is the same as the worst possible
estimate for the corresponding cluster (see
Lemmas~\ref{L:black},~\ref{L:grey},~\ref{L:white}). This justifies
the application of Lemma~\ref{norm2/3} in the proof of
Lemma~\ref{L:2/3-1ind}.

By construction, white clusters are separated from each other by the
distance no less than $k^{\gamma r_1/6}$. Grey and black clusters
are separated by the distance at least $k^{\gamma
r_1/2+2\delta_0r_1}$ and $k^{\gamma r_1+\delta_0r_1}$, respectively.
Consider first a white cluster. Let $\Pi_w$ be a singular white
cluster, namely,
\begin{equation}\label{ap9-1}
\|(P_w(H-k^{2l})P_w)^{-1}\|>k^\xi,\ \ \ \xi>k^{\gamma
r_1/6-2\delta},
\end{equation}
here and below $H=H(\k^{(2)}(\varphi_0))$, $P_w$ is the projector corresponding to $\Pi_w$. By construction,
$\Pi_w$ belongs to a small white box and its neighbors. Let us refer to
it as expanded small white box. Its size is $3k^{\gamma
r_1/2+2\delta_0r_1}$ and it contains less than $k^{\gamma
r_1/6-\delta_0r_1}$ elements of $\MM^{(2)}$.
\begin{lemma}\label{propw}
If \eqref{ap9-1} holds for a white cluster $\Pi_w$ then
\begin{equation}\label{ap9-2}
\|(P(H-k^{2l})P)^{-1}\|>ck^{k^{\gamma r_1/6-2\delta}},
\end{equation}
$P$ being the projector corresponding to the expanded small white
box $\Pi $ containing $\Pi_w$. The box  $\Pi $ has the size $3k^{\gamma
r_1/2+2\delta_0r_1}$ and contains less than $k^{\gamma
r_1/6-\delta_0r_1}$ elements of $\MM^{(2)}$.
\end{lemma}
\begin{proof}
Assume \eqref{ap9-1} holds, but \eqref{ap9-2} does not. Let $f\in
P_w\ell^2$ be such that $\|f\|=1$, $P_w(H-k^{2l})f=o(k^{-\xi})$, $\xi =k^{\gamma
r_1/6-2\delta}$.
Let us define
$$
g:=f-(P(H-k^{2l})P)^{-1}(P-P_w)Vf.
$$
Now we have:
\begin{equation} \label{392}
\begin{split}&
P(H-k^{2l})Pg=P(H-k^{2l})f-(P-P_w)Vf=\cr &
 P_w(H-k^{2l})f+(P-P_w)(H-k^{2l})f -(P-P_w)Vf=\cr &
P_w(H-k^{2l})P_w f+(P-P_w)(H_0-k^{2l})P_wf=P_w(H-k^{2l})P_w f=o(k^{-\xi}).
\end{split}
\end{equation}
If we show that
\begin{equation}\label{ap9-3}
\|(P(H-k^{2l})P)^{-1}(P-P_w)Vf\|=o(1),
\end{equation}
which means $\|g\|=1+o(1)$, then the lemma easily follows by the way
of contradiction. Thus, it remains to prove \eqref{ap9-3}. Denote
$\tilde f :=(P-P_w)Vf$. Let $\tilde{H}^{(2)}_w$ be the operator
consisting of $k^{\delta }$-clusters in $\Pi $. Namely,
$\tilde{H}^{(2)}_w=\sum _iP_iHP_i$,  $P_i$ being projectors onto
$k^{\delta }$-clusters.
 Formally,
\begin{equation}\label{393}
\begin{split}&
(P(H-k^{2l})P)^{-1}\tilde f =
\sum\limits_{r=0}^{r_0}(\tilde{H}^{(2)}_w-k^{2l})^{-1}\left(-(H-\tilde{H}^{(2)}_w)(\tilde{H}^{(2)}_w-k^{2l})^{-1}\right)^r\tilde
f +\cr &
\left(P(H-k^{2l})P\right)^{-1}\left(-(H-\tilde{H}^{(2)}_w)(\tilde{H}^{(2)}_w-k^{2l})^{-1}\right)^{r_0+1}\tilde
f,\ \ \ r_0=[k^{\gamma r_1/6-\delta}]-1.
\end{split}
\end{equation}
Some of $k^{\delta }$-clusters $P_iHP_i$ are resonant. However, their distance to the boundary of any white cluster is greater than $k^{\gamma r_1/6}$. Using this fact and considering
as in the proof of
\eqref{||A||}, we obtain
$$
\left\|\left(-(H-\tilde{H}^{(2)}_w)(\tilde{H}^{(2)}_w-k^{2l})^{-1}\right)^r\tilde
f \right\|<ck^{-\beta r}\ \ \  \mbox{when  }\ \ r\leq r_0+1, $$
since $(\tilde f)_{\m}\neq 0$ only near the boundary of a white cluster.
Hence, the right hand part of \eqref{393} is well
defined.
%%%%%on $\tilde f$ even though some blocks (corresponding to
%%%%%other white clusters in the expanded small white box under
%%%%%consideration) of the operator $\tilde{H}^{(2)}_w$ may be not
%%%%%defined (here we use the fact that
%%%%%since the distance between white
%%%%%clusters is at least $k^{\gamma r_1/6}$).
Now, substituting the last estimate into \eqref{393}, applying the same
arguments as in the proof of Theorem~\ref{Thm2} and using the
estimate opposite to \eqref{ap9-2} we obtain \eqref{ap9-3}.
\end{proof}

For singular grey and black clusters the proof is very similar. So,
we just introduce corresponding objects and formulate the results.

Let $\Pi_g$ be a singular grey cluster, i.e.
\begin{equation}\label{ap9-4}
\|(P_g(H-k^{2l})P_g)^{-1}\|>k^\xi,\ \ \ \xi>k^{\gamma
r_1/2+2\delta_0r_1-2\delta},
\end{equation}
$P_g$ being the projector corresponding to $\Pi_g$. By construction,
$\Pi_g$ belongs to a big white box and its neighbors. We refer to it as
expanded big white box. Its size is $3k^{\gamma r_1}$ and it
contains less than $k^{\gamma r_1/2+\delta_0r_1}$ elements of
$\MM^{(2)}$.
\begin{lemma}\label{propg}
If \eqref{ap9-4} holds for a grey cluster $\Pi_g$, such that all the white clusters imbedded into it do not satisfy \eqref{ap9-2}, then
\begin{equation}\label{ap9-5}
\|(P(H-k^{2l})P)^{-1}\|>ck^{k^{\gamma r_1/2+2\delta_0r_1-2\delta}},
\end{equation}
$P$ being the projector corresponding to the expanded big white box $\Pi $
containing $\Pi_g$. The box $\Pi $ has the size  $3k^{\gamma r_1}$ and it
contains less than $k^{\gamma r_1/2+\delta_0r_1}$ elements of
$\MM^{(2)}$.

\end{lemma}

The proof is analogous to that of Lemma \ref{propw} up to the
obvious changes: instead of $P_w$ we take $P_g$, $r_0=k^{\frac 12
\gamma r_1+2\delta _0 r_1-\delta }$ and $\tilde H^{(2)}_w$ is
replaced by $\tilde H^{(2)}_g$ which consists of $k^{\delta }$
non-resonance clusters and white clusters, which do not satisfy
\eqref{ap9-2}.

Let $\Pi_b$ be a singular black cluster, i.e.
\begin{equation}\label{ap9-6}
\|(P_b(H-k^{2l})P_b)^{-1}\|>k^\xi,\ \ \ \xi>k^{\gamma
r_1+\delta_0r_1-2\delta},
\end{equation}
$P_b$ being the projector corresponding to $\Pi_b$. By
Lemma~\ref{L:black} any black cluster can be covered by a box of the
size $ck^{3\gamma r_1/2+3}$ containing less than $ck^{\gamma r_1
+3}$ elements of $\MM^{(2)}$. We refer to it as expanded black box.
\begin{lemma}\label{propb}
If \eqref{ap9-6} holds for a black cluster $\Pi_b$, such that all the white and grey clusters imbedded into it do not satisfy \eqref{ap9-2},  \eqref{ap9-5}, then
\begin{equation}\label{ap9-7}
\|(P(H-k^{2l})P)^{-1}\|>ck^{k^{\gamma r_1 + \delta_0r_1-2\delta}},
\end{equation}
$P$ being the projector corresponding to the expanded black box
containing $\Pi_b$. The box  $\Pi $ has the size $ck^{3\gamma
r_1/2+3}$ and it contains less than $ck^{\gamma
r_1+3}$ elements of $\MM^{(2)}$.

\end{lemma}

\end{document}